\documentclass[10pt]{article}%
\usepackage{amscd}
\usepackage{amsmath}
\usepackage{amsfonts}
\usepackage{amssymb}
\usepackage[margin=0.5in]{geometry}%
\providecommand{\U}[1]{\protect\rule{.1in}{.1in}}

\numberwithin{equation}{section}
\newtheorem{theorem}{Theorem}[section]
\newtheorem{corollary}[theorem]{Corollary}

\newtheorem{lemma}[theorem]{Lemma}
\newtheorem{proposition}[theorem]{Proposition}
\newtheorem{condition}[theorem]{Condition}
\newtheorem{rem}[theorem]{Remark}
\newtheorem{defin}[theorem]{Definition}
\newenvironment{proof}[1][Proof]{\noindent\textbf{#1.} }{\ \rule{0.5em}{0.5em}}

\def\p2{\mathcal A_{\Phi,2\pi}(B)}
\def\0p2{\mathcal A_{\Phi,2\pi}(0)}
\def\sp2{\mathcal A_{\Phi,2\pi,\hbox{\rm SR}}(B)}
\def\beq{\begin{equation}}
\def\ene{\end{equation}}

\def\qed{\ifhmode\unskip\nobreak\fi\ifmmode\ifinner
\else\hskip5pt\fi\fi\hbox{\hskip5pt\vrule width4pt height6pt
depth1.5pt\hskip1pt}}

\def\+out{x^{\rm out}}
\begin{document}

\title{High-energy and smoothness asymptotic expansion of the scattering amplitude
for the Dirac equation and applications\thanks{ Research partially supported
by CONACYT under Project CB-2008-01-99100.}}
\author{Ivan Naumkin\thanks{ Electronic Mail: ivannaumkinkaikin@gmail.com} and Ricardo
Weder\thanks{Fellow, Sistema Nacional de Investigadores. Electronic mail:
weder@unam.mx }\\Departamento de F\'{\i}sica Matem\'{a}tica,\\Instituto de Investigaciones en Matem\'aticas Aplicadas y en Sistemas. \\Universidad Nacional Aut\'onoma de M\'exico.\\Apartado Postal 20-726, M\'exico DF 01000, M\'exico.}
\date{}
\maketitle

\qquad


\vspace{.5cm} \centerline{{\bf Abstract}} \bigskip

\noindent We obtain an explicit formula for the diagonal singularities of the
scattering amplitude for the Dirac equation with short-range electromagnetic
potentials. Using this expansion we uniquely reconstruct an electric potential
and magnetic field from the high-energy limit of the scattering amplitude.
Moreover, supposing that the electric potential and magnetic field are
asymptotic sums of homogeneous terms we give the unique reconstruction
procedure for these asymptotics from the scattering amplitude, known for some
energy $E.$ Furthermore, we prove that the set of the averaged scattering
solutions to the Dirac equation is dense in the set of all solutions to the
Dirac equation that are in $L^{2}\left(  \Omega\right)  ,$ where $\Omega$ is
any connected bounded open set in $\mathbb{R}^{3}$ with smooth boundary, and
we show that if we know an electric potential and a magnetic field for
$\mathbb{R}^{3}\diagdown\Omega,$ then the scattering amplitude, given for some
energy $E$, uniquely determines these electric potential and magnetic field
everywhere in $\mathbb{R}^{3}$. Combining this uniqueness result with the
reconstruction procedure for the asymptotics of the electric potential and the
magnetic field we show that the scattering amplitude, known for some $E$,
uniquely determines a electric potential and a magnetic field, that are
asymptotic sums of homogeneous terms, which converges to the electric
potential and the magnetic field respectively. Moreover, we discuss the
symmetries of the kernel of the scattering matrix, which follow from the
parity, charge-conjugation and time-reversal transformations for the Dirac operator.

\section{Introduction.}

The free Dirac operator $H_{0}$ is given by
\begin{equation}
H_{0}=-i\alpha\cdot\nabla+m\alpha_{4}, \label{basicnotions1}%
\end{equation}
where $m$ is the mass of the particle, $\alpha=(\alpha_{1},\alpha_{2}%
,\alpha_{3})$ and $\alpha_{j},$ $j=1,2,3,4,$ are $4\times4$ Hermitian matrices
that satisfy the relation:%
\begin{equation}
\alpha_{j}\alpha_{k}+\alpha_{k}\alpha_{j}=2\delta_{jk},\text{ }j,k=1,2,3,4,
\label{basicnotions2}%
\end{equation}
where $\delta_{jk}$ denotes the Kronecker symbol. The standard choice of
$\alpha_{j}$ is (\cite{16}):
\[
\alpha_{j}=%
\begin{pmatrix}
0 & \sigma_{j}\\
\sigma_{j} & 0
\end{pmatrix}
,\text{ \ }j\leq3,\text{ \ \ \ \ }\alpha_{4}=%
\begin{pmatrix}
I_{2} & 0\\
0 & -I_{2}%
\end{pmatrix}
=\beta,
\]
where $I_{n}$ is $n\times n$ unit matrix and
\begin{equation}
\sigma_{1}=%
\begin{pmatrix}
0 & 1\\
1 & 0
\end{pmatrix}
,\sigma_{2}=%
\begin{pmatrix}
0 & -i\\
i & 0
\end{pmatrix}
,\sigma_{3}=%
\begin{pmatrix}
1 & 0\\
0 & -1
\end{pmatrix}
\label{basicnotions22}%
\end{equation}
are the Pauli matrices.

The perturbed Dirac operator is defined by
\begin{equation}
H=H_{0}+\mathbf{V}. \label{basicnotions6}%
\end{equation}
Here the potential $\mathbf{V}\left(  x\right)  $ is an Hermitian $4\times4$
matrix valued function defined for $x\in\mathbb{R}^{3}$.

In this paper we consider the scattering theory for the Dirac operator with
electromagnetic potential. This problem has been extensively studied. Results
on existence and completeness of the wave operators have been proved for a
wide class of perturbations, including the long-range case, by the stationary
method in (\cite{60},\cite{61},\cite{9},\cite{33}, and the references
therein), and by the time-dependent method in \cite{49} and \cite{16}. Also,
in \cite{33}, radiation estimates and asymptotics for large time of solutions
of the time-dependent Dirac equation are obtained. The study of the point
spectrum was made in (\cite{12},\cite{54},\cite{50},\cite{11}, and the
references quoted there). The limiting absorption principle was proved in
\cite{51},\cite{14} and \cite{13}. Meromorphic extensions and resonances of
the scattering matrix were treated in \cite{63} and \cite{14}. The high-energy
behavior of the resolvent and the scattering amplitude was studied by the
stationary method in \cite{15}, and by the time-dependent method of \cite{62},
in \cite{52} and \cite{57}. High-energy and low-energy behavior of the
solutions of the Dirac equation, as well as Levinson theorem, were obtained in
\cite{18} for spherically symmetric potentials. Holomorphy of the scattering
matrix at fixed energy with respect to $c^{-2}$ for abstract Dirac operators
was studied in \cite{59}. Finally, a detailed study of the Dirac equation is
made in the monographs of Thaller \cite{16} and, Balinsky and Evans \cite{46}.

The inverse scattering problem consists in establishing a relation between the
scattering data (for example the scattering amplitude) and the potential. The
most complete and difficult problem is to find a one-to-one correspondence
between the scattering amplitude and the potential; that is to say, to give
necessary and sufficient conditions on a scattering amplitude, such that it is
associated to the scattering matrix $S\left(  E\right)  $ of a unique
potential $V(x)$ in a given class. This is the characterization problem (see
\cite{53},\cite{56} and \cite{55} for a discussion of this problem for the
Schr\"{o}dinger equation).

Another problem is uniqueness and reconstruction of the potential from the
scattering data.\ We note that this problem is not well defined in the
electromagnetic case, because the scattering amplitude is invariant under
gauge transformation of the magnetic potential and thus, the problem cannot be
solved in a unique way. Nevertheless, we can ask about uniqueness and
reconstruction of the electric potential $V\left(  x\right)  $ and the
magnetic field $B\left(  x\right)  =\operatorname*{rot}A\left(  x\right)  $,
associated to a magnetic potential $A\left(  x\right)  $.

The uniqueness and reconstruction problem has different settings. One of them
is uniqueness and reconstruction from the high-energy limit of the scattering
amplitude. Using the stationary approach Ito \cite{15} solved this problem for
electromagnetic potentials of the form%
\begin{equation}
\mathbf{V}\left(  x\right)  =%
\begin{pmatrix}
V & \sigma\cdot A\\
\sigma\cdot A & V
\end{pmatrix}
, \label{intro4}%
\end{equation}
decaying faster then $\left\vert x\right\vert ^{-3}$ at infinity. Using an
adaptation to Dirac operators of the Green function of Faddeev (\cite{53}) and
the Green function introduced by Eskin and Ralston (\cite{64}) in the inverse
scattering theory for Schr\"{o}dinger operators, the same problem was solved
by Isozaki \cite{22}, for electric potentials of the form%
\begin{equation}
\mathbf{V}\left(  x\right)  =%
\begin{pmatrix}
V_{+} & 0\\
0 & V_{-}%
\end{pmatrix}
, \label{intro1}%
\end{equation}
decaying faster than $\left\vert x\right\vert ^{-2}$. The time-dependent
method (\cite{62}) was used by Jung \cite{52} to solve the uniqueness and
reconstruction problem for continuous, Hermitian matrix valued potentials
$\mathbf{V}\left(  x\right)  ,$ with integrable decay, which satisfy
\begin{equation}
\lbrack\left(  \mathbf{V}\left(  x_{1}\right)  +\left(  \alpha\cdot
\omega\right)  \mathbf{V}\left(  x_{1}\right)  \left(  \alpha\cdot
\omega\right)  \right)  ,\left(  \mathbf{V}\left(  x_{2}\right)  +\left(
\alpha\cdot\omega\right)  \mathbf{V}\left(  x_{2}\right)  \left(  \alpha
\cdot\omega\right)  \right)  ]=0, \label{intro7}%
\end{equation}
for all $x_{1},x_{2}\in\mathbb{R}^{n}.$ Also using the time-dependent method
Ito \cite{57} gave conditions on a time-dependent potential of the form
(\ref{intro4}), so that it can be reconstructed from the scattering operator.
In particular, if a short-range potential is time-independent, then the
electric potential $V\left(  x\right)  $ and the magnetic field $B\left(
x\right)  $ can be completely reconstructed from the scattering operator.

A similar problem is uniqueness and reconstruction of the electric potential
and the magnetic field from the scattering amplitude at a fixed energy
(ISPFE). Isozaki \cite{22} solved the ISPFE problem for potentials of the form
(\ref{intro1}) and decaying exponentially. A problem related to the ISPFE is
the inverse boundary value problem (IBVP), where the Dirac operator at a fixed
energy is considered in a bounded domain $\Omega$ in $\mathbb{R}^{3},$ and the
uniqueness and reconstruction of the electric potential and the magnetic field
from the Dirichlet-to-Dirichlet (D-D) map (see (\ref{completeness9}), Section
8) is studied. Tsuchida \cite{20} solved the IBVP problem for potentials
$\mathbf{V}$ of the form%
\begin{equation}
\mathbf{V}\left(  x\right)  =%
\begin{pmatrix}
V_{+} & \sigma\cdot A\\
\sigma\cdot A & V_{-}%
\end{pmatrix}
, \label{intro3}%
\end{equation}
and showed that the D-D map determines uniquely small potentials $V_{+},V_{-}$
and $B=\operatorname*{rot}A.$ Nakamura and Tsuchida \cite{19} solved the IBVP
for potentials $\mathbf{V}\in C^{\infty}\left(  \Omega\right)  ,$ of the form
(\ref{intro3}). Moreover, establishing a relation between the D-D map and the
scattering amplitude, they showed the uniqueness of the ISPFE problem for
potentials $\mathbf{V}\in C^{\infty}\left(  \mathbb{R}^{3}\right)  $ of the
form (\ref{intro3}) with compact support. Goto \cite{23} solved the
ISPFE\ problem for exponentially-decaying potentials of the form
(\ref{intro4}), and Li \cite{28} considered both\ the IBVP and the ISPFE
problems for potentials that are general Hermitian matrices with compact support.

We note that in case of the Schr\"{o}dinger equation with general short-range
potentials the scattering matrix at a fixed positive energy does not uniquely
determine the potential. Indeed, Chadan and Sabatier (see pag 207 of
\cite{24}) give examples of non-trivial radial oscillating potentials decaying
as $\left\vert x\right\vert ^{-3/2}$ at infinity, such that the corresponding
scattering amplitude is identically zero at some positive energy. Thus, for
the Schr\"{o}dinger equation with general short-range potentials, it is
necessary to impose extra conditions for the uniqueness of the ISPFE\ problem.
Weder and Yafaev \cite{25} considered the ISPFE problem for the
Schr\"{o}dinger equation with general short-range potentials (see also
\cite{65} for the case of long-range potentials) which are asymptotic sums of
homogeneous terms (see (\ref{basicnotions16}) and (\ref{basicnotions17})).
They used an explicit formula for the singularities of the scattering
amplitude, obtained by Yafaev in \cite{27} and \cite{30}, to show that the
asymptotics of the electric potential and the magnetic field can be recovered
from the diagonal singularities of the scattering amplitude (see also
\cite{68}-\cite{67} and \cite{69}). Moreover, Weder in \cite{58} and \cite{26}
proved that if two short-range electromagnetic potentials $\left(  V_{1}%
,A_{1}\right)  $ and $\left(  V_{2},A_{2}\right)  $ in $\mathbb{R}^{n}$,
$n\geq3$ have the same scattering matrix at a fixed positive energy and if the
electric potentials $V_{j}$ and the magnetic fields $B_{j}=\operatorname*{rot}%
A_{j,}$ $j=1,2,$ coincide outside of some ball they necessarily coincide
everywhere. The combination of this uniqueness result and the result of Weder
and Yafaev \cite{25} implies that the scattering matrix at a fixed positive
energy uniquely determines electric potentials and magnetic fields that are a
finite sum of homogeneous terms at infinity, or more generally, that are
asymptotic sums of homogeneous terms that actually converge, respectively, to
the electric potential and to the magnetic field (\cite{26}). We proceed
similarly in the Dirac case.

We follow the methods of the works of Weder \cite{58}, \cite{26}, Yafaev
\cite{27},\cite{30}, and Weder and Yafaev \cite{25}, \cite{65} for the
Schr\"{o}dinger equation. For the Schr\"{o}dinger equation with short-range
potentials, the approximate solutions are given by $u_{N}\left(  x,\xi\right)
=e^{i\left\langle x,\xi\right\rangle }+e^{i\left\langle x,\xi\right\rangle
}a_{N}\left(  x,\xi\right)  ,$ where $a$ solves the \textquotedblleft
transport\textquotedblright\ equation (see \cite{27}). These solutions are
constructed in such a way that outside some conical neighborhood of the
direction $x=\pm\xi,$ the \textquotedblleft remainder\textquotedblright%
\ satisfies the estimate
\begin{equation}
\left\vert \partial_{x}^{\alpha}\partial_{\xi}^{\beta}r_{N}\left(
x,\xi\right)  \right\vert \leq C_{\alpha,\beta}\left\langle x\right\rangle
^{-p\left(  N\right)  }\left\langle \xi\right\rangle ^{-q\left(  N\right)
},\text{ where }p\left(  N\right)  ,q\left(  N\right)  \rightarrow
+\infty,\text{ as }N\rightarrow+\infty. \label{intro5}%
\end{equation}
In the case of the Schr\"{o}dinger equation with long-range potentials, the
approximate solutions are of the form $u_{N}\left(  x,\xi\right)
=e^{i\left\langle x,\xi\right\rangle +i\phi\left(  x,\xi\right)  }\left(
1+a_{N}\left(  x,\xi\right)  \right)  ,$ where $\phi\left(  x,\xi\right)  $
solves the \textquotedblleft eikonal\textquotedblright\ equation and $a\left(
x,\xi\right)  $ is the solution of the \textquotedblleft
transport\textquotedblright\ equation (\cite{30}). Again, $u_{N}\left(
x,\xi\right)  $ is constructed so that the \textquotedblleft
remainder\textquotedblright\ satisfies the estimate (\ref{intro5}).

We observe that G\^{a}tel and Yafaev \cite{33} constructed approximate
solutions for the Dirac equation with long-range potentials of the form
(\ref{intro4}) satisfying for all $\alpha$ the estimate
\begin{equation}
\left\vert \partial_{x}^{\alpha}V\left(  x\right)  \right\vert +\left\vert
\partial_{x}^{\alpha}A\left(  x\right)  \right\vert \leq C_{\alpha
}\left\langle x\right\rangle ^{-\rho-\left\vert \alpha\right\vert },\text{ for
some }\rho\in(0,1). \label{intro6}%
\end{equation}
These solutions are given by $e^{i\left\langle x,\xi\right\rangle
+i\Phi\left(  x,\xi;E\right)  }p\left(  x,\xi;E\right)  ,$ where the phase
$\Phi\left(  x,\xi;E\right)  $ solves the \textquotedblleft
eikonal\textquotedblright\ equation and satisfies, for all $\alpha$ and
$\beta,$ the estimate $\left\vert \partial_{x}^{\alpha}\partial_{\xi}^{\beta
}\Phi\left(  x,\xi;E\right)  \right\vert \leq C_{\alpha,\beta}\left\langle
x\right\rangle ^{1-\rho-\left\vert \alpha\right\vert }$. The function
$p\left(  x,\xi;E\right)  $ satisfies a transport equation and it results to
be explicit and exact. It satisfies $\left\vert \partial_{x}^{\alpha}%
\partial_{\xi}^{\beta}\left(  p\left(  x,\xi;E\right)  -P_{\omega}\left(
E\right)  \right)  \right\vert \leq C_{\alpha,\beta}\left\langle
x\right\rangle ^{-\rho-\left\vert \alpha\right\vert },$ for all $\alpha$ and
$\beta,$ where $P_{\omega}\left(  E\right)  $ is the amplitude in the plane
wave solution $P_{\omega}\left(  E\right)  e^{i\sqrt{E^{2}-m^{2}}\left\langle
\omega,x\right\rangle }$ of energy $E$ and momentum in direction $\omega$\ and
it is defined by (\ref{basicnotions43}) below. Then the \textquotedblleft
remainder\textquotedblright\ satisfies, outside a conical neighborhood of the
direction $x=\pm\xi,$ the estimate $\left\vert \partial_{x}^{\alpha}%
\partial_{\xi}^{\beta}r_{N}\left(  x,\xi;E\right)  \right\vert \leq
C_{\alpha,\beta}\left\langle x\right\rangle ^{-1-\varepsilon},$ for some
$\varepsilon>0,$ for all $\alpha$ and $\beta.$ Using these solutions, they
construct special identities $J_{\pm},$ in order to prove the existence and
completeness of the wave operators and to obtain the asymptotics for large
time of the solutions of the time-dependent Dirac equation.

Even in the case of the Dirac equation with short-range electric potentials,
it is not enough to consider only the \textquotedblleft
transport\textquotedblright\ equation, in order to obtain the estimate
(\ref{intro5}). Thus, we also need to consider the \textquotedblleft
eikonal\textquotedblright\ equation. It also turns out that we need to
decompose the \textquotedblleft transport\textquotedblright\ equation in two
parts, one for the positive energies and another for the negative energies.

We consider potentials of the form (\ref{intro4}) satisfying the estimate
(\ref{intro6}), with $\rho>1$. We take the approximate solutions of the form
$e^{i\left\langle x,\xi\right\rangle +i\Phi\left(  x,\xi;E\right)  }%
w_{N}\left(  x,\xi;E\right)  ,$ where $\Phi\left(  x,\xi;E\right)  $ solves
the \textquotedblleft eikonal\textquotedblright\ equation and satisfies the
estimate $\left\vert \partial_{x}^{\alpha}\partial_{\xi}^{\beta}\Phi\left(
x,\xi;E\right)  \right\vert \leq C_{\alpha,\beta}\left(  1+\left\vert
x\right\vert \right)  ^{-\left(  \rho-1\right)  -\left\vert \alpha\right\vert
}\left\vert \xi\right\vert ^{-\left\vert \beta\right\vert }.$ The function
$w_{N}$ decomposes in the sum $\left(  w_{1}\right)  _{N}+\left(
w_{2}\right)  _{N},$ where the functions $\left(  w_{1}\right)  _{N}$ $\ $and
$\left(  w_{2}\right)  _{N}$ satisfy two different \textquotedblleft
transport\textquotedblright\ equations and the estimate $\left\vert
\partial_{x}^{\alpha}\partial_{\xi}^{\beta}\left(  w_{j}\right)  _{N}\left(
x,\xi;E\right)  \right\vert \leq C_{\alpha,\beta}\left(  1+\left\vert
x\right\vert \right)  ^{-\left\vert \alpha\right\vert }\left\vert
\xi\right\vert ^{-\left\vert \beta\right\vert },$ for $j=1,2.$ This
construction of the approximate solutions to the Dirac equation assures that
the \textquotedblleft remainder\textquotedblright\ $r_{N}\left(
x,\xi;E\right)  $ satisfies the estimate (\ref{intro5}). We note that we do
not reduce the Dirac equation to the Schr\"{o}dinger type equation. We deal
directly with the Dirac equation to obtain the approximate solutions.

Following the idea of \cite{27} and \cite{30}, and using our approximate
solutions to the Dirac equation, we construct special identities $J_{\pm}%
\ $and use the stationary equation for the scattering matrix to find an
explicit formula for the singularities of the scattering amplitude in a
neighborhood of the diagonal for potentials satisfying the estimate
(\ref{intro6}) with $\rho>1$. Using this formula we express the leading
singularity of the scattering amplitude in terms of the Fourier transform of
the electric and the magnetic potentials. Furthermore, we obtain an error
bound for the difference between the scattering amplitude and the leading
singularity. We also show that this error bound is optimal in the case of
electric and magnetic potentials that are homogeneous outside a sphere. We use
the formula for the singularities of the scattering amplitude also to study
the high-energy limit for the scattering matrix. Moreover, we uniquely
reconstruct the electric potential and the magnetic field from the high-energy
limit of the scattering amplitude. We recall that this result was proved in
\cite{15}, by studying the high-energy behavior of the resolvent, for smooth
potentials which satisfy (\ref{intro6}), for $\left\vert \alpha\right\vert
\leq d,$ $d\geq2,$ with $\rho>3.$ Also we show that for potentials,
homogeneous and non-trivial outside of a sphere, satisfying the estimate
(\ref{intro6}) with $1<\rho\leq2,$ the total scattering cross-section is infinite.

If the scattering amplitude is given for some energy $E$, then, in particular
we know all its singularities for this energy. As in \cite{25}, assuming that
the electric potential $V$ and the magnetic field $B=\operatorname*{rot}A$
satisfy estimates (\ref{eig32}) and (\ref{basicnotions18}), for all $d,$
respectively, and they are asymptotic to a sum of homogeneous terms we
uniquely recover these asymptotics from the scattering amplitude singularities.

On the other hand, inspired by \cite{58},\cite{26}, we consider a special set
of solutions to the Dirac equation for fixed energy, called \textquotedblleft
averaged scattering solutions\textquotedblright\ and show that for potentials
$\mathbf{V,}$ satisfying Condition \ref{basicnotions26} and estimate
(\ref{completeness2}), for some connected bounded open set $\Omega,$ with
smooth boundary, this set of solutions is strongly dense in the set of all
solutions to the Dirac equation that are in $L^{2}\left(  \Omega\right)  $.
This fact allows us to prove that if $\mathbf{V}_{j},$ $j=1,2,$ are of the
form (\ref{intro3}), where $V_{\pm}^{\left(  j\right)  }\in C^{\infty}\left(
\mathbb{R}^{3}\right)  ,$ $j=1,2,$ satisfying (\ref{completeness2}) and
$B_{j}\in C^{\infty}\left(  \mathbb{R}^{3}\right)  ,$ $j=1,2,$ satisfying
(\ref{basicnotions18}), for $d=1,$ are such that $V_{\pm}^{\left(  1\right)
}=V_{\pm}^{\left(  2\right)  },$ and $B_{1}=B_{2},$ for $x$ outside some
connected bounded open set $\Omega^{\prime},$ with smooth boundary, and the
scattering amplitudes for $\mathbf{V}_{1}$ and $\mathbf{V}_{2}$ coincide for
some energy, then $V_{\pm}^{\left(  1\right)  }=V_{\pm}^{\left(  2\right)  }$
and $B_{1}=B_{2}$ for all $x$. 

Finally, if the asymptotic decomposition of the
electric potential $V\in C^{\infty}\left(  \mathbb{R}^{3}\right)  $ and the
magnetic field $B\in C^{\infty}\left(  \mathbb{R}^{3}\right)  $ actually
converge, respectively, to $V$ and $B,$ outside some bounded set, then
combining both results, we show that the scattering matrix given for some
fixed energy uniquely determines the electric potential $V$ and the magnetic
field $B$. \

The paper is organized as follows. In Section 2 we give some known results
about scattering theory for the Dirac operators. In Section 3, we define the
scattering solutions, we calculate their asymptotics for large $x,$ and we
give a relation between the coefficient of the term, decaying as $\frac
{1}{\left\vert x\right\vert },$ for $x\rightarrow\infty,$ in this asymptotic
expansion and the kernel of the scattering matrix. In Section 4, we present
symmetries of the kernel of the scattering matrix, that follow from the
time-reversal, parity and charge conjugation transformations of the Dirac
equation. These symmetries, interesting on their own, can be useful in a study
of the characterization problem. In Section 5, we construct approximate
generalized eigenfunctions for the Dirac equation that we use in Section 6 to
obtain an explicit formula for the singularities of the scattering amplitude.
Applications of the formula for the singularities of the kernel of the
scattering matrix are presented in Section 7. The result on the completeness
of the averaged scattering solutions is given in Section 8. In Section 9 we
present the results on the uniqueness of ISPFE problem.

\textbf{Acknowledgement} \textit{We thank Osanobu Yamada for informing us of
reference \cite{45}.}

\section{Basic notions.}

The free Dirac operator $H_{0}$ (\ref{basicnotions1}) is a self-adjoint
operator on $L^{2}\left(  \mathbb{R}^{3};\mathbb{C}^{4}\right)  $ with domain
$D\left(  H_{0}\right)  =\mathcal{H}^{1}\left(  \mathbb{R}^{3};\mathbb{C}%
^{4}\right)  ,$ the Sobolev space of order $1$ (\cite{16}). When there is no
place of confusion we will write $L^{2}$ and$\ \mathcal{H}^{1}$ to simplify
the notation. We can diagonalize $H_{0}$ by the Fourier transform
$\mathcal{F}$ given by $\left(  \mathcal{F}f\right)  \left(  x\right)
=\left(  2\pi\right)  ^{-3/2}\int_{\mathbb{R}^{3}}e^{-i\left\langle
x,\xi\right\rangle }f\left(  \xi\right)  d\xi.$ Actually, $\mathcal{F}%
H_{0}\mathcal{F}^{\ast}$ acts as multiplication by the matrix $h_{0}\left(
\xi\right)  =\alpha\cdot\xi+m\beta.$ This matrix has two eigenvalues
$E=\pm\sqrt{\xi^{2}+m^{2}}$ and each eigenspace $X^{\pm}\left(  \xi\right)  $
is a two-dimensional subspace of $\mathbb{C}^{4}.$ The orthogonal projections
onto these eigenspaces are given by (see \cite{16}, page 9) $P^{\pm}\left(
\xi\right)  :=\frac{1}{2}\left(  I_{4}\pm\left(  \xi^{2}+m^{2}\right)
^{-1/2}\left(  \alpha\cdot\xi+m\beta\right)  \right)  .$ The spectrum of
$H_{0}$ is purely absolutely continuous and it is given by $\sigma\left(
H_{0}\right)  =\sigma_{ac}\left(  H_{0}\right)  =(-\infty,-m]\cup\lbrack
m,\infty).$

Let us introduce the weighted $L^{2}$ spaces for $s\in\mathbb{R},$ $L_{s}%
^{2}:=\{f:\left\langle x\right\rangle ^{s}f\left(  x\right)  \in L^{2}\},$
$\left\Vert f\right\Vert _{L_{s}^{2}}:=\left\Vert \left\langle x\right\rangle
^{s}f\left(  x\right)  \right\Vert _{L^{2}},$ where $\left\langle
x\right\rangle =\left(  1+\left\vert x\right\vert ^{2}\right)  ^{1/2}.$
Moreover, for any $\alpha,s\in\mathbb{R}$ we define $\mathcal{H}^{\alpha
,s}:=\{f:\left\langle x\right\rangle ^{s}f\left(  x\right)  \in\mathcal{H}%
^{\alpha}\},$ $\left\Vert f\right\Vert _{\mathcal{H}^{\alpha,s}}:=\left\Vert
\left\langle x\right\rangle ^{s}f\left(  x\right)  \right\Vert _{\mathcal{H}%
^{\alpha}},$ where $\left\Vert f\left(  x\right)  \right\Vert _{\mathcal{H}%
^{\alpha}}=\left(  \int_{\mathbb{R}^{3}}\left\langle \xi\right\rangle
^{2\alpha}\left\vert \hat{f}\left(  \xi\right)  \right\vert ^{2}d\xi\right)
^{1/2}.$

Let us now consider the perturbed Dirac operator $H$, given by
(\ref{basicnotions6}). We make the following assumption on the Hermitian
$4\times4$ matrix valued potential $\mathbf{V,}$ defined for $x\in
\mathbb{R}^{3}$:

\begin{condition}
\label{basicnotions26}For some $s_{0}>1/2,$ $\left\langle x\right\rangle
^{2s_{0}}\mathbf{V}$ is a compact operator from $\mathcal{H}^{1}$ to $L^{2}.$
\end{condition}

The assumptions on a potential $\mathbf{V,}$ assuring Condition
\ref{basicnotions26} are well known (see, for example, \cite{81}). In
particular, Condition \ref{basicnotions26} for $\mathbf{V}$ holds, if for some
$\varepsilon>0,$ $\sup_{x\in\mathbb{R}^{3}}\int_{\left\vert x-y\right\vert
\leq1}\left\vert \left\langle y\right\rangle ^{2s_{0}}\mathbf{V}\left(
y\right)  \right\vert ^{3+\varepsilon}dy<\infty$ and $\int_{\left\vert
x-y\right\vert \leq1}\left\vert \left\langle y\right\rangle ^{2s_{0}%
}\mathbf{V}\left(  y\right)  \right\vert ^{3+\varepsilon}dy\rightarrow0,$ as
$\left\vert x\right\vert \rightarrow0$ (see Theorem 9.6, Chapter 6, of
\cite{81}). Of course, the last two relations are true if the following
estimate is valid
\begin{equation}
\left\vert \mathbf{V}\left(  x\right)  \right\vert \leq C\left\langle
x\right\rangle ^{-\rho},\text{ for some }\rho>1. \label{basicnotions46}%
\end{equation}

Since $\mathbf{V}$ is an Hermitian $4\times4$ matrix valued potential
$\mathbf{V,}$ Condition \ref{basicnotions26} implies assumptions (A$_{1}%
$)-(A$_{3}$) of \cite{14}. Thus, under Condition \ref{basicnotions26} $H$ is a
self-adjoint operator on $D\left(  H\right)  =\mathcal{H}^{1}$ and the
essential spectrum $\sigma_{ess}\left(  H\right)  =\sigma\left(  H_{0}\right)
$. The wave operators (WO), defined as the following strong limit%
\begin{equation}
W_{\pm}\left(  H,H_{0}\right)  :=s-\lim_{t\rightarrow\pm\infty}e^{iHt}%
e^{-iH_{0}t}, \label{basicnotions9}%
\end{equation}
exist and are complete, i.e., $\operatorname*{Range}W_{\pm}=\mathcal{H}_{ac}$
(the subspace of absolutely continuity of $H$) and the singular continuous
spectrum of $H$ is absent.

We recall that the study about the absence of eigenvalues embedded in the
absolutely continuous spectrum was made in \cite{12},\cite{54},\cite{50}%
,\cite{11}, and the references quoted there. For example, it follows from
Theorem 6 of \cite{12} that $\{(-\infty,-m)\cup(m,\infty)\}\cap\sigma
_{p}\left(  H\right)  =\varnothing$, if $\mathbf{V}\in L_{\operatorname*{loc}%
}^{5}\left(  \mathbb{R}^{3}\right)  \ $satisfies the following: for any
$\varepsilon>0,$ there exists $R\left(  \varepsilon\right)  >0,$ such that for
any $u\in\mathcal{H}_{\operatorname*{c}}^{1}\left(  \left\vert x\right\vert
>R\left(  \varepsilon\right)  \right)  $ (the set of functions from
$\mathcal{H}^{1}\left(  \left\vert x\right\vert >R\left(  \varepsilon\right)
\right)  $ with compact support in $\left\vert x\right\vert >R\left(
\varepsilon\right)  $)
\begin{equation}
\left\Vert \left\vert x\right\vert \mathbf{V}u\right\Vert _{L^{2}}%
\leq\varepsilon\left\Vert u\right\Vert _{\mathcal{H}^{1}},
\label{basicnotions47}%
\end{equation}
and, moreover,%
\begin{equation}
\left\Vert \left\vert x\right\vert ^{1/2}\mathbf{V}\right\Vert _{L^{\infty
}\left(  \left\vert x\right\vert \geq R\right)  }<\infty,\text{ for some }R>0.
\label{basicnotions48}%
\end{equation}
Note that $\mathbf{V}$ satisfies these relations, if the estimate
(\ref{basicnotions46}) holds.

From the existence of the WO it follows that $HW_{\pm}=W_{\pm}H_{0}$
(intertwining relations). The scattering operator, defined by%
\begin{equation}
\mathbf{S=S}\left(  H,H_{0}\right)  :=W_{+}^{\ast}W_{-},
\label{basicnotions10}%
\end{equation}
commutes with $H_{0}$ and it is unitary.

Let $H_{0S}:=-\triangle$ be the free Schr\"{o}dinger operator in $L^{2}\left(
\mathbb{R}^{3};\mathbb{C}^{4}\right)  .$ The limiting absorption principle
(LAP) is the following statement. For $z$ in the resolvent set of $H_{0S}$ let
$R_{0S}\left(  z\right)  :=\left(  H_{0S}-z\right)  ^{-1}$ be the resolvent.
The limits $R_{0S}\left(  \lambda\pm i0\right)  =\lim_{\varepsilon
\rightarrow+0}R_{0S}\left(  \lambda\pm i\varepsilon\right)  ,$ ($\varepsilon
\rightarrow+0$ means $\varepsilon\rightarrow0$ with $\varepsilon>0$) exist in
the uniform operator topology in $\mathcal{B}\left(  L_{s}^{2},\mathcal{H}%
^{\alpha,-s}\right)  ,$ $s>1/2,$ $\left\vert \alpha\right\vert \leq2$
(\cite{4},\cite{71},\cite{37},\cite{8}) and, moreover, $\left\Vert
R_{0S}\left(  \lambda\pm i0\right)  f\right\Vert _{\mathcal{H}^{\alpha,-s}%
}\leq C_{s,\delta}\lambda^{-\left(  1-\left\vert \alpha\right\vert \right)
/2}\left\Vert f\right\Vert _{L_{s}^{2}},$ for $\lambda\in\lbrack\delta
,\infty),$ $\delta>0$. Here for any pair of Banach spaces $X,Y,$
$\mathcal{B}\left(  X,Y\right)  $ denotes the Banach space of all bounded
operators from $X$ into $Y.$ The functions $R_{0S}^{\pm}\left(  \lambda
\right)  ,$ given by $R_{0S}\left(  \lambda\right)  $ if $\operatorname{Im}%
\lambda\neq0,$ and $R_{0S}\left(  \lambda\pm i0\right)  ,$ if $\lambda
\in(0,\infty),$ are defined for $\lambda\in\mathbb{C}^{\pm}\cup\left(
0,\infty\right)  $ ($\mathbb{C}^{\pm}$ denotes, respectively, the upper,
lower, open complex half-plane) with values in $\mathcal{B}\left(  L_{s}%
^{2},\mathcal{H}^{\alpha,-s}\right)  $ and they are analytic for
$\operatorname{Im}\lambda\neq0$ and locally H\"{o}lder continuous for
$\lambda\in(0,\infty)$ with exponent $\vartheta$ satisfying the estimates
$0<\vartheta\leq s-1/2$ and $\vartheta<1.$

For $z$ in the resolvent set of $H_{0}$ let $R_{0}\left(  z\right)  :=\left(
H_{0}-z\right)  ^{-1}$ be the resolvent. From the LAP for $H_{0S}$ it follows
that the limits (see Lemma 3.1 of \cite{14})%
\begin{equation}
\left.  R_{0}\left(  E\pm i0\right)  =\lim_{\varepsilon\rightarrow+0}%
R_{0}\left(  E\pm i\varepsilon\right)  =\left\{
\begin{array}
[c]{c}%
\left(  H_{0}+E\right)  R_{0S}\left(  \left(  E^{2}-m^{2}\right)  \pm
i0\right)  \text{ for }E>m\\
\left(  H_{0}+E\right)  R_{0S}\left(  \left(  E^{2}-m^{2}\right)  \mp
i0\right)  \text{ for }E<-m,
\end{array}
\right.  \right.  \label{basicnotions27}%
\end{equation}
exist for $E\in(-\infty,-m)\cup(m,\infty)$ in the uniform operator topology in
$\mathcal{B}\left(  L_{s}^{2},\mathcal{H}^{\alpha,-s}\right)  ,$ $s>1/2,$
$\alpha\leq1,$ and $\left\Vert R_{0}\left(  E\pm i0\right)  f\right\Vert
_{\mathcal{H}^{\alpha,-s}}$ $\leq C_{s,\delta}\left\vert E\right\vert
^{\left\vert \alpha\right\vert }\left\Vert f\right\Vert _{L_{s}^{2}},$ for
$\left\vert E\right\vert \in\lbrack m+\delta,\infty),$ $\delta>0$.
Furthermore, the functions, $R_{0}^{\pm}\left(  E\right)  ,$ given by
$R_{0}\left(  E\right)  ,$ if $\operatorname{Im}E\neq0,$ and by $R_{0}\left(
E\pm i0\right)  ,$ if $E\in(-\infty,-m)\cup(m,\infty),$ are defined for
$E\in\mathbb{C}^{\pm}\cup\left(  -\infty,-m\right)  \cup\left(  m,+\infty
\right)  $ with values in $\mathcal{B}\left(  L_{s}^{2},\mathcal{H}%
^{\alpha,-s}\right)  ,$ and moreover, they are analytic for $\operatorname{Im}%
E\neq0$ and locally H\"{o}lder continuous for $E\in(-\infty,-m)\cup(m,\infty)$
with exponent $\vartheta$ such that $0<\vartheta\leq s-1/2$ and $\vartheta<1.$

Next we consider the resolvent $R\left(  z\right)  :=\left(  H-z\right)
^{-1}$ for $z$ in the resolvent set of $H.$ The following limits exist for
$E\in\{(-\infty,-m)\cup(m,\infty)\}\backslash\sigma_{p}\left(  H\right)  $ in
the uniform operator topology in $\mathcal{B}\left(  L_{s}^{2},\mathcal{H}%
^{\alpha,-s}\right)  ,$ $s\in\left(  1/2,s_{0}\right]  ,$ $\left\vert
\alpha\right\vert \leq1,$ where $s_{0}$ is defined by Condition
\ref{basicnotions26} (see Theorem 3.9 of \cite{14})%
\begin{equation}
\left.  R\left(  E\pm i0\right)  =\lim_{\varepsilon\rightarrow+0}R\left(  E\pm
i\varepsilon\right)  =R_{0}\left(  E\pm i0\right)  \left(  1+\mathbf{V}%
R_{0}\left(  E\pm i0\right)  \right)  ^{-1}.\right.  \label{basicnotions13}%
\end{equation}
From this relation and the properties of $R_{0}^{\pm}\left(  E\right)  $ it
follows that the functions, $R^{\pm}\left(  E\right)  :=\{R\left(  E\right)  $
if $\operatorname{Im}E\neq0,$ and $R\left(  E\pm i0\right)  ,$ $E\in
\{(-\infty,-m)\cup(m,\infty)\}\backslash\sigma_{p}\left(  H\right)  \},$
defined for $E\in\mathbb{C}^{\pm}\cup\{(-\infty,-m)\cup(m,\infty
)\}\backslash\sigma_{p}\left(  H\right)  ,$ with values in $\mathcal{B}\left(
L_{s}^{2},\mathcal{H}^{\alpha,-s}\right)  $ are analytic for
$\operatorname{Im}E\neq0$ and locally H\"{o}lder continuous for $E\in
\{(-\infty,-m)\cup(m,\infty)\}\backslash\sigma_{p}\left(  H\right)  $ with
exponent $\vartheta$ such that $0<\vartheta\leq s-1/2,$ $s<\min\{s_{0},3/2\}.$

The Foldy-Wouthuysen (F-W) transform \cite{48}, that diagonalizes the free
Dirac operator, is defined as follows: Let $\hat{G}\left(  \xi\right)  $ be
the unitary $4\times4$ matrix defined by $\hat{G}\left(  \xi\right)
=\exp\{\beta\left(  \alpha\cdot\xi\right)  \theta\left(  \left\vert
\xi\right\vert \right)  \},$ where $\theta\left(  t\right)  =\frac{1}%
{2t}\arctan\frac{t}{m}$ for $t>0.$ Note that
\begin{equation}
\tilde{h}_{0}\left(  \xi\right)  :=\hat{G}\left(  \xi\right)  h_{0}\left(
\xi\right)  \hat{G}\left(  \xi\right)  ^{-1}=\left(  \xi^{2}+m^{2}\right)
^{1/2}\beta. \label{basicnotions30}%
\end{equation}

The F-W transform is the unitary operator $G$ on $L^{2}$ given by
$G=\mathcal{F}^{-1}\hat{G}\left(  \xi\right)  \mathcal{F}$. $G$ transforms
$H_{0}$ into $\tilde{H}_{0}:=GH_{0}G^{-1}=\left(  H_{0S}+m^{2}\right)
^{1/2}\beta.$ We use the F-W transform to define the trace operator
$T_{0}\left(  E\right)  $ for the free Dirac operator (\cite{14}). Let
$P_{+}=\left(
\begin{array}
[c]{cc}%
I_{2} & 0\\
0 & 0
\end{array}
\right)  $ and $P_{-}=\left(
\begin{array}
[c]{cc}%
0 & 0\\
0 & I_{2}%
\end{array}
\right)  $. We define $T_{0}^{\pm}\left(  E\right)  \in\mathcal{B}\left(
L_{s}^{2};L^{2}\left(  \mathbb{S}^{2};\mathbb{C}^{4}\right)  \right)  $ by
$\left(  T_{0}^{\pm}\left(  E\right)  f\right)  \left(  \omega\right)
=\left(  2\pi\right)  ^{-\frac{3}{2}}\upsilon\left(  E\right)
{\displaystyle\int_{\mathbb{R}^{3}}}
e^{-i\nu\left(  E\right)  \left\langle \omega,x\right\rangle }$ $\ \times
P_{\pm}Gf\left(  x\right)  dx,$ where $\upsilon\left(  E\right)  =\left(
E^{2}\left(  E^{2}-m^{2}\right)  \right)  ^{\frac{1}{4}}$ and $\nu\left(
E\right)  =\sqrt{E^{2}-m^{2}}.$ The trace operator $T_{0}\left(  E\right)  $
for the free Dirac operator is defined by $T_{0}\left(  E\right)  =T_{0}%
^{+}\left(  E\right)  ,$ for $E>m$, and $T_{0}\left(  E\right)  =T_{0}%
^{-}\left(  E\right)  $, for $E<-m$. The operator valued function
$T_{0}\left(  E\right)  :(-\infty,-m)\cup(m,\infty)\rightarrow\mathcal{B}%
\left(  L_{s}^{2};L^{2}\left(  \mathbb{S}^{2};\mathbb{C}^{4}\right)  \right)
$ is locally H\"{o}lder continuous with exponent $\vartheta$ satisfying
$0<\vartheta\leq s-1/2$ and $\vartheta<1$ (\cite{70}, \cite{80}, \cite{8}).
Moreover, the operator $\left(  \mathcal{\tilde{F}}_{0}f\right)  \left(
E,\omega\right)  :=\left(  T_{0}\left(  E\right)  f\right)  \left(
\omega\right)  $ extends to a unitary operator from $L^{2}$ onto
$\mathcal{\hat{H}}^{\prime}:=L^{2}\left(  \left(  -\infty,-m\right)
;L^{2}\left(  \mathbb{S}^{2};P_{-}\mathbb{C}^{4}\right)  \right)  \oplus
L^{2}\left(  \left(  m,+\infty\right)  ;L^{2}\left(  \mathbb{S}^{2}%
;P_{+}\mathbb{C}^{4}\right)  \right)  =\left(  \int_{\left(  -\infty
,-m\right)  }^{\oplus}L^{2}\left(  \mathbb{S}^{2};P_{-}\mathbb{C}^{4}\right)
dE\right)  $ \ $\oplus\left(  \int_{\left(  m,+\infty\right)  }^{\oplus}%
L^{2}\left(  \mathbb{S}^{2};P_{+}\mathbb{C}^{4}\right)  dE\right)  ,$ that
gives a spectral representation of $H_{0},$ i.e., $\mathcal{\tilde{F}}%
_{0}H_{0}\mathcal{\tilde{F}}_{0}^{\ast}=E,$ the operator of multiplication by
$E$ in $\mathcal{\hat{H}}^{\prime}$.

The perturbed trace operators are defined by, $T_{\pm}\left(  E\right)
:=T_{0}\left(  E\right)  \left(  I-\mathbf{V}R\left(  E\pm i0\right)  \right)
,$ for $E\in\{(-\infty,-m)\cup(m,\infty)\}\backslash\sigma_{p}\left(
H\right)  .$ They are bounded from $L_{s}^{2}$ into $L^{2}\left(
\mathbb{S}^{2};\mathbb{C}^{4}\right)  ,$ for $s\in\left(  1/2,s_{0}\right]  $.
Furthermore, the operator valued functions $E\rightarrow T_{\pm}\left(
E\right)  $ from $\{(-\infty,-m)\cup(m,\infty)\}\backslash\sigma_{p}\left(
H\right)  $ into $\mathcal{B}\left(  L_{s}^{2};L^{2}\left(  \mathbb{S}%
^{2};\mathbb{C}^{4}\right)  \right)  $ are locally H\"{o}lder continuous for
$E\in\{(-\infty,-m)\cup(m,\infty)\}\backslash\sigma_{p}\left(  H\right)  $
with exponent $\vartheta$ satisfying the estimate $0<\vartheta\leq s-1/2,$
$s<\min\{s_{0},3/2\}$. The operators, $\left(  \mathcal{\tilde{F}}_{\pm
}f\right)  \left(  E,\omega\right)  :=\left(  T_{\pm}\left(  E\right)
f\right)  \left(  \omega\right)  $ extend to unitary operators from
$\mathcal{H}_{ac}$ onto $\mathcal{\hat{H}}^{\prime}$ and they give a spectral
representations for the restriction of $H$ to $\mathcal{H}_{ac},$
$\mathcal{F}_{\pm}H\mathcal{F}_{\pm}^{\ast}=E,$ the operator of multiplication
by $E$ in $\mathcal{\hat{H}}^{\prime}$.

Since the scattering operator $\mathbf{S}$ commutes with $H_{0},$ the operator
$\mathcal{\tilde{F}}_{0}\mathbf{S}\mathcal{\tilde{F}}_{0}^{\ast}$ acts as a
multiplication by the operator valued function $\tilde{S}\left(  E\right)
:L^{2}\left(  \mathbb{S}^{2};P_{\pm}\mathbb{C}^{4}\right)  \rightarrow
L^{2}\left(  \mathbb{S}^{2};P_{\pm}\mathbb{C}^{4}\right)  ,$ $\pm E>m,$ called
the scattering matrix. The scattering matrix satisfies the equality (see
Theorem 4.2 of \cite{14} and also \cite{42},\cite{37},\cite{8})%
\begin{equation}
\tilde{S}\left(  E\right)  T_{-}\left(  E\right)  =T_{+}\left(  E\right)  ,
\label{basicnotions33}%
\end{equation}
and, moreover, it has the following stationary representation for
$E\in\{(-\infty,-m)\cup(m,\infty)\}\backslash\sigma_{p}\left(  H\right)  ,$%
\begin{equation}
\tilde{S}\left(  E\right)  =I-2\pi iT_{0}\left(  E\right)  \left(
\mathbf{V}-\mathbf{V}R\left(  E+i0\right)  \mathbf{V}\right)  T_{0}\left(
E\right)  ^{\ast}. \label{basicnotions31}%
\end{equation}

We consider now another spectral representation of $H_{0}$ that we find more
convenient for our purposes. Let us define (\cite{15})%

\begin{equation}
\left(  \Gamma_{0}\left(  E\right)  f\right)  \left(  \omega\right)  :=\left(
2\pi\right)  ^{-\frac{3}{2}}\upsilon\left(  E\right)  P_{\omega}\left(
E\right)  \int_{\mathbb{R}^{3}}e^{-i\nu\left(  E\right)  \left\langle
\omega,x\right\rangle }f\left(  x\right)  dx, \label{basicnotions11}%
\end{equation}
with%
\begin{equation}
P_{\omega}\left(  E\right)  :=\left\{
\begin{array}
[c]{c}%
P^{+}\left(  \nu\left(  E\right)  \omega\right)  ,\text{ \ }E>m,\\
P^{-}\left(  \nu\left(  E\right)  \omega\right)  ,\text{ \ }E<-m,
\end{array}
\right.  \label{basicnotions43}%
\end{equation}
that is bounded from $L_{s}^{2},$ $s>1/2,$ into $L^{2}\left(  \mathbb{S}%
^{2};\mathbb{C}^{4}\right)  .$ The adjoint operator $\Gamma_{0}^{\ast}\left(
E\right)  :L^{2}\left(  \mathbb{S}^{2};\mathbb{C}^{4}\right)  \rightarrow
L_{-s}^{2},$ $s>1/2,$ is given by
\begin{equation}
\left(  \Gamma_{0}^{\ast}\left(  E\right)  f\right)  \left(  \omega\right)
:=\left(  2\pi\right)  ^{-\frac{3}{2}}\upsilon\left(  E\right)  \int%
_{\mathbb{S}^{2}}e^{i\nu\left(  E\right)  \left\langle x,\omega\right\rangle
}P_{\omega}\left(  E\right)  f\left(  \omega\right)  d\omega.
\label{basicnotions39}%
\end{equation}
Using (\ref{basicnotions30}) we have $\hat{G}\left(  \nu\left(  E\right)
\omega\right)  P_{\omega}\left(  E\right)  =\frac{1}{2}\left(  I_{4}\pm
\beta\right)  \hat{G}\left(  \nu\left(  E\right)  \omega\right)  =P_{\pm}%
\hat{G}\left(  \nu\left(  E\right)  \omega\right)  ,$ for $\pm E>m.$ This
relation means that $\hat{G}\left(  \nu\left(  E\right)  \omega\right)  $
takes $X^{\pm}\left(  \nu\left(  E\right)  \omega\right)  $ onto $P_{\pm
}\mathbb{C}^{4},$ for $\pm E>m,$ and moreover, it implies that
\begin{equation}
\Gamma_{0}\left(  E\right)  =\hat{G}\left(  \nu\left(  E\right)
\omega\right)  ^{-1}T_{0}\left(  E\right)  . \label{basicnotions28}%
\end{equation}
As $\hat{G}\left(  \nu\left(  E\right)  \omega\right)  ^{-1}$ is
differentiable on $E$, then from the last relation it follows that the
operator valued function $\Gamma_{0}\left(  E\right)  $ is locally H\"{o}lder
continuous on $(-\infty,-m)\cup(m,\infty)$ with the same exponent as
$T_{0}\left(  E\right)  .$

Let us define the unitary operator $\mathcal{U}$ from $\mathcal{\hat{H}%
}^{\prime}$ onto $\mathcal{\hat{H}}:=\int_{\left(  -\infty,-m\right)
\cup\left(  m,+\infty\right)  }^{\oplus}\mathcal{H}\left(  E\right)  dE,$
where%
\begin{equation}
\mathcal{H}\left(  E\right)  :=\int_{\mathbb{S}^{2}}^{\oplus}X^{\pm}\left(
\nu\left(  E\right)  \omega\right)  d\omega,\text{ \ }\pm E>m,
\label{basicnotions34}%
\end{equation}
by $\left(  \mathcal{U}f\right)  \left(  E,\omega\right)  :=\hat{G}\left(
\nu\left(  E\right)  \omega\right)  ^{-1}f\left(  E,\omega\right)  .$ Then the
operator $\left(  \mathcal{F}_{0}f\right)  \left(  E,\omega\right)  :=\left(
\Gamma_{0}\left(  E\right)  f\right)  \left(  \omega\right)  )=\left(
\mathcal{U\tilde{F}}_{0}f\right)  \left(  E,\omega\right)  $ extends to
unitary operator from $L^{2}$ onto $\mathcal{\hat{H}}$, that gives a spectral
representation of $H_{0}$%
\begin{equation}
\mathcal{F}_{0}H_{0}\mathcal{F}_{0}^{\ast}=E, \label{basicnotions12}%
\end{equation}
the operator of multiplication by $E$ in $\mathcal{\hat{H}}$.\ Note that in
$\mathcal{\hat{H}}^{\prime}$ the fibers $L^{2}\left(  \mathbb{S}^{2};P_{\pm
}\mathbb{C}^{4}\right)  $ can be written as $L^{2}\left(  \mathbb{S}%
^{2};P_{\pm}\mathbb{C}^{4}\right)  =\int_{\mathbb{S}^{2}}^{\oplus}P_{\pm
}\mathbb{C}^{4}d\omega,$ \ for $\pm E>m,$ where $P_{\pm}\mathbb{C}^{4}$ is
independent of $\omega.$ However, in $\mathcal{\hat{H}}$ the fibers are given
by (\ref{basicnotions34}), where $X^{\pm}\left(  \nu\left(  E\right)
\omega\right)  $ depend on $\omega.$

Let us define
\begin{equation}
\Gamma_{\pm}\left(  E\right)  :=\Gamma_{0}\left(  E\right)  \left(
I-\mathbf{V}R\left(  E\pm i0\right)  \right)  , \label{basicnotions25}%
\end{equation}
for $E\in\{(-\infty,-m)\cup(m,\infty)\}\backslash\sigma_{p}\left(  H\right)
.$ From relation (\ref{basicnotions28}) it follows that%
\begin{equation}
\Gamma_{\pm}\left(  E\right)  =\hat{G}\left(  \nu\left(  E\right)
\omega\right)  ^{-1}T_{\pm}\left(  E\right)  . \label{basicnotions32}%
\end{equation}
Then, the operator valued functions $\Gamma_{\pm}\left(  E\right)  $ are
bounded from $L_{s}^{2}$ into $L^{2}\left(  \mathbb{S}^{2};\mathbb{C}%
^{4}\right)  ,$ for $s\in\left(  1/2,s_{0}\right]  ,$ and they are locally
H\"{o}lder continuous on $(-\infty,-m)\cup(m,\infty)$ with the same exponent
as $T_{\pm}\left(  E\right)  .$ The operators,
\begin{equation}
\left(  \mathcal{F}_{\pm}f\right)  \left(  E,\omega\right)  :=\left(
\Gamma_{\pm}\left(  E\right)  f\right)  \left(  \omega\right)  =\left(
\mathcal{U\tilde{F}}_{\pm}f\right)  \left(  E,\omega\right)
\label{basicnotions40}%
\end{equation}
extend to unitary operators from $\mathcal{H}_{ac}$ onto $\mathcal{\hat{H}}$
and they give a spectral representations for the restriction of $H$ to
$\mathcal{H}_{ac},$ $\mathcal{F}_{\pm}H\mathcal{F}_{\pm}^{\ast}=E,$ the
operator of multiplication by $E$ in $\mathcal{\hat{H}}$.

In the spectral representation (\ref{basicnotions12}) the scattering matrix
acts as a multiplication by the operator valued function $S\left(  E\right)
:\mathcal{H}\left(  E\right)  \rightarrow\mathcal{H}\left(  E\right)  .$ Note
that relation (\ref{basicnotions28}) implies%
\begin{equation}
S\left(  E\right)  =\hat{G}\left(  \nu\left(  E\right)  \omega\right)
^{-1}\tilde{S}\left(  E\right)  \hat{G}\left(  \nu\left(  E\right)
\omega\right)  . \label{basicnotions35}%
\end{equation}
Thus, the scattering matrices $S\left(  E\right)  $ and $\tilde{S}\left(
E\right)  $ are unitary equivalent. Moreover, from relation
(\ref{basicnotions35}), (\ref{basicnotions28}) and representation
(\ref{basicnotions31}) we obtain the following stationary formula for
$S\left(  E\right)  ,$
\begin{equation}
S\left(  E\right)  =I-2\pi i\Gamma_{0}\left(  E\right)  \left(  \mathbf{V}%
-\mathbf{V}R\left(  E+i0\right)  \mathbf{V}\right)  \Gamma_{0}\left(
E\right)  ^{\ast}, \label{basicnotions14}%
\end{equation}
for $E\in\{(-\infty,-m)\cup(m,\infty)\}\backslash\sigma_{p}\left(  H\right)
.$ Here $I$ is the identity operator on $\mathcal{H}\left(  E\right)  .$

By the Schwartz Theorem (see, for example, Theorem 5.2.1 of \cite{77}) for
every continuous and linear operator $T$ from $C^{\infty}\left(
\mathbb{S}^{2};\mathbb{C}^{4}\right)  $ to $D^{\prime}\left(  \mathbb{S}%
^{2};\mathbb{C}^{4}\right)  $ (the set of distributions in $\mathbb{S}^{2}$)
there is one, and only one distribution $t\left(  \omega,\theta\right)  $ on
$\mathbb{S}^{2}\times\mathbb{S}^{2}$ such that for all $f\in C^{\infty}\left(
\mathbb{S}^{2};\mathbb{C}^{4}\right)  ,$ $\left(  Tf\right)  \left(
\omega\right)  =\int_{\mathbb{S}^{2}}t\left(  \omega,\theta\right)  f\left(
\theta\right)  d\theta.$ The \textquotedblleft integral\textquotedblright\ in
the R.H.S. of the last equation represents the duality parenthesis between the
test functions and distributions on the variable $\theta.$ The distribution
$t\left(  \omega,\theta\right)  $ is named the kernel of the operator $T.$

We say that the operator $T$ is an integral operator, if its kernel $t\left(
\omega,\theta\right)  $ is actually a continuous function for $\omega
\neq\theta$ which satisfies the following estimate $\left\vert t\left(
\omega,\theta\right)  \right\vert \leq\frac{C}{\left\vert \omega
-\theta\right\vert ^{2-\varepsilon}},$ for some $\varepsilon>0.$ Note that in
this case $t\left(  \omega,\theta\right)  \in L^{1}\left(  \mathbb{S}%
^{2}\times\mathbb{S}^{2}\right)  $ and moreover, using the Young's inequality
we have $\left\Vert Tf\right\Vert _{L^{2}\left(  \mathbb{S}^{2}\right)  }%
^{2}\leq\int_{\mathbb{S}^{2}}\left\vert \int_{\mathbb{S}^{2}}\frac
{C}{\left\vert \omega-\theta\right\vert ^{2-\varepsilon}}f\left(
\theta\right)  d\theta\right\vert ^{2}d\omega\leq C\left\Vert f\right\Vert
_{L^{2}\left(  \mathbb{S}^{2}\right)  }^{2}.$ Thus, if $T$ is an integral
operator, then it is bounded in $L^{2}\left(  \mathbb{S}^{2};\mathbb{C}%
^{4}\right)  .$ Below we show (Theorem \ref{representation25}) that if a
magnetic potential $A$ and an electric potential $V$ satisfy estimates
(\ref{eig31}) and (\ref{eig32}) respectively, then $S\left(  E\right)  -I$ is
an integral operator. We call \textquotedblleft scattering
amplitude\textquotedblright\ to the integral kernel $s^{\operatorname{int}%
}\left(  \omega,\theta;E\right)  $ of $S\left(  E\right)  -I.$

From the unitary of $S\left(  E\right)  $ it follows $\left(  S\left(
E\right)  -I\right)  ^{\ast}\left(  S\left(  E\right)  -I\right)  =-\left(
S\left(  E\right)  -I\right)  -\left(  S\left(  E\right)  -I\right)  ^{\ast}.$
Then,%
\begin{equation}%
{\displaystyle\int_{\mathbb{S}^{2}}}
s^{\operatorname{int}}\left(  \eta,\omega;E\right)  ^{\ast}%
s^{\operatorname{int}}\left(  \eta,\theta;E\right)  d\eta
=-s^{\operatorname{int}}\left(  \omega,\theta;E\right)  -s^{\operatorname{int}%
}\left(  \theta,\omega;E\right)  ^{\ast}. \label{representation103}%
\end{equation}
This equality is known in the physics literature as Optical Theorem (see for
example \cite{24}).

We had already mentioned in the introduction that the scattering operator
$\mathbf{S}$ and the scattering matrix $S\left(  E\right)  $ are invariant
under the gauge transformation $A\rightarrow A+\nabla\psi,$ for $\psi\in
C^{\infty}\left(  \mathbb{R}^{3}\right)  $ such that $\partial^{\alpha}%
\psi=O\left(  \left\vert x\right\vert ^{-\rho-\left\vert \alpha\right\vert
}\right)  $ for $0\leq\left\vert \alpha\right\vert \leq1$ and some $\rho>0$ as
$\left\vert x\right\vert \rightarrow\infty$. Indeed, note that $\tilde
{H}:=e^{-i\psi}He^{i\psi}$ satisfies $\tilde{H}=H+\left(  \alpha\cdot
\nabla\psi\right)  .$ Then we get $W_{\pm}\left(  \tilde{H},H_{0}\right)
=s-\lim_{t\rightarrow\pm\infty}e^{i\tilde{H}t}e^{-itH_{0}}=s-\lim
_{t\rightarrow\pm\infty}e^{-i\psi}e^{iHt}e^{i\psi}e^{-itH_{0}}.$ Under the
assumptions on $\psi$, the operator of multiplication by the function
$e^{i\psi}-1$ is a compact operator from $\mathcal{H}^{1}$ to $L^{2}.$ Since
$e^{-iH_{0}t}$ converges weakly to $0,$ as $t\rightarrow\pm\infty,$ then
$s-\lim_{t\rightarrow\pm\infty}\left(  e^{-i\psi}e^{iHt}e^{i\psi}e^{-itH_{0}%
}\right)  =s-\lim_{t\rightarrow\pm\infty}\left(  e^{-i\psi}e^{iHt}\left(
e^{i\psi}-1\right)  e^{-itH_{0}}\right)  +s-\lim_{t\rightarrow\pm\infty
}\left(  e^{-i\psi}e^{iHt}e^{-itH_{0}}\right)  =s-\lim_{t\rightarrow\pm\infty
}\left(  e^{-i\psi}e^{iHt}e^{-itH_{0}}\right)  ,$ and then $W_{\pm}\left(
\tilde{H},H_{0}\right)  =e^{i\psi}W_{\pm}\left(  H,H_{0}\right)  .$ Thus, we
conclude that $\mathbf{S}\left(  \tilde{H},H_{0}\right)  =\mathbf{S}\left(
H,H_{0}\right)  .$

It results convenient for us to associate the scattering matrix $S\left(
E\right)  $ directly with the magnetic field $B\left(  x\right)
=\operatorname*{rot}A\left(  x\right)  .$ However, as $S\left(  E\right)  $
was defined in terms of the magnetic potential $A\left(  x\right)  ,$ we
recall the procedure given in \cite{25} and \cite{31} for the construction of
a short-range magnetic potential from an arbitrary magnetic field satisfying
the condition for some $d\geq1$
\begin{equation}
\left\vert \partial^{\alpha}B\left(  x\right)  \right\vert \leq C_{\alpha
}\left(  1+\left\vert x\right\vert \right)  ^{-r-\left\vert \alpha\right\vert
},\text{ }r>2,\text{ }0\leq\left\vert \alpha\right\vert \leq d.
\label{basicnotions18}%
\end{equation}
We note that the magnetic potential can be reconstructed from the magnetic
field $B\left(  x\right)  $ such that $\operatorname*{div}B(x)=0$ only up to
arbitrary gauge transformations. It is not convenient to work in the standard
transversal gauge $\left\langle x,A_{tr}\left(  x\right)  \right\rangle =0$
since even for magnetic fields of compact support the potential $A_{tr}\left(
x\right)  $ decays only as $\left\vert x\right\vert ^{-1}$ at infinity.

Let $B(x)=(B_{1}(x),B_{2}(x),B_{3}(x))$ be a magnetic field satisfying
estimates (\ref{basicnotions18}) such that $\operatorname*{div}B(x)=0$. Let us
define the following matrix
\[
F=\left(
\begin{array}
[c]{ccc}%
0 & B_{3} & -B_{2}\\
-B_{3} & 0 & B_{1}\\
B_{2} & -B_{1} & 0
\end{array}
\right)
\]
and introduce the auxiliary potentials
\begin{equation}
A_{i}^{\left(  reg\right)  }\left(  x\right)  =%
{\displaystyle\int_{1}^{\infty}}
s\sum_{j=1}^{3}F^{\left(  ij\right)  }\left(  sx\right)  x_{j}ds,\text{
\ }A_{i}^{\left(  \infty\right)  }\left(  x\right)  =-%
{\displaystyle\int_{0}^{\infty}}
s\sum_{j=1}^{3}F^{\left(  ij\right)  }\left(  sx\right)  x_{j}ds,
\label{basicnotions19}%
\end{equation}
where $F^{\left(  ij\right)  }$ is the $\left(  i,j\right)  $-th element of
$F$. Note that $A^{\left(  \infty\right)  }\left(  x\right)  $ is a
homogeneous function of degree $-1,$ $A^{\left(  reg\right)  }\left(
x\right)  =O\left(  \left\vert x\right\vert ^{-\rho}\right)  $ with $\rho=r-1$
as $\left\vert x\right\vert \rightarrow\infty,$ $\operatorname*{rot}A^{\left(
\infty\right)  }\left(  x\right)  =0$ for $x\neq0$ and $A_{tr}=A^{\left(
reg\right)  }\left(  x\right)  +A^{\left(  \infty\right)  }\left(  x\right)
.$ We define the function $U\left(  x\right)  $ for $x\neq0$ as a curvilinear
integral
\begin{equation}
U\left(  x\right)  =\int_{\Gamma_{x_{0},x}}\left\langle A^{\left(
\infty\right)  }\left(  y\right)  ,dy\right\rangle \label{basicnotions20}%
\end{equation}
taken between some fixed point $x_{0}\neq0$ and a variable point $x.$ If
$0\notin\Gamma_{x_{0},x},$ it follows from the Stokes theorem that the
function $U\left(  x\right)  $ does not depend on the choice of a contour
$\Gamma_{x_{0},x}$ and $\operatorname*{grad}U\left(  x\right)  =A_{\infty
}\left(  x\right)  .$ Now we define the magnetic potential as
\begin{equation}
A\left(  x\right)  =A_{tr}\left(  x\right)  -\operatorname*{grad}\left(
\eta\left(  x\right)  U\left(  x\right)  \right)  =A_{reg}\left(  x\right)
+\left(  1-\eta\left(  x\right)  \right)  A_{\infty}\left(  x\right)
-U\left(  x\right)  \operatorname*{grad}\eta\left(  x\right)  ,
\label{basicnotions21}%
\end{equation}
where $\eta\left(  x\right)  \in C^{\infty}\left(  \mathbb{R}^{3}\right)  ,$
$\eta\left(  x\right)  =0$ in a neighborhood of zero and $\eta\left(
x\right)  =1$ for $\left\vert x\right\vert \geq1.$ Note that
$\operatorname*{rot}A\left(  x\right)  =B\left(  x\right)  ,$ $A\in C^{\infty
}$ if $B\in C^{\infty}$ and $A\left(  x\right)  =A_{reg}\left(  x\right)  $
for $\left\vert x\right\vert \geq1.$ Moreover, it follows from the assumption
(\ref{basicnotions18}) that $A\left(  x\right)  $ satisfies the estimates
$\left\vert \partial^{\alpha}A\left(  x\right)  \right\vert \leq C_{\alpha
}\left(  1+\left\vert x\right\vert \right)  ^{-\rho-\left\vert \alpha
\right\vert },$ $\rho=r-1,$ for all $0\leq\left\vert \alpha\right\vert \leq
d.$

For a given magnetic field $B\left(  x\right)  $ we associate the magnetic
potential $A\left(  x\right)  $ by formulae (\ref{basicnotions19}%
)--(\ref{basicnotions21}) and then construct the scattering matrix $S\left(
E\right)  $ in terms of the Dirac operator (\ref{basicnotions6}). As we showed
above, if another magnetic potential $\tilde{A}\left(  x\right)  $ satisfies
$\operatorname*{rot}\tilde{A}\left(  x\right)  $ $=B(x)$, then the scattering
matrices corresponding to potentials $A$ and $\tilde{A}$ coincide. This allows
us to speak about the scattering matrix $S(E)$ corresponding to the magnetic
field $B$.

Now we introduce some notation. Let $\mathit{\dot{S}}^{-\rho}=\mathit{\dot{S}%
}^{-\rho}\left(  \mathbb{R}^{3}\right)  $ be the set of $C^{\infty}\left(
\mathbb{R}^{3}\backslash\{0\}\right)  $ functions $f\left(  x\right)  $ such
that $\partial^{\alpha}f\left(  x\right)  =O\left(  \left\vert x\right\vert
^{-\rho-\left\vert \alpha\right\vert }\right)  ,$ as $\left\vert x\right\vert
\rightarrow\infty,$ for all $\alpha.$ Define also a subspace of $\mathit{\dot
{S}}^{-\rho}$ by $\mathit{S}^{-\rho}:=\mathit{\dot{S}}^{-\rho}\cap C^{\infty
}\left(  \mathbb{R}^{3}\right)  .$ An example of functions from the class
$\mathit{\dot{S}}^{-\rho}$ are the homogeneous functions $f\in C^{\infty
}\left(  \mathbb{R}^{3}\backslash\{0\}\right)  $ of order $-\rho$, i.e. such
that $f\left(  tx\right)  =t^{-\rho}f\left(  x\right)  $ for all $x\neq0$ and
$t>0.$

Let functions $f_{j}\in\mathit{\dot{S}}^{-\rho_{j}}$ with $\rho_{j}%
\rightarrow\infty.$ The notation
\begin{equation}
f\left(  x\right)  \simeq\sum_{j=1}^{\infty}f_{j}\left(  x\right)
\label{basicnotions16}%
\end{equation}
means that, for any $N,$ the remainder\
\begin{equation}
f\left(  x\right)  -\sum_{j=1}^{N}f_{j}\left(  x\right)  \in\mathit{\dot{S}%
}^{-\rho},\text{ where }\rho=\min_{j\geq N+1}\rho_{j}. \label{basicnotions17}%
\end{equation}
Note that the function $f\in C^{\infty}$ is determined by its expansion
(\ref{basicnotions16}) only up to a term from the Schwartz class
$\mathit{S=S}^{-\infty}\mathit{.}$

Finally, let us recall some facts about pseudodifferential operators. We
define a pseudodifferential operator (PDO) by (\cite{32},\cite{36})
\begin{equation}
\left(  Af\right)  \left(  x\right)  =\left(  2\pi\right)  ^{-d}%
\int_{\mathbb{R}^{d}}\int_{\mathbb{R}^{d}}e^{i\left\langle \xi,x-\xi^{\prime
}\right\rangle }a\left(  x,\xi\right)  f\left(  \xi^{\prime}\right)
d\xi^{\prime}d\xi, \label{basicnotions42}%
\end{equation}
where $f\in\mathcal{S(}\mathbb{R}^{d};\mathbb{C}^{4}\mathcal{)}$ and $a\left(
x,\xi\right)  $ is a $\left(  4\times4\right)  -$matrix. Here $d$ is the
dimension (equals $2$ or $3$). We denote by $\mathcal{S}^{n,m}$ the class of
PDO, which symbols are of $C^{\infty}\left(  \mathbb{R}^{d}\times
\mathbb{R}^{d}\right)  $ class and for all $x,\xi$ and for all multi-indices
$\alpha,\beta$,%
\begin{equation}
\left\vert \partial_{x}^{\alpha}\partial_{\xi}^{\beta}a\left(  x,\xi\right)
\right\vert \leq C_{\alpha,\beta}\left\langle x\right\rangle ^{n-\left\vert
\alpha\right\vert }\left\langle \xi\right\rangle ^{m-\left\vert \beta
\right\vert }. \label{representation197}%
\end{equation}
Note that $\mathcal{S}^{n,m}\subset\Gamma_{1}^{2m_{1}},$ for $m_{1}%
=\max\{n,m\}$, where the clases $\Gamma_{\rho}^{m_{1}}$ are defined in
\cite{36}. Moreover we need a more special class $\mathcal{S}_{\pm}%
^{n,m}\subset$ $\mathcal{S}^{n,m}$ satisfying the additional property
$a\left(  x,\xi\right)  =0$ if $\mp\left\langle \hat{x},\hat{\xi}\right\rangle
\leq\varepsilon,$ $\varepsilon>0,$ and $a\left(  x,\xi\right)  =0$ if
$\left\vert x\right\vert \leq\varepsilon_{1}$ or $\left\vert \xi\right\vert
\leq\varepsilon_{1},$ $\varepsilon_{1}>0.$ (Here $\hat{x}=x/\left\vert
x\right\vert $ and $\hat{\xi}=\xi/\left\vert \xi\right\vert $).

It is convenient for us to consider a more general formula for the action of
the PDO's, determined by their amplitude. We define a PDO $\mathbf{A}$ by
\begin{equation}
\left(  \mathbf{A}f\right)  \left(  x\right)  =\left(  2\pi\right)  ^{-d}%
\int_{\mathbb{R}^{d}}\int_{\mathbb{R}^{d}}e^{i\left\langle \xi,x-\xi^{\prime
}\right\rangle }a\left(  x,\xi,\xi^{\prime}\right)  f\left(  \xi^{\prime
}\right)  d\xi^{\prime}d\xi, \label{basicnotions41}%
\end{equation}
where $a\left(  x,\xi,\xi^{\prime}\right)  $ is called the amplitude of
$\mathbf{A}.$ We say that $a\left(  x,\xi,\xi^{\prime}\right)  $ belongs to
the class $\mathcal{S}^{n}$ if for all indices $\alpha,\beta\,,\gamma$ the
following estimate holds
\begin{equation}
\left\vert \partial_{x}^{\alpha}\partial_{\xi}^{\beta}\partial_{\xi^{\prime}%
}^{\gamma}a\left(  x,\xi,\xi^{\prime}\right)  \right\vert \leq C_{\alpha
,\beta,\gamma}\left\langle \left(  x,\xi,\xi^{\prime}\right)  \right\rangle
^{n-\left\vert \alpha+\beta+\gamma\right\vert },\text{ }\left(  x,\xi
,\xi^{\prime}\right)  \in\mathbb{R}^{3d}. \label{representation198}%
\end{equation}
We note that $\mathcal{S}^{n}$ is contained in $\Pi_{1}^{n}$ ($\Pi_{\rho}^{n}%
$~are defined in \cite{36})). We can make the passage from the amplitude to
the correspondent (left) symbol by the relation
\begin{equation}
a_{\operatorname*{left}}\left(  x,\xi\right)  \simeq\sum_{\alpha}\frac
{1}{\alpha!}\left.  \partial_{\xi}^{\alpha}\left(  -i\partial_{\xi^{\prime}%
}^{\alpha}\right)  a\left(  x,\xi,\xi^{\prime}\right)  \right\vert
_{\xi^{\prime}=x}. \label{basicnotions45}%
\end{equation}

For arbitrary $n,$ the integrals in the R.H.S. of relations
(\ref{basicnotions42}) and (\ref{basicnotions41}) are understood as
oscillating integrals.\ Furthermore, we recall the following results from the
PDO calculus (see \cite{32} or \cite{36})

\begin{proposition}
\label{basicnotions24}If $a\left(  x,\xi\right)  \in\mathcal{S}^{n,m}$ with
$n\leq0$ and $m\leq0,$ then the PDO $A$ can be extended to a bounded operator
in $L^{2}.$ The $L^{2}$- norm of $A\left\langle x\right\rangle ^{-n}$ is
estimated by some constant $C,$ that depends only on the constants
$C_{\alpha,\beta}$, given by (\ref{representation197})$.$ Moreover, if
$a\left(  x,\xi\right)  \in\mathcal{S}^{n,m}$ with $n<0$ and $m<0,$ then $A$
can be extended to a compact operator in $L^{2}.$
\end{proposition}

\begin{proposition}
\label{representation196}Let $A_{j}$ be PDO with symbols $a_{j}\in
\mathcal{S}^{n_{j},m_{j}},$ for $j=1,2.$ Then the symbol $a$ of the product
$A_{1}A_{2}$ belongs to the class $\mathcal{S}^{n_{1}+n_{2},m_{1}+m_{2}}$ and
it admits the following asymptotic expansion $a\left(  x,\xi\right)
=\sum_{\left\vert \alpha\right\vert <N}\frac{\left(  -i\right)  ^{\left\vert
\alpha\right\vert }}{\alpha!}\partial_{\xi}^{\alpha}a_{1}\left(  x,\xi\right)
\partial_{x}^{\alpha}a_{2}\left(  x,\xi\right)  +r^{\left(  N\right)  }\left(
x,\xi\right)  ,$ where $r^{\left(  N\right)  }\in\mathcal{S}^{n_{1}%
+n_{2}-N,m_{1}+m_{2}-N}.$
\end{proposition}

Note that, in particular, Propositions \ref{basicnotions24} and
\ref{representation196} imply the following result

\begin{proposition}
\label{representation131}Let $A_{j}$ be PDO with symbols $a_{j}\in
\mathcal{S}^{0,0},$ $j=1,2,$ and let $A$ be the PDO with symbol $a_{1}\left(
x,\xi\right)  a_{2}\left(  x,\xi\right)  .$ Then, $A_{1}A_{2}-A$ can be
extended to a compact operator.
\end{proposition}

For a PDO $A$ defined by its amplitude $a\left(  x,\xi,\xi^{\prime}\right)  $
we have

\begin{proposition}
\label{representation76}If $a\left(  x,\xi,\xi^{\prime}\right)  \in
\mathcal{S}^{n}$, then the PDO $A$ can be extended to a bounded operator in
$L^{2}$ for $n=0,$ and it operator norm is bounded by a constant $C$ depending
only on $C_{\alpha,\beta,\gamma}$ given by (\ref{representation198}).
Moreover, $A$ can be extended to a compact operator in $L^{2}$ for $n<0.$
\end{proposition}

\section{Scattering solutions.}

In the stationary approach to the scattering theory it is useful to consider
special solutions to the Dirac equation
\begin{equation}
\left(  H_{0}+\mathbf{V}\right)  u=Eu,\text{ for }x\in\mathbb{R}^{3},
\label{re7}%
\end{equation}
called scattering solutions, or generalized eigenfunctions of continuous
spectrum. Suppose that $\mathbf{V}$ satisfies Condition \ref{basicnotions26}
and $\mathbf{V}\in L_{s}^{2},$ for some $s>1/2.$ Then, for all $E\in
\{(-\infty,-m)\cup(m,\infty)\}\backslash\sigma_{p}\left(  H\right)  ,$ we
define the scattering solutions to equation (\ref{re7}) by%
\begin{equation}
u_{\pm}\left(  x,\theta;E\right)  =P_{\theta}\left(  E\right)  e^{i\nu\left(
E\right)  \left\langle x,\theta\right\rangle }-\left(  R\left(  E\pm
i0\right)  \mathbf{V}\left(  \cdot\right)  P_{\theta}\left(  E\right)
e^{i\nu\left(  E\right)  \left\langle \cdot,\theta\right\rangle }\right)
\left(  x\right)  . \label{representation29}%
\end{equation}
We observe that under suitable assumptions on the solutions $u$ to
(\ref{re7}),\ known as "radiation conditions", and on the potential
$\mathbf{V,}$ $u_{\pm}$ is characterized as the unique solution to
(\ref{re7}), satisfying these "radiation conditions". The problem of existence
and unicity of solutions to (\ref{re7}) was treated in \cite{41}, by studying
the formula (\ref{basicnotions27}), for $f\in L_{s}^{2},$ $s>1/2,$ with
radiation conditions $v_{\pm}\in L_{-s}^{2}\ $and$\ (\frac{\partial}{\partial
x_{j}}v_{\pm}\left(  x\right)  \mp iv_{\pm}\left(  E\right)  \frac{x_{j}%
}{\left\vert x\right\vert }u\left(  x\right)  )\in L_{s-1}^{2},$ $1/2<s\leq1,$
for $j=1,2,3.$ The radiation estimates, in the sense that the operators
$\left\langle x\right\rangle ^{-1/2}\left(  \frac{\partial}{\partial x_{j}%
}-\left\vert x\right\vert ^{-2}x_{j}\sum_{k=1}^{3}x_{k}\frac{\partial
}{\partial x_{k}}\right)  $, for $j=1,2,3,$ are $H-$smooth, for the Dirac
equation with long-range potentials, were obtained in \cite{33}.

We want to give an asymptotic formula for $u_{\pm}$, as $\left\vert
x\right\vert \rightarrow\infty,$ where the asymptotic is understood in an
appropriate sense. We denote by $\tilde{o}\left(  \left\vert x\right\vert
^{-1}\right)  $ a function $g\left(  x\right)  $ such that $\lim
_{r\rightarrow\infty}\left(  \frac{1}{r}\int_{\left\vert x\right\vert \leq
r}\left\vert g\left(  x\right)  \right\vert ^{2}dx\right)  =0.$ Of course a
$o\left(  \left\vert x\right\vert ^{-1}\right)  $ function is also a
$\tilde{o}\left(  \left\vert x\right\vert ^{-1}\right)  $ function. We prove
the following

\begin{theorem}
\label{re17}Suppose that $\mathbf{V}$ satisfies Condition \ref{basicnotions26}
and $\mathbf{V}\in L_{s}^{2},$ for some $s>1/2.$ Then, the scattering
solutions admit the asymptotic expansion%
\begin{equation}
u_{\pm}\left(  x,\theta;E\right)  =P_{\theta}\left(  E\right)  e^{i\nu\left(
E\right)  \left\langle x,\theta\right\rangle }+a_{\pm}\left(  \hat{x}%
,\theta;E\right)  \frac{e^{\pm i\left(  \operatorname*{sgn}E\right)
\nu\left(  E\right)  \left\vert x\right\vert }}{\left\vert x\right\vert
}+\tilde{o}\left(  \left\vert x\right\vert ^{-1}\right)  , \label{re16}%
\end{equation}
where the functions $a_{\pm}\left(  \hat{x},\theta;E\right)  :=-\left(
\operatorname*{sgn}E\right)  \left(  \frac{2\pi\left\vert E\right\vert }%
{\nu\left(  E\right)  }\right)  ^{1/2}\left(  \Gamma_{\pm}\left(  E\right)
\mathbf{V}\left(  \cdot\right)  P_{\theta}\left(  E\right)  e^{i\nu\left(
E\right)  \left\langle \cdot,\theta\right\rangle }\right)  \left(  \pm\left(
\operatorname*{sgn}E\right)  \hat{x}\right)  \ $can be recovered from $u_{\pm
}\left(  x,\theta;E\right)  $ by the formula
\begin{equation}
a_{\pm}\left(  \hat{x},\theta;E\right)  =-\left(  \operatorname*{sgn}E\right)
\left(  \frac{2\pi\left\vert E\right\vert }{\nu\left(  E\right)  }\right)
^{1/2}\left(  \Gamma_{0}\left(  E\right)  \mathbf{V}u_{\pm}\right)  \left(
\pm\left(  \operatorname*{sgn}E\right)  \hat{x}\right)  . \label{re18}%
\end{equation}
Moreover $a_{+}\left(  x,\theta;E\right)  $ is related to the scattering
amplitude $s^{\operatorname{int}}\left(  \hat{x},\theta;E\right)  $ by the
formula $a_{+}\left(  \hat{x},\theta;E\right)  =-i\left(  \operatorname*{sgn}%
E\right)  \left(  2\pi\right)  \nu\left(  E\right)  ^{-1}$ $\times
s^{\operatorname{int}}\left(  \left(  \operatorname*{sgn}E\right)  \hat
{x},\theta;E\right)  .$
\end{theorem}

Note that the coefficient of the leading term in the asymptotics (\ref{re16})
is explicit. Similar asymptotic was obtained in \cite{41}, but the expression
for the asymptotics (\ref{re16}) is not explicit there.

In order to prove Theorem \ref{re17} we need some results.

\begin{lemma}
\label{re13}For all $\left\vert E\right\vert >m$ and all functions $f\in
L_{2}$ with compact support\bigskip\ we have
\begin{equation}
R_{0}\left(  E\pm i0\right)  f\left(  x\right)  =\left(  \operatorname*{sgn}%
E\right)  \left(  \frac{2\pi\left\vert E\right\vert }{\nu\left(  E\right)
}\right)  ^{1/2}\left(  \Gamma_{0}\left(  E\right)  f\right)  \left(
\pm\left(  \operatorname*{sgn}E\right)  \hat{x}\right)  \frac{e^{\pm i\left(
\operatorname*{sgn}E\right)  \nu\left(  E\right)  \left\vert x\right\vert }%
}{\left\vert x\right\vert }+O\left(  \left\vert x\right\vert ^{-2}\right)  ,
\label{re1}%
\end{equation}
and
\begin{equation}
\left(  \partial_{\left\vert x\right\vert }R_{0}\left(  E\pm i0\right)
\right)  f\left(  x\right)  =\pm i\left(  2\pi\right)  ^{1/2}\upsilon\left(
E\right)  \left(  \Gamma_{0}\left(  E\right)  f\right)  \left(  \pm\left(
\operatorname*{sgn}E\right)  \hat{x}\right)  \frac{e^{\pm i\left(
\operatorname*{sgn}E\right)  \nu\left(  E\right)  \left\vert x\right\vert }%
}{\left\vert x\right\vert }+O\left(  \left\vert x\right\vert ^{-2}\right)  .
\label{re2}%
\end{equation}

\end{lemma}

\begin{proof}
From the relation $H_{0}^{2}=-\Delta+m^{2}$ we get, for all $z\in\mathbb{C},$
$\left(  H_{0}-z\right)  \left(  H_{0}+z\right)  =-\Delta-\left(  z^{2}%
-m^{2}\right)  ,$ and hence $R_{0}\left(  z\right)  =\left(  H_{0}+z\right)
R_{0S}\left(  z^{2}-m^{2}\right)  .$ Recall that $R_{0S}\left(  z^{2}%
-m^{2}\right)  $ is an integral operator. Its kernel (Green function) is given
by $\left(  4\pi\left\vert x-y\right\vert \right)  ^{-1}e^{i\sqrt{z^{2}-m^{2}%
}\left\vert x-y\right\vert }.$ Therefore, $R_{0}\left(  z\right)  $ is an
integral operator with kernel $R_{0}\left(  x,y;z\right)  :=\left(
H_{0}+z\right)  (\left(  4\pi\left\vert x-y\right\vert \right)  ^{-1}$ $\times
e^{i\sqrt{z^{2}-m^{2}}\left\vert x-y\right\vert }).$ Proceeding similarly to
the proof of Lemma 2.6 of \cite{8} (page 79), we get (\ref{re1}) and
(\ref{re2}).
\end{proof}

From the LAP for the perturbed Dirac operator and Lemma \ref{re13}, we get

\begin{lemma}
\label{re10}For all $f\in L_{s}^{2},$ $s>1/2,$ asymptotics (\ref{re1}) and
(\ref{re2}) hold, with $O\left(  \left\vert x\right\vert ^{-2}\right)  $
replaced by $\tilde{o}\left(  \left\vert x\right\vert ^{-1}\right)  .$ Suppose
that $\mathbf{V}$ satisfies Condition \ref{basicnotions26}. Then, for all
$f\in L_{s}^{2},$ $s>1/2,$ and $E\in\{(-\infty,-m)\cup(m,\infty)\}\backslash
\sigma_{p}\left(  H\right)  \}$ we have%
\[
R\left(  E\pm i0\right)  f\left(  x\right)  =\left(  \operatorname*{sgn}%
E\right)  \left(  \frac{2\pi\left\vert E\right\vert }{\nu\left(  E\right)
}\right)  ^{1/2}\left(  \Gamma_{\pm}\left(  E\right)  f\right)  \left(
\pm\left(  \operatorname*{sgn}E\right)  \hat{x}\right)  \frac{e^{\pm i\left(
\operatorname*{sgn}E\right)  \nu\left(  E\right)  \left\vert x\right\vert }%
}{\left\vert x\right\vert }+\tilde{o}\left(  \left\vert x\right\vert
^{-1}\right)  ,
\]
and
\[
\left(  \partial_{\left\vert x\right\vert }R\left(  E\pm i0\right)  \right)
f\left(  x\right)  =\pm i\left(  2\pi\right)  ^{1/2}\upsilon\left(  E\right)
\left(  \Gamma_{\pm}\left(  E\right)  f\right)  \left(  \pm\left(
\operatorname*{sgn}E\right)  \hat{x}\right)  \frac{e^{\pm i\left(
\operatorname*{sgn}E\right)  \nu\left(  E\right)  \left\vert x\right\vert }%
}{\left\vert x\right\vert }+\tilde{o}\left(  \left\vert x\right\vert
^{-1}\right)  .
\]

\end{lemma}

\begin{proof}
The proof is the same as in the case of the Schr\"{o}dinger equation (see
Proposition 4.3 and Theorem 4.4 of \cite{8}, page 240).
\end{proof}

We now are able to prove Theorem \ref{re17}\textbf{.}

\textbf{Proof of Theorem \ref{re17}. \ }Note that $\mathbf{V}\left(  x\right)
P_{\theta}\left(  E\right)  e^{i\nu\left(  E\right)  \left\langle
x,\theta\right\rangle }\in L_{s}^{2}\mathbf{,}$ for some $s>1/2.$ The
asymptotic expansion (\ref{re16}) is consequence of (\ref{re18}). Multiplying
(\ref{representation29}) from the left side by $I+R_{0}\left(  E\pm i0\right)
\mathbf{V}$ and using the resolvent identity $R=R_{0}-R_{0}\mathbf{V}R,$ we
get that $u_{\pm}\left(  x,\theta;E\right)  =P_{\theta}\left(  E\right)
e^{i\nu\left(  E\right)  \left\langle x,\theta\right\rangle }-R_{0}\left(
E\pm i0\right)  \mathbf{V}u_{\pm}.$ Then relation (\ref{re18}) follows from
Lemma \ref{re10}. Using the representation (\ref{basicnotions14}) of the
scattering matrix $S\left(  E\right)  ,$ the definition (\ref{basicnotions25})
of $\Gamma_{+}\left(  E\right)  $ and the relation (\ref{basicnotions39}) for
$\Gamma_{0}^{\ast}\left(  E\right)  $ we get $s^{\operatorname{int}}\left(
\omega,\theta;E\right)  =-i\left(  2\pi\right)  ^{-\frac{1}{2}}\upsilon\left(
E\right)  (\Gamma_{+}\left(  E\right)  \mathbf{V}\left(  \cdot\right)
P_{\theta}\left(  E\right)  e^{i\nu\left(  E\right)  \left\langle \cdot
,\theta\right\rangle })\left(  \omega\right)  .$ Then, the relation between
$a_{+}$ and $s^{\operatorname{int}}$, follows from the definition of $a_{+}.$

\section{Symmetries of the kernel of the scattering matrix.}

In this Section we discuss the parity, charge-conjugation and time-reversal
transformations for the Dirac operator (see, for example \cite{48}). In
particular, we want to study the symmetries that these transformations imply
for the kernel $s\left(  \omega,\theta;E\right)  $ of the scattering matrix
$S\left(  E\right)  $. These symmetries can give necessary conditions when one
studies the characterization problem. We consider the Dirac operator with
potential $\mathbf{V}$ of the form (\ref{intro4}). Suppose that $\mathbf{V}$
satisfies Condition \ref{basicnotions26}. The latter assumption is made only
in order to guarantee the existence of the wave operators $W_{\pm}$ and the
scattering operator $\mathbf{S},$ and may be relaxed.

The parity transformation is defined as $\mathcal{P=}e^{i\phi}\beta\varkappa,$
where $\left(  \varkappa f\right)  \left(  x\right)  =f\left(  -x\right)  $ is
the space reflection operator and $\phi$ is a fixed phase. The transformation
$\mathcal{P}$ commutes with $H_{0}$. For the perturbed operator $H$ we have
$\mathcal{P}\left(  -i\alpha\nabla+m\beta+\alpha A\left(  x\right)  +V\left(
x\right)  \right)  =\left(  -i\alpha\nabla+m\beta-\alpha A\left(  -x\right)
+V\left(  -x\right)  \right)  \mathcal{P}.$ Therefore, for the operator $H$ of
the form (\ref{intro4}) with general electromagnetic potential $\mathbf{V}$
the parity transformation $\mathcal{P}$ is not a symmetry. If we consider an
even electric potential, $V\left(  x\right)  =V\left(  -x\right)  ,$ and an
odd magnetic potential, $A\left(  -x\right)  =-A\left(  x\right)  $, then
$\mathcal{P}H=H\mathcal{P}$. In this case it follows that the parity
transformation $\mathcal{P}$ commutes with the wave operators $W_{\pm}$ and
the scattering operator $\mathbf{S.}$

Noting that%
\begin{equation}
\beta\hat{\varkappa}\Gamma_{0}\left(  E\right)  =\Gamma_{0}\left(  E\right)
e^{-i\phi}\mathcal{P}, \label{sym1}%
\end{equation}
and%
\begin{equation}
e^{-i\phi}\mathcal{P}\Gamma_{0}^{\ast}\left(  E\right)  =\Gamma_{0}^{\ast
}\left(  E\right)  \beta\hat{\varkappa}, \label{sym7}%
\end{equation}
where $\left(  \hat{\varkappa}b\right)  \left(  \omega\right)  =b\left(
-\omega\right)  ,$ is the reflection operator on the unit sphere, we obtain
the equality $\beta\hat{\varkappa}S\left(  E\right)  =\beta\hat{\varkappa
}\Gamma_{0}\left(  E\right)  \mathbf{S}\Gamma_{0}^{\ast}\left(  E\right)
=\Gamma_{0}\left(  E\right)  \mathbf{S}\Gamma_{0}^{\ast}\left(  E\right)
\hat{\varkappa}\beta$ $=S\left(  E\right)  \hat{\varkappa}\beta.$ This means
that the kernel $s\left(  \omega,\theta;E\right)  $ of the scattering matrix
$S\left(  E\right)  $ satisfy the relation%
\begin{equation}
s\left(  \omega,\theta;E\right)  =\beta s\left(  -\omega,-\theta;E\right)
\beta. \label{sym4}%
\end{equation}

For the Dirac operator the charge-conjugation transformation is defined as
$\mathcal{C}=i\left(  \beta\alpha_{2}\right)  \mathbf{C,}$ where $\mathbf{C}$
is the complex conjugation.\ Note that $\mathcal{C}\left(  -i\alpha
\nabla+m\beta+\alpha A+V\right)  =i\left(  \beta\alpha_{2}\right)  \left(
i\overline{\alpha}\nabla+m\beta+\overline{\alpha}A+V\right)  \mathbf{C}%
=-\left(  -i\alpha\nabla+m\beta-\alpha A-V\right)  \mathcal{C}\mathbf{.}$ We
consider now an odd electric potential $V\left(  x\right)  $ and an even
magnetic potential $A\left(  x\right)  .$ Using the charge-conjugation
transformation $\mathcal{C}$ we see that $\mathcal{CP}H=-H\mathcal{CP}$,
$\mathcal{CP}W_{\pm}=W_{\pm}\mathcal{CP}$ and $\mathcal{CP}\mathbf{S=S}%
\mathcal{CP}$. Moreover, as%
\begin{equation}
\mathcal{C}\hat{\varkappa}\Gamma_{0}\left(  E\right)  =\Gamma_{0}\left(
-E\right)  \mathcal{C}, \label{representation199}%
\end{equation}
and%
\begin{equation}
\mathcal{C}\Gamma_{0}^{\ast}\left(  E\right)  =\Gamma_{0}^{\ast}\left(
-E\right)  \mathcal{C}\hat{\varkappa}, \label{representation200}%
\end{equation}
then using (\ref{sym1}) and (\ref{sym7}) we have $\mathcal{C}\beta\Gamma
_{0}\left(  E\right)  =\mathcal{C}\hat{\varkappa}\Gamma_{0}\left(  E\right)
e^{-i\phi}\mathcal{P}=\Gamma_{0}\left(  -E\right)  \mathcal{C}e^{-i\phi
}\mathcal{P}$, and $\mathcal{C}e^{-i\phi}\mathcal{P}\Gamma_{0}^{\ast}\left(
E\right)  =\mathcal{C}\Gamma_{0}^{\ast}\left(  E\right)  \beta\hat{\varkappa
}=\Gamma_{0}^{\ast}\left(  -E\right)  \mathcal{C}\beta$. \ The last two
equalities imply $\mathcal{C}\beta S\left(  E\right)  =\mathcal{C}\beta
\Gamma_{0}\left(  E\right)  \mathbf{S}\Gamma_{0}^{\ast}\left(  E\right)
=\Gamma_{0}\left(  -E\right)  \mathbf{S}\Gamma_{0}^{\ast}\left(  -E\right)
\mathcal{C}\beta=S\left(  -E\right)  \mathcal{C}\beta.$ Thus, we obtain the
following relation for the kernel $s\left(  \omega,\theta;E\right)  $ of the
scattering matrix $S\left(  E\right)  :$%
\begin{equation}
\overline{s\left(  \omega,\theta;E\right)  }=\alpha_{2}s\left(  \omega
,\theta;-E\right)  \alpha_{2}. \label{sym5}%
\end{equation}

Another symmetry of the free Dirac operator is the time-reversal
transformation $\mathcal{T}=-i\left(  \alpha_{1}\alpha_{3}\right)  \mathbf{C}%
$. Note that
\begin{equation}
\left.  \mathcal{T}\left(  -i\alpha\nabla+m\beta+\alpha A+V\right)  =-i\left(
\alpha_{1}\alpha_{3}\right)  \left(  i\overline{\alpha}\nabla+m\beta
+\overline{\alpha}A+V\right)  \mathbf{C}=\left(  -i\alpha\nabla+m\beta-\alpha
A+V\right)  \mathcal{T}\mathbf{.}\right.  \label{representation104}%
\end{equation}
If $A=0,$ then relation (\ref{representation104}) implies that $\mathcal{T}%
H=H\mathcal{T}$ and $\mathcal{T}e^{itH}=e^{-itH}\mathcal{T}$. Thus, we have
the following relations
\begin{equation}
\mathcal{T}W_{\pm}=W_{\mp}\mathcal{T}\text{ and }\mathcal{T}\mathbf{S=S}%
^{\ast}\mathcal{T}. \label{representation105}%
\end{equation}
Noting that%
\begin{equation}
\mathcal{T}\hat{\varkappa}\Gamma_{0}\left(  E\right)  =\Gamma_{0}\left(
E\right)  \mathcal{T}, \label{sym2}%
\end{equation}
and%
\begin{equation}
\mathcal{T}\Gamma_{0}^{\ast}\left(  E\right)  =\Gamma_{0}^{\ast}\left(
E\right)  \mathcal{T}\hat{\varkappa} \label{sym9}%
\end{equation}
we obtain $\mathcal{T}\hat{\varkappa}S\left(  E\right)  =\mathcal{T}%
\hat{\varkappa}\Gamma_{0}\left(  E\right)  \mathbf{S}\Gamma_{0}^{\ast}\left(
E\right)  =\Gamma_{0}\left(  E\right)  \mathbf{S}^{\ast}\Gamma_{0}^{\ast
}\left(  E\right)  \mathcal{T}\hat{\varkappa}=S\left(  E\right)  ^{\ast
}\mathcal{T}\hat{\varkappa}.$ The last equality for the scattering matrix
$S\left(  E\right)  $ leads to the following symmetry relation for the kernel
$s\left(  \omega,\theta;E\right)  :$%
\begin{equation}
\left(  \alpha_{1}\alpha_{3}\right)  \overline{s\left(  \omega,\theta
;E\right)  }=\left(  s\left(  -\theta,-\omega;E\right)  \right)  ^{\ast
}\left(  \alpha_{1}\alpha_{3}\right)  . \label{representation106}%
\end{equation}

If $A\neq0$, then relation (\ref{representation105}) is not satisfied. In this
case, in addition to $\mathcal{T}$, we need to apply some other transformation
to $H$ to get a relation similar to (\ref{representation105}). Note that the
parity transformation $\mathcal{P}$ changes the sign of the magnetic potential
$A.$ Therefore, in case of even potentials $\mathbf{V}$ we have $\mathcal{TP}%
H=H\mathcal{TP}$, and hence, $\mathcal{TP}W_{\pm}=W_{\mp}\mathcal{TP}$ and
$\mathcal{TP}\mathbf{S=S}^{\ast}\mathcal{TP}.$ Using relations (\ref{sym1})
and (\ref{sym2}) we get $\mathcal{T}\beta\Gamma_{0}\left(  E\right)
=\mathcal{T}\hat{\varkappa}\Gamma_{0}\left(  E\right)  e^{-i\phi}%
\mathcal{P}=\Gamma_{0}\left(  E\right)  \mathcal{T}e^{-i\phi}\mathcal{P},$
and, by using (\ref{sym7}) and (\ref{sym9}) we obtain $\mathcal{T}e^{-i\phi
}\mathcal{P}\Gamma_{0}^{\ast}\left(  E\right)  =\mathcal{T}\Gamma_{0}^{\ast
}\left(  E\right)  \beta\hat{\varkappa}=\Gamma_{0}^{\ast}\left(  E\right)
\mathcal{T}\beta.$ From the above equalities we have $\mathcal{T}\beta
S\left(  E\right)  =\mathcal{T}\beta\Gamma_{0}\left(  E\right)  \mathbf{S}%
\Gamma_{0}^{\ast}\left(  E\right)  =\Gamma_{0}\left(  E\right)  \mathbf{S}%
^{\ast}\Gamma_{0}^{\ast}\left(  E\right)  \mathcal{T}\beta=S\left(  E\right)
^{\ast}\mathcal{T}\beta,$ and, thus,%
\begin{equation}
\left(  \alpha_{1}\alpha_{3}\beta\right)  \overline{s\left(  \omega
,\theta;E\right)  }=\left(  s\left(  \theta,\omega;E\right)  \right)  ^{\ast
}\left(  \alpha_{1}\alpha_{3}\beta\right)  . \label{sym6}%
\end{equation}

Let us consider the case when the electric potential $V=0$ and the magnetic
potential $A$ is a general function of $x$. As the charge-conjugation
transformation changes the sign of the magnetic potential $A$, we get a
relation, similar to (\ref{representation105}) for the following
transformation $\Lambda=\mathcal{CT}.$ As $\Lambda\left(  iH\right)
=-iH\Lambda,$ then $\Lambda e^{itH}=e^{-itH}\Lambda,$ which implies that
$\Lambda W_{\pm}=W_{\mp}\Lambda$ and $\Lambda\mathbf{S=S}^{\ast}\Lambda.$ From
relations (\ref{representation199}) and (\ref{sym2}) it follows that
$\Lambda\Gamma_{0}\left(  E\right)  =\Gamma_{0}\left(  -E\right)  \Lambda
.\ $Moreover, using equalities (\ref{representation200}) and (\ref{sym9}) we
get $\Lambda\Gamma_{0}^{\ast}\left(  E\right)  =\Gamma_{0}^{\ast}\left(
-E\right)  \Lambda.$ Therefore, we obtain $\Lambda S\left(  E\right)
=\Lambda\Gamma_{0}\left(  E\right)  \mathbf{S}\Gamma_{0}^{\ast}\left(
E\right)  =\Gamma_{0}\left(  -E\right)  \mathbf{S}^{\ast}\Gamma_{0}^{\ast
}\left(  -E\right)  \Lambda=S\left(  -E\right)  ^{\ast}\Lambda.$ Therefore we
obtain the following symmetry relation%
\begin{equation}
s\left(  \omega,\theta;E\right)  =\gamma\left(  s\left(  \theta,\omega
;-E\right)  \right)  ^{\ast}\gamma, \label{representation107}%
\end{equation}
where $\gamma=\alpha_{1}\alpha_{2}\alpha_{3}\beta.$

Finally suppose that $\mathbf{V}\left(  x\right)  $ is an odd function. Then
the transformation $\Pi=\mathcal{CTP}$ satisfies the equality $\Pi H=-H\Pi$
and $\Pi e^{itH}=e^{-itH}\Pi.$ This implies that $\Pi W_{\pm}=W_{\mp}\Pi$ and
$\Pi\mathbf{S=S}^{\ast}\Pi.$ Moreover, as $\Lambda\beta\hat{\varkappa}%
\Gamma_{0}\left(  E\right)  =\Gamma_{0}\left(  -E\right)  e^{-i\phi}\Pi,$ and
$e^{-i\phi}\Pi\Gamma_{0}^{\ast}\left(  E\right)  =\Gamma_{0}^{\ast}\left(
-E\right)  \Lambda\beta\hat{\varkappa},$ then we have $\Lambda\beta
\hat{\varkappa}S\left(  E\right)  =\Lambda\beta\hat{\varkappa}\Gamma
_{0}\left(  E\right)  \mathbf{S}\Gamma_{0}^{\ast}\left(  E\right)  =\Gamma
_{0}\left(  -E\right)  \mathbf{S}^{\ast}\Gamma_{0}^{\ast}\left(  -E\right)
\Lambda\beta\hat{\varkappa}=S\left(  -E\right)  ^{\ast}\Lambda\beta
\hat{\varkappa},$ and%
\begin{equation}
s\left(  \omega,\theta;E\right)  =\beta\gamma\left(  s\left(  -\theta
,-\omega;-E\right)  \right)  ^{\ast}\gamma\beta. \label{sym10}%
\end{equation}

\section{Approximate solutions.}

In this Section we construct approximate generalized eigenfunctions for the
Dirac equation. For the Schr\"{o}dinger equation with short-range potentials,
the approximate solutions are given by $u\left(  x,\xi\right)
=e^{i\left\langle x,\xi\right\rangle }+e^{i\left\langle x,\xi\right\rangle
}a\left(  x,\xi\right)  ,$ where $a$ solves the \textquotedblleft
transport\textquotedblright\ equation (see \cite{27}). In the case of the
Schr\"{o}dinger equation with long-range potentials, the approximate solutions
are of the form $u\left(  x,\xi\right)  =e^{i\left\langle x,\xi\right\rangle
+i\phi}\left(  1+a\left(  x,\xi\right)  \right)  ,$ where $\phi$ solves the
\textquotedblleft eikonal\textquotedblright\ equation and $a\left(
x,\xi\right)  $ is the solution of the \textquotedblleft
transport\textquotedblright\ equation (\cite{30}).

For the Dirac equation with short-range potentials it is not enough to
consider only the \textquotedblleft transport\textquotedblright\ equation, in
order to obtain the desired estimates. Thus, we need to consider the
\textquotedblleft eikonal\textquotedblright\ equation too. It also turns out
that we need to decompose the \textquotedblleft transport\textquotedblright%
\ equation in two equations, one for the positive energies and another for the
negative energies, to obtain a smoothness and high-energy expansion of the
generalized eigenfunctions for the Dirac equation.

For an arbitrary $\xi\in\mathbb{R}^{3}$ let us consider the Dirac equation
\begin{equation}
Hu=\left(  \alpha\left(  -i\triangledown+A\right)  +m\beta+V\right)
u=Eu,\text{ }E=\pm\sqrt{\xi^{2}+m^{2}}, \label{eig1}%
\end{equation}
where $A=\left(  A_{1},A_{2},A_{3}\right)  $ is a magnetic potential
satisfying the estimate%
\begin{equation}
\left\vert \partial_{x}^{\alpha}A\left(  x\right)  \right\vert \leq
C_{\alpha,\beta}\left(  1+\left\vert x\right\vert \right)  ^{-\rho
_{m}-\left\vert \alpha\right\vert },\text{ }\rho_{m}>1, \label{eig31}%
\end{equation}
for all $\alpha,$ and $V$ is a scalar electric potential, which satisfies, for
all $\alpha,$ the estimate%
\begin{equation}
\left\vert \partial_{x}^{\alpha}V\left(  x\right)  \right\vert \leq
C_{\alpha,\beta}\left(  1+\left\vert x\right\vert \right)  ^{-\rho
_{e}-\left\vert \alpha\right\vert },\text{ }\rho_{e}>1. \label{eig32}%
\end{equation}

\begin{defin}
\label{eig40}Let $\omega=\xi/\left\vert \xi\right\vert ,$ $\hat{x}%
=x/\left\vert x\right\vert $ and $\Xi^{\pm}\left(  E\right)  :=\Xi^{\pm
}\left(  \varepsilon_{0},R;E\right)  \subset\mathbb{R}^{3}\times\mathbb{R}%
^{3}$ be the domain $\Xi^{\pm}\left(  E\right)  =[\left(  x,\xi\right)
\in\mathbb{R}^{3}\times\mathbb{R}^{3}\mid$ $\pm\left(  \operatorname*{sgn}%
E\right)  \left\langle \hat{x},\omega\right\rangle \geq-1+\varepsilon_{0}$ for
$\left\vert x\right\vert \geq R\},$ for some $0<\varepsilon_{0}<1$ and
$0<R<\infty.$
\end{defin}

We aim to construct $4\times4$ matrices $u_{N}^{\pm}\left(  x,\xi;E\right)  $
whose columns are approximate solutions to equation (\ref{eig1}) in such way
that the remainders
\begin{equation}
r_{N}^{\pm}\left(  x,\xi;E\right)  \left.  :=\right.  e^{-i\left\langle
x,\xi\right\rangle }\left(  H-E\right)  u_{N}^{\pm}\left(  x,\xi;E\right)  ,
\label{eig30}%
\end{equation}
satisfy the following estimates
\begin{equation}
\left\vert \partial_{x}^{\alpha}\partial_{\xi}^{\beta}r_{N}^{\pm}\left(
x,\xi;E\right)  \right\vert \leq C_{\alpha,\beta}\left(  1+\left\vert
x\right\vert \right)  ^{-\rho-N-\left\vert \alpha\right\vert }\left\vert
\xi\right\vert ^{-N-\left\vert \beta\right\vert },\text{ }N\geq0,
\label{eig28}%
\end{equation}
for $\rho=\min\{\rho_{e},\rho_{m}\},$ $\left(  x,\xi\right)  \in\Xi^{\pm
}\left(  E\right)  $ and all multi-indices $\alpha$ and $\beta.$ \

It is natural for us to seek the matrices $u_{N}^{\pm}\left(  x,\xi;E\right)
$ as%
\begin{equation}
u_{N}^{\pm}\left(  x,\xi;E\right)  =e^{i\phi^{\pm}\left(  x,\xi;E\right)
}w_{N}^{\pm}\left(  x,\xi;E\right)  , \label{eig2}%
\end{equation}
where $\phi^{\pm}\left(  x,\xi;E\right)  $ is a real-valued function
and$\ w_{N}^{\pm}\left(  x,\xi;E\right)  $ are $4\times4$ matrix-valued
functions. Introducing (\ref{eig2}) into equation (\ref{eig1}) and using
(\ref{eig30}) we get%
\begin{equation}
\left(  \alpha\left(  -i\triangledown+\triangledown\phi^{\pm}+A\right)
+m\beta+V-E\right)  w_{N}^{\pm}=e^{i\left\langle x,\xi\right\rangle
-i\phi^{\pm}}r_{N}^{\pm}. \label{eig3}%
\end{equation}
Let us write $\phi^{\pm}$ as%
\begin{equation}
\phi^{\pm}\left(  x,\xi;E\right)  =\left\langle x,\xi\right\rangle +\Phi^{\pm
}\left(  x,\xi;E\right)  , \label{eig25}%
\end{equation}
where $\Phi^{\pm}$ tends to $0$ as $\left\vert x\right\vert \rightarrow\infty$
for $\left(  x,\xi\right)  \in\Xi^{\pm}\left(  E\right)  .$ Then, (\ref{eig3})
takes the form%
\begin{equation}
\left(  \alpha\cdot\xi+m\beta-E+\alpha\left(  -i\triangledown+\triangledown
\Phi^{\pm}+A\right)  +V\right)  w_{N}^{\pm}=e^{-i\Phi^{\pm}}r_{N}^{\pm}.
\label{eig4}%
\end{equation}
Now let us decompose $w_{N}^{\pm}$ as
\begin{equation}
w_{N}^{\pm}=\left(  w_{1}\right)  _{N}^{\pm}+P_{\omega}\left(  E\right)
\left(  w_{2}\right)  _{N}^{\pm}. \label{eig42}%
\end{equation}
Then, we get $\left(  -2EP_{\omega}\left(  -E\right)  +\alpha\left(
-i\triangledown+\triangledown\Phi^{\pm}+A\right)  +V\right)  \left(
w_{1}\right)  _{N}^{\pm}+\left(  \alpha\left(  -i\triangledown+\triangledown
\Phi^{\pm}+A\right)  +V\right)  P_{\omega}\left(  E\right)  \left(
w_{2}\right)  _{N}^{\pm}=e^{-i\Phi^{\pm}}r_{N}^{\pm},$ where we used that
$\alpha\cdot\xi+m\beta-E=-2EP_{\omega}\left(  -E\right)  .$ Using the algebra
of the matrices $\alpha_{j}$ we get the equality $\alpha\left(
-i\triangledown+\triangledown\Phi^{\pm}+A\right)  \left(  \alpha\cdot
\xi\right)  $ $=2\left\langle \xi,\left(  -i\triangledown+\triangledown
\Phi^{\pm}+A\right)  \right\rangle -\left(  \alpha\cdot\xi\right)
\alpha\left(  -i\triangledown+\triangledown\Phi^{\pm}+A\right)  .$ This
relation and equality $P_{\omega}^{2}\left(  E\right)  =P_{\omega}\left(
E\right)  $ imply
\begin{equation}
\left.
\begin{array}
[c]{c}%
\left(  -2EP_{\omega}\left(  -E\right)  +\alpha\left(  -i\triangledown
+\triangledown\Phi^{\pm}+A\right)  +V\right)  \left(  w_{1}\right)  _{N}^{\pm
}+P_{\omega}\left(  -E\right)  \alpha\left(  -i\triangledown+\triangledown
\Phi^{\pm}+A\right)  P_{\omega}\left(  E\right)  \left(  w_{2}\right)
_{N}^{\pm}\\
+\frac{1}{E}\left\langle \xi,\left(  -i\triangledown+\triangledown\Phi^{\pm
}+A\right)  \right\rangle P_{\omega}\left(  E\right)  \left(  w_{2}\right)
_{N}^{\pm}+VP_{\omega}\left(  E\right)  \left(  w_{2}\right)  _{N}^{\pm
}=e^{-i\Phi^{\pm}}r_{N}^{\pm}.
\end{array}
\right.  \label{eig5}%
\end{equation}

Let the functions $\Phi^{\pm}$ satisfy the \textquotedblleft
eikonal\textquotedblright\ equation%
\begin{equation}
\left\langle \omega,\triangledown\Phi^{\pm}+A\right\rangle +\frac
{E}{\left\vert \xi\right\vert }V=0. \label{eig6}%
\end{equation}
Then, we need that the functions $\left(  w_{1}\right)  _{N}^{\pm}$ and
$\left(  w_{2}\right)  _{N}^{\pm}$ are approximate solutions for the
\textquotedblleft transport\textquotedblright\ equation%
\begin{equation}
\left.
\begin{array}
[c]{c}%
\left(  -2EP_{\omega}\left(  -E\right)  +\alpha\left(  -i\triangledown
+\triangledown\Phi^{\pm}+A\right)  +V\right)  \left(  w_{1}\right)  _{N}^{\pm
}\\
+P_{\omega}\left(  -E\right)  \alpha\left(  -i\triangledown+\triangledown
\Phi^{\pm}+A\right)  P_{\omega}\left(  E\right)  \left(  w_{2}\right)
_{N}^{\pm}+\frac{1}{E}\left\langle \xi,\left(  -i\triangledown\right)
\right\rangle P_{\omega}\left(  E\right)  \left(  w_{2}\right)  _{N}^{\pm
}=e^{-i\Phi^{\pm}}r_{N}^{\pm}.
\end{array}
\right.  \label{eig7}%
\end{equation}
We will search the functions $\left(  w_{1}\right)  _{N}^{\pm}$ and $\left(
w_{2}\right)  _{N}^{\pm}~$as
\begin{equation}
\left(  w_{1}\right)  _{N}^{\pm}=\sum_{j=1}^{N}\frac{1}{\left\vert
\xi\right\vert ^{j}}b_{j}^{\pm}\left(  x,\xi;E\right)  \label{eig8}%
\end{equation}
and%
\begin{equation}
\left(  w_{2}\right)  _{N}^{\pm}=\sum_{j=0}^{N}\frac{1}{\left\vert
\xi\right\vert ^{j}}c_{j}^{\pm}\left(  x,\xi;E\right)  . \label{eig9}%
\end{equation}
We want that the functions $w_{N}^{\pm}$ in relation (\ref{eig2}) tend to
$P_{\omega}\left(  E\right)  ,$ as $\left\vert x\right\vert \rightarrow
\infty.$ Therefore we set $b_{0}^{\pm}=0$ and $c_{0}^{\pm}=I$. Plugging
(\ref{eig8}) and (\ref{eig9}) in (\ref{eig7}) and multiplying the resulting
equation on the left-hand side by $P_{\omega}\left(  -E\right)  $ we get%
\begin{equation}
\left.
\begin{array}
[c]{c}%
\sum_{j=1}^{N}\frac{1}{\left\vert \xi\right\vert ^{j}}\left(  -2EP_{\omega
}\left(  -E\right)  +P_{\omega}\left(  -E\right)  \left(  \alpha\left(
-i\triangledown+\triangledown\Phi^{\pm}+A\right)  +V\right)  \right)
b_{j}^{\pm}\\
+\sum_{j=0}^{N}\frac{1}{\left\vert \xi\right\vert ^{j}}P_{\omega}\left(
-E\right)  \alpha\left(  -i\triangledown+\triangledown\Phi^{\pm}+A\right)
P_{\omega}\left(  E\right)  c_{j}^{\pm}=e^{-i\Phi^{\pm}}P_{\omega}\left(
-E\right)  r_{N}^{\pm},
\end{array}
\right.  \label{eig11}%
\end{equation}
and by multiplying by $P_{\omega}\left(  E\right)  $ we obtain%
\begin{equation}
\left.  \sum_{j=1}^{N}\frac{1}{\left\vert \xi\right\vert ^{j}}\left(  \left(
P_{\omega}\left(  E\right)  \left(  \alpha\left(  -i\triangledown
+\triangledown\Phi^{\pm}+A\right)  +V\right)  \right)  b_{j}^{\pm}%
+\frac{\left\vert \xi\right\vert }{E}\left\langle \omega,\left(
-i\triangledown\right)  \right\rangle P_{\omega}\left(  E\right)  c_{j}^{\pm
}\right)  =e^{-i\Phi^{\pm}}P_{\omega}\left(  E\right)  r_{N}^{\pm}.\right.
\label{eig10}%
\end{equation}
In order to get the desired estimates for $r_{N}^{\pm}$ we need that the terms
in (\ref{eig11}) and (\ref{eig10}), which contain powers of $\frac
{1}{\left\vert \xi\right\vert }$ smaller than $N,$ are equal to $0.$ Then,
comparing the terms of the same power of $\frac{1}{\left\vert \xi\right\vert
}$ in (\ref{eig11}) and (\ref{eig10}) we obtain the following equations ($E$
behaves like $\left(  \operatorname*{sgn}E\right)  \left\vert \xi\right\vert $
for large $\left\vert \xi\right\vert $)
\begin{equation}
\left.  b_{j+1}^{\pm}\left(  x,\xi;E\right)  =\frac{\left\vert \xi\right\vert
}{2E}P_{\omega}\left(  -E\right)  \left(  \alpha\left(  -i\triangledown
+\triangledown\Phi^{\pm}+A\right)  +V\right)  b_{j}^{\pm}+\frac{\left\vert
\xi\right\vert }{2E}P_{\omega}\left(  -E\right)  \alpha\left(  -i\triangledown
+\triangledown\Phi^{\pm}+A\right)  P_{\omega}\left(  E\right)  c_{j}^{\pm
},\right.  \label{eig12}%
\end{equation}
for $0\leq j\leq N-1$ and%
\begin{equation}
\left\langle \omega,\triangledown c_{j}^{\pm}\right\rangle =-i\frac
{E}{\left\vert \xi\right\vert }P_{\omega}\left(  E\right)  \left(
\alpha\left(  -i\triangledown+\triangledown\Phi^{\pm}+A\right)  +V\right)
b_{j}^{\pm}, \label{eig37}%
\end{equation}
for $0\leq j\leq N.$ It follows from (\ref{eig12}) that $b_{j}^{\pm}%
=P_{\omega}\left(  -E\right)  b_{j}^{\pm}.$ Then, the term $P_{\omega}\left(
E\right)  Vb_{j}^{\pm}$ in equation (\ref{eig37}) is equals to zero. Thus,
$c_{j}^{\pm}$ satisfies equation%
\begin{equation}
\left\langle \omega,\triangledown c_{j}^{\pm}\right\rangle =-i\frac
{E}{\left\vert \xi\right\vert }P_{\omega}\left(  E\right)  \left(
\alpha\left(  -i\triangledown+\triangledown\Phi^{\pm}+A\right)  \right)
b_{j}^{\pm}. \label{eig13}%
\end{equation}

Our problem is reduced now to solve equations (\ref{eig6}) and (\ref{eig13}).
Both of the equations are of the form%
\begin{equation}
\left\langle \omega,\triangledown d\right\rangle =F\left(  x,\xi;E\right)  ,
\label{eig39}%
\end{equation}
where $d$ and $F$ are either scalars or matrices. A simple substitution shows
that the functions
\begin{equation}
\Phi^{\pm}\left(  x,\xi;E\right)  =\left\{  \left.
\begin{array}
[c]{c}%
\pm\int\limits_{0}^{\infty}\left(  \frac{\left\vert E\right\vert }{\left\vert
\xi\right\vert }V\left(  x\pm t\omega\right)  +\left\langle \omega,A\left(
x\pm t\omega\right)  \right\rangle \right)  dt,\text{ }E>m,\\
\pm\int\limits_{0}^{\infty}\left(  \frac{\left\vert E\right\vert }{\left\vert
\xi\right\vert }V\left(  x\mp t\omega\right)  -\left\langle \omega,A\left(
x\mp t\omega\right)  \right\rangle \right)  dt,\text{ }-E>m,
\end{array}
\right.  \right.  \label{eig15}%
\end{equation}
formally satisfy equation (\ref{eig6}) and the matrices%
\begin{equation}
c_{j}^{\pm}\left(  x,\xi;E\right)  =\left\{  \left.
\begin{array}
[c]{c}%
\pm\int\limits_{0}^{\infty}F_{j}^{\pm}\left(  x\pm t\omega,\xi;E\right)
dt,\text{ }E>m,\\
\pm\int\limits_{0}^{\infty}F_{j}^{\pm}\left(  x\mp t\omega,\xi;E\right)
dt,\text{ }-E>m,
\end{array}
\right.  \right.  \label{eig16}%
\end{equation}
for $j\geq1,$ where%
\begin{equation}
F_{j}^{\pm}\left(  x,\xi;E\right)  =i\frac{\left\vert E\right\vert
}{\left\vert \xi\right\vert }P_{\omega}\left(  E\right)  \left(  \alpha\left(
-i\triangledown+\triangledown\Phi^{\pm}\left(  x,\xi;E\right)  +A\right)
\right)  b_{j}^{\pm}\left(  x,\xi;E\right)  , \label{eig41}%
\end{equation}
solve, at least formally, equation (\ref{eig13}).

Note that relations (\ref{eig12}) and (\ref{eig16}) imply, by induction that
\begin{equation}
b_{j}^{\pm}P_{\omega}\left(  E\right)  =b_{j}^{\pm}\text{ and }c_{j}^{\pm
}P_{\omega}\left(  E\right)  =c_{j}^{\pm},\text{ }j\geq1. \label{eig35}%
\end{equation}

We need the following result to give a precise sense to expressions
(\ref{eig15}) and (\ref{eig16}), and to get the desired estimates for the
functions $\Phi^{\pm},$ $b_{j}^{\pm}$ and $c_{j}^{\pm}$ (see Lemma 2.1,
\cite{30})

\begin{lemma}
\label{eig17}Suppose that the function (or matrix) $F$ satisfies the estimate
\begin{equation}
\left\vert \partial_{x}^{\alpha}\partial_{\xi}^{\beta}F\left(  x,\xi;E\right)
\right\vert \leq C_{\alpha,\beta}\left(  1+\left\vert x\right\vert \right)
^{-\rho-\left\vert \alpha\right\vert }\left\vert \xi\right\vert ^{-\left\vert
\beta\right\vert }\text{,} \label{eig34}%
\end{equation}
for $\left(  x,\xi\right)  \in\Xi^{\pm}\left(  E\right)  $ and some $\rho>1.$
Then the scalar (or a matrix-valued) functions $d^{\pm}\left(  x,\xi;E\right)
=\pm\int\limits_{0}^{\infty}F\left(  x\pm t\omega,\xi;E\right)  dt$ satisfy
equation (\ref{eig39}) and the estimates $\left\vert \partial_{x}^{\alpha
}\partial_{\xi}^{\beta}d^{\pm}\left(  x,\xi;E\right)  \right\vert \leq
C_{\alpha,\beta}\left(  1+\left\vert x\right\vert \right)  ^{-\left(
\rho-1\right)  -\left\vert \alpha\right\vert }\left\vert \xi\right\vert
^{-\left\vert \beta\right\vert },$ on $\Xi^{\pm}\left(  E\right)  $ for all
$\alpha$ and $\beta.$
\end{lemma}

It follows from Lemma \ref{eig17} that under assumptions (\ref{eig31}) and
(\ref{eig32}) the phase functions $\Phi^{\pm},$ defined by relation
(\ref{eig15}) are solutions to equation (\ref{eig6}) and satisfy on $\Xi^{\pm
}\left(  E\right)  $ the estimates%
\begin{equation}
\left\vert \partial_{x}^{\alpha}\partial_{\xi}^{\beta}\Phi^{\pm}\left(
x,\xi;E\right)  \right\vert \leq C_{\alpha,\beta}\left(  1+\left\vert
x\right\vert \right)  ^{-\left(  \rho-1\right)  -\left\vert \alpha\right\vert
}\left\vert \xi\right\vert ^{-\left\vert \beta\right\vert },\text{ \ }%
\rho=\min\{\rho_{e},\rho_{m}\}. \label{eig22}%
\end{equation}

Moreover, we obtain from Lemma \ref{eig17} the following

\begin{proposition}
\label{eig18}Suppose that the magnetic potential $A$ and the electric
potential $V$ satisfy the estimates (\ref{eig31}) and (\ref{eig32}),
respectively. Then, $c_{j}^{\pm}$ defined by (\ref{eig16}) solve equation
(\ref{eig13}) and the following estimates hold on $\Xi^{\pm}\left(  E\right)
$
\begin{equation}
\left\vert \partial_{x}^{\alpha}\partial_{\xi}^{\beta}b_{j}^{\pm}\left(
x,\xi;E\right)  \right\vert \leq C_{\alpha,\beta}\left(  1+\left\vert
x\right\vert \right)  ^{-\rho-j+1-\left\vert \alpha\right\vert }\left\vert
\xi\right\vert ^{-\left\vert \beta\right\vert },\text{ }j\geq1,\text{ }
\label{eig19}%
\end{equation}
and%
\begin{equation}
\left\vert \partial_{x}^{\alpha}\partial_{\xi}^{\beta}c_{j}^{\pm}\left(
x,\xi;E\right)  \right\vert \leq C_{\alpha,\beta}\left(  1+\left\vert
x\right\vert \right)  ^{-\rho-j+1-\left\vert \alpha\right\vert }\left\vert
\xi\right\vert ^{-\left\vert \beta\right\vert },\text{ }j\geq1.\text{ }
\label{eig20}%
\end{equation}

\end{proposition}

\begin{proof}
We argue by induction in $j.$ Set $j=1.$ First note that
\begin{equation}
\left\vert \partial_{\xi}^{\beta}P_{\omega}\left(  \pm E\right)  \right\vert
\leq C\left\vert \xi\right\vert ^{-\left\vert \beta\right\vert }\text{ \ and
}\left\vert \partial_{\xi}^{\beta}\left(  \frac{\left\vert \xi\right\vert }%
{E}P_{\omega}\left(  \pm E\right)  \right)  \right\vert \leq C\left\vert
\xi\right\vert ^{-\left\vert \beta\right\vert }. \label{eig23}%
\end{equation}
Differentiating the relation (\ref{eig12}) we see that $\partial_{x}^{\alpha
}\partial_{\xi}^{\beta}b_{1}^{\pm}$ is a sum of terms of the form
$\partial_{\xi}^{\beta_{1}}\left(  \frac{\left\vert \xi\right\vert }%
{2E}P_{\omega}\left(  -E\right)  \right)  \left(  \partial_{x}^{\alpha
}\partial_{\xi}^{\beta_{2}}\left(  \alpha\left(  \triangledown\Phi^{\pm
}+A\right)  \right)  \right)  $ $\times\left(  \partial_{\xi}^{\beta_{3}%
}P_{\omega}\left(  E\right)  \right)  ,$ with $\sum_{j=1}^{3}\beta_{j}=\beta.$
Then, from inequalities (\ref{eig22}) and (\ref{eig23}) it follows the
estimate (\ref{eig19}) for $b_{1}^{\pm}.$ Using (\ref{eig22}),(\ref{eig23})
and (\ref{eig19}) with $j=1$ we see that $F_{1}^{\pm}$ in equality
(\ref{eig16}) satisfies the estimate $\left\vert \partial_{x}^{\alpha}%
\partial_{\xi}^{\beta}F_{1}^{\pm}\left(  x,\xi;E\right)  \right\vert \leq
C_{\alpha,\beta}\left(  1+\left\vert x\right\vert \right)  ^{-\rho
-1-\left\vert \alpha\right\vert }\left\vert \xi\right\vert ^{-\left\vert
\beta\right\vert }.$ Then, using Lemma \ref{eig17} it follow that $c_{1}^{\pm
}\left(  x,\xi;E\right)  $ solve equation (\ref{eig13}) and satisfy estimates
(\ref{eig20})$.$

By induction assume that (\ref{eig19}) and (\ref{eig20}) are true for $j=n-1.$
Differentiating (\ref{eig12}) it follows that $\partial_{x}^{\alpha}%
\partial_{\xi}^{\beta}b_{n}^{\pm}$ is a sum of terms of the form
\begin{equation}
\left.
\begin{array}
[c]{c}%
\partial_{\xi}^{\beta_{1}^{\prime}}\left(  \frac{\left\vert \xi\right\vert
}{2E}P_{\omega}\left(  -E\right)  \right)  \left(  \partial_{x}^{\alpha
_{1}^{\prime}}\partial_{\xi}^{\beta_{2}^{\prime}}\left(  \alpha\left(
-i\triangledown+\triangledown\Phi^{\pm}+A\right)  +V\right)  \right)
\partial_{x}^{\alpha_{2}^{\prime}}\partial_{\xi}^{\beta_{3}^{\prime}}b_{n-1}\\
+\partial_{\xi}^{\beta_{1}}\left(  \frac{\left\vert \xi\right\vert }%
{2E}P_{\omega}\left(  -E\right)  \right)  \left(  \partial_{x}^{\alpha_{1}%
}\partial_{\xi}^{\beta_{2}}\left(  \alpha\left(  -i\triangledown
+\triangledown\Phi^{\pm}+A\right)  \right)  \right)  \left(  \partial_{\xi
}^{\beta_{3}}P_{\omega}\left(  E\right)  \right)  \partial_{x}^{\alpha_{2}%
}\partial_{\xi}^{\beta_{4}}c_{n-1}%
\end{array}
\right.  \label{eig29}%
\end{equation}
with $\alpha_{1}+\alpha_{2}=\alpha,$ $\alpha_{1}^{\prime}+\alpha_{2}^{\prime
}=\alpha,$ $\sum_{j=1}^{4}\beta_{j}=\beta$ and $\sum_{j=1}^{3}\beta
_{j}^{\prime}=\beta.$ Therefore, from the hypothesis of induction, and
inequalities (\ref{eig22}) and (\ref{eig23}) we get estimates (\ref{eig19})
for $b_{n}^{\pm}.$ Similarly we see that $F_{n}$ satisfy the estimate
$\left\vert \partial_{x}^{\alpha}\partial_{\xi}^{\beta}F_{n}\left(
x,\xi;E\right)  \right\vert \leq C_{\alpha,\beta}\left(  1+\left\vert
x\right\vert \right)  ^{-\rho-n-\left\vert \alpha\right\vert }\left\vert
\xi\right\vert ^{-\left\vert \beta\right\vert }.$ Using Lemma \ref{eig17} we
conclude that $c_{n}^{\pm}\left(  x,\xi;E\right)  $ are solutions to equation
(\ref{eig13}) satisfying estimates (\ref{eig20}).
\end{proof}

Proposition \ref{eig18} implies that the solutions to equation (\ref{eig1})
are given by (\ref{eig2}). Let us define the functions%
\begin{equation}
a_{N}^{\pm}\left(  x,\xi;E\right)  :=e^{i\Phi^{\pm}\left(  x,\xi;E\right)
}w_{N}^{\pm}\left(  x,\xi;E\right)  . \label{eig33}%
\end{equation}
Note that
\begin{equation}
u_{N}^{\pm}\left(  x,\xi;E\right)  =e^{i\left\langle x,\xi\right\rangle }%
a_{N}^{\pm}\left(  x,\xi;E\right)  . \label{eig43}%
\end{equation}
Relation (\ref{eig35}) implies that
\begin{equation}
u_{N}^{\pm}P_{\omega}(E)=u_{N}^{\pm}\text{, }a_{N}^{\pm}P_{\omega}\left(
E\right)  =a_{N}^{\pm}\text{ and }r_{N}^{\pm}P_{\omega}\left(  E\right)
=r_{N}^{\pm}. \label{eig36}%
\end{equation}

We conclude this Section with the following result

\begin{theorem}
Suppose that the magnetic potential $A$ and the electric potential $V$ satisfy
the estimates (\ref{eig31}) and (\ref{eig32}), respectively. Then, for every
$\left(  x,\xi\right)  \in\Xi^{\pm}\left(  E\right)  $ the following estimates
hold%
\begin{equation}
\left\vert \partial_{x}^{\alpha}\partial_{\xi}^{\beta}w_{N}^{\pm}\left(
x,\xi;E\right)  \right\vert \leq C_{\alpha,\beta}\left(  1+\left\vert
x\right\vert \right)  ^{-\left\vert \alpha\right\vert }\left\vert
\xi\right\vert ^{-\left\vert \beta\right\vert }, \label{eig26}%
\end{equation}%
\begin{equation}
\left\vert \partial_{x}^{\alpha}\partial_{\xi}^{\beta}a_{N}^{\pm}\left(
x,\xi;E\right)  \right\vert \leq C_{\alpha,\beta}\left(  1+\left\vert
x\right\vert \right)  ^{-\left\vert \alpha\right\vert }\left\vert
\xi\right\vert ^{-\left\vert \beta\right\vert }. \label{eig27}%
\end{equation}
Moreover the remainder\ $r_{N}^{\pm}\left(  x,\xi;E\right)  $ satisfy estimate
(\ref{eig28}).
\end{theorem}

\begin{proof}
Estimate (\ref{eig26}) is a consequence of Proposition \ref{eig18}.

Note that $\partial_{x}^{\alpha}\partial_{\xi}^{\beta}a_{N}^{\pm}\left(
x,\xi;E\right)  $ is a sum of terms of the form $\left(  \partial_{x}%
^{\alpha_{1}}\partial_{\xi}^{\beta_{1}}e^{i\Phi^{\pm}\left(  x,\xi;E\right)
}\right)  \left(  \partial_{x}^{\alpha_{2}}\partial_{\xi}^{\beta_{2}}%
w_{N}^{\pm}\left(  x,\xi;E\right)  \right)  ,$ with $\alpha_{1}+\alpha
_{2}=\alpha$ and $\beta_{1}+\beta_{2}=\beta.$ Thus, relation (\ref{eig27})
follows from the estimates (\ref{eig22}) and (\ref{eig26}).

To prove the estimate (\ref{eig28}) we observe first that relations
(\ref{eig11}) and (\ref{eig10}) imply $r_{N}^{\pm}\left(  x,\xi;E\right)
=\left\vert \xi\right\vert ^{-N}e^{i\Phi^{\pm}\left(  x,\xi;E\right)
}\{P_{\omega}\left(  -E\right)  $ $\times\left(  \alpha\left(  -i\triangledown
+\triangledown\Phi^{\pm}+A\right)  +V\right)  b_{N}^{\pm}+\left(  P_{\omega
}\left(  -E\right)  \alpha\left(  -i\triangledown+\triangledown\Phi^{\pm
}+A\right)  \right)  P_{\omega}\left(  E\right)  c_{N}^{\pm}\}.$
Differentiating $r_{N}^{\pm}$ we see that $\partial_{x}^{\alpha}\partial_{\xi
}^{\beta}r_{N}^{\pm}$ is a sum of terms similar to (\ref{eig29}). Thus, using
estimates (\ref{eig22}), (\ref{eig23}), (\ref{eig19}) and (\ref{eig20}) we
obtain (\ref{eig28}).
\end{proof}

\section{Estimates for the scattering matrix kernel.}

\subsection{Statement of the results.}

In this Section we study the diagonal singularities and the high-energy
behavior of the scattering amplitude for potentials of the form (\ref{intro4})
satisfying estimates (\ref{eig31}) and (\ref{eig32}). We follow the method of
Yafaev \cite{27} and \cite{30} for the Schr\"{o}dinger operator for this
problem, that consist in defining special identifications $J_{\pm}$ and in
studying the perturbed stationary formula for the scattering matrix. In other
words, we will use the approximate solutions (\ref{eig2}) to construct
explicit functions $s_{N}\left(  \omega,\theta;E\right)  ,$ such that the
difference $s-s_{N}$ is increasingly smoother as $N\rightarrow\infty.$
Moreover, as $N\rightarrow\infty,$ the difference $s-s_{N}$ decays
increasingly faster when $E\rightarrow\infty.$

Let us announce the main result of this Section. First we prepare some results.

For an arbitrary point $\omega_{0}\in\mathbb{S}^{2}$ let $\Pi_{\omega_{0}}$ be
the plane orthogonal to $\omega_{0}$ and
\begin{equation}
\Omega_{\pm}\left(  \omega_{0},\delta\right)  :=\{\omega\in\mathbb{S}^{2}%
|\pm\left\langle \omega,\omega_{0}\right\rangle >\delta\},
\label{representation64}%
\end{equation}
for some $0<\delta<1$. For any $\omega_{j}\in\mathbb{S}^{2}$ let us define
$O_{j}^{\pm}=\Omega_{\pm}\left(  \omega_{j},\sqrt{\left(  1+\delta\right)
/2}\right)  $ and set $O_{j}:=O_{j}^{+}\cup O_{j}^{-}.$ Let us prove the
following result

\begin{lemma}
\label{representation221}Let $j$ and $k$ be such that $O_{j}\cap O_{k}%
\neq\varnothing.$ Then, if $\omega_{jk}\in O_{j}^{\pm}\cap O_{k}^{\pm},$ we
get $O_{j}^{+}\cup O_{k}^{+}\subseteq\Omega_{\pm}\left(  \omega_{jk}%
,\delta\right)  $ and $O_{j}^{-}\cup O_{k}^{-}\subseteq\Omega_{\mp}\left(
\omega_{jk},\delta\right)  .$ Moreover, if $\omega_{jk}\in O_{j}^{\pm}\cap
O_{k}^{\mp},$ we have $O_{j}^{+}\cup O_{k}^{-}\subseteq\Omega_{\pm}\left(
\omega_{jk},\delta\right)  $ and $O_{j}^{-}\cup O_{k}^{+}\subseteq\Omega_{\mp
}\left(  \omega_{jk},\delta\right)  .$
\end{lemma}

\begin{proof}
Let $\omega_{jk}\in O_{j}^{+}$ and $\omega$ be in $O_{j}^{+}.$ If $\omega
_{j}=\omega_{jk}$ or $\omega=\omega_{j}$, then $\omega$ belongs to $\Omega
_{+}\left(  \omega_{jk},\delta\right)  $. Thus, we can suppose that
$\omega_{j}\neq\omega_{jk}$ and $\omega\neq\omega_{j}.$ Let $\theta_{\omega}$
be a unit vector in the plane generated by $\omega$ and $\omega_{j},$ that is
orthogonal to $\omega_{j}:\left\langle \omega_{j},\theta_{\omega}\right\rangle
=0.$ We decompose $\omega$ as $\omega=\left\langle \omega,\omega
_{j}\right\rangle \omega_{j}+\left\langle \omega,\theta_{\omega}\right\rangle
\theta_{\omega}.$ Similarly, for $\omega_{j}\neq\omega_{jk}$ we take a unit
vector $\theta_{\omega_{jk}}$ such that $\omega_{jk}=\left\langle \omega
_{jk},\omega_{j}\right\rangle \omega_{j}+\left\langle \omega_{jk}%
,\theta_{\omega_{jk}}\right\rangle \theta_{\omega_{jk}},$ $\left\langle
\omega_{j},\theta_{\omega_{jk}}\right\rangle =0.$ Then, we have $\left\langle
\omega,\omega_{jk}\right\rangle =\left\langle \omega,\omega_{j}\right\rangle
\left\langle \omega_{jk},\omega_{j}\right\rangle +\left\langle \omega
,\theta_{\omega}\right\rangle \left\langle \omega_{jk},\theta_{\omega_{jk}%
}\right\rangle \left\langle \theta_{\omega},\theta_{\omega_{jk}}\right\rangle
.$ As $\left\vert \left\langle \omega,\omega_{j}\right\rangle \right\vert
>\sqrt{\left(  1+\delta\right)  /2},$ we get $\left\vert \left\langle
\omega,\theta_{\omega}\right\rangle \right\vert <\sqrt{1-\left(
1+\delta\right)  /2}$ and, similarly $\left\vert \left\langle \omega
_{jk},\theta_{\omega_{jk}}\right\rangle \right\vert <\sqrt{1-\left(
1+\delta\right)  /2}.$ Using this inequalities and the estimate $\left\langle
\theta_{\omega},\theta_{\omega_{jk}}\right\rangle \geq-1$, we obtain
$\left\langle \omega,\omega_{jk}\right\rangle >\left(  1+\delta\right)
/2-\left(  1-\left(  1+\delta\right)  /2\right)  =\delta,$ and then,
$\omega\in\Omega_{+}\left(  \omega_{jk},\delta\right)  .$ If $\omega$ belongs
to $O_{j}^{-},$ $\left(  -\omega\right)  \in O_{j}^{+}.$ Thus, $-\omega$
belongs to $\Omega_{+}\left(  \omega_{jk},\delta\right)  $ and hence,
$\omega\in\Omega_{-}\left(  \omega_{jk},\delta\right)  .$ If $\omega_{jk}\in
O_{j}^{-}$, then $-\omega_{jk}\in O_{j}^{+}$, that implies $O_{j}^{\pm
}\subseteq\Omega_{\pm}\left(  -\omega_{jk},\delta\right)  =\Omega_{\mp}\left(
\omega_{jk},\delta\right)  .$ Proceeding similarly for $\omega_{jk},\omega\in
O_{k},$\ we obtain the result of Lemma \ref{representation221}$.$
\end{proof}

Let us take $\{O_{j}\}_{j=1,2,...,n},$ such that for some $n,$ they are an
open cover of $\mathbb{S}^{2}$ with the following property: if $O_{j}\cap
O_{k}=\varnothing,$ then $\operatorname*{dist}\left(  O_{j},O_{k}\right)  >0.$
For every $j$ we take $\chi_{j}\left(  \omega\right)  \in C^{\infty}\left(
\mathbb{S}^{2}\right)  ,$ $\ \chi_{j}\left(  \omega\right)  =\chi_{j}\left(
-\omega\right)  ,$ such that $\sum_{j=1}^{n}\chi_{j}\left(  \omega\right)
=1,$ for any $\omega\in\mathbb{S}^{2}.$

We decompose\ $S\left(  E\right)  $ as the sum
\begin{equation}
S\left(  E\right)  =\sum_{j,k=1}^{n}\chi_{j}S\left(  E\right)  \chi_{k}.
\label{representation170}%
\end{equation}
Then the kernel $s\left(  \omega,\theta;E\right)  ~$of the scattering matrix
$S\left(  E\right)  $ decomposes as the sum
\begin{equation}
s\left(  \omega,\theta;E\right)  =\sum_{j,k=1}^{n}\chi_{j}\left(
\omega\right)  s\left(  \omega,\theta;E\right)  \chi_{k}\left(  \theta\right)
. \label{representation187}%
\end{equation}

Suppose that $O_{j}\cap O_{k}\neq\varnothing$ and let $\omega_{jk}\in
O_{j}\cap O_{k}$ be fixed. We take $\omega_{kj}=\omega_{jk}.$ Let us define
\begin{equation}
\chi_{jk}\left(  \omega,\theta\right)  :=\chi_{jk}^{+}\left(  \omega\right)
\chi_{jk}^{+}\left(  \theta\right)  -\chi_{jk}^{-}\left(  \omega\right)
\chi_{jk}^{-}\left(  \theta\right)  , \label{representation225}%
\end{equation}
where $\chi_{jk}^{\pm}\left(  \omega\right)  \in C^{\infty}\left(
\mathbb{S}^{2}\right)  $ are such that $\chi_{jk}^{\pm}\left(  \omega\right)
=\chi_{kj}^{\pm}\left(  \omega\right)  ,$ $\chi_{jk}^{\pm}\left(
\omega\right)  =1$ for $\omega\in\Omega_{\pm}\left(  \omega_{jk}%
,\delta\right)  $ and $\chi_{jk}^{\pm}\left(  \omega\right)  =0$ for
$\pm\left\langle \omega,\omega_{jk}\right\rangle <0.$ Note that $\chi
_{jk}\left(  \omega,\theta\right)  =\chi_{jk}\left(  \theta,\omega\right)  $
and $\chi_{jk}\left(  \omega,\theta\right)  =-\chi_{jk}\left(  -\omega
,-\theta\right)  $ Moreover, Lemma \ref{representation221} implies the
following properties of the function $\chi_{jk}\left(  \omega,\theta\right)
:$ if $\omega_{jk}\in O_{j}^{+}\cap O_{k}^{+},$ $\chi_{jk}\left(
\omega,\theta\right)  =\pm1$ for $\left(  \omega,\theta\right)  \in O_{j}%
^{\pm}\times O_{k}^{\pm}\subseteq\Omega_{\pm}\left(  \omega_{jk}%
,\delta\right)  \times\Omega_{\pm}\left(  \omega_{jk},\delta\right)  ,$ and
$\chi_{jk}\left(  \omega,\theta\right)  =0$ for $\left(  \omega,\theta\right)
\in O_{j}^{\pm}\times O_{k}^{\mp}\subseteq\Omega_{\pm}\left(  \omega
_{jk},\delta\right)  \times\Omega_{\mp}\left(  \omega_{jk},\delta\right)  $;
if $\omega_{jk}\in O_{j}^{-}\cap O_{k}^{-},$ $\chi_{jk}\left(  \omega
,\theta\right)  =\mp1$ for $\left(  \omega,\theta\right)  \in O_{j}^{\pm
}\times O_{k}^{\pm}\subseteq\Omega_{\mp}\left(  \omega_{jk},\delta\right)
\times\Omega_{\mp}\left(  \omega_{jk},\delta\right)  ,$ and $\chi_{jk}\left(
\omega,\theta\right)  =0$ for $\left(  \omega,\theta\right)  \in O_{j}^{\pm
}\times O_{k}^{\mp}\subseteq\Omega_{\mp}\left(  \omega_{jk},\delta\right)
\times\Omega_{\pm}\left(  \omega_{jk},\delta\right)  ;$ if $\omega_{jk}\in
O_{j}^{+}\cap O_{k}^{-},$ $\chi_{jk}\left(  \omega,\theta\right)  =\pm1$ for
$\left(  \omega,\theta\right)  \in O_{j}^{\pm}\times O_{k}^{\mp}%
\subseteq\Omega_{\pm}\left(  \omega_{jk},\delta\right)  \times\Omega_{\pm
}\left(  \omega_{jk},\delta\right)  ,$ and $\chi_{jk}\left(  \omega
,\theta\right)  =0$ for $\left(  \omega,\theta\right)  \in O_{j}^{\pm}\times
O_{k}^{\pm}\subseteq\Omega_{\pm}\left(  \omega_{jk},\delta\right)
\times\Omega_{\mp}\left(  \omega_{jk},\delta\right)  ;$ if $\omega_{jk}\in
O_{j}^{-}\cap O_{k}^{+},$ $\chi_{jk}\left(  \omega,\theta\right)  =\pm1$ for
$\left(  \omega,\theta\right)  \in O_{j}^{\mp}\times O_{k}^{\pm}%
\subseteq\Omega_{\pm}\left(  \omega_{jk},\delta\right)  \times\Omega_{\pm
}\left(  \omega_{jk},\delta\right)  ,$ and $\chi_{jk}\left(  \omega
,\theta\right)  =0$ for $\left(  \omega,\theta\right)  \in O_{j}^{\pm}\times
O_{k}^{\pm}\subseteq\Omega_{\mp}\left(  \omega_{jk},\delta\right)
\times\Omega_{\pm}\left(  \omega_{jk},\delta\right)  .$ We set
\begin{equation}
s_{N,jk}\left(  \omega,\theta;E\right)  :=\left(  2\pi\right)  ^{-2}%
\upsilon\left(  E\right)  ^{2}\chi_{jk}\left(  \omega,\theta\right)  \chi
_{j}\left(  \omega\right)  \chi_{k}\left(  \theta\right)  \int\limits_{\Pi
_{\omega_{jk}}}e^{i\nu\left(  E\right)  \left\langle y,\theta-\omega
\right\rangle }\mathbf{h}_{N,jk}\left(  y,\omega,\theta;E\right)  dy,
\label{representation26}%
\end{equation}
where
\begin{equation}
\mathbf{h}_{N,jk}\left(  y,\omega,\theta;E\right)  :=\left(
\operatorname*{sgn}E\right)  \left(  a_{N}^{+}\left(  y,\nu\left(  E\right)
\omega;E\right)  \right)  ^{\ast}\left(  \alpha\cdot\omega_{jk}\right)
\left(  a_{N}^{-}\left(  y,\nu\left(  E\right)  \theta;E\right)  \right)  ,
\label{representation27}%
\end{equation}
where $a_{N}^{\pm}\left(  x,\xi;E\right)  $ are the functions (\ref{eig33}).
The integral in (\ref{representation26}) is understood as an oscillatory integral.

\begin{rem}
\label{representation75}\rm{ Note that the operator $S_{\operatorname{pr}%
}\left(  E\right)  $ with kernel $\sum\limits_{O_{j}\cap O_{k}\neq\varnothing
}s_{N,jk}\left(  \omega,\theta;E\right)  $ is a PDO on the sphere $S^{2},$
with amplitude of the class $S^{0}.$ Indeed, let us denote by $\zeta$ the
orthogonal projection of $\omega\in\Omega_{+}\left(  \omega_{jk}%
,\delta\right)  $ on the two-dimensional plane $\Pi_{\omega_{jk}},$ and let
$\Sigma$ be the projection of $\Omega_{+}\left(  \omega_{jk},\delta\right)  $
on $\Pi_{\omega_{jk}}.$ We identify below the points $\omega\in\Omega
_{+}\left(  \omega_{jk},\delta\right)  $ and $\zeta\in\Sigma$ and for any
function $f\left(  \omega\right)  $, $\omega\in\Omega_{+}\left(  \omega
_{jk},\delta\right)  ,$ we define $\tilde{f}\left(  \zeta\right)  :=f\left(
\omega\right)  .$ If $f\left(  \omega,\theta\right)  $ is a function of two
variables $\omega,\theta\in\Omega_{+}\left(  \omega_{jk},\delta\right)  ,$
then we put $\tilde{f}\left(  \zeta,\zeta^{\prime}\right)  :=f\left(
\omega,\theta\right)  .$ We have,%
\begin{equation}
\left.
{\displaystyle\int}
s_{N,jk}\left(  \omega,\theta;E\right)  f\left(  \theta\right)  d\theta
=\left(  2\pi\right)  ^{-2}\upsilon\left(  E\right)  ^{2}%
{\displaystyle\int_{\Pi_{\omega_{jk}}}}
{\displaystyle\int_{\Pi_{\omega_{jk}}}}
e^{i\left\langle y,\zeta^{\prime}-\zeta\right\rangle }\mathbf{\tilde{h}%
}_{N,jk}^{\prime}\left(  y,\zeta,\zeta^{\prime}\right)  \tilde{f}\left(
\zeta^{\prime}\right)  d\zeta^{\prime}dy,\right.  \label{representation74}%
\end{equation}
with $\tilde{h}_{N,jk}^{\prime}\left(  y,\zeta,\zeta^{\prime}\right)
:=\frac{\tilde{\chi}_{jk}\left(  \zeta,\zeta^{\prime}\right)  \tilde{\chi}%
_{j}\left(  \zeta\right)  \tilde{\chi}_{k}\left(  \zeta^{\prime}\right)
}{\left(  1-\left\vert \zeta^{\prime}\right\vert ^{2}\right)  ^{1/2}}\tilde
{h}_{N,jk}\left(  y,\zeta,\zeta^{\prime};E\right)  .$ Note that for
$\omega,\theta\in$ $\Omega_{\pm}\left(  \omega_{jk},\delta\right)  ,$ the
functions $a_{N}^{\pm}$ satisfy (\ref{eig27}), for all $y\in\Pi_{\omega_{jk}%
}.$ Therefore, the amplitude $\tilde{h}_{N,jk}^{\prime}\left(  y,\zeta
,\zeta^{\prime};E\right)  $ of $S_{\operatorname{pr}}\left(  E\right)  $
belongs to the class $S^{0}.$}
\end{rem}

We define also the function $g_{N,jk}^{\operatorname{int}}\left(
\omega,\theta;E\right)  $ as
\begin{equation}
\mathbf{g}_{N,jk}^{\operatorname{int}}\left(  \omega,\theta;E\right)
:=s_{N,jk}\left(  \omega,\theta;E\right)  -s_{00}^{(jk)}\left(  \omega
,\theta;E\right)  , \label{representation181}%
\end{equation}
where%
\begin{equation}
\left.  s_{00}^{(jk)}\left(  \omega,\theta;E\right)  :=\left(
\operatorname*{sgn}E\right)  \left(  2\pi\right)  ^{-2}\upsilon\left(
E\right)  ^{2}\chi_{jk}\left(  \omega,\theta\right)  \chi_{j}\left(
\omega\right)  \chi_{k}\left(  \theta\right)  \int\limits_{\Pi_{\omega_{jk}}%
}e^{i\nu\left(  E\right)  \left\langle y,\theta-\omega\right\rangle }%
P_{\omega}\left(  E\right)  \left(  \alpha\cdot\omega_{jk}\right)  P_{\theta
}\left(  E\right)  dy.\right.  \label{representation182}%
\end{equation}
Proposition \ref{representation89} shows that $\sum\limits_{O_{j}\cap
O_{k}\neq\varnothing}s_{00}^{(jk)}\left(  \omega,\theta;E\right)  $ is a
Dirac-function on $\mathcal{H}\left(  E\right)  .$

We now formulate the results that we will prove in this Section. For
$\omega_{0}\in\mathbb{S}^{2}$ we introduce cut-off function $\Psi_{\pm}\left(
\omega,\theta;\omega_{0}\right)  \in C^{\infty}\left(  \mathbb{S}^{2}%
\times\mathbb{S}^{2}\right)  ,$ supported on $\Omega_{\pm}\left(  \omega
_{0},\delta\right)  \times\Omega_{\pm}\left(  \omega_{0},\delta\right)  .$
Moreover, let $\Psi_{1}\left(  \omega,\theta\right)  \in C^{\infty}\left(
\mathbb{S}^{2}\times\mathbb{S}^{2}\right)  $ be supported on $O\times
O^{\prime},$ where $O,O^{\prime}\subseteq\mathbb{S}^{2}$ are open sets such
that $\overline{O}\cap\overline{O}^{\prime}=\varnothing.$ We define
\begin{equation}
s_{\operatorname{sing}}^{(N)}\left(  \omega,\theta;E;\omega_{0}\right)
:=\pm\Psi_{\pm}\left(  \omega,\theta;\omega_{0}\right)  \left(  2\pi\right)
^{-2}\upsilon\left(  E\right)  ^{2}\int\limits_{\Pi_{\omega_{0}}}%
e^{i\nu\left(  E\right)  \left\langle y,\theta-\omega\right\rangle }%
\mathbf{h}_{N}\left(  y,\omega,\theta;E;\omega_{0}\right)  dy,
\label{representation248}%
\end{equation}
as an oscillatory integral, where
\begin{equation}
\mathbf{h}_{N}\left(  y,\omega,\theta;E;\omega_{0}\right)  :=\left(
\operatorname*{sgn}E\right)  \left(  a_{N}^{+}\left(  y,\nu\left(  E\right)
\omega;E\right)  \right)  ^{\ast}\left(  \alpha\cdot\omega_{0}\right)  \left(
a_{N}^{-}\left(  y,\nu\left(  E\right)  \theta;E\right)  \right)
\label{representation252}%
\end{equation}
and $s_{\operatorname*{reg}}\left(  \omega,\theta;E\right)  :=\Psi_{1}\left(
\omega,\theta\right)  s\left(  \omega,\theta;E\right)  .$

\begin{theorem}
\label{representation179}Let the magnetic potential $A\left(  x\right)  $ and
the electric potential $V\left(  x\right)  $ satisfy the estimates
(\ref{eig31}) and (\ref{eig32}), respectively. For any $p$ and $q,$
$s_{\operatorname*{reg}}\left(  \omega,\theta;E\right)  \ $belongs to the
class $C^{p}\left(  \mathbb{S}^{2}\times\mathbb{S}^{2}\right)  $ and its
$C^{p}$-norm is a $O\left(  E^{-q}\right)  $ function. Moreover, for any $p$
and $q$ there exists $N,$ sufficiently large, such that, $\Psi_{\pm}\left(
\omega,\theta;\omega_{0}\right)  s\left(  \omega,\theta;E\right)
-s_{\operatorname{sing}}^{(N)}\left(  \omega,\theta;E\right)  $ belongs to the
class $C^{p}\left(  \mathbb{S}^{2}\times\mathbb{S}^{2}\right)  $, and
moreover, its $C^{p}-$norm is bounded by $C\left\vert E\right\vert ^{-q}$, as
$\left\vert E\right\vert \rightarrow\infty$. These estimates are uniform in
$\omega_{0},$ in the case when the $C^{p}$-norms of the function\ $\Psi_{\pm
}\left(  \omega,\theta;\omega_{0}\right)  $ are uniformly bounded on
$\omega_{0}\in\mathbb{S}^{2}.$
\end{theorem}

Let us decompose the scattering matrix $S\left(  E\right)  $ in the sum
(\ref{representation170}). Then, taking $\Psi_{\pm}\left(  \omega
,\theta;\omega_{jk}\right)  =\chi_{jk}^{\pm}\left(  \omega\right)  \chi
_{jk}^{\pm}\left(  \theta\right)  \chi_{j}\left(  \omega\right)  \chi
_{k}\left(  \theta\right)  ,$ $\omega_{0}=\omega_{jk},$ in the definition of
$s_{\operatorname{sing}}^{(N)},$ and noting that in this case
$s_{\operatorname{sing}}^{(N)}=s_{N,jk},$ we obtain

\begin{corollary}
If $O_{j}\cap O_{k}=\varnothing,$ then for any $p$ and $q,$ $\chi_{j}\left(
\omega\right)  s\left(  \omega,\theta;E\right)  \chi_{k}\left(  \theta\right)
$ belongs to the class $C^{p}\left(  \mathbb{S}^{2}\times\mathbb{S}%
^{2}\right)  $ and its $C^{p}$-norm is a $O\left(  E^{-q}\right)  $ function.
If $O_{j}\cap O_{k}\neq\varnothing,$ then for any $p$ and $q$ there exists
$N,$ sufficiently large, such that, $\chi_{j}\left(  \omega\right)  s\chi
_{k}\left(  \theta\right)  -s_{N,jk}$ belongs to the class $C^{p}\left(
\mathbb{S}^{2}\times\mathbb{S}^{2}\right)  $, and moreover, its $C^{p}-$norm
is bounded by $C\left\vert E\right\vert ^{-q}$, as $\left\vert E\right\vert
\rightarrow\infty$.
\end{corollary}

\begin{theorem}
\label{representation25}Let the magnetic potential $A\left(  x\right)  $ and
the electric potential $V\left(  x\right)  $ satisfy the estimates
(\ref{eig31}) and (\ref{eig32}), respectively.\ Then, the scattering matrix
$S\left(  E\right)  $ admits the following decomposition%
\begin{equation}
S\left(  E\right)  =I+\mathcal{G+R}\text{,} \label{representation180}%
\end{equation}
where $I$ is the identity in $\mathcal{H}\left(  E\right)  ,$ $\mathcal{G}$ is
an integral operator with kernel $\mathbf{g}_{N}\left(  \omega,\theta
;E\right)  \mathbf{:=}\sum\limits_{O_{j}\cap O_{k}\neq\varnothing}\chi
_{j}\left(  \omega\right)  \mathbf{g}_{N,jk}^{\operatorname{int}}\left(
\omega,\theta;E\right)  \chi_{k}\left(  \theta\right)  ,$ which satisfies the
estimate
\begin{equation}
\left\vert \mathbf{g}_{N}\left(  \omega,\theta;E\right)  \right\vert \leq
C\left\vert \omega-\theta\right\vert ^{-\left(  3-\rho\right)  },\text{
}\omega\neq\theta,\text{ for }\rho=\min\{\rho_{e},\rho_{m}\}<3,
\label{representation83}%
\end{equation}
and it is a continuous function of $\omega$ and $\theta,$ for $\rho>3;$ and
$\mathcal{R}$ is an integral operator with kernel $r_{N}\left(  \omega
,\theta;E\right)  $. For any $p$ and $q$ there exists $N,$ sufficiently large,
such that $r_{N}\left(  \omega,\theta;E\right)  $ belongs to the class
$C^{p}\left(  \mathbb{S}^{2}\times\mathbb{S}^{2}\right)  $, and moreover, its
$C^{p}-$norm is bounded by $C\left\vert E\right\vert ^{-q}$, as $\left\vert
E\right\vert \rightarrow\infty$.
\end{theorem}

\subsection{Symmetries of the approximate kernel of the scattering matrix.}

We note that the approximate kernels $\sum\limits_{O_{j}\cap O_{k}%
\neq\varnothing}s_{N,jk}$ satisfy the symmetry relations (\ref{sym4}),
(\ref{sym5}), (\ref{representation106}), (\ref{sym6}),
\ (\ref{representation107}) and (\ref{sym10}). Below we suppose that $E>m.$
The case $E<-m$ is analogous.

Let us first show that the approximate kernels $\sum\limits_{O_{j}\cap
O_{k}\neq\varnothing}s_{N,jk}$ are invariant under the gauge transformation
$A\rightarrow A+\nabla\psi,$ for $\psi\in C^{\infty}\left(  \mathbb{R}%
^{3}\right)  $ such that $\partial^{\alpha}\psi=O\left(  \left\vert
x\right\vert ^{-\rho-\left\vert \alpha\right\vert }\right)  $ for
$0\leq\left\vert \alpha\right\vert \leq1$ and some $\rho>0$ as $\left\vert
x\right\vert \rightarrow\infty$.\ We emphasize the dependence of different
functions on $A.$ We get from (\ref{eig15}) that%
\begin{equation}
\left.  \Phi^{\pm}\left(  x,\xi;E;A+\nabla\psi\right)  =\pm\left(
\int\limits_{0}^{\infty}\left(  \frac{\left\vert E\right\vert }{\left\vert
\xi\right\vert }V\left(  x\pm t\omega\right)  +\left\langle \omega,A\left(
x\pm t\omega\right)  \right\rangle +\left\langle \omega,\nabla\psi\left(  x\pm
t\omega\right)  \right\rangle \right)  dt\right)  =\Phi^{\pm}\left(
x,\xi;E;A\right)  -\psi\left(  x\right)  .\right.  \label{sym40}%
\end{equation}
Next we show that%
\begin{equation}
b_{j}^{\pm}\left(  x,\xi;E;A+\nabla\psi\right)  =b_{j}^{\pm}\left(
x,\xi;E;A\right)  \text{,} \label{sym41}%
\end{equation}
and%
\begin{equation}
c_{j}^{\pm}\left(  x,\xi;E;A+\nabla\psi\right)  =c_{j}^{\pm}\left(
x,\xi;E;A\right)  \text{,} \label{sym42}%
\end{equation}
for $j\geq1.$ First we prove relations (\ref{sym41}) and (\ref{sym42}) for
$j=1.$ From (\ref{eig12}) and using (\ref{sym40}) we get $b_{1}^{\pm}\left(
x,\xi;E;A+\nabla\psi\right)  =\frac{\left\vert \xi\right\vert }{2E}P_{\omega
}\left(  -E\right)  \alpha\left(  \triangledown\Phi^{\pm}\left(
x,\xi;E;A\right)  -\triangledown\psi\left(  x\right)  +A+\nabla\psi\right)
P_{\omega}\left(  E\right)  =b_{1}^{\pm}\left(  x,\xi;E;A\right)  .$ Then, by
(\ref{eig41}), we have $F_{1}^{\pm}\left(  x,\xi;E;A+\nabla\psi\right)  $
$=F_{1}^{\pm}\left(  x,\xi;E;A\right)  $. Moreover, using the relation
(\ref{eig16}), we get $c_{1}^{\pm}\left(  x,\xi;E;A+\nabla\psi\right)
=\pm\int\limits_{0}^{\infty}F_{1}^{\pm}\left(  x\pm\omega t,\xi;E;A+\nabla
\psi\right)  dt=\pm\int\limits_{0}^{\infty}F_{1}^{\pm}\left(  x\pm\omega
t,\xi;E;A\right)  dt=c_{1}^{\pm}\left(  x,\xi;E;A\right)  .$ By an argument
similar to the case $j=1$ we prove relations (\ref{sym41}) and (\ref{sym42})
by induction for any $j.$ The definition (\ref{representation27}) of
$\mathbf{h}_{N,jk}$ and relations (\ref{eig33}), (\ref{eig42}), (\ref{eig8}),
(\ref{eig9}), (\ref{sym40})-(\ref{sym42}) imply $\mathbf{h}_{N,jk}\left(
y,\omega,\theta;E;A+\nabla\psi\right)  =\mathbf{h}_{N,jk}\left(
y,\omega,\theta;E;A\right)  ,$ and hence, $s_{N,jk}\left(  y,\omega
,\theta;E;A+\nabla\psi\right)  =s_{N,jk}\left(  y,\omega,\theta;E;A\right)  .$

Now we show that relation (\ref{sym4}) for $\sum\limits_{O_{j}\cap O_{k}%
\neq\varnothing}s_{N,jk}$ with an even electric potential $V$ and an odd
magnetic potential $A$ holds. From (\ref{representation26}) and the relation
$\chi_{jk}\left(  \omega,\theta\right)  \chi_{j}\left(  \omega\right)
\chi_{k}\left(  \theta\right)  =-\chi_{jk}\left(  -\omega,-\theta\right)
\chi_{j}\left(  -\omega\right)  \chi_{k}\left(  -\theta\right)  $ we get
$\beta s_{N,jk}\left(  -\omega,-\theta;E\right)  \beta=-\left(  2\pi\right)
^{-2}\upsilon\left(  E\right)  ^{2}\chi_{jk}\left(  \omega,\theta\right)
\chi_{j}\left(  \omega\right)  \chi_{k}\left(  \theta\right)  \int%
_{\Pi_{\omega_{jk}}}e^{i\nu\left(  E\right)  \left\langle y,\theta
-\omega\right\rangle }\beta\mathbf{h}_{N,jk}\left(  -y,-\omega,-\theta
;E\right)  \beta dy.$ Thus, we need to show that
\begin{equation}
\beta\mathbf{h}_{N,jk}\left(  y,\omega,\theta;E\right)  =-\mathbf{h}%
_{N,jk}\left(  -y,-\omega,-\theta;E\right)  \beta. \label{representation120}%
\end{equation}
By relation (\ref{representation27}), equation (\ref{representation120}) is
equivalent to%
\begin{equation}
\left.  \beta\left(  a_{N}^{+}\left(  y,\nu\left(  E\right)  \omega;E\right)
\right)  ^{\ast}\left(  \alpha\cdot\omega_{jk}\right)  \left(  a_{N}%
^{-}\left(  y,\nu\left(  E\right)  \theta;E\right)  \right)  =-\left(
a_{N}^{+}\left(  -y,-\nu\left(  E\right)  \omega;E\right)  \right)  ^{\ast
}\left(  \alpha\cdot\omega_{jk}\right)  \left(  a_{N}^{-}\left(
-y,-\nu\left(  E\right)  \theta;E\right)  \right)  \beta.\right.
\label{sym43}%
\end{equation}
Under the assumptions on $V$ and $A$ we get from (\ref{eig15}) that%
\begin{equation}
\left.
\begin{array}
[c]{c}%
\Phi^{\pm}\left(  x,\xi;E\right)  =\pm\left(  \int\limits_{0}^{\infty}\left(
\frac{\left\vert E\right\vert }{\left\vert \xi\right\vert }V\left(  x\pm
t\omega\right)  +\left\langle \omega,A\left(  x\pm t\omega\right)
\right\rangle \right)  dt\right) \\
=\pm\left(  \int\limits_{0}^{\infty}\left(  \frac{\left\vert E\right\vert
}{\left\vert \xi\right\vert }V\left(  -x\pm t\left(  -\omega\right)  \right)
+\left\langle \left(  -\omega\right)  ,A\left(  -x\pm t\left(  -\omega\right)
\right)  \right\rangle \right)  dt\right)  =\Phi^{\pm}\left(  -x,-\xi
;E\right)  .
\end{array}
\right.  \label{sym32}%
\end{equation}
Let us prove that%
\begin{equation}
\beta b_{j}^{\pm}\left(  x,\xi;E\right)  =b_{j}^{\pm}\left(  -x,-\xi;E\right)
\beta\text{ and }\beta c_{j}^{\pm}\left(  x,\xi;E\right)  =c_{j}^{\pm}\left(
-x,-\xi,E\right)  \beta, \label{sym34}%
\end{equation}
for $j\geq1.$ By (\ref{eig12}), for $j=1,$ we have $\beta b_{1}^{\pm}\left(
x,\xi;E\right)  =\frac{\left\vert \xi\right\vert }{2E}\beta P_{\omega}\left(
-E\right)  \alpha\left(  \triangledown\Phi^{\pm}\left(  x,\xi;E\right)
+A\right)  P_{\omega}\left(  E\right)  .$ Using that $\beta\alpha=-\alpha
\beta$, $-A\left(  x\right)  =A\left(  -x\right)  ,$ relation (\ref{sym32})
and $-\triangledown\Phi^{\pm}\left(  x,\xi;E\right)  =\left(  \triangledown
\Phi^{\pm}\right)  \left(  -x,-\xi;E\right)  $ we obtain%
\begin{equation}
\left.
\begin{array}
[c]{c}%
\beta b_{1}^{\pm}\left(  x,\xi;E\right)  =-\frac{\left\vert \xi\right\vert
}{2E}P_{-\omega}\left(  -E\right)  \alpha\left(  \triangledown\Phi^{\pm
}\left(  x,\xi;E\right)  +A\left(  x\right)  \right)  P_{-\omega}\left(
E\right)  \beta\\
=\frac{\left\vert \xi\right\vert }{2E}P_{-\omega}\left(  -E\right)
\alpha\left(  \left(  \triangledown\Phi^{\pm}\right)  \left(  -x,-\xi
;E\right)  +A\left(  -x\right)  \right)  P_{-\omega}\left(  E\right)
\beta=b_{1}^{\pm}\left(  -x,-\xi;E\right)  \beta.
\end{array}
\right.  \label{sym33}%
\end{equation}
Using (\ref{eig41}), (\ref{sym33}) and equality $i\beta\triangledown
b_{1}^{\pm}\left(  x,\xi;E\right)  =-\left(  i\triangledown b_{1}^{\pm
}\right)  \left(  -x,-\xi;E\right)  \beta,$ we obtain $\beta F_{1}^{\pm
}\left(  x,\xi;E\right)  =-i\frac{\left\vert E\right\vert }{\left\vert
\xi\right\vert }P_{-\omega}\left(  E\right)  \alpha(-i\triangledown
+\triangledown\Phi^{\pm}\left(  x,\xi;E\right)  $ $+A\left(  x\right)  )\beta
b_{1}^{\pm}\left(  x,\xi;E\right)  =F_{1}^{\pm}\left(  -x,-\xi;E\right)
\beta$ and therefore we get $\beta c_{1}^{\pm}\left(  x,\xi;E\right)  =\pm
\int\limits_{0}^{\infty}\beta F_{1}^{\pm}\left(  x\pm\omega t,\xi;E\right)
dt=\pm\int\limits_{0}^{\infty}F_{1}^{\pm}\left(  -x\pm\left(  -\omega\right)
t,-\xi;E\right)  \beta dt=c_{1}^{\pm}\left(  -x,-\xi;E\right)  \beta.$ By an
argument similar to the case $j=1$ we prove relation (\ref{sym34}) by
induction$.$ Using (\ref{sym32}), (\ref{sym34}) and equality $\beta
\alpha=-\alpha\beta$ we obtain (\ref{sym43}) and then, we get relation
(\ref{representation120}).

Let us consider now an odd electric potential $V\left(  x\right)  $ and an
even magnetic potential $A\left(  x\right)  $ and prove equality (\ref{sym5}).
From (\ref{representation26}) we get $\overline{s_{N,jk}\left(  \omega
,\theta;E\right)  }=\left(  2\pi\right)  ^{-2}\upsilon\left(  E\right)
^{2}\chi_{jk}\left(  \omega,\theta\right)  \chi_{j}\left(  \omega\right)
\chi_{k}\left(  \theta\right)  \int\limits_{\Pi_{\omega_{jk}}}e^{i\nu\left(
E\right)  \left\langle y,\theta-\omega\right\rangle }\overline{\mathbf{h}%
_{N,jk}\left(  -y,\omega,\theta;E\right)  }dy.$ Thus, we have to show that
$\alpha_{2}\left(  \overline{\mathbf{h}_{N,jk}\left(  -y,\omega,\theta
;E\right)  }\right)  =\mathbf{h}_{N,jk}\left(  y,\omega,\theta;-E\right)
\alpha_{2}.$ That is%
\begin{equation}
\left.
\begin{array}
[c]{c}%
\alpha_{2}\overline{\left(  a_{N}^{+}\left(  -y,\nu\left(  E\right)
\omega;E\right)  \right)  ^{\ast}\left(  \alpha\cdot\omega_{jk}\right)
\left(  a_{N}^{-}\left(  -y,\nu\left(  E\right)  \theta;E\right)  \right)  }\\
=-\left(  a_{N}^{+}\left(  y,\nu\left(  E\right)  \omega;-E\right)  \right)
^{\ast}\left(  \alpha\cdot\omega_{jk}\right)  \left(  a_{N}^{-}\left(
y,\nu\left(  E\right)  \theta;-E\right)  \right)  \alpha_{2}.
\end{array}
\right.  \label{sym36}%
\end{equation}
Let us show that
\begin{equation}
\alpha_{2}\left(  \overline{a_{N}^{\pm}\left(  -y,\nu\left(  E\right)
\omega;E\right)  }\right)  =a_{N}^{\pm}\left(  y,\nu\left(  E\right)
\omega;-E\right)  \alpha_{2}. \label{sym15}%
\end{equation}
For the phase functions $\Phi^{\pm}$ we have the following equality%
\begin{equation}
\left.
\begin{array}
[c]{c}%
-\Phi^{\pm}\left(  x,\xi;E\right)  =\pm\left(  -\int\limits_{0}^{\infty
}\left(  \frac{\left\vert E\right\vert }{\left\vert \xi\right\vert }V\left(
x\pm t\omega\right)  +\left\langle \omega,A\left(  x\pm t\omega\right)
\right\rangle \right)  dt\right) \\
=\pm\left(  \int\limits_{0}^{\infty}\left(  \frac{\left\vert E\right\vert
}{\left\vert \xi\right\vert }V\left(  -x\mp t\omega\right)  -\left\langle
\omega,A\left(  -x\mp t\omega\right)  \right\rangle \right)  dt\right)
=\Phi^{\pm}\left(  -x,\xi;-E\right)  .
\end{array}
\right.  \label{sym11}%
\end{equation}
Let us prove that%
\begin{equation}
\alpha_{2}\left(  \overline{b_{j}^{\pm}\left(  x,\xi;E\right)  }\right)
=b_{j}^{\pm}\left(  -x,\xi;-E\right)  \alpha_{2} \label{sym12}%
\end{equation}
and%
\begin{equation}
\alpha_{2}\left(  \overline{c_{j}^{\pm}\left(  x,\xi;E\right)  }\right)
=c_{j}^{\pm}\left(  -x,\xi;-E\right)  \alpha_{2}, \label{sym35}%
\end{equation}
for any $j\geq1.$ For $j=1,$ using that $\alpha_{2}\overline{\alpha}%
=-\alpha\alpha_{2}$, $\alpha_{2}\overline{P_{\omega}\left(  E\right)
}=P_{\omega}\left(  -E\right)  \alpha_{2},$ $A\left(  x\right)  =A\left(
-x\right)  $ and (\ref{sym11}) we obtain%
\begin{equation}
\left.
\begin{array}
[c]{c}%
\alpha_{2}\overline{b_{1}^{\pm}\left(  x,\xi;E\right)  }=-\frac{\left\vert
\xi\right\vert }{2E}P_{\omega}\left(  E\right)  \alpha\left(  \triangledown
\Phi^{\pm}\left(  x,\xi;E\right)  +A\left(  x\right)  \right)  P_{\omega
}\left(  -E\right)  \alpha_{2}\\
=-\frac{\left\vert \xi\right\vert }{2E}P_{\omega}\left(  E\right)
\alpha\left(  \left(  \triangledown\Phi^{\pm}\right)  \left(  -x,\xi
;-E\right)  +A\left(  -x\right)  \right)  P_{\omega}\left(  -E\right)
\alpha_{2}=b_{1}^{\pm}\left(  -x,\xi;-E\right)  \alpha_{2}.
\end{array}
\right.  \label{sym14}%
\end{equation}
From (\ref{sym11}) and (\ref{sym14}) we have the equality $\alpha_{2}%
\overline{F_{1}^{\pm}\left(  x,\xi;E\right)  }=i\frac{\left\vert E\right\vert
}{\left\vert \xi\right\vert }P_{\omega}\left(  -E\right)  \left(
\alpha\left(  i\triangledown+\triangledown\Phi^{\pm}\left(  x,\xi;E\right)
+A\left(  x\right)  \right)  \right)  \alpha_{2}\overline{b_{1}^{\pm}\left(
x,\xi;E\right)  }=F_{1}^{\pm}\left(  -x,\xi;-E\right)  \alpha_{2}$ and
$\alpha_{2}\overline{c_{1}^{\pm}\left(  x,\xi;E\right)  }=\pm\int%
\limits_{0}^{\infty}\alpha_{2}\overline{F_{1}^{\pm}\left(  x\pm\omega
t,\xi;E\right)  }dt=\pm\int\limits_{0}^{\infty}F_{1}^{\pm}\left(  -x\mp\omega
t,\xi;-E\right)  \alpha_{2}dt=c_{1}^{\pm}\left(  -x,\xi;-E\right)  \alpha
_{2}.$ Relation (\ref{sym12}) and (\ref{sym35}) for any $j$ can be proved
similarly by induction$.$ From (\ref{sym11}), (\ref{sym12}) and (\ref{sym35})
using the definition of $a_{N}^{\pm}$ (see (\ref{eig33})) and recalling that
$\alpha_{2}\overline{P_{\omega}\left(  E\right)  }=P_{\omega}\left(
-E\right)  \alpha_{2},$ we get (\ref{sym15}). Multiplying equation
(\ref{sym15}) on the left and on the right by $\alpha_{2}$ and taking adjoint,
we prove that
\begin{equation}
\alpha_{2}\left(  \overline{a_{N}^{\pm}\left(  -y,\nu\left(  E\right)
\omega;E\right)  }\right)  ^{\ast}=\left(  a_{N}^{\pm}\left(  y,\nu\left(
E\right)  \omega;-E\right)  \right)  ^{\ast}\alpha_{2}. \label{sym37}%
\end{equation}
By (\ref{sym15}), (\ref{sym37}) and $\alpha_{2}\overline{\alpha}=-\alpha
\alpha_{2}$ we obtain (\ref{sym36}), what proves (\ref{sym5}).

We suppose now that $A=0$ and prove the equality (\ref{representation106}) for
$\sum\limits_{O_{j}\cap O_{k}\neq\varnothing}s_{N,jk}$. Note that
\begin{equation}
\chi_{jk}\left(  \omega,\theta\right)  \chi_{j}\left(  \omega\right)  \chi
_{k}\left(  \theta\right)  +\chi_{kj}\left(  \omega,\theta\right)  \chi
_{k}\left(  \omega\right)  \chi_{j}\left(  \theta\right)  =-(\chi_{jk}\left(
-\theta,-\omega\right)  \chi_{j}\left(  -\theta\right)  \chi_{k}\left(
-\omega\right)  +\chi_{kj}\left(  -\theta,-\omega\right)  \chi_{k}\left(
-\theta\right)  \chi_{j}\left(  -\omega\right)  ). \label{sym45}%
\end{equation}
Thus, it is enough to prove that $\left(  \alpha_{1}\alpha_{3}\right)
\overline{\mathbf{h}_{N,jk}\left(  y,\omega,\theta;E\right)  }=-\left(
\mathbf{h}_{N,jk}\left(  y,-\theta,-\omega;E\right)  \right)  ^{\ast}\left(
\alpha_{1}\alpha_{3}\right)  ,$ or%
\begin{equation}
\left.
\begin{array}
[c]{c}%
\left(  \alpha_{1}\alpha_{3}\right)  \overline{\left(  a_{N}^{+}\left(
y,\nu\left(  E\right)  \omega;E\right)  \right)  ^{\ast}\left(  \alpha
\cdot\omega_{jk}\right)  \left(  a_{N}^{-}\left(  y,\nu\left(  E\right)
\theta;E\right)  \right)  }\\
=-\left(  a_{N}^{-}\left(  y,-\nu\left(  E\right)  \omega;E\right)  \right)
^{\ast}\left(  \alpha\cdot\omega_{jk}\right)  \left(  a_{N}^{+}\left(
y,-\nu\left(  E\right)  \theta;E\right)  \right)  \left(  \alpha_{1}\alpha
_{3}\right)  .
\end{array}
\right.  \label{sym20}%
\end{equation}
First of all note that%
\begin{equation}
\left.  -\Phi^{\pm}\left(  x,\xi;E\right)  =\pm\left(  -\int\limits_{0}%
^{\infty}\frac{\left\vert E\right\vert }{\left\vert \xi\right\vert }V\left(
x\pm t\omega\right)  dt\right)  =\mp\left(  \int\limits_{0}^{\infty}%
\frac{\left\vert E\right\vert }{\left\vert \xi\right\vert }V\left(  x\mp
t\left(  -\omega\right)  \right)  dt\right)  =\Phi^{\mp}\left(  x,-\xi
;E\right)  .\right.  \label{sym16}%
\end{equation}
Let us prove that%
\begin{equation}
\left(  \alpha_{1}\alpha_{3}\right)  \overline{b_{j}^{\pm}\left(
x,\xi;E\right)  }=b_{j}^{\mp}\left(  x,-\xi;E\right)  \left(  \alpha_{1}%
\alpha_{3}\right)  \label{sym18}%
\end{equation}
and that%
\begin{equation}
\left(  \alpha_{1}\alpha_{3}\right)  \overline{c_{j}^{\pm}\left(
x,\xi;E\right)  }=c_{j}^{\mp}\left(  x,-\xi;E\right)  \left(  \alpha_{1}%
\alpha_{3}\right)  , \label{sym19}%
\end{equation}
for $j\geq1.$ Consider the case $j=1.$ We have $\left(  \alpha_{1}\alpha
_{3}\right)  \overline{b_{1}^{\pm}\left(  x,\xi;E\right)  }=\frac{\left\vert
\xi\right\vert }{2E}\left(  \alpha_{1}\alpha_{3}\right)  \overline{P_{\omega
}\left(  -E\right)  \left(  \alpha\cdot\triangledown\Phi^{\pm}\left(
x,\xi;E\right)  \right)  P_{\omega}\left(  E\right)  }.$ Using that $\left(
\alpha_{1}\alpha_{3}\right)  \overline{\alpha}=-\alpha\left(  \alpha_{1}%
\alpha_{3}\right)  $, $\left(  \alpha_{1}\alpha_{3}\right)  \beta=\beta\left(
\alpha_{1}\alpha_{3}\right)  $ and relation (\ref{sym16}) we have%
\begin{equation}
\left.  \left(  \alpha_{1}\alpha_{3}\right)  \overline{b_{1}^{\pm}\left(
x,\xi;E\right)  }=\frac{\left\vert \xi\right\vert }{2E}P_{-\omega}\left(
-E\right)  \alpha\left(  \triangledown\left(  \Phi^{\mp}\left(  x,-\xi
;E\right)  \right)  \right)  P_{-\omega}\left(  E\right)  \left(  \alpha
_{1}\alpha_{3}\right)  =b_{1}^{\mp}\left(  x,-\xi;E\right)  \left(  \alpha
_{1}\alpha_{3}\right)  .\right.  \label{sym17}%
\end{equation}
From (\ref{sym17}) we get $\left(  \alpha_{1}\alpha_{3}\right)  \overline
{F_{1}^{\pm}\left(  x,\xi;E\right)  }=i\frac{\left\vert E\right\vert
}{\left\vert \xi\right\vert }P_{-\omega}\left(  E\right)  \left(
\alpha\left(  i\triangledown+\triangledown\Phi^{\pm}\left(  x,\xi;E\right)
\right)  \right)  \left(  \alpha_{1}\alpha_{3}\right)  \overline{b_{1}^{\pm
}\left(  x,\xi;E\right)  }=-F_{1}^{\mp}\left(  x,-\xi;E\right)  \left(
\alpha_{1}\alpha_{3}\right)  $ and $\left(  \alpha_{1}\alpha_{3}\right)
\overline{c_{1}^{\pm}\left(  x,\xi;E\right)  }=\pm\int\limits_{0}^{\infty
}\left(  \alpha_{1}\alpha_{3}\right)  \overline{F_{1}^{\pm}\left(  x\pm\omega
t,\xi;E\right)  }dt=\mp\int\limits_{0}^{\infty}F_{1}^{\mp}\left(  x\mp\left(
-\omega\right)  t,-\xi;E\right)  \left(  \alpha_{1}\alpha_{3}\right)
dt=c_{1}^{\mp}\left(  x,-\xi;E\right)  \left(  \alpha_{1}\alpha_{3}\right)  .$
Similarly we prove relations (\ref{sym18}) and (\ref{sym19}) for any $j.$ From
(\ref{sym16})-(\ref{sym19}), using the identity $\left(  \alpha_{1}\alpha
_{3}\right)  \overline{\alpha}=-\alpha\left(  \alpha_{1}\alpha_{3}\right)  $
we obtain (\ref{sym20}).

Let us prove (\ref{sym6}). Suppose that $\mathbf{V}$ is even. From
(\ref{representation26}) we get $\overline{s_{N,jk}\left(  \omega
,\theta;E\right)  }=\left(  2\pi\right)  ^{-2}\upsilon\left(  E\right)
^{2}\chi_{jk}\left(  \omega,\theta\right)  \chi_{j}\left(  \omega\right)
\chi_{k}\left(  \theta\right)  \int\limits_{\Pi_{\omega_{jk}}}$ $\times
e^{i\nu\left(  E\right)  \left\langle y,\theta-\omega\right\rangle }%
\overline{\mathbf{h}_{N,jk}\left(  -y,\omega,\theta;E\right)  }dy.$ Moreover,
note that
\begin{equation}
\chi_{jk}\left(  \omega,\theta\right)  \chi_{j}\left(  \omega\right)  \chi
_{k}\left(  \theta\right)  +\chi_{kj}\left(  \omega,\theta\right)  \chi
_{k}\left(  \omega\right)  \chi_{j}\left(  \theta\right)  =\chi_{jk}\left(
\theta,\omega\right)  \chi_{j}\left(  \theta\right)  \chi_{k}\left(
\omega\right)  +\chi_{kj}\left(  \theta,\omega\right)  \chi_{k}\left(
\theta\right)  \chi_{j}\left(  \omega\right)  ). \label{sym44}%
\end{equation}
Thus, in order to prove relation (\ref{sym6}) we need to show that%
\begin{equation}
\left(  \alpha_{1}\alpha_{3}\beta\right)  \overline{\mathbf{h}_{N,jk}\left(
-y,\omega,\theta;E\right)  }=\left(  \mathbf{h}_{N,jk}\left(  y,\theta
,\omega;E\right)  \right)  ^{\ast}\left(  \alpha_{1}\alpha_{3}\beta\right)  ,
\label{sym24}%
\end{equation}
which follows from%
\begin{equation}
\left.
\begin{array}
[c]{c}%
\left(  \alpha_{1}\alpha_{3}\beta\right)  \overline{\left(  a_{N}^{+}\left(
-y,\nu\left(  E\right)  \omega;E\right)  \right)  ^{\ast}\left(  \alpha
\cdot\omega_{jk}\right)  \left(  a_{N}^{-}\left(  -y,\nu\left(  E\right)
\theta;E\right)  \right)  }\\
=\left(  a_{N}^{-}\left(  y,\nu\left(  E\right)  \omega;E\right)  \right)
^{\ast}\left(  \alpha\cdot\omega_{jk}\right)  \left(  a_{N}^{+}\left(
y,\nu\left(  E\right)  \theta;E\right)  \right)  \left(  \alpha_{1}\alpha
_{3}\beta\right)  .
\end{array}
\right.  \label{sym25}%
\end{equation}
Note that%
\begin{equation}
\left.
\begin{array}
[c]{c}%
-\Phi^{\pm}\left(  x,\xi;E\right)  =\pm\left(  -\int\limits_{0}^{\infty
}\left(  \frac{\left\vert E\right\vert }{\left\vert \xi\right\vert }V\left(
x\pm t\omega\right)  +\left\langle \omega,A\left(  x\pm t\omega\right)
\right\rangle \right)  dt\right) \\
=\mp\left(  \int\limits_{0}^{\infty}\left(  \frac{\left\vert E\right\vert
}{\left\vert \xi\right\vert }V\left(  -x\mp t\omega\right)  +\left\langle
\omega,A\left(  -x\mp t\omega\right)  \right\rangle \right)  dt\right)
=\Phi^{\mp}\left(  -x,\xi;E\right)  .
\end{array}
\right.  \label{sym21}%
\end{equation}
For $j=1,$ using that $\left(  \alpha_{1}\alpha_{3}\beta\right)
\overline{\alpha}=\alpha\left(  \alpha_{1}\alpha_{3}\beta\right)  $ and
(\ref{sym21}) we have $\left(  \alpha_{1}\alpha_{3}\beta\right)
\overline{b_{1}^{\pm}\left(  x,\xi;E\right)  }=\frac{\left\vert \xi\right\vert
}{2E}P_{\omega}\left(  -E\right)  \alpha(\left(  \left(  \triangledown
\Phi^{\mp}\right)  \left(  -x,\xi;E\right)  \right)  $ $+A\left(  -x\right)
)P_{\omega}\left(  E\right)  \left(  \alpha_{1}\alpha_{3}\beta\right)
=b_{1}^{\mp}\left(  -x,\xi;E\right)  \left(  \alpha_{1}\alpha_{3}\beta\right)
,$ and $\left(  \alpha_{1}\alpha_{3}\beta\right)  \overline{F_{1}^{\pm}\left(
x,\xi;E\right)  }=-i\frac{\left\vert E\right\vert }{\left\vert \xi\right\vert
}P_{\omega}\left(  E\right)  \left(  \alpha\left(  i\triangledown
+\triangledown\Phi^{\pm}\left(  x,\xi;E\right)  \right)  +A\right)  $
$\times\left(  \alpha_{1}\alpha_{3}\beta\right)  \overline{b_{1}^{\pm}\left(
x,\xi;E\right)  }=-F_{1}^{\mp}\left(  -x,\xi;E\right)  \left(  \alpha
_{1}\alpha_{3}\beta\right)  .$ Therefore, we get $\left(  \alpha_{1}\alpha
_{3}\beta\right)  \overline{c_{1}^{\pm}\left(  x,\xi;E\right)  }=\pm\int%
_{0}^{\infty}\left(  \alpha_{1}\alpha_{3}\beta\right)  \overline{F_{1}^{\pm
}\left(  x\pm\omega t,\xi;E\right)  }$ $=\mp\int_{0}^{\infty}F_{1}^{\mp
}\left(  -x\mp\omega t,\xi;E\right)  \left(  \alpha_{1}\alpha_{3}\beta\right)
dt=c_{1}^{\mp}\left(  -x,\xi;E\right)  \left(  \alpha_{1}\alpha_{3}%
\beta\right)  .$ Then, by induction in $j$ we obtain%
\begin{equation}
\left(  \alpha_{1}\alpha_{3}\beta\right)  \overline{b_{j}^{\pm}\left(
x,\xi;E\right)  }=b_{j}^{\mp}\left(  -x,\xi;E\right)  \left(  \alpha_{1}%
\alpha_{3}\beta\right)  , \label{sym22}%
\end{equation}
and%
\begin{equation}
\left(  \alpha_{1}\alpha_{3}\beta\right)  \overline{c_{j}^{\pm}\left(
x,\xi;E\right)  }=c_{j}^{\mp}\left(  -x,\xi;E\right)  \left(  \alpha_{1}%
\alpha_{3}\beta\right)  , \label{sym23}%
\end{equation}
for any $j\geq1.$ As before, relations (\ref{sym21}), (\ref{sym22}) and
(\ref{sym23}) imply (\ref{sym25}).

Now suppose that $V$ is equal to zero and prove relation
(\ref{representation107}). As relation (\ref{sym44}) holds, we have to show
that $\gamma\mathbf{h}_{N,jk}\left(  y,\omega,\theta;E\right)  =\left(
\mathbf{h}_{N,jk}\left(  y,\theta,\omega;-E\right)  \right)  ^{\ast}\gamma,$
or, which is the same%
\begin{equation}
\left.
\begin{array}
[c]{c}%
\gamma\left(  a_{N}^{+}\left(  y,\nu\left(  E\right)  \omega;E\right)
\right)  ^{\ast}\left(  \alpha\cdot\omega_{jk}\right)  \left(  a_{N}%
^{-}\left(  y,\nu\left(  E\right)  \theta;E\right)  \right) \\
=-\left(  a_{N}^{-}\left(  y,\nu\left(  E\right)  \omega;-E\right)  \right)
^{\ast}\left(  \alpha\cdot\omega_{jk}\right)  \left(  a_{N}^{+}\left(
y,\nu\left(  E\right)  \theta;-E\right)  \right)  \gamma.
\end{array}
\right.  \label{sym27}%
\end{equation}
Noting that $\gamma\alpha=-\alpha\gamma,$ $\gamma\beta=-\beta\gamma$ and%
\begin{equation}
\left.  \Phi^{\pm}\left(  x,\xi;E\right)  =\pm\left(  \int\limits_{0}^{\infty
}\left\langle \omega,A\left(  x\pm t\omega\right)  \right\rangle dt\right)
=\mp\left(  -\int\limits_{0}^{\infty}\left\langle \omega,A\left(  x\pm
t\omega\right)  \right\rangle dt\right)  =\Phi^{\mp}\left(  x,\xi;-E\right)
,\right.  \label{sym26}%
\end{equation}
we have $\gamma b_{1}^{\pm}\left(  x,\xi;E\right)  =-\frac{\left\vert
\xi\right\vert }{2E}P_{\omega}\left(  E\right)  \alpha\left(  \left(
\triangledown\Phi^{\mp}\left(  x,\xi;-E\right)  \right)  +A\right)  P_{\omega
}\left(  -E\right)  \gamma=b_{1}^{\mp}\left(  x,\xi;-E\right)  \gamma.$ It
follows that $\gamma F_{1}^{\pm}\left(  x,\xi;E\right)  =-i\frac{\left\vert
E\right\vert }{\left\vert \xi\right\vert }P_{\omega}\left(  -E\right)  \left(
\alpha\left(  -i\triangledown+\triangledown\Phi^{\pm}\left(  x,\xi;E\right)
\right)  +A\right)  \gamma b_{1}^{\pm}\left(  x,\xi;E\right)  =-F_{1}^{\mp
}\left(  x,\xi;-E\right)  \gamma\ $and $\gamma c_{1}^{\pm}\left(
x,\xi;E\right)  =\pm\int\limits_{0}^{\infty}\gamma F_{1}^{\pm}\left(
x\pm\omega t,\xi;E\right)  dt$ $=\mp\int\limits_{0}^{\infty}F_{1}^{\mp}\left(
x\pm\omega t,\xi;-E\right)  \gamma dt=c_{1}^{\mp}\left(  x,\xi;-E\right)
\gamma.$ By induction in $j$ we get the equalities $\gamma b_{j}^{\pm}\left(
x,\xi;E\right)  =b_{j}^{\mp}\left(  x,\xi;-E\right)  \gamma,$ and $\gamma
c_{j}^{\pm}\left(  x,\xi;E\right)  =c_{j}^{\mp}\left(  x,\xi;-E\right)
\gamma,$ with $j\geq1.$ These two relations, together with (\ref{sym26}) and
identities $\gamma\alpha=-\alpha\gamma,$ $\gamma\beta=-\beta\gamma$ imply
equality (\ref{sym27}).

Finally, we consider an odd function $\mathbf{V.}$ From
(\ref{representation26}) we get $\left(  s_{N,jk}\left(  -\theta
,-\omega;-E\right)  \right)  ^{\ast}=\left(  2\pi\right)  ^{-2}\chi
_{jk}\left(  -\theta,-\omega\right)  \chi_{j}\left(  -\theta\right)  \chi
_{k}\left(  -\omega\right)  $ $\times\int_{\Pi_{\omega_{jk}}}e^{i\nu\left(
E\right)  \left\langle y,\theta-\omega\right\rangle }\left(  \mathbf{h}%
_{N,jk}\left(  -y,-\theta,-\omega;-E\right)  \right)  ^{\ast}dy.$ Then, by
(\ref{sym45}), relation (\ref{sym10}) for $\sum\limits_{O_{j}\cap O_{k}%
\neq\varnothing}s_{N,jk}$ follows from
\begin{equation}
\left.
\begin{array}
[c]{c}%
\gamma\beta\left(  a_{N}^{+}\left(  y,\nu\left(  E\right)  \omega;E\right)
\right)  ^{\ast}\left(  \alpha\cdot\omega_{jk}\right)  \left(  a_{N}%
^{-}\left(  y,\nu\left(  E\right)  \theta;E\right)  \right) \\
=\left(  a_{N}^{-}\left(  -y,-\nu\left(  E\right)  \omega;-E\right)  \right)
^{\ast}\left(  \alpha\cdot\omega_{jk}\right)  \left(  a_{N}^{+}\left(
-y,-\nu\left(  E\right)  \theta;-E\right)  \right)  \gamma\beta.
\end{array}
\right.  \label{sym31}%
\end{equation}
Note that%
\begin{equation}
\left.
\begin{array}
[c]{c}%
\Phi^{\pm}\left(  x,\xi;E\right)  =\pm\left(  \int\limits_{0}^{\infty}\left(
\frac{\left\vert E\right\vert }{\left\vert \xi\right\vert }V\left(  x\pm
t\omega\right)  +\left\langle \omega,A\left(  x\pm t\omega\right)
\right\rangle \right)  dt\right) \\
=\mp\left(  \int\limits_{0}^{\infty}\left(  \frac{\left\vert E\right\vert
}{\left\vert \xi\right\vert }V\left(  -x\pm t\left(  -\omega\right)  \right)
-\left\langle -\omega,A\left(  -x\pm t\left(  -\omega\right)  \right)
\right\rangle \right)  dt\right)  =\Phi^{\mp}\left(  -x,-\xi;-E\right)  .
\end{array}
\right.  \label{sym28}%
\end{equation}
Then, we have%
\begin{equation}
\left.  \gamma\beta b_{1}^{\pm}\left(  x,\xi;E\right)  =-\frac{\left\vert
\xi\right\vert }{2E}P_{-\omega}\left(  E\right)  \alpha\left(  \left(  \left(
\triangledown\Phi^{\mp}\right)  \left(  -x,-\xi;-E\right)  \right)  +A\left(
-x\right)  \right)  P_{-\omega}\left(  -E\right)  \gamma\beta=b_{1}^{\mp
}\left(  -x,-\xi;-E\right)  \gamma\beta.\right.  \label{sym29}%
\end{equation}
Using relations (\ref{sym28}) and (\ref{sym29}) we get $\gamma\beta F_{1}%
^{\pm}\left(  x,\xi;E\right)  =\left(  i\frac{\left\vert E\right\vert
}{\left\vert \xi\right\vert }P_{-\omega}\left(  -E\right)  \left(
\alpha\left(  -i\triangledown+\triangledown\Phi^{\pm}\left(  x,\xi;E\right)
\right)  +A\right)  \right)  \gamma\beta b_{1}^{\pm}\left(  x,\xi;E\right)
=-F_{1}^{\mp}\left(  -x,-\xi;-E\right)  \gamma\beta.$ Thus, we have
$\gamma\beta c_{1}^{\pm}\left(  x,\xi;E\right)  =\pm\int\limits_{0}^{\infty
}\gamma\beta F_{1}^{\pm}\left(  x\pm\omega t,\xi;E\right)  dt$ $=\mp
\int\limits_{0}^{\infty}F_{1}^{\mp}\left(  -x\pm\left(  -\omega\right)
t,-\xi;-E\right)  \gamma\beta dt=c_{1}^{\mp}\left(  -x,-\xi;-E\right)
\gamma\beta.$ By induction in $j$ we obtain%
\begin{equation}
\gamma\beta b_{j}^{\pm}\left(  x,\xi;E\right)  =b_{j}^{\mp}\left(
-x,-\xi;-E\right)  \gamma\beta\text{,} \label{sym30}%
\end{equation}
and%
\begin{equation}
\gamma\beta c_{j}^{\pm}\left(  x,\xi;E\right)  =c_{j}^{\mp}\left(
-x,-\xi;-E\right)  \gamma\beta, \label{sym39}%
\end{equation}
with $j\geq1.$ Using (\ref{sym28}), (\ref{sym30}) and (\ref{sym39}) we get
equality (\ref{sym31}).

\subsection{The identification operators.}

The proofs of Theorems \ref{representation179} and \ref{representation25} are
based in a stationary formula for the scattering matrix $S\left(  E\right)  $.

In the general case, where $H_{0}$ and $H$ are self-adjoint operators in
different Hilbert spaces $\mathcal{H}_{0}$ and $\mathcal{H}$ respectively, for
$\Lambda\subset\sigma_{ac}\left(  H_{0}\right)  ,$ the wave operators are
defined by the relation
\begin{equation}
W_{\pm}\left(  H,H_{0};J;\Lambda\right)  :=s-\lim_{t\rightarrow\pm\infty
}e^{itH}Je^{-itH_{0}}E_{0}\left(  \Lambda\right)  , \label{representation35}%
\end{equation}
where $J$ is a bounded identification operator between the spaces
$\mathcal{H}_{0}$ and $\mathcal{H},$ and $E_{0}\left(  \Lambda\right)  $ is
the resolution of the identity for $H_{0}.$ When $\Lambda=\sigma_{ac}\left(
H_{0}\right)  $ we write $W_{\pm}\left(  H,H_{0};J\right)  $ instead of
$W_{\pm}\left(  H,H_{0};J;\Lambda\right)  .$ In our case, $\sigma_{ac}\left(
H_{0}\right)  =\sigma\left(  H_{0}\right)  ,$ the spaces $\mathcal{H}_{0}$ and
$\mathcal{H}$ coincide and the wave operators $W_{\pm}\left(  H,H_{0}\right)
=W_{\pm}\left(  H,H_{0};I\right)  $ ($I$ is the identity operator) exist and
are complete (\cite{9}, \cite{14}, \cite{33}).\ Thus, there is no need to
consider an identification operator $J.$ However, it is convenient for us to
introduce special identifications $J_{\pm}$ and to study the scattering matrix
$\tilde{S}\left(  E\right)  $ associated to the wave operators $W_{\pm}\left(
H,H_{0};J;\Lambda\right)  .$ We have to construct $J_{\pm}$ in such way that
for a given $E,$ $\tilde{S}\left(  E\right)  =S\left(  E\right)  .$

We begin by defining the identifications $J_{\pm}=J_{\pm}^{\left(  N\right)
}.$ We take $\varepsilon_{0}$ and $R$ as in Definition \ref{eig40}. Let
$\varepsilon>0$ be such that
\begin{equation}
\sqrt{1-\delta^{2}}<\varepsilon<1-\varepsilon_{0}, \label{representation63}%
\end{equation}
where $\delta$ is given in the definition of the sets $\Omega_{\pm}\left(
\omega_{0},\delta\right)  $ (see (\ref{representation64})). Let $\sigma_{+}\in
C^{\infty}\left[  -1,1\right]  ,$ be such that $\sigma_{+}\left(  \tau\right)
=1$ if $\tau\in(-\varepsilon,1]$ and $\sigma_{+}\left(  \tau\right)  =0$ if
$\tau\in\lbrack-1,-1+\varepsilon_{0}].$ We take $\sigma_{-}\left(
\tau\right)  =\sigma_{+}\left(  -\tau\right)  .$ Now let $\eta\in C^{\infty
}\left(  \mathbb{R}^{3}\right)  ,$ $0\leq\eta\leq1,$ be such that $\eta\left(
x\right)  =0$ in a neighborhood of zero and $\eta\left(  x\right)  =1$ for
$\left\vert x\right\vert \geq R$. Let the function $\theta\left(  t\right)
\in C^{\infty}\left(  \mathbb{R}_{+}\right)  $ be equal to zero if $t\leq c$
and equal to $1$ for $t\geq c_{1},$ with some $0<c<c_{1}<\nu\left(  E\right)
$. \ Finally, we define $\zeta_{\pm}^{+}\left(  x,\xi\right)  \left.
:=\right.  \sigma_{\pm}\left(  \eta\left(  x\right)  \left\langle \hat{x}%
,\hat{\xi}\right\rangle \right)  \theta\left(  \left\vert \xi\right\vert
\right)  ,$ and $\zeta_{\pm}^{-}\left(  x,\xi\right)  \left.  :=\right.
\sigma_{\mp}\left(  \eta\left(  x\right)  \left\langle \hat{x},\hat{\xi
}\right\rangle \right)  \theta\left(  \left\vert \xi\right\vert \right)  .$
Note that $\zeta_{\pm}^{+}$ is supported on $\Xi^{\pm}\left(  E\right)  ,$ for
$E>m$ and $\zeta_{\pm}^{-}$ is supported on $\Xi^{\pm}\left(  E\right)  ,$ for
$E<-m.$

We define the identifications $J_{\pm}=J_{\pm}^{\left(  N\right)  }$ as the
PDO's
\begin{equation}
\left(  J_{\pm}f\right)  \left(  x\right)  :=\left(  2\pi\right)  ^{-3/2}%
\int_{\mathbb{R}^{3}}e^{i\left\langle x,\xi\right\rangle }j_{N}^{\pm}\left(
x,\xi\right)  \hat{f}\left(  \xi\right)  d\xi, \label{representation15}%
\end{equation}
where $j_{N}^{\pm}\left(  x,\xi\right)  :=a_{N}^{\pm}\left(  x,\xi;\left\vert
\lambda\left(  \xi\right)  \right\vert \right)  \zeta_{\pm}^{+}\left(
x,\xi\right)  +a_{N}^{\pm}\left(  x,\xi;-\left\vert \lambda\left(  \xi\right)
\right\vert \right)  \zeta_{\pm}^{-}\left(  x,\xi\right)  ,$ $\lambda\left(
\xi\right)  =\lambda\left(  \xi;E\right)  :=\left(  \operatorname*{sgn}%
E\right)  \sqrt{\left\vert \xi\right\vert ^{2}+m^{2}}$ and the functions
$a_{N}^{\pm}\left(  x,\xi;E\right)  $ are given by (\ref{eig33}). As
$a_{N}^{\pm}\left(  x,\xi;E\right)  $ satisfies the estimate (\ref{eig27})
on\ $\Xi^{\pm}\left(  E\right)  $, then $j_{N}^{\pm}\left(  x,\xi\right)
\in\mathit{S}^{0,0}$. Thus, using Proposition \ref{basicnotions24} we see that
$J_{\pm}$ are bounded. It follows from relation (\ref{eig36}) that
\begin{equation}
j_{N}^{\pm}\left(  x,\xi\right)  =a_{N}^{\pm}\left(  x,\xi;\left\vert
\lambda\left(  \xi\right)  \right\vert \right)  P^{+}\left(  \xi\right)
\zeta_{\pm}^{+}\left(  x,\xi\right)  +a_{N}^{\pm}\left(  x,\xi;-\left\vert
\lambda\left(  \xi\right)  \right\vert \right)  P^{-}\left(  \xi\right)
\zeta_{\pm}^{-}\left(  x,\xi\right)  . \label{representation133}%
\end{equation}

\begin{rem}
\rm{ We take $\zeta_{\pm}^{+}$ for the projector $P^{+}\left(
\xi\right)  $ on the positive energies $E$ and $\zeta_{\pm}^{-}$ for the
projector $P^{-}\left(  \xi\right)  $ \ on the negative energies $E$, in order
to assure that $J_{\pm}$ correspond to $W_{\pm}.$ This also explains the
different definitions of $\Phi^{\pm}$ and $c_{j}^{\pm}$ for $E>m$ and $-E>m.$}
\end{rem}

The following result is analogous to Lemma 1.1 of \cite{72} for the
Schr\"{o}dinger operator, that can be proved by using a stationary phase argument

\begin{lemma}
\label{representation56}Let $A_{\pm}^{+}$ and $A_{\pm}^{-}$ be PDO operators
with symbols $a_{\pm}^{+}\left(  x,\xi\right)  ,a_{\pm}^{-}\left(
x,\xi\right)  \in\mathit{S}^{0,0},$ respectively, satisfying
\begin{equation}
a_{\pm}^{+}\left(  x,\xi\right)  =0\text{ if }\pm\left\langle \hat{x},\hat
{\xi}\right\rangle \leq-1+\varepsilon_{0},\text{ and }a_{\pm}^{-}\left(
x,\xi\right)  =0\text{ if }\mp\left\langle \hat{x},\hat{\xi}\right\rangle
\leq-1+\varepsilon_{0},\text{ }\varepsilon_{0}>0, \label{representation58}%
\end{equation}
for $\left\vert x\right\vert \geq R>0,$ and%
\begin{equation}
\text{ }a_{\pm}^{+}\left(  x,\xi\right)  =a_{\pm}^{-}\left(  x,\xi\right)
=0\text{ for }\left\vert \xi\right\vert \leq c, \label{representation122}%
\end{equation}
for some $c>0.$ Moreover, suppose that
\begin{equation}
a_{\pm}^{+}\left(  x,\xi\right)  P^{+}\left(  \xi\right)  =a_{\pm}^{+}\left(
x,\xi\right)  \text{ and }a_{\pm}^{-}\left(  x,\xi\right)  P^{-}\left(
\xi\right)  =a_{\pm}^{-}\left(  x,\xi\right)  . \label{representation59}%
\end{equation}
Then, for any $f\in\mathcal{S}\left(  \mathbb{R}^{3};\mathbb{C}^{4}\right)  $
and any $N$ there is a constant $C_{N,f}$ such that
\begin{equation}
\left.  \left\Vert A_{\pm}e^{-itH_{0}}f\right\Vert \leq C_{N,f}\left(
1+\left\vert t\right\vert \right)  ^{-N},\text{ }\mp t>0,\right.
\label{representation57}%
\end{equation}
where $A_{\pm}$ is either $A_{\pm}^{+}$ or $A_{\pm}^{-}.$
\end{lemma}

Note that the symbols $a_{N}^{\pm}\left(  x,\xi;\left\vert \lambda\left(
\xi\right)  \right\vert \right)  \zeta_{\pm}^{+}\left(  x,\xi\right)  $ and
$a_{N}^{\pm}\left(  x,\xi;-\left\vert \lambda\left(  \xi\right)  \right\vert
\right)  \zeta_{\pm}^{-}\left(  x,\xi\right)  $ satisfy the assumptions of
Lemma \ref{representation56}.

Let us define $\left(  \mathbf{J}_{\pm}f\right)  \left(  x\right)  :=\left(
2\pi\right)  ^{-3/2}\int_{\mathbb{R}^{3}}e^{i\left\langle x,\xi\right\rangle
}\mathbf{j}_{\pm}\left(  x,\xi\right)  \hat{f}\left(  \xi\right)  d\xi,$ with
$\mathbf{j}_{\pm}\left(  x,\xi\right)  :=P^{+}\left(  \xi\right)  \zeta_{\pm
}^{+}\left(  x,\xi\right)  +P^{-}\left(  \xi\right)  \zeta_{\pm}^{-}\left(
x,\xi\right)  .$ Then, we have $\left(  \left(  J_{\pm}-\mathbf{J}_{\pm
}\right)  f\right)  \left(  x\right)  =\left(  2\pi\right)  ^{-3/2}%
\int_{\mathbb{R}^{3}}e^{i\left\langle x,\xi\right\rangle }\tilde{j}_{\pm
}\left(  x,\xi\right)  \hat{f}\left(  \xi\right)  d\xi,$ where $\tilde{j}%
_{\pm}\left(  x,\xi\right)  \in\mathit{S}^{-\left(  \rho-1\right)  ,0}.$ Note
that
\begin{equation}
\lim_{\left\vert t\right\vert \rightarrow\infty}\left\Vert \left(  J_{\pm
}-\mathbf{J}_{\pm}\right)  e^{-iH_{0}t}f\right\Vert =0.
\label{representation82}%
\end{equation}
The last equality is consequence of the following Proposition, that also can
be proved by using a stationary phase argument

\begin{proposition}
\label{representation118}Let $A$ be a PDO with symbol $a\left(  x,\xi\right)
\in\mathit{S}^{-\sigma,0},$ for some $\sigma>0,$ such that $a\left(
x,\xi\right)  =0$ for $\left\vert \xi\right\vert \leq c,$ $c>0$. Then, for
$f\in\mathcal{S}\left(  \mathbb{R}^{3};\mathbb{C}^{4}\right)  $ the following
estimate holds%
\begin{equation}
\left\Vert Ae^{-iH_{0}t}f\right\Vert \leq C\left\langle t\right\rangle
^{-\sigma}. \label{representation203}%
\end{equation}
In particular, we get
\begin{equation}
\lim_{t\rightarrow\pm\infty}\left\Vert Ae^{-iH_{0}t}f\right\Vert =0,\text{ for
any }f\in L^{2}. \label{representation117}%
\end{equation}

\end{proposition}

Using Lemma \ref{representation56} and equality (\ref{representation82}) we
obtain, by a stationary phase argument, the following result

\begin{proposition}
The following equalities hold
\begin{equation}
s-\lim_{t\rightarrow\pm\infty}\left(  J_{\pm}-\theta\left(  \sqrt{H_{0}%
^{2}-m^{2}}\right)  \right)  e^{-iH_{0}t}=0, \label{representation17}%
\end{equation}%
\begin{equation}
s-\lim_{\left\vert t\right\vert \rightarrow\infty}J_{+}^{\ast}J_{-}%
e^{-iH_{0}t}=0. \label{representation18}%
\end{equation}

\end{proposition}

As the wave operators $W_{\pm}\left(  H,H_{0}\right)  $ exist, relation
(\ref{representation17}) implies that $W_{\pm}\left(  H,H_{0};J_{\pm}\right)
$ exist and the following equality holds%
\begin{equation}
W_{\pm}\left(  H,H_{0}\right)  \theta\left(  \sqrt{H_{0}^{2}-m^{2}}\right)
=W_{\pm}\left(  J_{\pm}\right)  . \label{representation91}%
\end{equation}
Note that the existence of the wave operators $W_{\pm}\left(  H,H_{0};J_{\pm
}\right)  $ can be proved in the same way as in \cite{33}, where similar
identification operators $J_{\pm}$ were defined.

We define the scattering operator $\mathbf{S}\left(  J_{+},J_{-}\right)  ,$
associated to the wave operators $W_{\pm}\left(  H,H_{0};J_{\pm}\right)  ,$ by
the relation $\mathbf{S}\left(  J_{+},J_{-}\,\right)  $ $:=W_{+}^{\ast}\left(
H,H_{0};J_{+}\right)  W_{-}\left(  H,H_{0};J_{-}\right)  .$ To simplify the
notation we denote $\mathbf{\tilde{S}=S}\left(  J_{+},J_{-}\right)  .$

Identity (\ref{representation91}) implies that the scattering operators
$\mathbf{S}$ and $\mathbf{\tilde{S}}$ are related by the equality
\begin{equation}
\theta\left(  \sqrt{H_{0}^{2}-m^{2}}\right)  \mathbf{S}\theta\left(
\sqrt{H_{0}^{2}-m^{2}}\right)  =\mathbf{\tilde{S}}. \label{representation206}%
\end{equation}

\subsection{The perturbation and the Mourre estimate.}

We define the perturbations $T_{\pm}$ as
\begin{equation}
T_{\pm}=HJ_{\pm}-J_{\pm}H_{0}. \label{representation34}%
\end{equation}
Note that $\left(  J_{\pm}H_{0}f\right)  \left(  x\right)  =\left(
2\pi\right)  ^{-3/2}\int_{\mathbb{R}^{3}}e^{i\left\langle x,\xi\right\rangle
}\left\vert \lambda\left(  \xi\right)  \right\vert \left(  a_{N}^{\pm}\left(
x,\xi;\left\vert \lambda\left(  \xi\right)  \right\vert \right)  \zeta_{\pm
}^{+}\left(  x,\xi\right)  -a_{N}^{\pm}\left(  x,\xi;-\left\vert
\lambda\left(  \xi\right)  \right\vert \right)  \zeta_{\pm}^{-}\left(
x,\xi\right)  \right)  \hat{f}\left(  \xi\right)  d\xi.$ Using relation
(\ref{eig30}) we have
\begin{equation}
\left.
\begin{array}
[c]{c}%
g_{\pm}\left(  x,\xi\right)  :=\left(  H-\left\vert \lambda\left(  \xi\right)
\right\vert \right)  \left(  u_{N}^{\pm}\left(  x,\xi;\left\vert
\lambda\left(  \xi\right)  \right\vert \right)  \zeta_{\pm}^{+}\left(
x,\xi\right)  \right)  +\left(  H+\left\vert \lambda\left(  \xi\right)
\right\vert \right)  \left(  u_{N}^{\pm}\left(  x,\xi;-\left\vert
\lambda\left(  \xi\right)  \right\vert \right)  \zeta_{\pm}^{-}\left(
x,\xi\right)  \right) \\
=e^{i\left\langle x,\xi\right\rangle }\left(  r_{N}^{\pm}\left(
x,\xi;\left\vert \lambda\left(  \xi\right)  \right\vert \right)  \zeta_{\pm
}^{+}\left(  x,\xi\right)  +r_{N}^{\pm}\left(  x,\xi;-\left\vert
\lambda\left(  \xi\right)  \right\vert \right)  \zeta_{\pm}^{-}\left(
x,\xi\right)  \right) \\
-i\sum_{j=1}^{3}\left(  \partial_{x_{j}}\zeta_{\pm}^{+}\left(  x,\xi\right)
\right)  \alpha_{j}u_{N}^{\pm}\left(  x,\xi;\left\vert \lambda\left(
\xi\right)  \right\vert \right)  -i\sum_{j=1}^{3}\left(  \partial_{x_{j}}%
\zeta_{\pm}^{-}\left(  x,\xi\right)  \right)  \alpha_{j}u_{N}^{\pm}\left(
x,\xi;-\left\vert \lambda\left(  \xi\right)  \right\vert \right)  .
\end{array}
\right.  \label{representation16}%
\end{equation}
Then, taking $t_{\pm}\left(  x,\xi\right)  =e^{-i\left\langle x,\xi
\right\rangle }g_{\pm}\left(  x,\xi\right)  $ and using relation (\ref{eig43})
we obtain the following representation for $T_{\pm}:$%
\begin{equation}
\left(  T_{\pm}f\right)  \left(  x\right)  =\int e^{i\left\langle
x,\xi\right\rangle }t_{\pm}\left(  x,\xi\right)  \hat{f}\left(  \xi\right)
d\xi=\left(  T_{\pm}^{1}f\right)  \left(  x\right)  +\left(  T_{\pm}%
^{2}f\right)  \left(  x\right)  , \label{representation36}%
\end{equation}
where the parts $T_{\pm}^{1}$ and $T_{\pm}^{2}$ have the symbols%
\begin{equation}
t_{\pm}^{1}=r_{N}^{\pm}\left(  x,\xi;\left\vert \lambda\left(  \xi\right)
\right\vert \right)  \zeta_{\pm}^{+}\left(  x,\xi\right)  +r_{N}^{\pm}\left(
x,\xi;-\left\vert \lambda\left(  \xi\right)  \right\vert \right)  \zeta_{\pm
}^{-}\left(  x,\xi\right)  \label{representation37}%
\end{equation}
and%

\begin{equation}
t_{\pm}^{2}=-i\sum_{j=1}^{3}\left(  \partial_{x_{j}}\zeta_{\pm}^{+}\left(
x,\xi\right)  \right)  \alpha_{j}a_{N}^{\pm}\left(  x,\xi;\left\vert
\lambda\left(  \xi\right)  \right\vert \right)  -i\sum_{j=1}^{3}\left(
\partial_{x_{j}}\zeta_{\pm}^{-}\left(  x,\xi\right)  \right)  \alpha_{j}%
a_{N}^{\pm}\left(  x,\xi;-\left\vert \lambda\left(  \xi\right)  \right\vert
\right)  , \label{representation38}%
\end{equation}
respectively. Using (\ref{eig28}) we get
\begin{equation}
t_{\pm}^{1}\in\mathcal{S}^{-\rho-N,-N},\text{ }N\geq0,
\label{representation23}%
\end{equation}
and, using (\ref{eig27}) we obtain%

\begin{equation}
t_{\pm}^{2}\in\mathcal{S}_{\pm}^{-1,0}. \label{representation24}%
\end{equation}

Let us introduce the operator $\mathbf{A}$, known as \textquotedblleft the
generator of dilation\textquotedblright, $\mathbf{A=}\frac{1}{2i}\sum
_{j=1}^{3}\left(  x\cdot\nabla+\nabla\cdot x\right)  .$ Note that
\begin{equation}
i[H_{0},\mathbf{A}]=H_{0}-m\beta, \label{representation22}%
\end{equation}
and%
\begin{equation}
i[H,\mathbf{A}]=H-m\beta-\mathbf{V}+i[\mathbf{V,A}]=H-m\beta-\mathbf{V}%
-\left\langle x,\left(  \nabla\mathbf{V}\right)  \left(  x\right)
\right\rangle . \label{representation121}%
\end{equation}
We recall that if$\ \mathbf{V}\ $satisfies the estimate (\ref{basicnotions46}%
), there are no eigenvalues embedded in the absolutely continuous spectrum of
$H$. Thus, the Mourre estimate%
\begin{equation}
\pm E_{H}\left(  I\right)  i[H,\mathbf{A}]E_{H}\left(  I\right)  \geq
cE_{H}\left(  I\right)  ,\text{ }c>0,\text{ }I=\left(  E-\eta_{E},E+\eta
_{E}\right)  ,\text{ }\pm E>m, \label{representation20}%
\end{equation}
is satisfied for some $\eta_{E}>0$ (Theorem 2.5 of \cite{33}).

The results we need below were proved in \cite{5} (see also \cite{34}) by
introducing the so-called conjugate operator. Condition $\left(  c_{n}\right)
$ of Definition 3.1 of \cite{5} for an operator $\mathbf{B}$ to be conjugate
to $H,$ asks the following: \textquotedblleft The form $i[H,\mathbf{B}]$,
defined on $D(\mathbf{B})\cap D(H)$, is bounded from below and closable. The
self-adjoint operator associated with its closure is denoted $iB_{1}$. Assume
$D(B_{1})\supset D(H)$. If $n>1$, assume for $j=2,3,...,n$ that the form
$i[iB_{j-1},\mathbf{B}],$ defined on $D(\mathbf{B})\cap D(H)$, is bounded from
below and closable. The associated self-adjoint operator is denoted $iB_{j}$,
and $D(B_{j})\supset D(H)$ is assumed.\textquotedblright\ If we know that the
forms $i[iB_{j-1},\mathbf{B}]$ extend to self-adjoint operators $iB_{j}$, for
all $j,$ and $D(B_{j})\supset D(H),$ then there is no need to ask the
boundeness from below of $i[iB_{j-1},\mathbf{B}],$ $j\geq1,$ in order to
obtain the results of \cite{5}.

For $\pm E>m,$ let us consider the operator $\pm\mathbf{A.}$ Note that
$i[iB_{j-1},\pm\mathbf{A}],$ $j\geq1,$ ($iB_{0}=H$) are not bounded from
below. Nevertheless,\ using equality (\ref{representation121}) we see that
$i[iB_{j-1},\mathbf{A}]=H-m\beta-\mathbf{V+V}_{j},$ $j\geq1,$ where
$\mathbf{V}_{0}=\mathbf{V}$ and $\mathbf{V}_{j}$ is defined recursively by
$\mathbf{V}_{j}=-\left\langle x,\left(  \nabla\mathbf{V}_{j-1}\right)  \left(
x\right)  \right\rangle ,$ $j\geq1.$ Then, the forms $i[iB_{j-1},\mathbf{A}],$
$j\geq1,$ extend to self-adjoint operators $iB_{j}$, and $D(B_{j})=D(H).$
Moreover $\pm\mathbf{A}$ satisfies relation (\ref{representation20}). Thus, we
can apply the results of \cite{5}, taking, for $\pm E>m,$ the operator
$\pm\mathbf{A}$ as conjugate to $H.$ Furthermore, we use the dilatation
transformation argument (see \cite{39},\cite{30}) in order to prove that these
results hold uniformly for $\left\vert E\right\vert \geq E_{0}>m,$ for any
$E_{0}.$ We get the following

\begin{proposition}
\label{representation9}Let estimates (\ref{eig31}) and (\ref{eig32}) hold.
Define $\mathbf{P}_{+}:=E_{\mathbf{A}}\left(  0,\infty\right)  ~$and
$\mathbf{P}_{-}:=E_{\mathbf{A}}\left(  -\infty,0\right)  $ as the spectral
projections of the operator $\mathbf{A}$ ($E_{\mathbf{A}}$ is the resolution
of the identity for $\mathbf{A}$). For $\pm\operatorname{Re}z>m$ and
$\operatorname{Im}z\geq0,$ the operators%
\begin{equation}
\left\langle \mathbf{A}\right\rangle ^{-p}R\left(  z\right)  \left\langle
\mathbf{A}\right\rangle ^{-p},p>\frac{1}{2}, \label{representation1}%
\end{equation}%
\begin{equation}
\left\langle \mathbf{A}\right\rangle ^{-1+p}\mathbf{P}_{\pm}R\left(  z\right)
\left\langle \mathbf{A}\right\rangle ^{-q},\text{ \ }\left\langle
\mathbf{A}\right\rangle ^{-q}R\left(  z\right)  \mathbf{P}_{\pm}\left\langle
\mathbf{A}\right\rangle ^{-1+p} \label{representation2}%
\end{equation}
with $q>\frac{1}{2},$ $p<q$ and%
\begin{equation}
\left\langle \mathbf{A}\right\rangle ^{p}\mathbf{P}_{\mp}R\left(  z\right)
\mathbf{P}_{\pm}\left\langle \mathbf{A}\right\rangle ^{p},\text{ \ }\forall p
\label{representation3}%
\end{equation}
are continuous in norm with respect to $z.$ Moreover, the norms of operators
(\ref{representation1})-(\ref{representation3}) at $z=E+i0$ are bounded by
$C\left\vert E\right\vert ^{-1}$ as $\left\vert E\right\vert \rightarrow
\infty.$
\end{proposition}

We now present the following two assertions (\cite{27},\cite{30}). We denote
by $T$ the PDO with symbol $t.$ Recall that the clases $\mathcal{S}_{\pm
}^{m,n}$ where defined below (\ref{representation197}).

\begin{proposition}
\label{representation7}Let $t\in\mathcal{S}_{\pm}^{0,0}$ for one of the signs
and let $p>0$ be an entire number. Then, the operator $\left\langle
x\right\rangle ^{p}\left\langle \nabla\right\rangle ^{p}T\left\langle
\mathbf{A}\right\rangle ^{-q}$ is bounded, for $q\geq p$.
\end{proposition}

\begin{proposition}
Let $t\in\mathcal{S}_{\pm}^{n,m}$ for some $n$ and $m.$ Then the operator
$\left\langle x\right\rangle ^{q}\left\langle \nabla\right\rangle
^{s}T\mathbf{P}_{\pm}\left\langle \mathbf{A}\right\rangle ^{p}$ is bounded for
all real numbers $p,q,s.$
\end{proposition}

Since for $\mathbf{V}\ $satisfying the estimate (\ref{basicnotions46}), there
are no eigenvalues embedded in the absolutely continuous spectrum of $H,$ the
resolvent $R\left(  E\pm i0\right)  $ is locally H\"{o}lder continuous on
$(-\infty,-m)\cup(m,\infty).$ Thus, from Proposition 4.1 of \cite{15} we obtain

\begin{proposition}
\label{representation8}Let $\mathbf{V}=\{V_{jk}\}_{j,k=1,2,3,4,}$ be an
Hermitian $4\times4$-matrix valued function, such that $\left\vert \left(
x\cdot\nabla\right)  ^{l}V_{jk}\left(  x\right)  \right\vert ,$ $j,k=1,2,3,4,$
are bounded for all $x\in\mathbb{R}^{3}$ and $l=0,1,2.$ Then, for any
$E_{0}>m,$ the following estimate holds%
\[
\sup_{\substack{0<\varepsilon<1\\\left\vert E\right\vert \geq E_{0}%
}}\left\Vert R\left(  E\pm i\varepsilon\right)  f\right\Vert _{L_{-s}^{2}}\leq
C_{s,E_{0}}\left\Vert f\right\Vert _{L_{s}^{2}},\text{ }1/2<s\leq1.
\]

\end{proposition}

Using Propositions \ref{representation9}-\ref{representation8} we get,
similarly to Proposition 3.5 of \cite{27} or Proposition 4.1 of \cite{30}, the
following result

\begin{lemma}
\label{representation33}For any $p$, $q,$ and $N$ such that $N>p-\rho+1/2,$
$N\geq q,$ for $\left\vert \operatorname{Re}z\right\vert >m$ and
$\operatorname{Im}z\geq0,$ the operator
\begin{equation}
\left\langle x\right\rangle ^{p}\left\langle \nabla\right\rangle ^{q}%
T_{+}^{\ast}R\left(  z\right)  T_{-}\left\langle \nabla\right\rangle
^{q}\left\langle x\right\rangle ^{p}, \label{representation21}%
\end{equation}
is continuous in norm~with respect to $z$ and, moreover, the operator
$\left\langle x\right\rangle ^{p}\left\langle \nabla\right\rangle ^{q}%
T_{+}^{\ast}R_{+}\left(  E\right)  T_{-}\left\langle \nabla\right\rangle
^{q}\left\langle x\right\rangle ^{p},$ is uniformly bounded for $\left\vert
E\right\vert \geq E_{0}>m$, for all $E_{0}.$
\end{lemma}

\subsection{The regular and singular parts of the scattering matrix.}

For $\left\vert E\right\vert >m,$ let us define the following operators%
\begin{equation}
S_{1}\left(  E\right)  =2\pi i\Gamma_{0}\left(  E\right)  T_{+}^{\ast}%
R_{+}\left(  E\right)  T_{-}\Gamma_{0}^{\ast}\left(  E\right)  ,
\label{representation14}%
\end{equation}
and for $f_{j}\in\mathcal{H}\left(  E\right)  $ such that $f_{1},f_{2}\in
C^{\infty}\left(  \mathbb{S}^{2};\mathbb{C}^{4}\right)  ,$ $j=1,2,$
$S_{2}\left(  E\right)  $ is defined as the following form%
\begin{equation}
\left(  S_{2}\left(  E\right)  f_{1},f_{2}\right)  :=-2\pi i\lim
_{\mu\downarrow0}\left(  J_{+}^{\ast}T_{-}\delta_{\mu}\left(  H_{0}-E\right)
g_{1},\Gamma_{0}^{\ast}\left(  E\right)  \Gamma_{0}\left(  E\right)
g_{2}\right)  , \label{representation134}%
\end{equation}
where $\delta_{\mu}\left(  H_{0}-E\right)  :=\left(  2\pi i\right)
^{-1}\left(  R_{0}\left(  E+i\mu\right)  -R_{0}\left(  E-i\mu\right)  \right)
$ and $g_{j}$ are such that $\hat{g}_{j}\left(  \xi\right)  =\upsilon\left(
E\right)  ^{-1}f_{j}\left(  \hat{\xi}\right)  \gamma\left(  \left\vert
\xi\right\vert \right)  $, $j=1,2,$ with $\hat{\xi}=\xi/\left\vert
\xi\right\vert ,$ $\gamma\in C_{0}^{\infty}\left(  \mathbb{R}^{+}%
\setminus\{0\}\right)  $ and $\gamma\left(  \nu\left(  E\right)  \right)  =1.$
Below we compute the limit in the R.H.S. of (\ref{representation134}) and we
prove that it is a bounded operator, that is independent of $\gamma.$

The kernel of the operator $S_{1}\left(  E\right)  $ is smooth:

\begin{theorem}
\label{representation32}Let the magnetic potential $A\left(  x\right)  $ and
the electric potential $V\left(  x\right)  $ satisfy the estimates
(\ref{eig31}) and (\ref{eig32}) respectively. For any $p$ and $q,$ there is
$N$ such that the kernel $s_{1}\left(  \omega,\theta;E\right)  $ of the
operator $S_{1}\left(  E\right)  $ belongs to the class of $C^{p}\left(
\mathbb{S}^{2}\times\mathbb{S}^{2}\right)  $-functions$,$ and furthermore its
$C^{p}-$norm is $O\left(  \left\vert E\right\vert ^{-q}\right)  $ when
$\left\vert E\right\vert \rightarrow\infty.$
\end{theorem}

\begin{proof}
Using the definition (\ref{basicnotions11}) of $\Gamma_{0}\left(  E\right)  $,
the relation (\ref{basicnotions39}) for $\Gamma_{0}^{\ast}\left(  E\right)  ,$
and $\Gamma_{0}\left(  E\right)  \left\langle \nabla\right\rangle ^{-q_{0}%
}=\left(  1+\nu\left(  E\right)  ^{2}\right)  ^{-\frac{q_{0}}{2}}\Gamma
_{0}\left(  E\right)  ,$ we get $\left(  S_{1}\left(  E\right)  f\right)
\left(  \omega\right)  =i\left(  2\pi\right)  ^{-2}\upsilon\left(  E\right)
^{2}\left(  1+\nu\left(  E\right)  ^{2}\right)  ^{-q_{0}}\int\int%
\limits_{\mathbb{S}^{2}}e^{-i\nu\left(  E\right)  \left\langle \omega
,x\right\rangle }\left[  P_{\omega}\left(  E\right)  \left\langle
\nabla\right\rangle ^{q_{0}}T_{+}^{\ast}R_{+}\left(  E\right)  T_{-}%
\left\langle \nabla\right\rangle ^{q_{0}}P_{\theta}\left(  E\right)  \right]
e^{i\nu\left(  E\right)  \left\langle \theta,x\right\rangle }$ $\times
f\left(  \theta\right)  d\theta dx.$ Note that for $s>3/2,$ $e^{i\nu\left(
E\right)  \left\langle \theta,x\right\rangle }\left\langle x\right\rangle
^{-s}\in L^{2}.$ Using Lemma \ref{representation33} with $N\geq3/2-\rho$ we
see that
\begin{equation}
\left.
\begin{array}
[c]{c}%
s_{1}\left(  \omega,\theta;E\right)  \left.  :=\right.  i\left(  2\pi\right)
^{-2}\upsilon\left(  E\right)  ^{2}\left(  1+\nu\left(  E\right)  ^{2}\right)
^{-q_{0}}%
{\displaystyle\int}
\left(  e^{-i\nu\left(  E\right)  \left\langle \omega,x\right\rangle
}\left\langle x\right\rangle ^{-s}\right) \\
\times\left[  P_{\omega}\left(  E\right)  \left\langle x\right\rangle
^{s}\left\langle \nabla\right\rangle ^{q_{0}}T_{+}^{\ast}R_{+}\left(
E\right)  T_{-}\left\langle \nabla\right\rangle ^{q_{0}}\left\langle
x\right\rangle ^{s}P_{\theta}\left(  E\right)  \right]  \left(  e^{i\nu\left(
E\right)  \left\langle \theta,x\right\rangle }\left\langle x\right\rangle
^{-s}\right)  dx,
\end{array}
\right.  \label{representation177}%
\end{equation}
is a continuous function of $\omega$ and$\ \theta.$ Differentiating
(\ref{representation177}) $p$ times with respect to $\omega$ or $\theta$ we
see that $\partial_{\omega}^{p}\partial_{\theta}^{p}\left(  s_{1}\left(
\omega,\theta;E\right)  \right)  $ is continuous in $\omega$ and $\theta,$ and
bounded by $C\left\vert E\right\vert ^{-q},$ if the operator
\begin{equation}
\left\langle x\right\rangle ^{p_{0}}\left\langle \nabla\right\rangle ^{q_{0}%
}T_{+}^{\ast}R_{+}\left(  E\right)  T_{-}\left\langle \nabla\right\rangle
^{q_{0}}\left\langle x\right\rangle ^{p_{0}} \label{representation19}%
\end{equation}
with $p_{0}>p+3/2$ and $q_{0}\geq1+q/2+p,$ is bounded uniformly for
$\left\vert E\right\vert \geq E_{0}>m$. Taking $N\geq p_{0}-\rho$ and $N\geq
q_{0}$ in Lemma \ref{representation33} we get the desired result.
\end{proof}

Let us study now the limit (\ref{representation134}). It follows from
(\ref{representation134}) that%
\begin{equation}
\left.  \left(  S_{2}\left(  E\right)  f_{1},f_{2}\right)  =-2\pi i\lim
_{\mu\downarrow0}\left(  T_{-}\mathcal{F}^{\ast}\tilde{\delta}_{\mu}^{\left(
\operatorname*{sgn}E\right)  }\hat{g}_{1},J_{+}\Gamma_{0}^{\ast}\left(
E\right)  f_{2}\right)  ,\right.  \label{representation146}%
\end{equation}
where $\tilde{\delta}_{\mu}^{\left(  \operatorname*{sgn}E\right)  }=\frac{\mu
}{\pi}\left(  \left(  \lambda\left(  \xi\right)  -E\right)  ^{2}+\mu
^{2}\right)  ^{-1}P^{\left(  \operatorname*{sgn}E\right)  }\left(  \xi\right)
,$ where $\lambda\left(  \xi\right)  $ are defined below
(\ref{representation15}). Note that the equality $\lambda\left(  \xi\right)
=E$ is valid only if $\left\vert \xi\right\vert =\nu\left(  E\right)  .$

Using the relation (\ref{representation133}) we obtain the following equation
$J_{+}\Gamma_{0}^{\ast}\left(  E\right)  f_{2}=\left(  2\pi\right)
^{-\frac{3}{2}}\upsilon\left(  E\right)  \int e^{i\left\langle x,\nu\left(
E\right)  \omega\right\rangle }j_{N}^{+}\left(  x,\nu\left(  E\right)
\omega;E\right)  $ $\times f_{2}\left(  \omega\right)  d\omega,$ where
$j_{N}^{+}\left(  x,\xi;E\right)  :=a_{N}^{+}\left(  x,\xi;\lambda\left(
\xi\right)  \right)  \zeta_{+}^{\left(  \operatorname*{sgn}E\right)  }\left(
x,\xi\right)  $ and moreover
\begin{equation}
\left.  \left(  T_{-}f,J_{+}\Gamma_{0}^{\ast}\left(  E\right)  f_{2}\right)
=\left(  2\pi\right)  ^{-3}\upsilon\left(  E\right)
{\displaystyle\int}
\left(
{\displaystyle\int}
{\displaystyle\int}
e^{i\left\langle x,\xi^{\prime}-\nu\left(  E\right)  \omega\right\rangle
}\left(  \left(  j_{N}^{+}\left(  x,\nu\left(  E\right)  \omega;E\right)
\right)  ^{\ast}t_{-}\left(  x,\xi^{\prime}\right)  \hat{f}\left(  \xi
^{\prime}\right)  ,f_{2}\left(  \omega\right)  \right)  d\omega d\xi^{\prime
}\right)  dx,\right.  \label{representation51}%
\end{equation}
for $f\in\mathcal{S}\left(  \mathbb{R}^{3};\mathbb{C}^{4}\right)  .$

Let us define $\varphi\in C_{0}^{\infty}\left(  \mathbb{R}^{3}\right)  $ such
that $\varphi\left(  0\right)  =1.$ Then, taking $f=\mathcal{F}^{\ast}%
\tilde{\delta}_{\mu}^{\left(  \operatorname*{sgn}E\right)  }\hat{g}_{1}$ in
(\ref{representation51}) and using relations (\ref{representation36}),
(\ref{eig36}) and (\ref{representation146}) we get%
\begin{equation}
\left.  \left(  S_{2}\left(  E\right)  f_{1},f_{2}\right)  =-i\left(
2\pi\right)  ^{-2}\upsilon\left(  E\right)  \lim_{\mu\downarrow0}%
\lim_{\varepsilon\rightarrow0}%
{\displaystyle\int}
{\displaystyle\int_{\mathbb{S}^{2}}}
\left(  G^{\left(  \varepsilon\right)  }\left(  \nu\left(  E\right)
\omega,\xi^{\prime};E\right)  \tilde{\delta}_{\mu}^{\left(
\operatorname*{sgn}E\right)  }\hat{g}_{1}\left(  \xi^{\prime}\right)
,f_{2}\left(  \omega\right)  \right)  d\omega d\xi^{\prime}.\right.
\label{representation188}%
\end{equation}
where, for some $\varsigma\in C_{0}^{\infty}\left(  \mathbb{R}^{+}\right)  ,$
such that $\varsigma\left(  t\right)  =1,$ in some neighborhood of $\nu\left(
E\right)  $ and $\varsigma\left(  t\right)  =0,$ for $t<c_{1},$
\begin{equation}
G^{\left(  \varepsilon\right)  }\left(  \xi,\xi^{\prime};E\right)
:=\varsigma\left(  \left\vert \xi\right\vert \right)  \varsigma\left(
\left\vert \xi^{\prime}\right\vert \right)  \int e^{i\left\langle
x,\xi^{\prime}-\xi\right\rangle }\left(  j_{N}^{+}\left(  x,\xi;E\right)
\right)  ^{\ast}t_{-}\left(  x,\xi^{\prime};E\right)  \varphi\left(
\varepsilon x\right)  dx, \label{representation40}%
\end{equation}
and
\begin{equation}
\left.  t_{-}\left(  x,\xi^{\prime};E\right)  :=t_{-}^{1}\left(  x,\xi
^{\prime};E\right)  +t_{-}^{2}\left(  x,\xi^{\prime};E\right)  ,\right.
\label{representation223}%
\end{equation}
with $t_{-}^{1}\left(  x,\xi^{\prime};E\right)  :=r_{N}^{-}\left(
x,\xi^{\prime};\lambda\left(  \xi^{\prime}\right)  \right)  \zeta_{-}^{\left(
\operatorname*{sgn}E\right)  }\left(  x,\xi^{\prime}\right)  $ and $t_{-}%
^{2}\left(  x,\xi^{\prime};E\right)  :=-i\sum_{j=1}^{3}\left(  \partial
_{x_{j}}\zeta_{-}^{\left(  \operatorname*{sgn}E\right)  }\left(  x,\xi
^{\prime}\right)  \right)  \alpha_{j}a_{N}^{-}\left(  x,\xi^{\prime}%
;\lambda\left(  \xi^{\prime}\right)  \right)  .$ Below we study the limit of
$G^{\left(  \varepsilon\right)  }\left(  \xi,\xi^{\prime};E\right)  $, as
$\varepsilon\rightarrow0$. Then, calculating the limit
(\ref{representation188}), we recover information about the smoothness and
behavior for $\left\vert E\right\vert \rightarrow\infty$ of the kernel
$s_{2}\left(  \omega,\theta;E\right)  $ of $S_{2}\left(  E\right)  .$

Let us denote $\hat{\xi}=\xi/\left\vert \xi\right\vert $ and $\hat{\xi
}^{\prime}=\xi^{\prime}/\left\vert \xi^{\prime}\right\vert .$ We use the
following result

\begin{lemma}
\label{representation54}Let $G\left(  \xi,\xi^{\prime};E\right)  ,$ defined on
$\mathbb{R}^{3}\times\mathbb{R}^{3},$ be such that for each $E,$ $\left\vert
E\right\vert >m,$ the function $G\left(  \nu\left(  E\right)  \hat{\xi}%
,\xi^{\prime};E\right)  $ is H\"{o}lder-continuous on $\xi^{\prime},$
uniformly for $\omega\in\mathbb{S}^{2}.$ Let $g_{1},$ $f_{1}$ and $f_{2}$ be
as in (\ref{representation134})$.$ Then, the following identity holds%
\begin{equation}
\left.  \lim_{\mu\downarrow0}%
{\displaystyle\int}
{\displaystyle\int_{\mathbb{S}^{2}}}
\left(  G\left(  \nu\left(  E\right)  \omega,\xi^{\prime};E\right)
\tilde{\delta}_{\mu}^{\left(  \operatorname*{sgn}E\right)  }\hat{g}_{1}\left(
\xi^{\prime}\right)  ,f_{2}\left(  \omega\right)  \right)  d\omega
d\xi^{\prime}=\upsilon\left(  E\right)
{\displaystyle\int_{\mathbb{S}^{2}}}
{\displaystyle\int_{\mathbb{S}^{2}}}
\left(  G\left(  \nu\left(  E\right)  \omega,\nu\left(  E\right)
\theta;E\right)  f_{1}\left(  \theta\right)  ,f_{2}\left(  \omega\right)
\right)  d\omega d\theta.\right.  \label{representation65}%
\end{equation}

\end{lemma}

\begin{proof}
Note that for $f\in L^{1}\left(  0,\infty\right)  ,$ H\"{o}lder-continuous in
the point $\nu\left(  E\right)  ,$ the following relation holds%
\begin{equation}
\lim_{\mu\downarrow0}\frac{\mu}{\pi}\int_{0}^{\infty}f\left(  r\right)
\left(  \left(  \sqrt{r^{2}+m^{2}}-\left\vert E\right\vert \right)  ^{2}%
+\mu^{2}\right)  ^{-1}dr=\frac{f\left(  \nu\left(  E\right)  \right)
\left\vert E\right\vert }{\nu\left(  E\right)  }. \label{representation66}%
\end{equation}
Then, passing to the polar coordinate system in (\ref{representation65}) and
using that $\hat{g}_{1}\left(  \xi^{\prime}\right)  =\upsilon\left(  E\right)
^{-1}f_{1}\left(  \hat{\xi}^{\prime}\right)  \gamma\left(  \left\vert
\xi^{\prime}\right\vert \right)  $ we get%
\begin{equation}
\left.  \lim_{\mu\downarrow0}%
{\displaystyle\int}
{\displaystyle\int}
\left(  G\left(  \nu\left(  E\right)  \omega,\xi^{\prime};E\right)
\tilde{\delta}_{\mu}^{\left(  \operatorname*{sgn}E\right)  }\hat{g}_{1}\left(
\xi^{\prime}\right)  ,f_{2}\left(  \omega\right)  \right)  d\omega
d\xi^{\prime}=\lim_{\mu\downarrow0}\frac{\mu}{\pi}%
{\displaystyle\int_{0}^{\infty}}
\left(  \left(  \sqrt{r^{2}+m^{2}}-\left\vert E\right\vert \right)  ^{2}%
+\mu^{2}\right)  ^{-1}W\left(  r\right)  r^{2}dr,\right.
\label{representation153}%
\end{equation}
where $W\left(  r\right)  =\upsilon\left(  E\right)  ^{-1}\gamma\left(
r\right)  \int_{\mathbb{S}^{2}}\int_{\mathbb{S}^{2}}\left(  G\left(
\nu\left(  E\right)  \omega,r\theta;E\right)  f_{1}\left(  \theta\right)
,f_{2}\left(  \omega\right)  \right)  d\omega d\theta.$ As $G\left(
\nu\left(  E\right)  \hat{\xi},\xi^{\prime};E\right)  $ is a
H\"{o}lder-continuous function of the variable $\xi^{\prime},$ uniformly in
$\hat{\xi},$ we conclude that $W\left(  r\right)  $ is a H\"{o}lder-continuous
function of $r.$ Therefore, applying the equality (\ref{representation66}) to
the R.H.S. of (\ref{representation153}) and recalling that $\upsilon\left(
E\right)  =\left(  \left\vert E\right\vert \nu\left(  E\right)  \right)
^{1/2},$ we get the desired result.
\end{proof}

Using (\ref{representation223}) we decompose $G^{\left(  \varepsilon\right)
}\left(  \xi,\xi^{\prime};E\right)  $ as
\begin{equation}
G^{\left(  \varepsilon\right)  }\left(  \xi,\xi^{\prime};E\right)
=G_{1}^{\left(  \varepsilon\right)  }\left(  \xi,\xi^{\prime};E\right)
+G_{2}^{\left(  \varepsilon\right)  }\left(  \xi,\xi^{\prime};E\right)  ,
\label{representation166}%
\end{equation}
where
\begin{equation}
\left.  G_{1}^{\left(  \varepsilon\right)  }\left(  \xi,\xi^{\prime};E\right)
=\varsigma\left(  \left\vert \xi\right\vert \right)  \varsigma\left(
\left\vert \xi^{\prime}\right\vert \right)
{\displaystyle\int}
e^{i\left\langle x,\xi^{\prime}-\xi\right\rangle }\left(  j_{N}^{+}\left(
x,\xi;E\right)  \right)  ^{\ast}t_{-}^{1}\left(  x,\xi^{\prime};E\right)
\varphi\left(  \varepsilon x\right)  dx,\right.  \label{representation61}%
\end{equation}
and%
\begin{equation}
\left.  G_{2}^{\left(  \varepsilon\right)  }\left(  \xi,\xi^{\prime};E\right)
=\varsigma\left(  \left\vert \xi\right\vert \right)  \varsigma\left(
\left\vert \xi^{\prime}\right\vert \right)
{\displaystyle\int}
e^{i\left\langle x,\xi^{\prime}-\xi\right\rangle }\left(  j_{N}^{+}\left(
x,\xi;E\right)  \right)  ^{\ast}t_{-}^{2}\left(  x,\xi^{\prime};E\right)
\varphi\left(  \varepsilon x\right)  dx.\right.  \label{representation62}%
\end{equation}
Let us prove now the following result

\begin{lemma}
\label{representation183}The limit $G_{1}^{\left(  0\right)  }\left(
\nu\left(  E\right)  \hat{\xi},\nu\left(  E\right)  \hat{\xi}^{\prime
};E\right)  :=\lim\limits_{\varepsilon\rightarrow\infty}G_{1}^{\left(
\varepsilon\right)  }\left(  \nu\left(  E\right)  \hat{\xi},\nu\left(
E\right)  \hat{\xi}^{\prime};E\right)  $ exists. For any $p$ and $q,$ there is
$N,$ such that $G_{1}^{\left(  0\right)  }\left(  \nu\left(  E\right)
\hat{\xi},\nu\left(  E\right)  \hat{\xi}^{\prime};E\right)  \in C^{p}\left(
\mathbb{S}^{2}\times\mathbb{S}^{2}\right)  $, and its $C^{p}-$norm is bounded
by $C\left\vert E\right\vert ^{-q},$ as $\left\vert E\right\vert
\rightarrow\infty.$ Moreover, the limit in the R.H.S. of
(\ref{representation188}) with $G_{1}^{\left(  \varepsilon\right)  },$ instead
of $G^{\left(  \varepsilon\right)  },$ exists, and $-i\left(  2\pi\right)
^{-2}\upsilon\left(  E\right)  \lim\limits_{\mu\downarrow0}\lim
\limits_{\varepsilon\rightarrow0}%
{\displaystyle\int}
{\displaystyle\int_{\mathbb{S}^{2}}}
(G_{1}^{\left(  \varepsilon\right)  }\left(  \nu\left(  E\right)  \omega
,\xi^{\prime};E\right)  \tilde{\delta}_{\mu}^{\left(  \operatorname*{sgn}%
E\right)  }\hat{g}_{1}\left(  \xi^{\prime}\right)  ,$ $f_{2}\left(
\omega\right)  )d\omega d\xi^{\prime}$ $=-i\left(  2\pi\right)  ^{-2}%
\upsilon\left(  E\right)  ^{2}%
{\displaystyle\int_{\mathbb{S}^{2}}}
{\displaystyle\int_{\mathbb{S}^{2}}}
\left(  G_{1}^{\left(  0\right)  }\left(  \nu\left(  E\right)  \omega
,\nu\left(  E\right)  \theta;E\right)  f_{1}\left(  \theta\right)
,f_{2}\left(  \omega\right)  \right)  d\theta d\omega.$
\end{lemma}

\begin{proof}
As $j_{N}^{\pm}\left(  x,\xi;E\right)  \in\mathit{S}^{0,0}$, using
(\ref{representation23}) we see that $\left\vert \left(  j_{N}^{+}\left(
x,\xi;E\right)  \right)  ^{\ast}t_{-}^{1}\left(  x,\xi^{\prime};E\right)
\right\vert \leq C_{\alpha,\beta}\left(  1+\left\vert x\right\vert \right)
^{-\rho-N+\left\vert \alpha\right\vert +\left\vert \beta\right\vert
}\left\vert \xi^{\prime}\right\vert ^{-N},$ for $N\geq0.$ Thus, for $N$ big
enough, the limit in (\ref{representation61}), as $\varepsilon\rightarrow0,$
exists, and, moreover, $G_{1}^{\left(  0\right)  }\left(  \xi,\xi^{\prime
};E\right)  $ is a function in $C^{p\left(  N\right)  }\left(  \mathbb{R}%
^{3}\times\mathbb{R}^{3}\right)  ,$ that decreases as $C\left\vert \xi
^{\prime}\right\vert ^{-q\left(  N\right)  },$ when $\left\vert \xi^{\prime
}\right\vert \rightarrow\infty$. Replacing $G^{\left(  \varepsilon\right)  }$
with $G_{1}^{\left(  0\right)  }$ in (\ref{representation188}) and using Lemma
\ref{representation54} to calculate the resulting limit we complete the proof.
\end{proof}

We now study the term $G_{2}^{\left(  \varepsilon\right)  }\left(  \xi
,\xi^{\prime};E\right)  .$ Let us consider first the function $G_{2}^{\left(
\varepsilon\right)  }\left(  \xi,\xi^{\prime};E\right)  $ for $\xi\neq
\xi^{\prime}.$ We prove the following result

\begin{lemma}
\label{representation145}Let $O,O^{\prime}\subseteq\mathbb{S}^{2}$ be open
sets such that $\overline{O}\cap\overline{O}^{\prime}=\varnothing.$ Then,
there exists a function $G_{2,O,O^{\prime}}^{\left(  0\right)  }\left(
\xi,\xi^{\prime};E\right)  ,$ such that for any $p$ and $q$, $G_{2,O,O^{\prime
}}^{\left(  0\right)  }\left(  \xi,\xi^{\prime};E\right)  $ is of
$C^{p}\left(  \mathbb{R}^{3}\times\mathbb{R}^{3}\right)  $-class, and its
$C^{p}-$norm\ is bounded by $C\left\vert E\right\vert ^{-q}$ as $\left\vert
E\right\vert \rightarrow\infty,$ and moreover, for $g_{1},$ $f_{1}$ and
$f_{2}$ as in (\ref{representation134}), with the additional property
$f_{1}\in C_{0}^{\infty}\left(  O^{\prime}\right)  $ and $f_{2}\in
C_{0}^{\infty}\left(  O\right)  ,$ $-i\left(  2\pi\right)  ^{-2}%
\upsilon\left(  E\right)  \lim\limits_{\mu\downarrow0}\lim\limits_{\varepsilon
\rightarrow0}%
{\displaystyle\int}
{\displaystyle\int\limits_{\mathbb{S}^{2}}}
$ $\times(G_{2}^{\left(  \varepsilon\right)  }\left(  \nu\left(  E\right)
\omega,\xi^{\prime};E\right)  \tilde{\delta}_{\mu}^{\left(
\operatorname*{sgn}E\right)  }\hat{g}_{1}\left(  \xi^{\prime}\right)
,f_{2}\left(  \omega\right)  )d\omega d\xi^{\prime}=-i\left(  2\pi\right)
^{-2}\upsilon\left(  E\right)  ^{2}%
{\displaystyle\int\limits_{\mathbb{S}^{2}}}
{\displaystyle\int\limits_{\mathbb{S}^{2}}}
(G_{2,O,O^{\prime}}^{\left(  0\right)  }\left(  \nu\left(  E\right)
\omega,\nu\left(  E\right)  \theta;E\right)  f_{1}\left(  \theta\right)
,f_{2}\left(  \omega\right)  )d\theta d\omega.$ In particular, the function
$G_{2,O,O^{\prime}}^{\left(  0\right)  }$ satisfies the estimate%
\begin{equation}
\left\Vert G_{2,O,O^{\prime}}^{\left(  0\right)  }\left(  \nu\left(  E\right)
\hat{\xi},\nu\left(  E\right)  \hat{\xi}^{\prime};E\right)  \right\Vert
_{C^{p}\left(  \mathbb{S}^{2}\times\mathbb{S}^{2}\right)  }\leq C_{p}\left(
O,O^{\prime}\right)  \left\vert E\right\vert ^{-q}, \label{representation224}%
\end{equation}
for any $p$ and $q.$
\end{lemma}

\begin{proof}
Choosing $l$ such that $\sqrt{3}\left\vert \xi_{l}-\xi_{l}^{\prime}\right\vert
\geq\left\vert \xi-\xi^{\prime}\right\vert >0$ and integrating
(\ref{representation62}) by parts $n$ times we get
\begin{equation}
G_{2}^{\left(  \varepsilon\right)  }\left(  \xi,\xi^{\prime};E\right)
=G_{2,O,O^{\prime}}^{\left(  \varepsilon\right)  }\left(  \xi,\xi^{\prime
};E\right)  +R_{jk}^{\left(  \varepsilon\right)  }\left(  \xi,\xi^{\prime
};E\right)  , \label{representation245}%
\end{equation}
where $G_{2,O,O^{\prime}}^{\left(  \varepsilon\right)  }\left(  \xi
,\xi^{\prime};E\right)  :=\left(  \xi_{l}^{\prime}-\xi_{l}\right)
^{-n}\varsigma\left(  \left\vert \xi\right\vert \right)  \varsigma\left(
\left\vert \xi^{\prime}\right\vert \right)  \int e^{i\left\langle
x,\xi^{\prime}-\xi\right\rangle }\varphi\left(  \varepsilon x\right)  \left(
i\partial_{x_{l}}\right)  ^{n}\left(  \left(  j_{N}^{+}\left(  x,\xi;E\right)
\right)  ^{\ast}t_{-}^{2}\left(  x,\xi^{\prime};E\right)  \right)  dx$ and
$R_{jk}^{\left(  \varepsilon\right)  }$ is given by $R_{jk}^{\left(
\varepsilon\right)  }\left(  \xi,\xi^{\prime};E\right)  :=\sum_{m=1}%
^{n}\varepsilon^{m}\left(  \xi_{l}^{\prime}-\xi_{l}\right)  ^{-n}\int
e^{i\left\langle x,\xi^{\prime}-\xi\right\rangle }g_{m}\left(  x,\xi
,\xi^{\prime};E\right)  \varphi^{\left(  m\right)  }\left(  \varepsilon
x\right)  dx,$ with $g_{m}\in\mathcal{S}^{m-n},$ for $m\leq n.$ Note that
\begin{equation}
\lim_{\varepsilon\rightarrow0}%
{\displaystyle\int}
{\displaystyle\int_{\mathbb{S}^{2}}}
\left(  R_{jk}^{\left(  \varepsilon\right)  }\left(  \nu\left(  E\right)
\omega,\xi^{\prime};E\right)  \tilde{\delta}_{\mu}^{\left(
\operatorname*{sgn}E\right)  }\hat{g}_{1}\left(  \xi^{\prime}\right)
,f_{2}\left(  \omega\right)  \right)  d\omega d\xi^{\prime}=0.
\label{representation243}%
\end{equation}
Indeed, substituting the definition of $R_{jk}^{\left(  \varepsilon\right)  }$
in (\ref{representation243}) and integrating several times by parts in the
variable $\xi^{\prime}$ the resulting expression, we obtain a product of an
absolutely convergent integral, uniformly bounded on $\varepsilon,$ and
$\varepsilon^{m}.$

For $n>2$ the integral in the relation for $G_{2,O,O^{\prime}}^{\left(
\varepsilon\right)  }$ is absolutely convergent. Hence, the limit of
$G_{2,O,O^{\prime}}^{\left(  \varepsilon\right)  },$ as $\varepsilon
\rightarrow0,$ exists and it is equal to the absolutely convergent integral%
\begin{equation}
G_{2,O,O^{\prime}}^{\left(  0\right)  }\left(  \xi,\xi^{\prime};E\right)
=\left(  \xi_{j}-\xi_{j}^{\prime}\right)  ^{-n}\varsigma\left(  \left\vert
\xi\right\vert \right)  \varsigma\left(  \left\vert \xi^{\prime}\right\vert
\right)  \int e^{i\left\langle x,\xi^{\prime}-\xi\right\rangle }\left(
i\partial_{x_{j}}\right)  ^{n}\left(  \left(  j_{N}^{+}\left(  x,\xi;E\right)
\right)  ^{\ast}t_{-}\left(  x,\xi^{\prime};E\right)  \right)  dx.
\label{representation41}%
\end{equation}
Moreover, for any $n$ we have the following estimate
\begin{equation}
\left\vert \left(  \partial_{\xi}^{\beta}\partial_{\xi^{\prime}}%
^{\beta^{\prime}}G_{2,O,O^{\prime}}^{\left(  0\right)  }\right)  \left(
\xi,\xi^{\prime};E\right)  \right\vert \leq C_{p,jk}\left\vert \xi-\xi
^{\prime}\right\vert ^{-n},\text{ }\left\vert \beta\right\vert +\left\vert
\beta^{\prime}\right\vert =p, \label{representation137}%
\end{equation}
for $p<n-2.$ Introducing decomposition (\ref{representation245}) in the R.H.S.
of (\ref{representation188}) and using relation (\ref{representation243}) for
the part corresponding to $R_{jk}^{\left(  \varepsilon\right)  }$ and Lemma
\ref{representation54} to calculate the limit, as $\mu\rightarrow0,$ of the
part $G_{2,O,O^{\prime}}^{\left(  0\right)  },$ we conclude the proof.
\end{proof}

Now let us study the singularities of $G_{2}^{\left(  \varepsilon\right)
}\left(  \xi,\xi^{\prime};E\right)  $ for $\xi^{\prime}=\xi.$ For an arbitrary
$\omega_{0}\in\mathbb{S}^{2},$ we introduce cut-off functions $\Psi_{\pm
}\left(  \hat{\xi},\hat{\xi}^{\prime};\omega_{0}\right)  \in C^{\infty}\left(
\mathbb{S}^{2}\times\mathbb{S}^{2}\right)  ,$ supported on $\{\left(  \hat
{\xi},\hat{\xi}^{\prime}\right)  \in\mathbb{S}^{2}\times\mathbb{S}^{2}%
|\hat{\xi},\hat{\xi}^{\prime}\in\Omega_{\pm}\left(  \omega_{0},\delta\right)
\}$ and consider $\Psi_{\pm}\left(  \hat{\xi},\hat{\xi}^{\prime};\omega
_{0}\right)  G_{2}^{\left(  \varepsilon\right)  }\left(  \xi,\xi^{\prime
};E\right)  .$ We need the following result

\begin{lemma}
Let us take $\omega,\theta\in\Omega_{+}\left(  \omega_{0},\delta\right)  $ or
$\omega,\theta\in\Omega_{-}\left(  \omega_{0},\delta\right)  $. Suppose that
$\left\langle \theta,\hat{x}\right\rangle >\varepsilon,$ where $\varepsilon
>\sqrt{1-\delta^{2}}.$ Then, $\left\langle \omega,\hat{x}\right\rangle
>-\varepsilon.$
\end{lemma}

\begin{proof}
The case $\theta=\omega_{0}$ is immediate. Suppose that $\theta\neq\omega
_{0}.$ We prove the case $\omega,\theta\in\Omega_{+}\left(  \omega_{0}%
,\delta\right)  .$ The proof for the case $\omega,\theta\in\Omega_{-}\left(
\omega_{0},\delta\right)  $ is analogous. Let $\hat{y}$ be a unit vector in
the plane generated by $\theta$ and $\omega_{0},$ that is orthogonal to
$\omega_{0}:\left\langle \omega_{0},\hat{y}\right\rangle =0.$ We write
$\theta$ as $\theta=\left\langle \theta,\omega_{0}\right\rangle \omega
_{0}+\left\langle \theta,\hat{y}\right\rangle \hat{y}.$ Let $x$ be such that
$\left\langle \theta,\hat{x}\right\rangle >\varepsilon.$ Then, from the
relation $\left\langle \theta,\hat{x}\right\rangle =\left\langle \theta
,\omega_{0}\right\rangle \left\langle \omega_{0},\hat{x}\right\rangle
+\left\langle \theta,\hat{y}\right\rangle \left\langle \hat{y},\hat
{x}\right\rangle ,$ and $\left\vert \left\langle \theta,\hat{y}\right\rangle
\right\vert <\sqrt{1-\delta^{2}},$ for $\theta\in\Omega_{+}\left(  \omega
_{0},\delta\right)  ,$ it follows that $\left\langle \omega_{0},\hat
{x}\right\rangle >0.$ For $\omega\in\Omega_{+}\left(  \omega_{0}%
,\delta\right)  $ and for $\hat{z}$ such that $\left\langle \hat{y},\hat
{z}\right\rangle =\left\langle \omega_{0},\hat{z}\right\rangle =0,$ we have
$\omega=\left\langle \omega,\omega_{0}\right\rangle \omega_{0}+\left\langle
\omega,\hat{y}\right\rangle \hat{y}+\left\langle \omega,\hat{z}\right\rangle
\hat{z}.$ Then, it follows%
\begin{equation}
\left.  \left\langle \omega,\hat{x}\right\rangle =\left\langle \omega
,\omega_{0}\right\rangle \left\langle \omega_{0},\hat{x}\right\rangle
+\left\langle \omega,\hat{y}\right\rangle \left\langle \hat{y},\hat
{x}\right\rangle +\left\langle \omega,\hat{z}\right\rangle \left\langle
\hat{z},\hat{x}\right\rangle >\left\langle \omega,\hat{y}\right\rangle
\left\langle \hat{y},\hat{x}\right\rangle +\left\langle \omega,\hat
{z}\right\rangle \left\langle \hat{z},\hat{x}\right\rangle .\text{ }\right.
\label{representation173}%
\end{equation}
If $\left\langle \omega,\hat{y}\right\rangle \hat{y}+\left\langle \omega
,\hat{z}\right\rangle \hat{z}=0\ $or $\left\langle \hat{x},\hat{y}%
\right\rangle \hat{y}+\left\langle \hat{x},\hat{z}\right\rangle \hat{z}=0,$
then from (\ref{representation173}) we get $\left\langle \omega,\hat
{x}\right\rangle >0>-\delta.$ Suppose that $\left\langle \omega,\hat
{y}\right\rangle \hat{y}+\left\langle \omega,\hat{z}\right\rangle \hat{z}%
\neq0$ and $\left\langle \hat{y},\hat{x}\right\rangle \hat{y}+\left\langle
\hat{z},\hat{x}\right\rangle \hat{z}\neq0.$ Let us define $\omega_{\hat
{y},\hat{z}}=\frac{1}{\sqrt{\left\langle \omega,\hat{y}\right\rangle
^{2}+\left\langle \omega,\hat{z}\right\rangle ^{2}}}\left(  \left\langle
\omega,\hat{y}\right\rangle \hat{y}+\left\langle \omega,\hat{z}\right\rangle
\hat{z}\right)  $ and $\hat{x}_{\hat{y},\hat{z}}=\frac{1}{\sqrt{\left\langle
\hat{x},\hat{y}\right\rangle ^{2}+\left\langle \hat{x},\hat{z}\right\rangle
^{2}}}\left(  \left\langle \hat{x},\hat{y}\right\rangle \hat{y}+\left\langle
\hat{x},\hat{z}\right\rangle \hat{z}\right)  .$ Then, using that
$\varepsilon>\sqrt{1-\delta^{2}}$, it follows from (\ref{representation173})
that $\left\langle \omega,\hat{x}\right\rangle >\left(  \sqrt{\left\langle
\omega,\hat{y}\right\rangle ^{2}+\left\langle \omega,\hat{z}\right\rangle
^{2}}\right)  \left(  \sqrt{\left\langle \hat{x},\hat{y}\right\rangle
^{2}+\left\langle \hat{x},\hat{z}\right\rangle ^{2}}\right)  \left\langle
\omega_{\hat{y},\hat{z}},\hat{x}_{\hat{y},\hat{z}}\right\rangle >-\sqrt
{1-\delta^{2}}>-\varepsilon.$
\end{proof}

Note that, $\partial_{x_{j}}\zeta_{-}^{\pm}\left(  x,\xi^{\prime}\right)  $ is
equal to $0$ for $\pm\left\langle \hat{\xi}^{\prime},\hat{x}\right\rangle
<\varepsilon$ and $\zeta_{+}^{\pm}\left(  x,\xi\right)  =1$ for $\pm
\left\langle \hat{\xi},\hat{x}\right\rangle >-\varepsilon$ and $\left\vert
\xi\right\vert \geq c_{1}.$ Then, Lemma above implies $\zeta_{+}^{\pm}\left(
x,\xi\right)  \partial_{x_{j}}\zeta_{-}^{\pm}\left(  x,\xi^{\prime}\right)
=\partial_{x_{j}}\zeta_{-}^{\pm}\left(  x,\xi^{\prime}\right)  ,$ for
$\left\vert \xi\right\vert \geq c_{1}$ and $\hat{\xi},\hat{\xi}^{\prime}%
\in\Omega_{+}\left(  \omega_{0},\delta\right)  $ or $\hat{\xi},\hat{\xi
}^{\prime}\in\Omega_{-}\left(  \omega_{0},\delta\right)  .$ Thus, from
(\ref{representation62}) we obtain the following equality, for $\hat{\xi}%
,\hat{\xi}^{\prime}\in\Omega_{+}\left(  \omega_{0},\delta\right)  $ or
$\hat{\xi},\hat{\xi}^{\prime}\in\Omega_{-}\left(  \omega_{0},\delta\right)
,$
\begin{equation}
\left.  G_{2}^{\left(  \varepsilon\right)  }\left(  \xi,\xi^{\prime};E\right)
=-i\varsigma\left(  \left\vert \xi\right\vert \right)  \varsigma\left(
\left\vert \xi^{\prime}\right\vert \right)
{\displaystyle\sum\limits_{j=1}^{3}}
{\displaystyle\int}
\left(  u_{N}^{+}\left(  x,\xi;\lambda\left(  \xi\right)  \right)  \right)
^{\ast}\alpha_{j}u_{N}^{-}\left(  x,\xi^{\prime};\lambda\left(  \xi^{\prime
}\right)  \right)  \left(  \partial_{x_{j}}\zeta_{-}^{\left(
\operatorname*{sgn}E\right)  }\left(  x,\xi^{\prime}\right)  \right)
\varphi\left(  \varepsilon x\right)  dx.\right.  \label{representation156}%
\end{equation}

Recall the notation of Remark \ref{representation75}. We need the two
following results (see Proposition 5.4 and 5.5 of \cite{30})

\begin{lemma}
\label{representation147}Let us consider $\mathcal{A}_{\pm}^{\left(
\varepsilon\right)  }\left(  \xi,\xi^{\prime};E\right)  =\int\limits_{\Pi
_{\omega_{0}}}e^{i\left\langle y,\xi^{\prime}-\xi\right\rangle }g_{\pm}\left(
y,\xi,\xi^{\prime};E\right)  \varphi\left(  \varepsilon y\right)  dy,$ where
$\varphi\in C_{0}^{\infty}\left(  \mathbb{R}^{3}\right)  $ is such that
$\varphi\left(  0\right)  =1,$ $\Pi_{\omega_{0}}:=\left\{  y\in\mathbb{R}%
^{3}|\left\langle y,\omega_{0}\right\rangle =0,\omega_{0}\in\mathbb{S}%
^{2}\right\}  $ and $g_{\pm}\in\mathcal{S}^{p}$ for some real $p,$ satisfies
\begin{equation}
\operatorname*{supp}g_{\pm}\subset\{\left(  \hat{\xi},\hat{\xi}^{\prime
}\right)  \in\mathbb{R}^{3}\times\mathbb{R}^{3}\mid\hat{\xi},\hat{\xi}%
^{\prime}\in\Omega_{\pm}\left(  \omega_{0},\delta\right)  ,\text{ }%
\delta>0,\text{ }\left\vert \xi^{\prime}\right\vert \geq c\}.
\label{representation228}%
\end{equation}
Let $g_{1},$ $f_{1}$ and $f_{2}$ be as in (\ref{representation134})$.$ Then,
for $\left\vert E\right\vert >m$ and even $n$ we have
\begin{equation}
\left.
\begin{array}
[c]{c}%
\lim_{\mu\downarrow0}\lim_{\varepsilon\rightarrow0}%
{\displaystyle\int_{\mathbb{S}^{2}}}
{\displaystyle\int}
\left(  \mathcal{A}_{\pm}^{\left(  \varepsilon\right)  }\left(  \nu\left(
E\right)  \omega,\xi^{\prime};E\right)  \tilde{\delta}_{\mu}^{\left(
\operatorname*{sgn}E\right)  }\hat{g}_{1}\left(  \xi^{\prime}\right)
,f_{2}\left(  \omega\right)  \right)  d\xi^{\prime}d\omega\\
=\upsilon\left(  E\right)
{\displaystyle\int_{\Pi_{\omega_{0}}}}
{\displaystyle\int_{\Pi_{\omega_{0}}}}
{\displaystyle\int_{\Pi_{\omega_{0}}}}
e^{i\nu\left(  E\right)  \left\langle y,\zeta^{\prime}-\zeta\right\rangle
}\left\langle \nu\left(  E\right)  y\right\rangle ^{-n}\left\langle
D_{\zeta^{\prime}}\right\rangle ^{n}\left(  \tilde{g}_{\pm}^{\prime}\left(
y,\nu\left(  E\right)  \zeta,\nu\left(  E\right)  \zeta^{\prime};E\right)
\tilde{f}_{1}\left(  \zeta^{\prime}\right)  ,\tilde{f}_{2}\left(
\zeta\right)  \right)  dyd\zeta^{\prime}d\zeta,
\end{array}
\right.  \label{representation151}%
\end{equation}
where $\left\langle \nu\left(  E\right)  y\right\rangle =\sqrt{1+\left(
\nu\left(  E\right)  y\right)  ^{2}},$ $\left\langle D_{\zeta^{\prime}%
}\right\rangle ^{2}=1-\partial_{\zeta^{\prime}}^{2}$ and $\tilde{g}_{\pm
}^{\prime}\left(  y,\nu\left(  E\right)  \zeta,\nu\left(  E\right)
\zeta^{\prime};E\right)  :=\frac{\tilde{g}_{\pm}\left(  y,\nu\left(  E\right)
\zeta,\nu\left(  E\right)  \zeta^{\prime};E\right)  }{\left(  \left(
1-\left\vert \zeta\right\vert ^{2}\right)  \left(  1-\left\vert \zeta^{\prime
}\right\vert ^{2}\right)  \right)  ^{1/2}}.$
\end{lemma}

We define $\Pi_{\omega_{0}}^{\pm}\left(  E\right)  :=\{x\in\mathbb{R}%
^{3}|x=z\omega_{0}+y,$ $y\in\Pi_{\omega_{0}}$ and $\pm\left(
\operatorname*{sgn}E\right)  z\geq0\}.$

\begin{lemma}
\label{representation176}Let us consider $\mathcal{A}_{\pm}^{\left(
\varepsilon\right)  }\left(  \xi,\xi^{\prime};E\right)  =\left(  \left\vert
\lambda\left(  \xi\right)  \right\vert -\left\vert \lambda\left(  \xi^{\prime
}\right)  \right\vert \right)
{\displaystyle\int\limits_{\Pi_{\omega_{0}}^{\pm}\left(  E\right)  }}
e^{i\left\langle x,\xi^{\prime}-\xi\right\rangle }g_{\pm}\left(  x,\xi
,\xi^{\prime};E\right)  \varphi\left(  \varepsilon x\right)  dx,$ where
$\varphi\in C_{0}^{\infty}\left(  \mathbb{R}^{3}\right)  $ is such that
$\varphi\left(  0\right)  =1,$ and $g_{\pm}$ satisfies the assumption
(\ref{representation228}) and the estimate
\begin{equation}
\left\vert \partial_{x}^{\alpha}\partial_{\xi}^{\beta}\partial_{\xi^{\prime}%
}^{\beta^{\prime}}g_{\pm}\left(  x,\xi,\xi^{\prime};E\right)  \right\vert \leq
C_{\alpha,\beta,\beta^{\prime}}\left(  1+\left\vert x\right\vert \right)
^{p-\left\vert \alpha\right\vert }, \label{representation227}%
\end{equation}
for some real $p,$ and all $x\in\Pi_{\omega_{0}}^{\pm}\left(  E\right)  $.
Moreover, the following relation holds
\begin{equation}
g_{\pm}\left(  x,\xi,\xi^{\prime};E\right)  =0\text{ if\ }\left(
\operatorname*{sgn}E\right)  \left\langle \eta,x\right\rangle \geq
c_{0}\left\vert \eta\right\vert \left\vert x\right\vert ,\text{ for }\eta
=\xi+\xi^{\prime},\text{ }c_{0}\in\left(  0,1\right)  ,
\label{representation226}%
\end{equation}
for all $x\in\Pi_{\omega_{0}}^{\pm}\left(  E\right)  $, $\left\vert
x\right\vert \geq R.$ Then, we have
\begin{equation}
\lim_{\mu\downarrow0}\lim_{\varepsilon\rightarrow0}%
{\displaystyle\int_{\mathbb{S}^{2}}}
{\displaystyle\int}
\left(  \mathcal{A}_{\pm}^{\left(  \varepsilon\right)  }\left(  \nu\left(
E\right)  \omega,\xi^{\prime};E\right)  \tilde{\delta}_{\mu}^{\left(
\operatorname*{sgn}E\right)  }\hat{g}_{1}\left(  \xi^{\prime}\right)
,f_{2}\left(  \omega\right)  \right)  d\xi^{\prime}d\omega=0.
\label{representation175}%
\end{equation}

\end{lemma}

We want to find an expression for $G_{2}^{\left(  \varepsilon\right)  }\left(
\xi,\xi^{\prime};E\right)  $ independent on the cut-off functions $\zeta
_{-}^{\left(  \operatorname*{sgn}E\right)  }$. If $\hat{\xi}^{\prime}\in
\Omega_{\pm}\left(  \omega_{0},\delta\right)  ,$ the function $\partial
_{x_{j}}\zeta_{-}^{\left(  \operatorname*{sgn}E\right)  }\left(  x,\xi
^{\prime}\right)  $ is equal to zero for $\pm\left(  \operatorname*{sgn}%
E\right)  z<0$, so we can consider the integral in (\ref{representation156})
only in the region $\Pi_{\omega_{0}}^{\pm}\left(  E\right)  .$ Integrating by
parts in $G_{2}^{\left(  \varepsilon\right)  }$ and noting that $\zeta
_{-}^{\left(  \operatorname*{sgn}E\right)  }\left(  y,\xi^{\prime}\right)
=1,$ for $\left\vert \xi^{\prime}\right\vert \geq c_{1},$ we get
\begin{equation}
\Psi_{\pm}\left(  \hat{\xi},\hat{\xi}^{\prime};\omega_{0}\right)
G_{2}^{\left(  \varepsilon\right)  }\left(  \xi,\xi^{\prime};E\right)
=\Psi_{\pm}\left(  \hat{\xi},\hat{\xi}^{\prime};\omega_{0}\right)  \left(
\pm\breve{G}_{2}^{\left(  \varepsilon\right)  }\left(  \xi,\xi^{\prime
};E\right)  +R_{\pm}^{\left(  \varepsilon\right)  }\left(  \xi,\xi^{\prime
};E\right)  \right)  , \label{representation159}%
\end{equation}
where%
\begin{equation}
\left.  \breve{G}_{2}^{\left(  \varepsilon\right)  }\left(  \xi,\xi^{\prime
};E\right)  :=i\left(  \operatorname*{sgn}E\right)  \varsigma\left(
\left\vert \xi\right\vert \right)  \varsigma\left(  \left\vert \xi^{\prime
}\right\vert \right)
{\displaystyle\int\limits_{\Pi_{\omega_{0}}}}
\left(  u_{N}^{+}\left(  y,\xi;\lambda\left(  \xi\right)  \right)  \right)
^{\ast}\left(  \alpha\cdot\omega_{0}\right)  u_{N}^{-}\left(  y,\xi^{\prime
};\lambda\left(  \xi^{\prime}\right)  \right)  \varphi\left(  \varepsilon
y\right)  dy,\right.  \label{representation150}%
\end{equation}
and $R_{\pm}^{\left(  \varepsilon\right)  }\left(  \xi,\xi^{\prime};E\right)
:=i\left(  \operatorname*{sgn}E\right)  \varsigma\left(  \left\vert
\xi\right\vert \right)  \varsigma\left(  \left\vert \xi^{\prime}\right\vert
\right)
{\displaystyle\sum\limits_{i=1}^{3}}
$ $\
{\displaystyle\int\limits_{\Pi_{\omega_{0}}^{\pm}\left(  E\right)  }}
\partial_{x_{i}}\left(  \left(  u_{N}^{+}\left(  x,\xi;\lambda\left(
\xi\right)  \right)  ^{\ast}\alpha_{i}u_{N}^{-}\left(  x,\xi^{\prime}%
;\lambda\left(  \xi^{\prime}\right)  \right)  \right)  \varphi\left(
\varepsilon x\right)  \right)  \zeta_{-}^{\left(  \operatorname*{sgn}E\right)
}\left(  x,\xi^{\prime}\right)  dx.$ Using the definition (\ref{eig30}) of the
functions $r_{N}^{_{\pm}}\left(  x,\xi;E\right)  $ we obtain $i\sum_{j=1}%
^{3}\partial_{x_{j}}\left(  \left(  u_{N}^{_{+}}\left(  x,\xi;\lambda\left(
\xi\right)  \right)  \right)  ^{\ast}\alpha_{j}u_{N}^{_{-}}\left(
x,\xi^{\prime};\lambda\left(  \xi^{\prime}\right)  \right)  \right)
=e^{i\left\langle x,\xi^{\prime}-\xi\right\rangle }$ $\times\lbrack\left(
r_{N}^{_{+}}\left(  x,\xi;\lambda\left(  \xi\right)  \right)  \right)  ^{\ast
}a_{N}^{_{-}}\left(  x,\xi^{\prime};\lambda\left(  \xi^{\prime}\right)
\right)  $ $-\left(  a_{N}^{_{+}}\left(  x,\xi;\lambda\left(  \xi\right)
\right)  \right)  ^{\ast}r_{N}^{_{-}}\left(  x,\xi^{\prime};\lambda\left(
\xi^{\prime}\right)  \right)  +\left(  \lambda\left(  \xi\right)
-\lambda\left(  \xi^{\prime}\right)  \right)  \left(  a_{N}^{_{+}}\left(
x,\xi;\lambda\left(  \xi\right)  \right)  \right)  ^{\ast}a_{N}^{_{-}}\left(
x,\xi^{\prime};\lambda\left(  \xi^{\prime}\right)  \right)  ].$

Let us decompose $\Psi_{\pm}\left(  \hat{\xi},\hat{\xi}^{\prime};\omega
_{0}\right)  R_{\pm}^{\left(  \varepsilon\right)  }$ in the sum
\begin{equation}
\Psi_{\pm}\left(  \hat{\xi},\hat{\xi}^{\prime};\omega_{0}\right)  R_{\pm
}^{\left(  \varepsilon\right)  }=\Psi_{\pm}\left(  \hat{\xi},\hat{\xi}%
^{\prime};\omega_{0}\right)  \left(  \left(  R_{1}^{\left(  \varepsilon
\right)  }\right)  _{\pm}+\left(  R_{2}^{\left(  \varepsilon\right)  }\right)
_{\pm}+\left(  R_{3}^{\left(  \varepsilon\right)  }\right)  _{\pm}\right)  ,
\label{representation249}%
\end{equation}
where\
\[
\left.
\begin{array}
[c]{c}%
\left(  R_{1}^{\left(  \varepsilon\right)  }\right)  _{\pm}\left(  \xi
,\xi^{\prime};E\right)  :=i\varepsilon\left(  \operatorname*{sgn}E\right)
\varsigma\left(  \left\vert \xi\right\vert \right)  \varsigma\left(
\left\vert \xi^{\prime}\right\vert \right) \\
\times%
{\displaystyle\sum\limits_{j=1}^{3}}
\text{ }\left.
{\displaystyle\int\limits_{\Pi_{\omega_{0}}^{\pm}\left(  E\right)  }}
e^{i\left\langle x,\xi^{\prime}-\xi\right\rangle }\left(  \left(  a_{N}^{_{+}%
}\left(  x,\xi;\lambda\left(  \xi\right)  \right)  \right)  ^{\ast}\alpha
_{j}a_{N}^{_{-}}\left(  x,\xi^{\prime};\lambda\left(  \xi^{\prime}\right)
\right)  \right)  \zeta_{-}^{\left(  \operatorname*{sgn}E\right)  }\left(
x,\xi^{\prime}\right)  \left(  \partial_{x_{j}}\varphi\right)  \left(
\varepsilon x\right)  dx\right.  ,
\end{array}
\right.
\]
$\ $%
\[
\left.
\begin{array}
[c]{c}%
\ \left(  R_{2}^{\left(  \varepsilon\right)  }\right)  _{\pm}\left(  \xi
,\xi^{\prime};E\right)  :=\left(  \operatorname*{sgn}E\right)  \varsigma
\left(  \left\vert \xi\right\vert \right)  \varsigma\left(  \left\vert
\xi^{\prime}\right\vert \right) \\
\times%
{\displaystyle\int\limits_{\Pi_{\omega_{0}}^{\pm}\left(  E\right)  }}
e^{i\left\langle x,\xi^{\prime}-\xi\right\rangle }\left(  \left(  r_{N}%
^{+}\left(  x,\xi;\lambda\left(  \xi\right)  \right)  \right)  ^{\ast}%
a_{N}^{-}\left(  x,\xi^{\prime};\lambda\left(  \xi^{\prime}\right)  \right)
-\left(  a_{N}^{+}\left(  x,\xi;\lambda\left(  \xi\right)  \right)  \right)
^{\ast}r_{N}^{-}\left(  x,\xi^{\prime};\lambda\left(  \xi^{\prime}\right)
\right)  \right)  \zeta_{-}^{\left(  \operatorname*{sgn}E\right)  }\left(
x,\xi^{\prime}\right)  \varphi\left(  \varepsilon x\right)  dx,
\end{array}
\right.
\]
and
\[
\left.
\begin{array}
[c]{c}%
\left(  R_{3}^{\left(  \varepsilon\right)  }\right)  _{\pm}\left(  \xi
,\xi^{\prime};E\right)  :=\left(  \left\vert \lambda\left(  \xi\right)
\right\vert -\left\vert \lambda\left(  \xi^{\prime}\right)  \right\vert
\right)  \varsigma\left(  \left\vert \xi\right\vert \right)  \varsigma\left(
\left\vert \xi^{\prime}\right\vert \right) \\
\times%
{\displaystyle\int\limits_{\Pi_{\omega_{0}}^{\pm}\left(  E\right)  }}
e^{i\left\langle x,\xi^{\prime}-\xi\right\rangle }\left(  a_{N}^{+}\left(
x,\xi;\lambda\left(  \xi\right)  \right)  \right)  ^{\ast}a_{N}^{-}\left(
x,\xi^{\prime};\lambda\left(  \xi^{\prime}\right)  \right)  \zeta_{-}^{\left(
\operatorname*{sgn}E\right)  }\left(  x,\xi^{\prime}\right)  \varphi\left(
\varepsilon x\right)  dx.
\end{array}
\right.
\]
We first prove the following

\begin{lemma}
\label{representation184}The limit $\Psi_{\pm}\left(  \hat{\xi},\hat{\xi
}^{\prime};\omega_{0}\right)  \left(  R_{2}^{\left(  0\right)  }\right)
_{\pm}\left(  \xi,\xi^{\prime};E\right)  :=\lim\limits_{\varepsilon
\rightarrow0}\Psi_{\pm}\left(  \hat{\xi},\hat{\xi}^{\prime};\omega_{0}\right)
\left(  R_{2}^{\left(  \varepsilon\right)  }\right)  _{\pm}\left(  \xi
,\xi^{\prime};E\right)  $ exists. For any $p$ and $q,$ there is $N,$ such that
the function $\Psi_{\pm}\left(  \hat{\xi},\hat{\xi}^{\prime};\omega
_{0}\right)  \left(  R_{2}^{\left(  0\right)  }\right)  _{\pm}\left(
\nu\left(  E\right)  \hat{\xi},\nu\left(  E\right)  \hat{\xi}^{\prime
};E\right)  $ is of the $C^{p}\left(  \mathbb{S}^{2}\times\mathbb{S}%
^{2}\right)  $ class, and its $C^{p}-$norm is bounded by $C\left\vert
E\right\vert ^{-q},$ as $\left\vert E\right\vert \rightarrow\infty.$ Moreover,
the following relation holds
\begin{equation}
\left.
\begin{array}
[c]{c}%
-i\left(  2\pi\right)  ^{-2}\upsilon\left(  E\right)  \lim_{\mu\downarrow
0}\lim_{\varepsilon\rightarrow0}%
{\displaystyle\int}
{\displaystyle\int_{\mathbb{S}^{2}}}
\left(  \Psi_{\pm}\left(  \omega,\hat{\xi}^{\prime}\,;\omega_{0}\right)
R_{\pm}^{\left(  \varepsilon\right)  }\left(  \nu\left(  E\right)  \omega
,\xi^{\prime};E\right)  \tilde{\delta}_{\mu}^{\left(  \operatorname*{sgn}%
E\right)  }\hat{g}_{1}\left(  \xi^{\prime}\right)  ,f_{2}\left(
\omega\right)  \right)  d\omega d\xi^{\prime}\\
=-i\left(  2\pi\right)  ^{-2}\upsilon\left(  E\right)  ^{2}%
{\displaystyle\int_{\mathbb{S}^{2}}}
{\displaystyle\int_{\mathbb{S}^{2}}}
\left(  \Psi_{\pm}\left(  \omega,\theta;\omega_{0}\right)  \left(
R_{2}^{\left(  0\right)  }\right)  _{\pm}\left(  \nu\left(  E\right)
\omega,\nu\left(  E\right)  \theta;E\right)  f_{1}\left(  \theta\right)
,f_{2}\left(  \omega\right)  \right)  d\theta d\omega.
\end{array}
\right.  \label{representation161}%
\end{equation}

\end{lemma}

\begin{proof}
Since $\zeta_{-}^{\left(  \operatorname*{sgn}E\right)  }$ is supported on
$\Xi^{-}\left(  E\right)  ,$ $a_{N}^{_{-}}\zeta_{-}^{\left(
\operatorname*{sgn}E\right)  }$ satisfies the estimate (\ref{eig27}). If
$x\in\Pi_{\omega_{0}}^{\pm}\left(  E\right)  $ and $\hat{\xi}\in\Omega_{\pm
}\left(  \omega_{0},\delta\right)  ,$ using (\ref{representation63}), we get
$\left(  x,\xi\right)  \in$ $\Xi^{+}\left(  E\right)  .$ Then, $\left(
a_{N}^{_{+}}\left(  x,\xi;\lambda\left(  \xi\right)  \right)  \right)  ^{\ast
}$ also satisfies (\ref{eig27}). Thus, for $x\in\Pi_{\omega_{0}}^{\pm}\left(
E\right)  $ and $\hat{\xi}\in\Omega_{\pm}\left(  \omega_{0},\delta\right)  $,
we obtain $\left\vert \partial_{x}^{\alpha}\partial_{\xi}^{\beta}\partial
_{\xi^{\prime}}^{\beta^{\prime}}\left(  \left(  a_{N}^{_{+}}\left(
x,\xi;\lambda\left(  \xi\right)  \right)  \right)  ^{\ast}\alpha_{j}%
a_{N}^{_{-}}\left(  x,\xi^{\prime};\lambda\left(  \xi^{\prime}\right)
\right)  \right)  \zeta_{-}^{\left(  \operatorname*{sgn}E\right)  }\left(
x,\xi^{\prime}\right)  \right\vert \leq C_{\alpha,\beta,\beta^{\prime}}\left(
1+\left\vert x\right\vert \right)  ^{-\left\vert \alpha\right\vert }\left\vert
\xi\right\vert ^{-\left\vert \beta\right\vert }\left\vert \xi^{\prime
}\right\vert ^{-\left\vert \beta^{\prime}\right\vert },$ for all indices
$\alpha$ and $\beta$ . This estimate implies that for all $f,g\in
\mathcal{S}\left(  \mathbb{R}^{3};\mathbb{C}^{4}\right)  ,$ (see relation
(\ref{representation243})
\begin{equation}
\lim_{\varepsilon\rightarrow0}\int\int\left(  \Psi_{\pm}\left(  \hat{\xi}%
,\hat{\xi}^{\prime};\omega_{0}\right)  \left(  R_{1}^{\left(  \varepsilon
\right)  }\right)  _{\pm}\left(  \xi,\xi^{\prime};E\right)  f\left(
\xi^{\prime}\right)  ,g\left(  \xi\right)  \right)  d\xi^{\prime}d\xi=0.
\label{representation157}%
\end{equation}
The proof of this relation is analogous to that of relation
(\ref{representation243}).

Now observe that the functions $r_{N}^{-}\left(  x,\xi^{\prime};\lambda\left(
\xi^{\prime}\right)  \right)  \zeta_{-}^{\left(  \operatorname*{sgn}E\right)
}\left(  x,\xi^{\prime}\right)  $ and $a_{N}^{-}\left(  x,\xi^{\prime}%
;\lambda\left(  \xi^{\prime}\right)  \right)  \zeta_{-}^{\left(
\operatorname*{sgn}E\right)  }\left(  x,\xi^{\prime}\right)  $ satisfy the
estimates (\ref{eig28}) and (\ref{eig27}), respectively, for all
$x,\xi^{\prime}\in\mathbb{R}^{3}$. Moreover, the estimate (\ref{eig28}) for
the function $\left(  r_{N}^{+}\left(  x,\xi;\lambda\left(  \xi\right)
\right)  \right)  ^{\ast}$ and the estimate (\ref{eig27}) for $\left(
a_{N}^{+}\left(  x,\xi;\lambda\left(  \xi\right)  \right)  \right)  ^{\ast}$
hold for $x\in\Pi_{\omega_{0}}^{\pm}\left(  E\right)  $ and $\hat{\xi}%
\in\Omega_{\pm}\left(  \omega_{0},\delta\right)  $. Hence, for $N$ big enough,
we conclude that the limit, as $\varepsilon\rightarrow0,$ of $\Psi_{\pm
}\left(  \hat{\xi},\hat{\xi}^{\prime};\omega_{0}\right)  \left(
R_{2}^{\left(  \varepsilon\right)  }\right)  _{\pm}\left(  \xi,\xi^{\prime
};E\right)  $ exists and it is equal to $\Psi_{\pm}\left(  \hat{\xi},\hat{\xi
}^{\prime};\omega_{0}\right)  \left(  R_{2}^{\left(  0\right)  }\right)
_{\pm}\left(  \xi,\xi^{\prime};E\right)  .$ Moreover, for any $p$ and $q$
there exist $N,$ such that $\Psi_{\pm}\left(  \hat{\xi},\hat{\xi}^{\prime
};\omega_{0}\right)  \left(  R_{2}^{\left(  0\right)  }\right)  _{\pm}\left(
\xi,\xi^{\prime};E\right)  $ is a $C^{p}-$function of variables $\xi$ and
$\xi^{\prime},$ and its $C^{p}-$norm is bounded by $C\left\vert E\right\vert
^{-q},$ as $\left\vert E\right\vert \rightarrow\infty$.

Note that $\left\vert \partial_{x}^{\alpha}\partial_{\xi}^{\beta}\partial
_{\xi^{\prime}}^{\beta^{\prime}}\left(  \left(  a_{N}^{+}\left(  x,\xi
;\lambda\left(  \xi\right)  \right)  \right)  ^{\ast}a_{N}^{-}\left(
x,\xi^{\prime};\lambda\left(  \xi^{\prime}\right)  \right)  \zeta_{-}^{\left(
\operatorname*{sgn}E\right)  }\left(  x,\xi^{\prime}\right)  \right)
\right\vert \leq C_{\alpha,\beta,\beta^{\prime}}\left(  1+\left\vert
x\right\vert \right)  ^{-\left\vert \alpha\right\vert }\left\vert
\xi\right\vert ^{-\left\vert \beta\right\vert }\left\vert \xi^{\prime
}\right\vert ^{-\left\vert \beta^{\prime}\right\vert },$ for all $x\in
\Pi_{\omega_{0}}^{\pm}\left(  E\right)  \ $and $\hat{\xi}\in\Omega_{\pm
}\left(  \omega_{0},\delta\right)  ,$ and all indices $\alpha,\beta
,\beta^{\prime}$. For some $0<\kappa_{0}<\kappa_{1},$ let $\chi\in
C_{0}^{\infty}\left(  \mathbb{R}^{+}\right)  $ be such that $\chi\left(
\kappa\right)  =1$ for $0\leq\kappa\leq\kappa_{0}$ and $\chi\left(
\kappa\right)  =0$ for $\kappa\geq\kappa_{1}.$ We split $\left(
R_{3}^{\left(  \varepsilon\right)  }\right)  _{\pm}$ in two parts,
$\chi\left(  \left\vert \xi-\xi^{\prime}\right\vert \right)  \left(
R_{3}^{\left(  \varepsilon\right)  }\right)  _{\pm}$ and $\left(
1-\chi\left(  \left\vert \xi-\xi^{\prime}\right\vert \right)  \right)  \left(
R_{3}^{\left(  \varepsilon\right)  }\right)  _{\pm}.$ Taking $\kappa_{1}$
small enough, we see that the cut-off function $\chi\left(  \left\vert \xi
-\xi^{\prime}\right\vert \right)  \zeta_{-}^{\left(  \operatorname*{sgn}%
E\right)  }\left(  x,\xi^{\prime}\right)  $ satisfies relation
(\ref{representation226}). Then, applying Lemma \ref{representation176} to the
term $\chi\left(  \left\vert \xi-\xi^{\prime}\right\vert \right)  \Psi_{\pm
}\left(  \hat{\xi},\hat{\xi}^{\prime};\omega_{0}\right)  \left(
R_{3}^{\left(  \varepsilon\right)  }\right)  _{\pm}$ and Lemma
\ref{representation145} to $\left(  1-\chi\left(  \left\vert \xi-\xi^{\prime
}\right\vert \right)  \right)  \Psi_{\pm}\left(  \hat{\xi},\hat{\xi}^{\prime
};\omega_{0}\right)  \left(  R_{3}^{\left(  \varepsilon\right)  }\right)
_{\pm}$ we have%
\begin{equation}
\lim_{\mu\downarrow0}\lim_{\varepsilon\rightarrow0}%
{\displaystyle\int_{\mathbb{S}^{2}}}
{\displaystyle\int}
\left(  \Psi_{\pm}\left(  \hat{\xi},\hat{\xi}^{\prime};\omega_{0}\right)
\left(  R_{3}^{\left(  \varepsilon\right)  }\right)  _{\pm}\left(  \xi
,\xi^{\prime};E\right)  \tilde{\delta}_{\mu}^{\left(  \operatorname*{sgn}%
E\right)  }\hat{g}_{1}\left(  \xi^{\prime}\right)  ,f_{2}\left(
\omega\right)  \right)  d\xi^{\prime}d\omega=0. \label{representation158}%
\end{equation}
Introducing decomposition (\ref{representation249}) in the L.H.S. of
(\ref{representation161}) and using relations (\ref{representation157}),
(\ref{representation158}) and Lemma \ref{representation54} to calculate the
resulting limit, we conclude that relation (\ref{representation161}) holds.\
\end{proof}

Let us now prove the following result

\begin{lemma}
\label{representation185}Let $g_{1},$ $f_{1}$ and $f_{2}$ be as in
(\ref{representation134}). For an arbitrary $\omega_{0}\in\mathbb{S}^{2},$ the
equality
\begin{equation}
\left.
\begin{array}
[c]{c}%
-i\left(  2\pi\right)  ^{-2}\upsilon\left(  E\right)  \lim_{\mu\downarrow
0}\lim_{\varepsilon\downarrow0}%
{\displaystyle\int}
{\displaystyle\int}
\left(  \pm\Psi_{\pm}\left(  \omega,\hat{\xi}^{\prime}\,;\omega_{0}\right)
\breve{G}_{2}^{\left(  \varepsilon\right)  }\left(  \nu\left(  E\right)
\omega,\xi^{\prime};E\right)  \tilde{\delta}_{\mu}^{\left(
\operatorname*{sgn}E\right)  }\hat{g}_{1}\left(  \xi^{\prime}\right)
,f_{2}\left(  \omega\right)  \right)  d\xi^{\prime}d\omega\\
=%
{\displaystyle\int_{\mathbb{S}^{2}}}
{\displaystyle\int_{\mathbb{S}^{2}}}
\left(  s_{\operatorname{sing}}^{(N)}\left(  \omega,\theta;E;\omega
_{0}\right)  f_{1}\left(  \theta\right)  ,f_{2}\left(  \omega\right)  \right)
d\theta d\omega,
\end{array}
\right.  \label{representation189}%
\end{equation}
holds, where $s_{\operatorname{sing}}^{(N)}\left(  \omega,\theta;E;\omega
_{0}\right)  $ is given by (\ref{representation248}).
\end{lemma}

\begin{proof}
Applying Lemma \ref{representation147} to $\pm\Psi_{\pm}\left(  \hat{\xi}%
,\hat{\xi}^{\prime};\omega_{0}\right)  \breve{G}_{2}^{\left(  \varepsilon
\right)  }\left(  \xi,\xi^{\prime};E\right)  $ we get
\begin{equation}
\left.
\begin{array}
[c]{c}%
-i\left(  2\pi\right)  ^{-2}\upsilon\left(  E\right)  \lim_{\mu\downarrow
0}\lim_{\varepsilon\downarrow0}%
{\displaystyle\int}
{\displaystyle\int}
\left(  \pm\Psi_{\pm}\left(  \omega,\hat{\xi}^{\prime}\,;\omega_{0}\right)
\breve{G}_{2}^{\left(  \varepsilon\right)  }\left(  \nu\left(  E\right)
\omega,\xi^{\prime};E\right)  \tilde{\delta}_{\mu}^{\left(
\operatorname*{sgn}E\right)  }\hat{g}_{1}\left(  \xi^{\prime}\right)
,f_{2}\left(  \omega\right)  \right)  d\xi^{\prime}d\omega\\
=\left(  2\pi\right)  ^{-2}\upsilon\left(  E\right)  ^{2}%
{\displaystyle\int_{\Pi_{\omega_{0}}}}
{\displaystyle\int_{\Pi_{\omega_{0}}}}
{\displaystyle\int_{\Pi_{\omega_{0}}}}
e^{i\nu\left(  E\right)  \left\langle y,\zeta^{\prime}-\zeta\right\rangle
}\left\langle \nu\left(  E\right)  y\right\rangle ^{-n}\left\langle
D_{\zeta^{\prime}}\right\rangle ^{n}\left(  \mathbf{\tilde{h}}_{N}^{\prime
}\left(  y,\zeta,\zeta^{\prime};E\right)  \tilde{f}_{1}\left(  \zeta^{\prime
}\right)  ,\tilde{f}_{2}\left(  \zeta\right)  \right)  d\zeta^{\prime}d\zeta,
\end{array}
\right.  \label{representation236}%
\end{equation}
for even $n,$ where $\mathbf{\tilde{h}}_{N,jk}^{\prime}\left(  y,\zeta
,\zeta^{\prime};E\right)  :=\pm\tilde{\Psi}_{\pm}\left(  \zeta,\zeta^{\prime
};\omega_{0}\right)  \frac{1}{\left(  1-\left\vert \zeta\right\vert
^{2}\right)  ^{1/2}}\frac{1}{\left(  1-\left\vert \zeta^{\prime}\right\vert
^{2}\right)  ^{1/2}}\mathbf{\tilde{h}}_{N}\left(  y,\zeta,\zeta^{\prime
};E\right)  .$ Integrating back by parts in the R.H.S of
(\ref{representation236}) and understanding the resulting expression as an
oscillatory integral, we obtain the expression (\ref{representation189}).
\end{proof}

Recall the function $\Psi_{1}\left(  \hat{\xi},\hat{\xi}^{\prime}\right)  ,$
defined above (\ref{representation248}). We are able to prove the following
result for $S_{2}\left(  E\right)  .$

\begin{theorem}
\label{representation241}Let $s_{2}\left(  \omega,\theta;E\right)  $ be the
kernel of the operator $S_{2}\left(  E\right)  ,$ defined as the limit
(\ref{representation134}). For any $p$ and $q,$ $\Psi_{1}\left(  \omega
,\theta\right)  $ $\times s_{2}\left(  \omega,\theta;E\right)  $ belongs to
the class $C^{p}\left(  \mathbb{S}^{2}\times\mathbb{S}^{2}\right)  $ and its
$C^{p}$-norm is a $O\left(  E^{-q}\right)  $ function. Moreover, for any $p$
and $q$ there exists $N,$ sufficiently large, such that, $\Psi_{\pm}\left(
\omega,\theta;\omega_{0}\right)  s_{2}\left(  \omega,\theta;E\right)
-s_{\operatorname{sing}}^{\left(  N\right)  }\left(  \omega,\theta
;E;\omega_{0}\right)  $ belongs to the class $C^{p}\left(  \mathbb{S}%
^{2}\times\mathbb{S}^{2}\right)  $, and moreover, its $C^{p}-$norm is bounded
by $C\left\vert E\right\vert ^{-q}$, as $\left\vert E\right\vert
\rightarrow\infty$. These estimates are uniform in $\omega_{0}\in
\mathbb{S}^{2}.$
\end{theorem}

\begin{proof}
From the relations (\ref{representation188}) and (\ref{representation166}) we
get
\begin{equation}
\left.
\begin{array}
[c]{c}%
\left(  \left(  S_{2}\left(  E\right)  \Psi\left(  \omega,\cdot;\omega
_{0}\right)  f_{1}\right)  \left(  \omega\right)  ,f_{2}\left(  \omega\right)
\right)  =-i\left(  2\pi\right)  ^{-2}\upsilon\left(  E\right) \\
\times\lim_{\mu\downarrow0}\lim_{\varepsilon\rightarrow0}%
{\displaystyle\int}
{\displaystyle\int_{\mathbb{S}^{2}}}
\left(  \Psi\left(  \omega,\hat{\xi}^{\prime}\right)  \left(  G_{1}^{\left(
\varepsilon\right)  }+G_{2}^{\left(  \varepsilon\right)  }\right)  \left(
\nu\left(  E\right)  \omega,\xi^{\prime};E\right)  \tilde{\delta}_{\mu
}^{\left(  \operatorname*{sgn}E\right)  }\hat{g}_{1}\left(  \xi^{\prime
}\right)  ,f_{2}\left(  \omega\right)  \right)  d\omega d\xi^{\prime},
\end{array}
\right.  \label{representation239}%
\end{equation}
where $\Psi\left(  \hat{\xi},\hat{\xi}^{\prime}\right)  $ is either $\Psi
_{\pm}\left(  \hat{\xi},\hat{\xi}^{\prime};\omega_{0}\right)  $ or $\Psi
_{1}\left(  \hat{\xi},\hat{\xi}^{\prime}\right)  .$ Suppose first that
$\Psi\left(  \hat{\xi},\hat{\xi}^{\prime}\right)  =\Psi_{1}\left(  \hat{\xi
},\hat{\xi}^{\prime}\right)  .$ By Lemma \ref{representation183} and Lemma
\ref{representation145} the limit (\ref{representation239}) exists and, for
any $p$ and $q$ there exists $N,$ such that $\Psi_{1}\left(  \hat{\xi}%
,\hat{\xi}^{\prime}\right)  \left(  G_{1}^{\left(  0\right)  }+G_{2}^{\left(
\varepsilon\right)  }\right)  \left(  \nu\left(  E\right)  \hat{\xi}%
,\nu\left(  E\right)  \hat{\xi}^{\prime};E\right)  $ belongs to the class
$C^{p}\left(  \mathbb{S}^{2}\times\mathbb{S}^{2}\right)  $, and its $C^{p}%
-$norm is bounded by $C\left\vert E\right\vert ^{-q}$, as $\left\vert
E\right\vert \rightarrow\infty$.

Now let us consider the case $\Psi\left(  \hat{\xi},\hat{\xi}^{\prime}\right)
=\Psi_{\pm}\left(  \hat{\xi},\hat{\xi}^{\prime};\omega_{0}\right)  .$ Again,
Lemma \ref{representation183} implies that the part in the limit
(\ref{representation239}) corresponding to the term $\Psi_{\pm}\left(
\hat{\xi},\hat{\xi}^{\prime};\omega_{0}\right)  G_{1}^{\left(  \varepsilon
\right)  }\left(  \nu\left(  E\right)  \hat{\xi},\xi^{\prime};E\right)  $
exists and the function $\Psi_{\pm}\left(  \hat{\xi},\hat{\xi}^{\prime}%
;\omega_{0}\right)  G_{1}^{\left(  \varepsilon\right)  }\left(  \nu\left(
E\right)  \hat{\xi},\xi^{\prime};E\right)  $ belongs to the class
$C^{p}\left(  \mathbb{S}^{2}\times\mathbb{S}^{2}\right)  $, and its $C^{p}%
-$norm is bounded by $C\left\vert E\right\vert ^{-q}$, as $\left\vert
E\right\vert \rightarrow\infty$. Using relation (\ref{representation159}) and
applying Lemma \ref{representation184} to the part $\Psi_{\pm}R_{\pm}^{\left(
\varepsilon\right)  }$ and Lemma \ref{representation185} to $\Psi_{\pm}%
\breve{G}_{2}^{\left(  \varepsilon\right)  },$ in order to calculate the limit
in the R.H.S of (\ref{representation239}) corresponding to the term $\Psi
_{\pm}\left(  \omega,\theta;\omega_{0}\right)  G_{2}^{\left(  \varepsilon
\right)  }\left(  \nu\left(  E\right)  \hat{\xi},\xi^{\prime};E\right)  ,$ and
noting that all the estimates are uniform on $\omega_{0}$ if the $C^{p}$-norms
of the function\ $\Psi_{\pm}\left(  \omega,\theta;\omega_{0}\right)  $ are
uniformly bounded on $\omega_{0}\in\mathbb{S}^{2},$ we complete the proof.
\end{proof}

Let us present a result that we use below (see Lemma 4.1 of \cite{47})

\begin{lemma}
\label{representation201}Let $f\left(  x,\xi\right)  \in C^{\infty}\left(
\mathbb{R}^{N}\times\mathbb{R}^{N}\right)  $ satisfy $\left\vert \partial
_{x}^{\alpha}f\left(  x,\xi\right)  \right\vert \leq C_{\alpha}\left\langle
x\right\rangle ^{-\rho-\left\vert \alpha\right\vert }$ for some $0<\rho<N.$
Then we have $\left\vert \int_{\mathbb{R}^{N}}(e^{-i\left\langle
x,\xi\right\rangle }f\left(  x,\xi\right)  )dx\right\vert \leq C\left\vert
\xi\right\vert ^{-\left(  N-\rho\right)  }$ as $\left\vert \xi\right\vert
\rightarrow0.$
\end{lemma}

Let us define the function\ $\mathbf{h}_{N}^{\operatorname{int}}$ by the
relation
\begin{equation}
\left.
\begin{array}
[c]{c}%
\mathbf{h}_{N}^{\operatorname{int}}\left(  y,\omega,\theta;E;\omega
_{0}\right)  :=\left(  \operatorname*{sgn}E\right)  (\left(  a^{+}\left(
y,\nu\left(  E\right)  \omega;E\right)  -P_{\omega}\left(  E\right)  \right)
^{\ast}\left(  \alpha\cdot\omega_{0}\right)  \left(  a^{_{-}}\left(
y,\nu\left(  E\right)  \theta;E\right)  -P_{\theta}\left(  E\right)  \right)
\\
+P_{\omega}\left(  E\right)  \left(  \alpha\cdot\omega_{0}\right)  (a^{_{-}%
}\left(  y,\nu\left(  E\right)  \theta;E\right)  -P_{\theta}\left(  E\right)
)+\left(  a^{+}\left(  y,\nu\left(  E\right)  \omega;E\right)  -P_{\omega
}\left(  E\right)  \right)  \left(  \alpha\cdot\omega_{0}\right)  P_{\theta
}\left(  E\right)  ).
\end{array}
\right.  \label{representation253}%
\end{equation}
We prove now the following

\begin{lemma}
\label{representation186}The function $\mathbf{b}_{N},$ given by
$\mathbf{b}_{N}\left(  \omega,\theta;E;\omega_{0}\right)  :=\pm\left(
2\pi\right)  ^{-2}\upsilon\left(  E\right)  ^{2}\Psi_{\pm}\left(
\omega,\theta;\omega_{0}\right)  \int\limits_{\Pi_{\omega_{0}}}e^{i\nu\left(
E\right)  \left\langle y,\theta-\omega\right\rangle }(\mathbf{h}_{N}\left(
y,\omega,\theta;E;\omega_{0}\right)  $ $-\left(  \operatorname*{sgn}E\right)
P_{\omega}\left(  E\right)  \left(  \alpha\cdot\omega_{0}\right)  P_{\theta
}\left(  E\right)  )dy,$ satisfies the estimate (\ref{representation83}).
\end{lemma}

\begin{proof}
We first note that $\mathbf{b}_{N}\left(  \omega,\theta;E;\omega_{0}\right)
=\pm\left(  2\pi\right)  ^{-2}\upsilon\left(  E\right)  ^{2}\Psi_{\pm}\left(
\omega,\theta;\omega_{0}\right)  \int\limits_{\Pi_{\omega_{0}}}e^{i\nu\left(
E\right)  \left\langle y,\theta-\omega\right\rangle }\mathbf{h}_{N}%
^{\operatorname{int}}\left(  y,\omega,\theta;E;\omega_{0}\right)  dy,$\ where
the function $\mathbf{h}_{N}^{\operatorname{int}}$ is defined by
$\mathbf{h}_{N}^{\operatorname{int}}\left(  y,\omega,\theta;E;\omega
_{0}\right)  :=\left(  \operatorname*{sgn}E\right)  (\left(  a^{+}\left(
y,\nu\left(  E\right)  \omega;E\right)  -P_{\omega}\left(  E\right)  \right)
^{\ast}\left(  \alpha\cdot\omega_{0}\right)  \left(  a^{_{-}}\left(
y,\nu\left(  E\right)  \theta;E\right)  -P_{\theta}\left(  E\right)  \right)
+P_{\omega}\left(  E\right)  \left(  \alpha\cdot\omega_{0}\right)  (a^{_{-}%
}\left(  y,\nu\left(  E\right)  \theta;E\right)  -P_{\theta}\left(  E\right)
)+\left(  a^{+}\left(  y,\nu\left(  E\right)  \omega;E\right)  -P_{\omega
}\left(  E\right)  \right)  \left(  \alpha\cdot\omega_{0}\right)  P_{\theta
}\left(  E\right)  ).$ Using the notation of Remark \ref{representation75} and
relations (\ref{eig22}), (\ref{eig42}), (\ref{eig8}), (\ref{eig9}),
(\ref{eig19}) and (\ref{eig20}) we get $\mathbf{b}_{N}\left(  y,\omega
,\theta;E;\omega_{0}\right)  =\pm\left(  2\pi\right)  ^{-2}\upsilon\left(
E\right)  ^{2}\int_{\Pi_{\omega_{0}}}e^{i\nu\left(  E\right)  \left(
y,\zeta^{\prime}-\zeta\right)  }$ $\times\left(  \mathbf{\tilde{h}}%
_{N}^{\operatorname{int}}\right)  ^{\prime}\left(  y,\zeta,\zeta^{\prime
};E;\omega_{0}\right)  dy,$ where $\left(  \mathbf{\tilde{h}}_{N}%
^{\operatorname{int}}\right)  ^{\prime}\in\mathcal{S}^{-\left(  \rho-1\right)
}.$ From Lemma \ref{representation201} and the inequality $\left\vert
\omega-\theta\right\vert $ $\leq\frac{1}{\delta}\left\vert \zeta^{\prime
}-\zeta\right\vert $ we obtain the estimate (\ref{representation83}) for
$\mathbf{b}_{N}.$
\end{proof}

We now use the partition of the unity that we introduce in
(\ref{representation170}). Recall the definition (\ref{representation182}) of
$s_{00}^{(jk)}.$ For $f,g\in\mathcal{H}\left(  E\right)  ,$ let us define
$I_{jk}:=\int_{\mathbb{S}^{2}}\int_{\mathbb{S}^{2}}\left(  s_{00}%
^{(jk)}\left(  \omega,\theta;E\right)  f\left(  \theta\right)  ,g\left(
\omega\right)  \right)  _{\mathcal{H}\left(  E\right)  }d\theta d\omega.$ We
prove the following

\begin{proposition}
\label{representation89}The function $\sum\limits_{O_{j}\cap O_{k}%
\neq\varnothing}s_{00}^{(jk)}\left(  \omega,\theta;E\right)  $ is a
Dirac-function over $\mathcal{H}\left(  E\right)  $. That is
\begin{equation}
\sum\limits_{O_{j}\cap O_{k}\neq\varnothing}I_{jk}=\left(  f,g\right)
_{\mathcal{H}\left(  E\right)  }. \label{representation88}%
\end{equation}

\end{proposition}

\begin{proof}
Observe that $I_{jk}=\left(  \operatorname*{sgn}E\right)  \left(  2\pi\right)
^{-2}\upsilon\left(  E\right)  ^{2}\int\limits_{\Pi_{\omega_{jk}}}%
\int\limits_{\Pi_{\omega_{jk}}}\int\limits_{\Pi_{\omega_{jk}}}e^{i\nu\left(
E\right)  \left\langle y,\zeta^{\prime}-\zeta\right\rangle }\left(
\mathbf{\tilde{h}}_{00,jk}^{\prime}\left(  \zeta,\zeta^{\prime};E\right)
\tilde{f}\left(  \zeta^{\prime}\right)  ,\tilde{g}\left(  \zeta\right)
\right)  d\zeta^{\prime}dyd\zeta,$ \ where $\mathbf{\tilde{h}}_{00,jk}%
^{\prime}=\frac{\tilde{\chi}_{j}\left(  \zeta\right)  }{\left(  1-\left\vert
\zeta\right\vert ^{2}\right)  ^{1/2}}\frac{\tilde{\chi}_{k}\left(
\zeta^{\prime}\right)  }{\left(  1-\left\vert \zeta^{\prime}\right\vert
^{2}\right)  ^{1/2}}\tilde{\chi}_{jk}\left(  \zeta,\zeta^{\prime}\right)
\tilde{P}_{\zeta}\left(  E\right)  \left(  \alpha\cdot\omega_{jk}\right)
\tilde{P}_{\zeta^{\prime}}\left(  E\right)  .$ (Here we used the notation of
Remark \ref{representation75}). Calculating the integrals over $\zeta^{\prime
}$ and $y$ we get $I_{jk}=\left(  \operatorname*{sgn}E\right)  \frac
{\upsilon\left(  E\right)  ^{2}}{\nu\left(  E\right)  ^{2}}\int_{\Pi
_{\omega_{jk}}}\frac{1}{1-\left\vert \zeta\right\vert ^{2}}\tilde{\chi}%
_{j}\left(  \zeta\right)  \tilde{\chi}_{k}\left(  \zeta^{\prime}\right)
\tilde{\chi}_{jk}\left(  \zeta,\zeta^{\prime}\right)  \left(  \tilde{P}%
_{\zeta}\left(  E\right)  \left(  \alpha\cdot\omega_{jk}\right)  \tilde
{P}_{\zeta^{\prime}}\left(  E\right)  \tilde{f}\left(  \zeta\right)
,\tilde{g}\left(  \zeta\right)  \right)  d\zeta.$ Since%
\begin{equation}
P_{\omega}\left(  E\right)  \left(  \alpha\cdot\omega_{jk}\right)  =\left(
\alpha\cdot\omega_{jk}\right)  P_{\omega}\left(  -E\right)  +\frac{\nu\left(
E\right)  }{E}\left\langle \omega,\omega_{jk}\right\rangle ,
\label{representation113}%
\end{equation}
we have $P_{\omega}\left(  E\right)  \left(  \alpha\cdot\omega_{jk}\right)
P_{\omega}\left(  E\right)  =\frac{\nu\left(  E\right)  }{E}\left\langle
\omega,\omega_{jk}\right\rangle P_{\omega}\left(  E\right)  .$ As
$\pm\left\langle \omega,\omega_{jk}\right\rangle =\sqrt{1-\left\vert
\zeta\right\vert ^{2}}$ for $\omega\in\Omega_{\pm}\left(  \omega_{jk}%
,\delta\right)  ,$ then $\chi_{j}\left(  \omega\right)  \chi_{k}\left(
\omega\right)  $ $\times\chi_{jk}\left(  \omega,\omega\right)  \left\langle
\omega,\omega_{jk}\right\rangle =\chi_{j}\left(  \omega\right)  \chi
_{k}\left(  \omega\right)  \chi_{jk}^{\prime}\left(  \omega,\omega\right)
\sqrt{1-\left\vert \zeta\right\vert ^{2}},$ where $\chi_{jk}^{\prime}\left(
\omega,\theta\right)  :=\chi_{jk}^{+}\left(  \omega\right)  \chi_{jk}%
^{+}\left(  \theta\right)  +\chi_{jk}^{-}\left(  \omega\right)  \chi_{jk}%
^{-}\left(  \theta\right)  .$ Thus, using the relation $\chi_{j}\left(
\omega\right)  \chi_{k}\left(  \omega\right)  \chi_{jk}^{\prime}\left(
\omega,\omega\right)  =\chi_{j}\left(  \omega\right)  \chi_{k}\left(
\omega\right)  ,$ we get $I_{jk}=\int_{\Pi_{\omega_{jk}}}\frac{1}%
{\sqrt{1-\left\vert \zeta\right\vert ^{2}}}\tilde{\chi}_{j}\left(
\zeta\right)  \tilde{\chi}_{k}\left(  \zeta\right)  \left(  \tilde{P}_{\zeta
}\left(  E\right)  \tilde{f}\left(  \zeta\right)  ,\tilde{g}\left(
\zeta\right)  \right)  d\zeta.$ Returning back to variable $\omega$ in the
last expression we obtain%
\begin{equation}
I_{jk}=\int_{\mathbb{S}^{2}}\chi_{j}\left(  \omega\right)  \chi_{k}\left(
\omega\right)  \left(  f_{1}\left(  \omega\right)  ,f_{2}\left(
\omega\right)  \right)  d\omega. \label{representation86}%
\end{equation}
Noting that $\sum\limits_{O_{j}\cap O_{k}\neq\varnothing}\chi_{j}\left(
\omega\right)  \chi_{k}\left(  \omega\right)  =\sum\limits_{j,k=1}^{n}\chi
_{j}\left(  \omega\right)  \chi_{k}\left(  \omega\right)  =1$ we obtain
(\ref{representation88}).
\end{proof}

We obtain the following

\begin{lemma}
\label{representation168}The limit (\ref{representation134}) exists and the
operator $S_{2}\left(  E\right)  $ is decomposed as follows%
\begin{equation}
S_{2}\left(  E\right)  =I+\mathcal{G+R}_{1}\text{,} \label{representation171}%
\end{equation}
where $I$ is the identity in $\mathcal{H}\left(  E\right)  ,$ $\mathcal{G}$ is
an integral operator with kernel $\sum\limits_{O_{j}\cap O_{k}\neq\varnothing
}\mathbf{g}_{N,jk}^{\operatorname{int}}\left(  \omega,\theta;E\right)  $
satisfying (\ref{representation83}) and $\mathcal{R}_{1}$ is an integral
operator with kernel $r_{1}\left(  \omega,\theta;E\right)  :=-i\left(
2\pi\right)  ^{-2}\upsilon\left(  E\right)  ^{2}\left(  \sum\limits_{O_{j}\cap
O_{k}=\varnothing}\chi_{j}\left(  \omega\right)  \left(  G_{1}^{\left(
0\right)  }+G_{2,O_{j},O_{k}}^{\left(  0\right)  }\right)  \left(  \nu\left(
E\right)  \omega,\nu\left(  E\right)  \theta;E\right)  \chi_{k}\left(
\theta\right)  \right.  $ $\left.  +\sum\limits_{O_{j}\cap O_{k}%
\neq\varnothing}r_{jk}\left(  \omega,\theta;E\right)  \right)  ,$ where the
function $r_{jk}$ is defined by $r_{jk}\left(  \omega,\theta;E\right)
:=\chi_{j}\left(  \omega\right)  \left(  G_{1}^{\left(  0\right)  }+\left(
1-\chi_{jk}^{\prime}\left(  \omega,\theta\right)  \right)  G_{2}^{\left(
0\right)  }+\chi_{jk}^{+}\left(  \omega\right)  \right.  $ $\left.  \times
\chi_{jk}^{+}\left(  \theta\right)  \left(  R_{2}^{\left(  0\right)  }\right)
_{+}+\chi_{jk}^{-}\left(  \omega\right)  \chi_{jk}^{-}\left(  \theta\right)
\left(  R_{2}^{\left(  0\right)  }\right)  _{-}\right)  \chi_{k}\left(
\theta\right)  ,$ with $\chi_{jk}^{\prime}\left(  \omega,\theta\right)
=\chi_{jk}^{+}\left(  \omega\right)  \chi_{jk}^{+}\left(  \theta\right)
+\chi_{jk}^{-}\left(  \omega\right)  \chi_{jk}^{-}\left(  \theta\right)  .$
For any $p$ and $q$ there exists $N,$ sufficiently large, such that,
$r_{1}\left(  \omega,\theta;E\right)  $ belongs to the class $C^{p}\left(
\mathbb{S}^{2}\times\mathbb{S}^{2}\right)  $, and moreover, its $C^{p}-$norm
is bounded by $C\left\vert E\right\vert ^{-q}$, as $\left\vert E\right\vert
\rightarrow\infty$. Furthermore, the operator $S_{2}\left(  E\right)  -I$ is a
compact operator on $\mathcal{H}\left(  E\right)  .$
\end{lemma}

\begin{proof}
Let us write $S_{2}\left(  E\right)  $ as in relation (\ref{representation170}%
). First we consider the case $O_{j}\cap O_{k}=\varnothing.$ The relation
(\ref{representation166}), Lemma \ref{representation183} and Lemma
\ref{representation145} imply that $\chi_{j}\left(  \omega\right)
s_{2}\left(  \omega,\theta;E\right)  \chi_{k}\left(  \theta\right)  =-i\left(
2\pi\right)  ^{-2}\upsilon\left(  E\right)  ^{2}\chi_{j}\left(  \omega\right)
\left(  G_{1}^{\left(  0\right)  }+G_{2,O_{j},O_{k}}^{\left(  0\right)
}\right)  \left(  \nu\left(  E\right)  \omega,\nu\left(  E\right)
\theta;E\right)  \chi_{k}\left(  \theta\right)  $, where the function
$\chi_{j}\left(  \omega\right)  \left(  G_{1}^{\left(  0\right)  }%
+G_{2,O_{j},O_{k}}^{\left(  0\right)  }\right)  \left(  \nu\left(  E\right)
\omega,\nu\left(  E\right)  \theta;E\right)  \chi_{k}\left(  \theta\right)  $
satisfies the estimate (\ref{representation224})$.$ Suppose now that
$O_{j}\cap O_{k}\neq\varnothing.$ We take $\Psi_{\pm}\left(  \hat{\xi}%
,\hat{\xi}^{\prime};\omega_{jk}\right)  =\chi_{j}\left(  \hat{\xi}\right)
\chi_{k}\left(  \hat{\xi}^{\prime}\right)  \chi_{jk}^{\pm}\left(  \hat{\xi
}\right)  \chi_{jk}^{\pm}\left(  \hat{\xi}^{\prime}\right)  $ (here we use
Lemma \ref{representation221}). Let us decompose $G_{2}^{\left(
\varepsilon\right)  }=\left(  1-\chi_{jk}^{\prime}\left(  \hat{\xi},\hat{\xi
}^{\prime}\right)  \right)  G_{2}^{\left(  \varepsilon\right)  }+\chi
_{jk}^{\prime}\left(  \hat{\xi},\hat{\xi}^{\prime}\right)  G_{2}^{\left(
\varepsilon\right)  }.$ Then using the relations (\ref{representation188}),
(\ref{representation166}) and (\ref{representation159}) we get
\begin{equation}
\left.  \left(  \chi_{j}\left(  S_{2}\left(  E\right)  \chi_{k}f_{1}\right)
,f_{2}\right)  =-i\left(  2\pi\right)  ^{-2}\upsilon\left(  E\right)
\lim_{\mu\downarrow0}\lim_{\varepsilon\rightarrow0}%
{\displaystyle\int}
{\displaystyle\int_{\mathbb{S}^{2}}}
\left(  G_{jk}^{\left(  \varepsilon\right)  }\left(  \nu\left(  E\right)
\omega,\xi^{\prime};E\right)  \tilde{\delta}_{\mu}^{\left(
\operatorname*{sgn}E\right)  }\hat{g}_{1}\left(  \xi^{\prime}\right)
,f_{2}\left(  \omega\right)  \right)  d\omega d\xi^{\prime},\right.
\label{representation250}%
\end{equation}
with $G_{jk}^{\left(  \varepsilon\right)  }\left(  \xi,\xi^{\prime};E\right)
:=\chi_{j}\left(  \hat{\xi}\right)  \left(  G_{1}^{\left(  \varepsilon\right)
}+\left(  1-\chi_{jk}^{\prime}\left(  \hat{\xi},\hat{\xi}^{\prime}\right)
\right)  G_{2}^{\left(  \varepsilon\right)  }+\chi_{jk}\left(  \hat{\xi}%
,\hat{\xi}^{\prime}\right)  \tilde{G}_{2}^{\left(  \varepsilon\right)  }%
+\sum_{\tau=-1}^{1}\chi_{jk}^{\operatorname*{sgn}\tau}\left(  \hat{\xi
}\right)  \chi_{jk}^{\operatorname*{sgn}\tau}\left(  \hat{\xi}^{\prime
}\right)  R_{\operatorname*{sgn}\tau}^{\left(  \varepsilon\right)  }\right)
\chi_{k}\left(  \hat{\xi}^{\prime}\right)  .$ Using Lemma
\ref{representation183} for the part in the limit (\ref{representation250})
corresponding to the term $\chi_{j}\left(  \hat{\xi}\right)  G_{1}^{\left(
\varepsilon\right)  }\chi_{k}\left(  \hat{\xi}^{\prime}\right)  ,$ Lemma
\ref{representation145} for $\chi_{j}\left(  \hat{\xi}\right)  $
$\times\left(  1-\chi_{jk}^{\prime}\left(  \hat{\xi},\hat{\xi}^{\prime
}\right)  \right)  G_{2}^{\left(  \varepsilon\right)  }\chi_{k}\left(
\hat{\xi}\right)  ,$ Lemma \ref{representation184} for $\chi_{j}\left(
\hat{\xi}\right)  \chi_{k}\left(  \hat{\xi}^{\prime}\right)  \chi_{jk}^{\pm
}\left(  \hat{\xi}\right)  \chi_{jk}^{\pm}\left(  \hat{\xi}^{\prime}\right)
R_{\pm}^{\left(  \varepsilon\right)  }$ and Lemma \ref{representation185} for
$\chi_{j}\left(  \hat{\xi}\right)  \chi_{k}\left(  \hat{\xi}^{\prime}\right)
$ $\times\chi_{jk}^{\pm}\left(  \hat{\xi}\right)  \chi_{jk}^{\pm}\left(
\hat{\xi}^{\prime}\right)  \tilde{G}_{2}^{\left(  \varepsilon\right)  },$ in
order to calculate the limit (\ref{representation250}), we get $\chi
_{j}\left(  \omega\right)  s_{2}\left(  \omega,\theta;E\right)  \chi
_{k}\left(  \theta\right)  =-i\left(  2\pi\right)  ^{-2}\upsilon\left(
E\right)  ^{2}r_{jk}\left(  \omega,\theta;E\right)  $ $+s_{N,jk}\left(
\omega,\theta;E\right)  .$ For any $p$ and $q$ there exists $N,$ sufficiently
large, such that, $r_{jk}\left(  \omega,\theta;E\right)  $ belongs to the
class $C^{p}\left(  \mathbb{S}^{2}\times\mathbb{S}^{2}\right)  $, and
moreover, its $C^{p}-$norm is bounded by $C\left\vert E\right\vert ^{-q}$, as
$\left\vert E\right\vert \rightarrow\infty$.

Taking the sum over $j$ and $k$, such that $O_{j}\cap O_{k}\neq\varnothing,$
of $s_{N,jk}\left(  \omega,\theta;E\right)  ,$ and using Proposition
\ref{representation89} and Lemma \ref{representation186}, with $\Psi_{\pm
}\left(  \omega,\theta;\omega_{jk}\right)  =\chi_{j}\left(  \omega\right)
\chi_{k}\left(  \theta\right)  \chi_{jk}^{\pm}\left(  \omega\right)  \chi
_{jk}^{\pm}\left(  \theta\right)  ,$ we obtain the term $I+\mathcal{G}$, where
$\mathcal{G}$ satisfies the assertions of Lemma \ref{representation168}.
Moreover, taking the sum over $j$ and $k$, such that $O_{j}\cap O_{k}%
\neq\varnothing,$ of the parts $\chi_{j}\left(  \omega\right)  s_{2}\left(
\omega,\theta;E\right)  \chi_{k}\left(  \theta\right)  $ we see that the
kernel of $\mathcal{R}_{1}$ is given by $r_{1}\left(  \omega,\theta;E\right)
.$

Noting that the amplitude of $S_{2}\left(  E\right)  -I$ belongs to the class
$\mathcal{S}^{-\left(  \rho-1\right)  },$ it follows from Proposition
\ref{representation76} that the operator $S_{2}\left(  E\right)  -I$ is compact.
\end{proof}

\subsection{Stationary representation for the scattering matrix and proofs of
Theorems \ref{representation179} and \ref{representation25}.}

In order to prove Theorem \ref{representation25} we need a stationary formula
for the scattering matrix $S\left(  E\right)  $ in\ terms of the
identifications $J_{\pm}$ (\cite{47}, \cite{42}, \cite{43}, \cite{37},
\cite{8}). The scattering operators $\mathbf{S}$ and $\mathbf{\tilde{S}}$ are
related by equation (\ref{representation206}) and hence, $\mathbf{\tilde{S}}$
commutes with $H_{0}.$ This implies that $\mathcal{F}_{0}\mathbf{\tilde{S}%
}\mathcal{F}_{0}^{\ast}$ acts as a multiplication by the operator valued
function $\tilde{S}\left(  E\right)  .$\ Let us first give a stationary
formula for $\tilde{S}\left(  E\right)  .$

\begin{lemma}
\label{representation135}For $\left\vert E\right\vert >m,$ the scattering
matrix $\tilde{S}\left(  E\right)  $ can be represented as
\begin{equation}
\tilde{S}\left(  E\right)  =S_{1}\left(  E\right)  +S_{2}\left(  E\right)  ,
\label{representation30}%
\end{equation}
where $S_{1}\left(  E\right)  $ is given by relation (\ref{representation14})
and $S_{2}\left(  E\right)  $ is defined by (\ref{representation171}).
\end{lemma}

\begin{proof}
We follow the proof of Theorem 3.3 of \cite{47} for the case of the
Schr\"{o}dinger equation. Let $\Lambda\subset(-\infty,-m)\cup(m,+\infty)$ be
bounded.\ We first note that for $g_{j}\in\mathcal{S}\left(  \mathbb{R}%
^{3};\mathbb{C}^{4}\right)  \cap E_{0}\left(  \Lambda\right)  L^{2},$ $j=1,2,$
we have
\begin{equation}
\int\limits_{\Lambda}\left(  \tilde{S}\left(  E\right)  \Gamma_{0}\left(
E\right)  g_{1},\Gamma_{0}\left(  E\right)  g_{2}\right)  _{\mathcal{H}\left(
E\right)  }dE=\left(  \mathbf{\tilde{S}}g_{1},g_{2}\right)  .
\label{representation142}%
\end{equation}
We will work with the form $\left(  \mathbf{\tilde{S}}g_{1},g_{2}\right)  ,$
and then use equality (\ref{representation142}) to get the desired result for
$\tilde{S}\left(  E\right)  .$ From (\ref{representation18}) we get
$W_{+}\left(  H,H_{0};J_{+}\right)  ^{\ast}W_{+}\left(  H,H_{0};J_{-}\right)
=0,$ and thus, $\mathbf{\tilde{S}}=W_{+}^{\ast}\left(  H,H_{0};J_{+}\right)
\left(  W_{-}\left(  H,H_{0};J_{-}\right)  -W_{+}\left(  H,H_{0};J_{-}\right)
\right)  .$ Noting that
\begin{equation}
W_{+}\left(  H,H_{0};J_{-}\right)  -W_{-}\left(  H,H_{0};J_{-}\right)
=\left(  i\int\limits_{-\infty}^{\infty}e^{itH}T_{-}e^{-itH_{0}}dt\right)  ,
\label{representation43}%
\end{equation}
and using the intertwining property of $W_{+}\left(  H,H_{0};J_{+}\right)  $
we have
\begin{equation}
\left.
\begin{array}
[c]{c}%
\left(  \mathbf{\tilde{S}}g_{1},g_{2}\right)  =\left(  \left(  W_{-}\left(
H,H_{0};J_{-}\right)  -W_{+}\left(  H,H_{0};J_{-}\right)  \right)  g_{1}%
,W_{+}\left(  H,H_{0};J_{+}\right)  g_{2}\right) \\
=-i%
{\displaystyle\int\limits_{-\infty}^{\infty}}
\left(  T_{-}e^{-itH_{0}}g_{1},W_{+}\left(  H,H_{0};J_{+}\right)  e^{-itH_{0}%
}g_{2}\right)  .
\end{array}
\right.  \label{representation60}%
\end{equation}
Let us split $T_{-}=T_{-}^{1}+T_{-}^{2}$ in the R.H.S. of
(\ref{representation60}). By relation (\ref{representation23}), using
Proposition \ref{representation118} we see that the integral in the R.H.S. of
(\ref{representation60}) for $T_{-}^{1}$ is well defined. Moreover, noting
that $t_{-}^{2}$ is equal to $0$ in some neighborhood of the directions
$\left\langle \hat{x},\hat{\xi}\right\rangle =\pm1,$ $\left\vert x\right\vert
\geq R,$ and applying Lemma \ref{representation56} to the operator $T_{-}^{2}$
we conclude that the integral in the R.H.S. of (\ref{representation60}) is
well defined$.$ Using the equality $W_{+}\left(  H,H_{0};J_{+}\right)
-J_{+}=i\int\limits_{0}^{\infty}e^{i\tau H}T_{+}e^{-i\tau H_{0}}d\tau,$ we
obtain%
\begin{equation}
\left.  \left(  \mathbf{\tilde{S}}g_{1},g_{2}\right)  =i\int\limits_{0}%
^{\infty}i\left(  \int\limits_{-\infty}^{\infty}\left(  T_{-}e^{-itH_{0}}%
g_{1},e^{i\tau H}T_{+}e^{-i\left(  \tau+t\right)  H_{0}}g_{2}\right)
dt\right)  d\tau-i%
{\displaystyle\int\limits_{-\infty}^{\infty}}
\left(  T_{-}e^{-itH_{0}}g_{1},J_{+}e^{-itH_{0}}g_{2}\right)  dt.\right.
\label{representation44}%
\end{equation}
By the same argument as that we used in (\ref{representation60}) we show the
convergence of the integrals in (\ref{representation44}).

Let us consider the first term of the R.H.S. of (\ref{representation44}). We
represent this term as the following double limit%
\begin{equation}
\left.
\begin{array}
[c]{c}%
i\int\limits_{0}^{\infty}i\left(  \int\limits_{-\infty}^{\infty}\left(
T_{-}e^{-itH_{0}}g_{1},e^{i\tau H}T_{+}e^{-i\left(  \tau+t\right)  H_{0}}%
g_{2}\right)  dt\right)  d\tau\\
=\lim\limits_{\mu,\mu^{\prime}\downarrow0}i%
{\displaystyle\int\limits_{0}^{\infty}}
e^{-\mu\tau}i\left(
{\displaystyle\int\limits_{-\infty}^{\infty}}
e^{-\mu^{\prime}\left\vert t\right\vert }\left(  e^{i\left(  \tau+t\right)
H_{0}}T_{+}^{\ast}e^{-i\tau H}T_{-}e^{-itH_{0}}g_{1},g_{2}\right)  dt\right)
d\tau.
\end{array}
\right.  \label{representation45}%
\end{equation}

As $\mathcal{F}_{0}$ gives a spectral representation of $H_{0}$
(see\ \ref{basicnotions12}), then
\begin{equation}
\left(  g_{1},g_{2}\right)  =\int\limits_{\Lambda}\left(  \Gamma_{0}\left(
E\right)  g_{1}\left(  \cdot\right)  ,\Gamma_{0}\left(  E\right)  g_{2}\left(
\cdot\right)  \right)  _{\mathcal{H}\left(  E\right)  }dE,
\label{representation46}%
\end{equation}
for $g_{1},g_{2}\in L_{s}^{2}\cap E_{0}\left(  \Lambda\right)  L^{2},$
$s>1/2,$ where $\Gamma_{0}\left(  E\right)  $ is given by the relation
(\ref{basicnotions11}). Applying (\ref{representation46}) to the R.H.S. of
(\ref{representation45}) and using the identities $i\int\limits_{0}^{\infty
}e^{-\mu\tau}e^{-i\tau\left(  H-E\right)  }d\tau=R\left(  E+i\mu\right)  $ and
$i\int\limits_{-\infty}^{\infty}e^{-\mu^{\prime}\left\vert t\right\vert
}e^{-it\left(  H_{0}-E\right)  }dt=2\pi i\delta_{\mu^{\prime}}\left(
H_{0}-E\right)  ,$ to calculate the integrals on $t$ and $\tau$ of the
resulting expression we get
\begin{equation}
\left.
\begin{array}
[c]{c}%
i\int\limits_{0}^{\infty}i\left(  \int\limits_{-\infty}^{\infty}\left(
T_{-}e^{-itH_{0}}g_{1},e^{i\tau H}T_{+}e^{-i\left(  \tau+t\right)  H_{0}}%
g_{2}\right)  dt\right)  d\tau\\
=\lim\limits_{\mu,\mu^{\prime}\downarrow0}2\pi i\int\limits_{\Lambda}\left(
\Gamma_{0}\left(  E\right)  T_{+}^{\ast}R\left(  E+i\mu\right)  T_{-}%
\delta_{\mu^{\prime}}\left(  H_{0}-E\right)  g_{1},\Gamma_{0}\left(  E\right)
g_{2}\right)  dE.
\end{array}
\right.  \label{representation247}%
\end{equation}
Let $f_{1},f_{2}\in\mathcal{H}\left(  E\right)  $ be $C^{\infty}\left(
\mathbb{S}^{2};\mathbb{C}^{4}\right)  $ functions and take $g_{j}$ as in
(\ref{representation134})$.$ Note that $g_{j}\in\mathcal{S}\left(
\mathbb{R}^{3};\mathbb{C}^{4}\right)  $ and $\Gamma_{0}\left(  E\right)
g_{j}=f_{j}.$ Moreover, observe that the operator $\Gamma_{0}\left(  E\right)
\left\langle x\right\rangle ^{-s},$ for $s>1/2$, is bounded$.$ Then, relation
$\lim_{\mu\downarrow0}\left\Vert (\delta_{\mu}\left(  H_{0}-E\right)  \right.
$ $\left.  -\Gamma_{0}^{\ast}\left(  E\right)  \Gamma_{0}\left(  E\right)
)f\right\Vert _{L_{-s}^{2}}=0$ and Lemma \ref{representation33} imply that the
limit in the R.H.S. of (\ref{representation247}) exists. Hence, we obtain%
\begin{equation}
\left.  i\int\limits_{0}^{\infty}i\left(  \int\limits_{-\infty}^{\infty
}\left(  T_{-}e^{-itH_{0}}g_{1},e^{i\tau H}T_{+}e^{-i\left(  \tau+t\right)
H_{0}}g_{2}\right)  dt\right)  d\tau=2\pi i%
{\displaystyle\int\limits_{\Lambda}}
\left(  \Gamma_{0}\left(  E\right)  T_{+}^{\ast}R_{+}\left(  E\right)
T_{-}\Gamma_{0}\left(  E\right)  ^{\ast}f_{1},f_{2}\right)  _{\mathcal{H}%
\left(  E\right)  }dE.\right.  \label{representation143}%
\end{equation}

Applying (\ref{representation46}) to the second term of the R.H.S. of
(\ref{representation44}) we have
\begin{equation}
\left.  i%
{\displaystyle\int\limits_{-\infty}^{\infty}}
\left(  T_{-}e^{-itH_{0}}g_{1},J_{+}e^{-itH_{0}}g_{2}\right)  dt=%
{\displaystyle\int\limits_{\Lambda}}
\left(  i%
{\displaystyle\int\limits_{-\infty}^{\infty}}
\left(  \Gamma_{0}\left(  E\right)  J_{+}^{\ast}T_{-}e^{-it\left(
H_{0}-E\right)  }g_{1},\Gamma_{0}\left(  E\right)  g_{2}\right)
_{\mathcal{H}\left(  E\right)  }dt\right)  dE.\right.
\label{representation204}%
\end{equation}
As the operators $\Gamma_{0}\left(  E\right)  \left\langle x\right\rangle
^{-s}$ and $\left\langle x\right\rangle ^{s}J_{+}^{\ast}\left\langle
x\right\rangle ^{-s},$ for $s>1/2$, are bounded, splitting $\left\langle
x\right\rangle ^{s}T_{-}=\left\langle x\right\rangle ^{s}T_{-}^{1}%
+\left\langle x\right\rangle ^{s}T_{-}^{2}$ and using Proposition
\ref{representation118} for $\left\langle x\right\rangle ^{s}T_{-}^{1}$ (with
$N\geq1$) and Lemma \ref{representation56} for $\left\langle x\right\rangle
^{s}T_{-}^{2}$ we conclude that the integral in the variable $t$ in the R.H.S.
of (\ref{representation204}) is absolutely convergent. Then we get
$i\int\limits_{-\infty}^{\infty}\left(  T_{-}e^{-itH_{0}}g_{1},J_{+}%
e^{-itH_{0}}g_{2}\right)  dt=\int\limits_{\Lambda}\lim\limits_{\mu\downarrow
0}i\int\limits_{-\infty}^{\infty}e^{-\mu\left\vert t\right\vert }(J_{+}^{\ast
}T_{-}e^{-it\left(  H_{0}-E\right)  }g_{1},$ $\Gamma_{0}^{\ast}\left(
E\right)  \Gamma_{0}\left(  E\right)  g_{2})dtdE.$ Calculating the integral on
the variable $t$ we finally obtain%
\begin{equation}
\left.  i%
{\displaystyle\int\limits_{-\infty}^{\infty}}
\left(  T_{-}e^{-itH_{0}}g_{1},J_{+}e^{-itH_{0}}g_{2}\right)  dt=2\pi i%
{\displaystyle\int\limits_{\Lambda}}
\lim_{\mu\downarrow0}\left(  J_{+}^{\ast}T_{-}\delta_{\mu}\left(
H_{0}-E\right)  g_{1},\Gamma_{0}^{\ast}\left(  E\right)  \Gamma_{0}\left(
E\right)  g_{2}\right)  dE.\right.  \label{representation144}%
\end{equation}

Using equalities (\ref{representation142}), (\ref{representation44}),
(\ref{representation143}) and (\ref{representation144}) we obtain%
\begin{equation}
\left.
\begin{array}
[c]{c}%
{\displaystyle\int\limits_{\Lambda}}
\left(  \tilde{S}\left(  E\right)  f_{1},f_{2}\right)  _{\mathcal{H}\left(
E\right)  }dE\\
=2\pi i%
{\displaystyle\int\limits_{\Lambda}}
\left(  \Gamma_{0}\left(  E\right)  T_{+}^{\ast}R_{+}\left(  E\right)
T_{-}\Gamma_{0}\left(  E\right)  ^{\ast}f_{1},f_{2}\right)  _{\mathcal{H}%
\left(  E\right)  }dE-2\pi i%
{\displaystyle\int\limits_{\Lambda}}
\lim_{\mu\downarrow0}\left(  J_{+}^{\ast}T_{-}\delta_{\mu}\left(
H_{0}-E\right)  g_{1},\Gamma_{0}^{\ast}\left(  E\right)  \Gamma_{0}\left(
E\right)  g_{2}\right)  dE.
\end{array}
\right.  \label{representation190}%
\end{equation}
Note that the limit in the second term in the R.H.S. of
(\ref{representation190}) is equals to the limit in relation
(\ref{representation134}). Then applying Lemma \ref{representation168} to
calculate this limit we get%
\begin{equation}%
{\displaystyle\int\limits_{\Lambda}}
\left(  \tilde{S}\left(  E\right)  f_{1},f_{2}\right)  _{\mathcal{H}\left(
E\right)  }dE=%
{\displaystyle\int\limits_{\Lambda}}
\left(  \left(  S_{1}\left(  E\right)  +S_{2}\left(  E\right)  \right)
f_{1},f_{2}\right)  _{\mathcal{H}\left(  E\right)  }dE,
\label{representation205}%
\end{equation}
where $S_{1}\left(  E\right)  $ is given by relation (\ref{representation14})
and $S_{2}\left(  E\right)  $ is defined by (\ref{representation171}). Since
relation (\ref{representation205}) holds for all bounded $\Lambda
\subset(-\infty,-m)\cup(m,+\infty)$ and all $f_{1},f_{2}\in\mathcal{H}\left(
E\right)  \cap C^{\infty}\left(  \mathbb{S}^{2};\mathbb{C}^{4}\right)  $, we
obtain relation (\ref{representation30}).
\end{proof}

\begin{corollary}
\label{representation191}For any $\left\vert E\right\vert >m$ the scattering
matrix $S\left(  E\right)  $ satisfy the relation
\[
S\left(  E\right)  =S_{1}\left(  E\right)  +S_{2}\left(  E\right)  .
\]

\end{corollary}

\begin{proof}
Recall that the scattering operators $\mathbf{S}$ and $\mathbf{\tilde{S}}$ are
related by equation (\ref{representation206}). As $\tilde{S}\left(  E\right)
$ satisfies equation (\ref{representation30}) and since $c_{1}$ in the
definition of $\theta\left(  t\right)  $ is arbitrary, we conclude that for
every $\left\vert E\right\vert >m$ there is $c_{1}$ such that the scattering
matrix $S\left(  E\right)  $ is equal to the scattering matrix $\tilde
{S}\left(  E\right)  .$ Thus, we get the relation (\ref{representation30})
also for $S\left(  E\right)  .$
\end{proof}

Theorem \ref{representation25} is now consequence of Corollary
\ref{representation191}, Theorem \ref{representation32} and Lemma
\ref{representation168}. Theorem \ref{representation179} follows from
Corollary \ref{representation191}, Theorem \ref{representation32} and Theorem
\ref{representation241}.

\section{Applications of the formula for the singularities of the kernel of
the scattering matrix.}

\subsection{High energy limit of the scattering matrix.}

Let us consider the principal part $S_{0}\left(  E\right)  ,$ of $S\left(
E\right)  .$ We take $N=0$ in the relation (\ref{representation26}). Then, the
kernel $s_{0}\left(  \omega,\theta;E\right)  $ of $S_{0}\left(  E\right)  $ is
given by $s_{0}\left(  \omega,\theta;E\right)  =\sum\limits_{O_{j}\cap
O_{k}\neq\varnothing}s_{0,jk}\left(  \omega,\theta;E\right)  $

Let us define the operator $\mathbf{P}\left(  E\right)  :L^{2}\left(
\mathbb{S}^{2};\mathbb{C}^{4}\right)  \rightarrow L^{2}\left(  \mathbb{S}%
^{2};\mathbb{C}^{4}\right)  $ by the relation $\left(  \mathbf{P}\left(
E\right)  f\right)  \left(  \omega\right)  :=P_{\omega}\left(  E\right)
f\left(  \omega\right)  .$ Moreover we denote $\mathbf{P}\left(  \pm
\infty\right)  :L^{2}\left(  \mathbb{S}^{2};\mathbb{C}^{4}\right)  \rightarrow
L^{2}\left(  \mathbb{S}^{2};\mathbb{C}^{4}\right)  $ by the relation $\left(
\mathbf{P}\left(  \pm\infty\right)  f\right)  \left(  \omega\right)
:=P_{\omega}\left(  \pm\infty\right)  f\left(  \omega\right)  ,$ where
$P_{\omega}\left(  \pm\infty\right)  =\frac{1}{2}\left(  1\pm\left(
\alpha\cdot\omega\right)  \right)  .$ Note that $\mathbf{P}\left(  E\right)  $
converges in $L^{2}\left(  \mathbb{S}^{2};\mathbb{C}^{4}\right)  $ norm to
$\mathbf{P}\left(  \pm\infty\right)  ,$ as $\pm E\rightarrow\infty.$ We prove
the following result.

\begin{proposition}
The operator $S_{0}\left(  E\right)  $ is uniformly bounded in $L^{2}\left(
\mathbb{S}^{2};\mathbb{C}^{4}\right)  ,$ for all $\left\vert E\right\vert \geq
E_{0}$, $E_{0}>m,$ and the following estimate holds%
\begin{equation}
\left\Vert S\left(  E\right)  -S_{0}\left(  E\right)  \right\Vert =O\left(
\left\vert E\right\vert ^{-1}\right)  ,\text{ }\left\vert E\right\vert
\rightarrow\infty. \label{representation77}%
\end{equation}
Moreover, $S\left(  E\right)  \mathbf{P}\left(  E\right)  $ converges
strongly$,$ as $\pm E\rightarrow\infty,$ in $L^{2}\left(  \mathbb{S}%
^{2};\mathbb{C}^{4}\right)  $ to the operator $S\left(  \pm\infty\right)
\mathbf{P}\left(  \pm\infty\right)  $ where $S\left(  \pm\infty\right)  $ is
the operator of multiplication by the function $\sum\limits_{O_{j}\cap
O_{k}\neq\varnothing}\chi_{j}\left(  \omega\right)  \chi_{k}\left(
\omega\right)  e^{-i\int_{-\infty}^{\infty}\left(  V\left(  t\omega\right)
\pm\left\langle \omega,A\left(  t\omega\right)  \right\rangle \right)  dt}.$
\end{proposition}

\begin{proof}
Using the notation of Remark \ref{representation75} and making the change
$z=\nu\left(  E\right)  y$ in the relation for $\tilde{s}_{0}\left(
\zeta,\zeta^{\prime};E\right)  ,$ we obtain $\tilde{s}_{0}\left(  \zeta
,\zeta^{\prime};E\right)  =\left(  2\pi\right)  ^{-2}\left(  \frac
{\upsilon\left(  E\right)  }{\nu\left(  E\right)  }\right)  ^{2}%
{\displaystyle\int\limits_{\Pi_{\omega_{jk}}}}
e^{i\left\langle z,\zeta^{\prime}-\zeta\right\rangle }\mathbf{\tilde{h}%
}_{N,jk}^{\prime}\left(  \frac{z}{\nu\left(  E\right)  },\zeta,\zeta^{\prime
};E\right)  dz.$ Since $\left\vert \partial_{z}^{\alpha}\partial_{\zeta
}^{\beta}\partial_{\zeta^{\prime}}^{\gamma}\mathbf{\tilde{h}}_{N,jk}^{\prime
}\left(  \frac{z}{\nu\left(  E\right)  },\zeta,\zeta^{\prime};E\right)
\right\vert \leq C_{\alpha,\beta,\,\gamma}\left\langle z\right\rangle
^{-\left\vert \alpha\right\vert },$ where $C_{\alpha,\beta,\,\gamma}$ is
independent on $E,$ for $\left\vert E\right\vert \geq E_{0}$, and
$\mathbf{\tilde{h}}_{N,jk}^{\prime}$ is a compact-supported function of
$\zeta$ and $\zeta^{\prime},$ it follows from Proposition
\ref{representation76} that $S_{0}\left(  E\right)  $ is uniformly bounded in
$L^{2}\left(  \mathbb{S}^{2};\mathbb{C}^{4}\right)  ,$ for $\left\vert
E\right\vert \geq E_{0}.$

Using (\ref{eig42}), (\ref{eig8}), (\ref{eig9}), (\ref{eig22}), (\ref{eig19}),
(\ref{eig20}) and (\ref{eig33}) we see that
\begin{equation}
\left\vert \partial_{z}^{\alpha}\partial_{\zeta}^{\beta}\partial
_{\zeta^{\prime}}^{\gamma}\left(  \mathbf{\tilde{h}}_{N,jk}^{\prime
}-\mathbf{\tilde{h}}_{0,jk}^{\prime}\right)  \left(  \frac{z}{\nu\left(
E\right)  },\zeta,\zeta^{\prime};E\right)  \right\vert \leq C_{\alpha
,\beta,\,\gamma}\nu\left(  E\right)  ^{-1}\left\langle z\right\rangle
^{-\rho-\left\vert \alpha\right\vert }, \label{representation193}%
\end{equation}
where $C_{\alpha,\beta,\,\gamma}$ is independent on $E,$ for $\left\vert
E\right\vert \geq E_{0}.$ Then, equality (\ref{representation77}) follows from
Proposition \ref{representation76}.

Now we calculate the limit of $S\left(  E\right)  $ as $\pm E\rightarrow
\infty.$ Note that integrating by parts on the variable $\zeta^{\prime}$ we
have $%
{\displaystyle\int_{\mathbb{S}^{2}}}
s_{0,jk}\left(  \omega,\theta;E\right)  $ $\times f\left(  \theta\right)
d\theta=\left(  2\pi\right)  ^{-2}%
{\displaystyle\int_{\Pi_{\omega_{jk}}}}
{\displaystyle\int_{\Pi_{\omega_{jk}}}}
e^{i\left\langle z,\zeta^{\prime}-\zeta\right\rangle }\left\langle
z\right\rangle ^{-n}\left\langle D_{\zeta^{\prime}}\right\rangle ^{n}\left(
\mathbf{\tilde{h}}_{0,jk}^{\prime}\left(  \frac{z}{\nu\left(  E\right)
},\zeta,\zeta^{\prime};E\right)  \tilde{f}\left(  \zeta^{\prime}\right)
\right)  d\zeta^{\prime}dz$, for even $n$ and any $f\in C^{\infty}.$ Thus,
taking the limit, as $\pm E\rightarrow\infty,$ in the R.H.S. of the last
relation and then integrating back by parts we get
\begin{equation}
\left.  \lim_{\pm E\rightarrow\infty}%
{\displaystyle\int_{\mathbb{S}^{2}}}
s_{0}\left(  \omega,\theta;E\right)  f\left(  \theta\right)  d\theta=\left(
2\pi\right)  ^{-2}%
{\displaystyle\int_{\Pi_{\omega_{jk}}}}
{\displaystyle\int_{\Pi_{\omega_{jk}}}}
e^{i\left\langle z,\zeta^{\prime}-\zeta\right\rangle }\mathbf{\tilde{h}%
}_{0,jk}^{\prime}\left(  0,\zeta,\zeta^{\prime};\infty\right)  \tilde
{f}\left(  \zeta^{\prime}\right)  d\zeta^{\prime}dz,\right.
\label{representation207}%
\end{equation}
where $\mathbf{\tilde{h}}_{0,jk}^{\prime}\left(  0,\zeta,\zeta^{\prime}%
;\pm\infty\right)  :=\pm\frac{\tilde{\chi}_{jk}\left(  \zeta,\zeta^{\prime
}\right)  \tilde{\chi}_{j}\left(  \zeta\right)  \tilde{\chi}_{k}\left(
\zeta^{\prime}\right)  }{\left(  1-\left\vert \zeta^{\prime}\right\vert
^{2}\right)  ^{1/2}}\tilde{P}_{\zeta}\left(  \pm\infty\right)  \left(
\alpha\cdot\omega_{jk}\right)  \tilde{P}_{\zeta^{\prime}}\left(  \pm
\infty\right)  .$ The integral in the variable $\zeta^{\prime}$ in
(\ref{representation207}) is the inverse of the Fourier transform of the
function $\mathbf{\tilde{h}}_{0,jk}^{\prime}\left(  0,\zeta,\zeta^{\prime
};\infty\right)  \tilde{f}\left(  \zeta^{\prime}\right)  .$ As $\mathbf{\tilde
{h}}_{0,jk}^{\prime}\left(  0,\zeta,\zeta^{\prime};\infty\right)  \tilde
{f}\left(  \zeta^{\prime}\right)  $ has a compact support in $\zeta^{\prime},$
then calculating the integral with respect to $\zeta^{\prime}$ in
(\ref{representation207}) we obtain a function of the variable $z$ that
belongs to $\mathcal{S}.$ Calculating the integral in $z$ we get back the
function $\mathbf{\tilde{h}}_{0,jk}^{\prime}\left(  0,\zeta,\zeta
;\infty\right)  \tilde{f}\left(  \zeta\right)  .$ From relation
(\ref{representation113}) we get
\begin{equation}
P_{\omega}\left(  \pm\infty\right)  \left(  \alpha\cdot\omega_{jk}\right)
P_{\omega}\left(  \pm\infty\right)  =\pm\left\langle \omega,\omega
_{jk}\right\rangle P_{\omega}\left(  \pm\infty\right)
.\ \label{representation251}%
\end{equation}
As $\pm\left\langle \omega,\omega_{jk}\right\rangle =\left(  1-\left\vert
\zeta\right\vert ^{2}\right)  ^{1/2},$ for $\omega\in\Omega_{\pm}\left(
\omega_{jk},\delta\right)  ,$ then $\left\langle \omega,\omega_{jk}%
\right\rangle \chi_{jk}\left(  \omega,\omega\right)  \chi_{j}\left(
\omega\right)  \chi_{k}\left(  \omega\right)  =\left(  1-\left\vert
\zeta\right\vert ^{2}\right)  ^{1/2}\chi_{j}\left(  \omega\right)  \chi
_{k}\left(  \omega\right)  .$ Therefore, using these relations in the
expression for $\mathbf{\tilde{h}}_{0,jk}^{\prime}\left(  0,\zeta,\zeta
;\pm\infty\right)  ,$ substituting the result in (\ref{representation207}),
and taking in account that $\sum\limits_{O_{j}\cap O_{k}\neq\varnothing}%
\chi_{j}\left(  \omega\right)  \chi_{k}\left(  \omega\right)  =1,$ we complete
the proof.
\end{proof}

\subsection{Leading diagonal singularity of the kernel of the scattering
matrix.}

Recall that for $\rho=\min\{\rho_{e},\rho_{m}\}>3,$ $s^{\operatorname{int}%
}\left(  \omega,\theta;E\right)  \in C^{0}\left(  \mathbb{S}^{2}%
\times\mathbb{S}^{2}\right)  ,$ where $s^{\operatorname{int}}$ is the kernel
of the operator $S\left(  E\right)  -I$ (see Theorem \ref{representation25}).
Let us now calculate the leading term on the diagonal of
$s^{\operatorname{int}}\left(  \omega,\theta;E\right)  ,$ for $1<\rho<3,$ as
$\omega-\theta\rightarrow0,$ with $E$ fixed$.$ For a fixed $\omega_{0}%
\in\mathbb{S}^{2},$ we take a cut-off function $\Psi_{+}\left(  \omega
,\theta;\omega_{0}\right)  ,$ supported on $\Omega_{+}\left(  \omega
_{0},\delta\right)  \times\Omega_{+}\left(  \omega_{0},\delta\right)  ,$ such
that it is equal to $1$ in $\Omega_{+}\left(  \omega_{0},\delta^{\prime
}\right)  $, for some $\delta^{\prime}>\delta.$ We define $\mathcal{V}%
_{V,A;\omega}^{\left(  E\right)  }\left(  x\right)  :=\frac{\left\vert
E\right\vert }{\nu\left(  E\right)  }V\left(  x\right)  +\left(
\operatorname*{sgn}E\right)  \left\langle \omega,A\left(  x\right)
\right\rangle .$ The following result is similar to Theorem 1.2 of \cite{47}
for the Schr\"{o}dinger equation:

\begin{proposition}
\label{representation231}Let the magnetic potential $A\left(  x\right)  $ and
the electric potential $V\left(  x\right)  $ satisfy the estimates
(\ref{eig31}) and (\ref{eig32}), with $1<\rho<3,$ respectively. Then, for all
fixed $\omega\in\mathbb{S}^{2}$ and $\theta\in\Omega_{+}\left(  \omega
,\delta^{\prime}\right)  ,$ $\omega\neq\theta,$ we have%
\begin{equation}
\left.  \left\vert \left(  s^{\operatorname{int}}\left(  \omega,\theta
;E\right)  -\frac{1}{i}\left(  2\pi\right)  ^{-1/2}\upsilon\left(  E\right)
^{2}\frac{\nu\left(  E\right)  }{\left\vert E\right\vert }\left(
\mathcal{FV}_{V,A;\omega}^{\left(  E\right)  }\right)  \left(  -\nu\left(
E\right)  \tilde{\theta}\right)  P_{\omega}\left(  E\right)  \right)
\right\vert \leq C\left\vert \omega-\theta\right\vert ^{-2+\rho_{1}},\right.
\label{representation102}%
\end{equation}
where $\tilde{\theta}=\theta-\left\langle \theta,\omega\right\rangle \omega,$
$\rho_{1}=2\left(  \rho-1\right)  ,$ if $\rho<2$ and $\rho_{1}=2-\varepsilon,$
with $\varepsilon>0,$ for $\rho=2.$ Here the constant $C$ is independent on
$\omega.$ If $\rho>2,$ then the difference in the L.H.S. of
(\ref{representation102}) is continuous$.$
\end{proposition}

\begin{proof}
Note that
\begin{equation}
\left\vert \omega-\theta\right\vert ^{2}=2\left(  1-\left\langle \theta
,\omega\right\rangle \right)  =2\frac{\left\vert \tilde{\theta}\right\vert
^{2}}{1+\sqrt{1-\left\vert \tilde{\theta}\right\vert ^{2}}}.
\label{representation99}%
\end{equation}
Let us define $h\left(  y,\omega,\theta;E\right)  :=-i\left(
\operatorname*{sgn}E\right)  \left(  \Phi^{+}\left(  y,\nu\left(  E\right)
\omega;E\right)  -\Phi^{-}\left(  y,\nu\left(  E\right)  \theta;E\right)
\right)  P_{\omega}\left(  E\right)  \left(  \alpha\cdot\omega\right)
P_{\theta}\left(  E\right)  .$ Putting $\omega=\omega_{0}$ in
(\ref{representation252}) and using estimates (\ref{eig22}), (\ref{eig19}) and
(\ref{eig20}) we have $\left(  \mathbf{h}_{N}-\left(  \operatorname*{sgn}%
E\right)  P_{\omega}\left(  E\right)  \left(  \alpha\cdot\omega\right)
P_{\theta}\left(  E\right)  -h\right)  \in\mathcal{S}^{-\rho_{1}}.$ Then,
decomposing $\theta\neq\omega$ as $\theta=\left\langle \theta,\omega
\right\rangle \omega+\tilde{\theta},$ $\tilde{\theta}\in\Pi_{\omega}\ $and
using Lemma \ref{representation201} and (\ref{representation99}) we get
\begin{equation}
\left.
\begin{array}
[c]{c}%
\left\vert \left(  2\pi\right)  ^{-2}\upsilon\left(  E\right)  ^{2}%
\int\limits_{\Pi_{\omega}}e^{i\nu\left(  E\right)  \left\langle y,\tilde
{\theta}\right\rangle }\left(  \mathbf{h}_{N}\left(  y,\omega,\theta
;E;\omega\right)  -\left(  \operatorname*{sgn}E\right)  P_{\omega}\left(
E\right)  \left(  \alpha\cdot\omega\right)  P_{\theta}\left(  E\right)
-h\left(  y,\omega,\theta;E\right)  \right)  dy\right\vert \\
\leq\left\{
\begin{array}
[c]{c}%
C\left\vert \omega-\theta\right\vert ^{-2+\rho_{1}},\text{ }\rho<2,\\
C\left\vert \omega-\theta\right\vert ^{-\varepsilon},\text{ }\varepsilon
>0,\text{ }\rho=2.
\end{array}
\right.
\end{array}
\right.  \label{representation96}%
\end{equation}
Moreover, for $\rho>2,$ as the integral in the L.H.S. of
(\ref{representation96}) is absolutely convergent, it is a continuous function
of $\omega$ and $\theta.$

Let us show that for all $\alpha,$
\begin{equation}
\left\vert \int\limits_{0}^{\infty}\partial_{y}^{\alpha}\left(  \mathcal{V}%
_{V,A;\omega}^{\left(  E\right)  }\left(  y\pm t\omega\right)  -\mathcal{V}%
_{V,A;\theta}^{\left(  E\right)  }\left(  y\pm t\theta\right)  \right)
dt\right\vert \leq C_{\alpha}\left\vert \omega-\theta\right\vert \left\langle
y\right\rangle ^{-\left(  \rho-1\right)  -\left\vert \alpha\right\vert },
\label{representation93}%
\end{equation}
Since $\mathcal{V}_{V,A;\omega}^{\left(  E\right)  }\left(  y\pm
t\theta\right)  -\mathcal{V}_{V,A;\theta}^{\left(  E\right)  }\left(  y\pm
t\theta\right)  =\left(  \operatorname*{sgn}E\right)  \left\langle
\omega-\theta,A\left(  y\pm t\theta\right)  \right\rangle ,$ and $A$ satisfies
the estimate (\ref{eig31}), then it is enough to prove the following relation
\begin{equation}
\left\vert \int\limits_{0}^{\infty}\partial_{y}^{\alpha}\left(  \mathcal{V}%
_{V,A;\omega}^{\left(  E\right)  }\left(  y\pm t\omega\right)  -\mathcal{V}%
_{V,A;\omega}^{\left(  E\right)  }\left(  y\pm t\theta\right)  \right)
dt\right\vert \leq C_{\alpha}\left\vert \omega-\theta\right\vert \left\langle
y\right\rangle ^{-\left(  \rho-1\right)  -\left\vert \alpha\right\vert }.
\label{representation242}%
\end{equation}
First take $\alpha=0.$ Using the mean value theorem we have%
\begin{equation}
\left.  \mathcal{V}_{V,A;\omega}^{\left(  E\right)  }\left(  y\pm
t\omega\right)  -\mathcal{V}_{V,A;\omega}^{\left(  E\right)  }\left(  y\pm
t\theta\right)  =\pm t\left\langle \left(  \left(  \triangledown
\mathcal{V}_{V,A;\omega}^{\left(  E\right)  }\right)  \left(  \pm ct\left(
\theta-\omega\right)  +\left(  y\pm t\omega\right)  \right)  \right)
,\omega-\theta\right\rangle ,\right.  \label{representation94}%
\end{equation}
for some $0\leq c\leq1.$ Estimates (\ref{eig31}) and (\ref{eig32}) for $A$ and
$V$ imply
\begin{equation}
\left.  \left\vert \left(  \triangledown\mathcal{V}_{V,A;\omega}^{\left(
E\right)  }\right)  \left(  \pm ct\left(  \theta-\omega\right)  +\left(  y\pm
t\omega\right)  \right)  \right\vert \leq C\left(  1+\left\vert ct\left(
\theta-\omega\right)  \pm\left(  y\pm t\omega\right)  \right\vert \right)
^{-\rho-1}.\right.  \label{representation95}%
\end{equation}
Let us take $\left\vert \tilde{\theta}\right\vert \leq\sqrt{1-\delta^{2}}$.
Then, for $\delta$ close enough to $1,$ we get $\left\vert ct\left(
\theta-\omega\right)  \pm\left(  y\pm t\omega\right)  \right\vert ^{2}%
=c^{2}t^{2}\left\vert \theta-\omega\right\vert ^{2}\pm2ct\left\vert
y\right\vert \left\langle \hat{y},\theta\right\rangle -2ct^{2}+2ct^{2}%
\left\langle \omega,\theta\right\rangle +\left\vert y\right\vert ^{2}+t^{2}%
~$\ $\geq\left(  1-\sqrt{1-\delta^{2}}\right)  \left\vert y\right\vert
^{2}+t^{2}\left(  1-\eta+2\left(  \delta-1\right)  \right)  \geq c_{1}\left(
\left\vert y\right\vert ^{2}+t^{2}\right)  ,$ for some $c_{1}>0.$ Using this
estimate in (\ref{representation95}) and substituting the resulting inequality
in (\ref{representation94}) we obtain estimate (\ref{representation242}), and
hence, relation (\ref{representation93}), for $\alpha=0$. The proof of
(\ref{representation93}) for $\alpha>0$ is analogous.

From (\ref{representation113}) we have that $\frac{\nu\left(  E\right)
}{\left\vert E\right\vert }P_{\omega}\left(  E\right)  =\left(
\operatorname*{sgn}E\right)  P_{\omega}\left(  E\right)  \left(  \alpha
\cdot\omega\right)  P_{\omega}\left(  E\right)  .$ Then, using $\left\vert
\partial_{y}^{\alpha}(-i\left(  \operatorname*{sgn}E\right)  (\Phi^{+}\left(
y,\nu\left(  E\right)  \omega;E\right)  \right.  $ $\left.  -\Phi^{-}\left(
y,\nu\left(  E\right)  \theta;E\right)  ))P_{\omega}\left(  E\right)  \left(
\alpha\cdot\omega\right)  (P_{\theta}\left(  E\right)  -P_{\omega}\left(
E\right)  )\right\vert \leq C_{\alpha}\left\vert \omega-\theta\right\vert
\left\langle y\right\rangle ^{-\left(  \rho-1\right)  -\left\vert
\alpha\right\vert }$ and (\ref{representation93}) we get, for all $\alpha,$%
\begin{equation}
\left\vert \partial_{y}^{\alpha}\left(  h\left(  y,\omega,\theta;E\right)
-i\frac{\nu\left(  E\right)  }{\left\vert E\right\vert }R\left(
y,\omega;E\right)  P_{\omega}\left(  E\right)  \right)  \right\vert \leq
C_{\alpha}\left\vert \omega-\theta\right\vert \left\langle y\right\rangle
^{-\left(  \rho-1\right)  -\left\vert \alpha\right\vert },
\label{representation92}%
\end{equation}
where%
\begin{equation}
R\left(  y,\omega;E\right)  :=\int\limits_{-\infty}^{\infty}\left(
\mathcal{V}_{V,A;\omega}^{\left(  E\right)  }\left(  y+t\omega\right)
\right)  dt. \label{representation219}%
\end{equation}

Using (\ref{representation92}), Lemma \ref{representation201} and
(\ref{representation99}) we have%
\begin{equation}
\left.  \left\vert \int\nolimits_{\Pi_{\omega}}e^{i\nu\left(  E\right)
\left\langle y,\tilde{\theta}\right\rangle }\left(  h\left(  y,\omega
,\theta;E\right)  +i\frac{\nu\left(  E\right)  }{\left\vert E\right\vert
}R\left(  y,\omega;E\right)  P_{\omega}\left(  E\right)  \right)
dy\right\vert \leq C\left\vert \omega-\theta\right\vert ^{-2+\rho}\right.  .
\label{representation97}%
\end{equation}
Then, relation (\ref{representation102}) follows from Theorem
\ref{representation179}, estimates (\ref{representation96}),
(\ref{representation97}) and equation $-i\left(  2\pi\right)  ^{-2}%
\upsilon\left(  E\right)  ^{2}\int\limits_{\Pi_{\omega}}e^{i\nu\left(
E\right)  \left\langle y,\tilde{\theta}\right\rangle }R\left(  y,\omega
;E\right)  dy$ $=-i\left(  2\pi\right)  ^{-1/2}\upsilon\left(  E\right)
^{2}\left(  \mathcal{FV}_{V,A;\omega}^{\left(  E\right)  }\right)  \left(
-\nu\left(  E\right)  \tilde{\theta}\right)  .$
\end{proof}

\begin{rem}
\rm{Suppose that $V,A\in C^{\infty}$ are such that $V=V_{0}\left\vert
x\right\vert ^{-\rho}\ $and $A=A_{0}\left\vert x\right\vert ^{-\rho},$
$1<\rho<3,$ for $\left\vert x\right\vert \geq R,$ for some $R>0,$ and $V_{0}$
is a real constant and $A_{0}$ is a constant, real vector, satisfying
$V_{0}+\left\langle \omega,A_{0}\right\rangle $ $\neq0,$ for all $\omega\in
S^{2}.$ Since $FV_{V,A;\omega}^{\left(  E\right)  }=F\left(  \mathcal{V}%
_{V,A;\omega}^{\left(  E\right)  }\mathbf{-}\left(  V_{0}+\left\langle
\omega,A_{0}\right\rangle \right)  \left\vert x\right\vert ^{-\rho}\right)
+\left(  V_{0}+\left\langle \omega,A_{0}\right\rangle \right)  F\left(
\left\vert x\right\vert ^{-\rho}\right)  $ and $F\left(  \left\vert
x\right\vert ^{-\rho}\right)  =2^{3-\rho}\pi^{\frac{3}{2}}\frac{\Gamma\left(
\frac{3-\rho}{2}\right)  }{\Gamma\left(  \frac{\rho}{2}\right)  }\left\vert
\xi\right\vert ^{-\left(  3-\rho\right)  }=4\pi\rho\left(  \rho-1\right)
\Gamma\left(  -\rho\right)  \left(  \sin\frac{\pi\rho}{2}\right)  \left\vert
\xi\right\vert ^{-\left(  3-\rho\right)  },$ where $\Gamma$ is the Gamma
function (see \cite{76}), then as in the non-relativistic case \cite{47},
relations (\ref{representation102}) and (\ref{representation99}) imply that
the estimate (\ref{representation83}) is optimal. This implies that the
relation $\left\vert s^{\operatorname{int}}\left(  \omega,\theta;E\right)
\right\vert \leq C\left\vert \omega-\theta\right\vert ^{-3+\rho}$ is the best
possible.}
\end{rem}

\subsection{The scattering cross-section.}

As before let $s^{\operatorname{int}}$ be the kernel of the operator $S\left(
E\right)  -I.$ Let us consider the principal part $\mathcal{G}_{0}\left(
E\right)  ,$ of $S\left(  E\right)  -I.$ We take $N=0$ in the relation
(\ref{representation181}). Then, using the definition (\ref{eig33}) and
relations (\ref{eig42}), (\ref{eig8}) and (\ref{eig9}) we get that the kernel
$\mathbf{g}_{0}^{\operatorname{int}}\left(  \omega,\theta;E\right)  $ of
$\mathcal{G}_{0}\left(  E\right)  $ is given by%
\begin{equation}
\left.  \mathbf{g}_{0}^{\operatorname{int}}\left(  \omega,\theta;E\right)
=\sum\limits_{O_{j}\cap O_{k}\neq\varnothing}\mathbf{g}_{0,jk}%
^{\operatorname{int}}\left(  \omega,\theta;E\right)  ,\right.
\label{representation80}%
\end{equation}
with $\mathbf{g}_{0,jk}^{\operatorname{int}}\left(  \omega,\theta;E\right)
=\left(  2\pi\right)  ^{-2}\upsilon\left(  E\right)  ^{2}\chi_{jk}\left(
\omega,\theta\right)  \chi_{j}\left(  \omega\right)  \chi_{k}\left(
\theta\right)
{\displaystyle\int_{\Pi_{\omega_{jk}}}}
e^{i\nu\left(  E\right)  \left\langle y,\theta-\omega\right\rangle }%
\mathbf{h}_{pr}\left(  y,\omega,\theta;E\right)  dy,$ and $\mathbf{h}%
_{pr}\left(  y,\omega,\theta;E\right)  :=\left(  \operatorname*{sgn}E\right)
$ $\times\left(  e^{-i\Phi^{+}\left(  y,\nu\left(  E\right)  \omega;E\right)
+i\Phi^{-}\left(  y,\nu\left(  E\right)  \theta;E\right)  }-1\right)
P_{\omega}\left(  E\right)  \left(  \alpha\cdot\omega_{jk}\right)  P_{\theta
}\left(  E\right)  .$

If $\rho>3,~$then Theorem \ref{representation25} assures that $\mathbf{g}%
_{N,jk}^{\operatorname{int}}\left(  \omega,\theta;E\right)  $ is a continuous
function of $\omega$ and$\ \theta$. Thus, we can consider the limit of
$\mathbf{g}_{0,jk}^{\operatorname{int}}\left(  \omega,\theta;E\right)  $ as
$\left\vert E\right\vert \rightarrow\infty$ on the diagonal $\omega
=\theta=\omega_{jk}.$ Taking in account (\ref{representation193}) and relation
(\ref{representation80}), and using (\ref{representation251}), we have
\begin{equation}
\lim_{\pm E\rightarrow\infty}\left\Vert \frac{s^{\operatorname{int}}\left(
\omega,\omega;E\right)  }{\upsilon\left(  E\right)  ^{2}}-\left(  2\pi\right)
^{-2}\int\limits_{\Pi_{\omega}}m_{\pm}\left(  y,\omega\right)  dy\right\Vert
_{\mathcal{B}\left(  X^{\pm}\left(  \nu\left(  E\right)  \omega\right)
\right)  }=0, \label{representation79}%
\end{equation}
where $m_{\pm}\left(  y,\omega\right)  =\left(  e^{-i\int_{-\infty}^{\infty
}\left(  V\left(  y+t\omega\right)  \pm\left\langle \omega,A\left(
y+t\omega\right)  \right\rangle \right)  dt}-1\right)  .$ Equality
(\ref{representation79}) was proved in \cite{15} by studying the high-energy
limit of the resolvent.

Now let us prove the following result

\begin{proposition}
\label{representation210}Suppose that $\rho>2.$ Then, the function
$s^{\operatorname{int}}\left(  \omega,\theta;E\right)  +s^{\operatorname{int}%
}\left(  \theta,\omega;E\right)  ^{\ast}$ is continuous on $\mathbb{S}%
^{2}\times\mathbb{S}^{2}.$
\end{proposition}

\begin{proof}
It follows from estimates (\ref{eig22}), (\ref{eig19}) and (\ref{eig20}), and
definition (\ref{eig33}) that $\mathbf{h}_{N,jk}-\mathbf{h}_{0,jk}%
\in\mathcal{S}^{-\rho}.$ Then, if $\rho>2,$ we have $\chi_{j}\left(
\omega\right)  s\left(  \omega,\theta;E\right)  \chi_{k}\left(  \theta\right)
-\mathbf{g}_{0,jk}^{\operatorname{int}}\left(  \omega,\theta;E\right)  \in
C^{0}\left(  \mathbb{S}^{2}\times\mathbb{S}^{2}\right)  .$ Thus, we only need
to show that the sum $\mathbf{g}_{0,jk}^{\operatorname{int}}\left(
\omega,\theta;E\right)  +\mathbf{g}_{0,jk}^{\operatorname{int}}\left(
\theta,\omega;E\right)  ^{\ast}$ is continuous on $\mathbb{S}^{2}%
\times\mathbb{S}^{2}.$ From the definition of $\mathbf{g}_{0,jk}%
^{\operatorname{int}}$ we have%
\begin{equation}
\left.
\begin{array}
[c]{c}%
\mathbf{g}_{0,jk}^{\operatorname{int}}\left(  \omega,\theta;E\right)
+\mathbf{g}_{0,jk}^{\operatorname{int}}\left(  \theta,\omega;E\right)  ^{\ast
}=\left(  \operatorname*{sgn}E\right)  \left(  2\pi\right)  ^{-2}%
\upsilon\left(  E\right)  ^{2}\chi_{jk}\left(  \omega,\theta\right)  \chi
_{j}\left(  \omega\right)  \chi_{k}\left(  \theta\right) \\
\times%
{\displaystyle\int\limits_{\Pi_{\omega_{jk}}}}
e^{i\nu\left(  E\right)  \left\langle y,\theta-\omega\right\rangle }\left(
e^{-i\Phi^{+}\left(  y,\nu\left(  E\right)  \omega;E\right)  +i\Phi^{-}\left(
y,\nu\left(  E\right)  \theta;E\right)  }+e^{i\Phi^{+}\left(  y,\nu\left(
E\right)  \omega;E\right)  -i\Phi^{-}\left(  y,\nu\left(  E\right)
\theta;E\right)  }-2\right. \\
\left.  +e^{i\Phi^{+}\left(  y,\nu\left(  E\right)  \theta;E\right)
-i\Phi^{-}\left(  y,\nu\left(  E\right)  \omega;E\right)  }-e^{i\Phi
^{+}\left(  y,\nu\left(  E\right)  \omega;E\right)  -i\Phi^{-}\left(
y,\nu\left(  E\right)  \theta;E\right)  }\right)  P_{\omega}\left(  E\right)
\left(  \alpha\cdot\omega_{jk}\right)  P_{\theta}\left(  E\right)  dy.
\end{array}
\right.  \label{representation211}%
\end{equation}
Note that%
\begin{equation}
\left.
\begin{array}
[c]{c}%
{\displaystyle\int\limits_{\Pi_{\omega_{jk}}}}
e^{i\nu\left(  E\right)  \left\langle y,\theta-\omega\right\rangle }\left(
e^{-i\Phi^{+}\left(  y,\nu\left(  E\right)  \omega;E\right)  +i\Phi^{-}\left(
y,\nu\left(  E\right)  \theta;E\right)  }+e^{i\Phi^{+}\left(  y,\nu\left(
E\right)  \omega;E\right)  -i\Phi^{-}\left(  y,\nu\left(  E\right)
\theta;E\right)  }-2\right)  dy\\
=2%
{\displaystyle\int\limits_{\Pi_{\omega_{jk}}}}
e^{i\nu\left(  E\right)  \left\langle y,\theta-\omega\right\rangle }\left(
\cos\left(  \int\limits_{0}^{\infty}\mathcal{V}_{V,A;\omega}^{\left(
E\right)  }\left(  y+t\omega\right)  +\mathcal{V}_{V,A;\theta}^{\left(
E\right)  }\left(  y-t\theta\right)  dt\right)  -1\right)  dy.
\end{array}
\right.  \label{representation111}%
\end{equation}
The R.H.S. of relation (\ref{representation111}) is absolutely convergent if
$\rho>2.$ Thus, to complete the proof, it is enough to show that $%
{\displaystyle\int\limits_{\Pi_{\omega_{jk}}}}
e^{i\nu\left(  E\right)  \left\langle y,\theta-\omega\right\rangle }%
(e^{i\Phi^{+}\left(  y,\nu\left(  E\right)  \theta;E\right)  -i\Phi^{-}\left(
y,\nu\left(  E\right)  \omega;E\right)  }-e^{i\Phi^{+}\left(  y,\nu\left(
E\right)  \omega;E\right)  -i\Phi^{-}\left(  y,\nu\left(  E\right)
\theta;E\right)  })dy$ is continuous. Since $\left\vert \left(  \Phi^{\pm
}\left(  y,\nu\left(  E\right)  \omega;E\right)  \right)  ^{n}\right\vert $
$\leq C_{n}\left\langle y\right\rangle ^{-\left(  \rho-1\right)  n},$ we only
have to prove the continuity of the following integral%
\begin{equation}
i%
{\displaystyle\int\limits_{\Pi_{\omega_{jk}}}}
e^{i\nu\left(  E\right)  \left\langle y,\theta-\omega\right\rangle }\left(
\left(  \Phi^{+}\left(  y,\nu\left(  E\right)  \theta;E\right)  -\Phi
^{-}\left(  y,\nu\left(  E\right)  \omega;E\right)  \right)  -\left(  \Phi
^{+}\left(  y,\nu\left(  E\right)  \omega;E\right)  -\Phi^{-}\left(
y,\nu\left(  E\right)  \theta;E\right)  \right)  \right)  dy.
\label{representation208}%
\end{equation}
Note that estimate (\ref{representation93}) implies $\left\vert \partial
_{y}^{\alpha}\left(  \left(  \Phi^{+}\left(  y,\nu\left(  E\right)
\theta;E\right)  -\Phi^{+}\left(  y,\nu\left(  E\right)  \omega;E\right)
\right)  +\left(  \Phi^{-}\left(  y,\nu\left(  E\right)  \theta;E\right)
-\Phi^{-}\left(  y,\nu\left(  E\right)  \omega;E\right)  \right)  \right)
\right\vert $ $\leq C_{\alpha}\left\vert \omega-\theta\right\vert \left\langle
y\right\rangle ^{-\left(  \rho-1\right)  -\left\vert \alpha\right\vert }.$
Then, it follows from Lemma \ref{representation201} that integral
(\ref{representation208}) is estimated by $C\left\vert \omega-\theta
\right\vert ^{\rho-2},$ and thus, it is continuous.\
\end{proof}

We define the scattering cross-section for a fixed incoming direction $\theta$
and all outgoing directions $\omega$ by the following relation $\sigma\left(
\theta;E;u\right)  =\int_{\mathbb{S}^{2}}\left\vert s^{\operatorname{int}%
}\left(  \omega,\theta;E\right)  u\right\vert ^{2}d\omega,$ for a normalized
initial state $u\in X^{\pm}\left(  \nu\left(  E\right)  \theta\right)  ,$
$\left\vert u\right\vert _{\mathbb{C}^{4}}=1.$ Using (\ref{representation103})
we have
\begin{equation}
\sigma\left(  \theta;E;u\right)  =-\left(  \left(  s^{\operatorname{int}%
}\left(  \theta,\theta;E\right)  +s^{\operatorname{int}}\left(  \theta
,\theta;E\right)  ^{\ast}\right)  u,u\right)  . \label{representation109}%
\end{equation}

The following Lemma is consequence of the relation (\ref{representation109})
and Proposition \ref{representation210}.

\begin{lemma}
\label{representation234}\bigskip The scattering cross-section $\sigma\left(
\theta;E;u\right)  $ is a continuous function of $\theta,$ for $\rho>2.$
Furthermore, the total scattering cross-section, given by the relation
$\int_{\mathbb{S}^{2}}\int_{\mathbb{S}^{2}}\left\vert s^{\operatorname{int}%
}\left(  \omega,\theta;E\right)  u\right\vert ^{2}d\theta d\omega,$ for a
normalized initial state $u\in X^{\pm}\left(  \nu\left(  E\right)
\theta\right)  ,$ $\left\vert u\right\vert _{\mathbb{C}^{4}}=1,$ is finite if
$\rho>2.$
\end{lemma}

The estimate%
\begin{equation}
\left\vert \partial_{y}^{\alpha}\partial_{\zeta}^{\beta}\partial
_{\zeta^{\prime}}^{\gamma}\left(  \mathbf{\tilde{h}}_{N}^{\operatorname{int}%
}\left(  y,\zeta,\zeta^{\prime};E;\omega_{jk}\right)  -\mathbf{\tilde{h}}%
_{pr}\left(  y,\zeta,\zeta^{\prime};E\right)  \right)  \right\vert \leq
C_{\alpha,\beta,\,\gamma}\nu\left(  E\right)  ^{-1}\left\langle y\right\rangle
^{-\rho-\left\vert \alpha\right\vert }, \label{representation213}%
\end{equation}
where $\mathbf{h}_{N}^{int}$ is defined (\ref{representation253}), and Theorem
\ref{representation179} imply
\begin{equation}
\upsilon\left(  E\right)  ^{-2}\left(  \chi_{j}\left(  \omega\right)
s^{\operatorname{int}}\left(  \omega,\theta;E\right)  \chi_{k}\left(
\theta\right)  -\mathbf{g}_{0,jk}^{\operatorname{int}}\left(  \omega
,\theta;E\right)  \right)  =O\left(  \left\vert E\right\vert ^{-1}\right)
,\text{ as }\left\vert E\right\vert \rightarrow\infty.
\label{representation112}%
\end{equation}
Then, for any $u\in X^{\pm}\left(  \nu\left(  E\right)  \theta\right)  ,$
$\left\vert u\right\vert _{\mathbb{C}^{4}}=1,$ taking $\omega=\omega
_{jk}=\theta$ and using relations (\ref{representation109}),
(\ref{representation112}), (\ref{representation211}), (\ref{representation111}%
), (\ref{representation113}) and equalities $P_{\theta}\left(  E\right)  u=u,$
$\chi_{jk}\left(  \theta,\theta\right)  \chi_{j}\left(  \theta\right)
\chi_{k}\left(  \theta\right)  =\chi_{j}\left(  \theta\right)  \chi_{k}\left(
\theta\right)  $ and $\sum\limits_{O_{j}\cap O_{k}\neq\varnothing}$ $\chi
_{j}\left(  \theta\right)  \chi_{k}\left(  \theta\right)  =1,$ we get $\left(
2\pi\right)  ^{2}\upsilon\left(  E\right)  ^{-2}\sigma\left(  \theta
;E;u\right)  $ $=2\int\limits_{\Pi_{\theta}}\left(  1-\cos\int\limits_{-\infty
}^{\infty}\mathcal{V}_{V,A;\theta}^{\left(  E\right)  }\left(  y+t\theta
\right)  dt\right)  dy+O\left(  \left\vert E\right\vert ^{-1}\right)  ,$ as
$\left\vert E\right\vert \rightarrow\infty.$ A similar result was obtained in
\cite{15} by studying the high-energy limit of the resolvent.

The following result is a consequence of Theorem \ref{representation25} and
Proposition \ref{representation231}

\begin{proposition}
Let the electric potential $V$ satisfy estimate (\ref{eig32}) with some
$\rho_{e}>1$ and the magnetic field $B$ satisfy the estimate
(\ref{basicnotions18}) with $r=\rho_{m}+1,$ $\rho_{m}>1$ and all $d.$ Let $V$
and $B$ be homogeneous functions of order $-\rho_{e}$ and $-\rho_{m}-1$,
respectively, for $\left\vert x\right\vert \geq R,$ for some $R>0,$ and at
least one of them is non-trivial for $\left\vert x\right\vert \geq R.$ Then
the total scattering cross-section is infinite if and only if $\rho\leq2,$
where if both $V$ and $B$ are non-trivial for $\left\vert x\right\vert \geq
R,$ then $\rho=\min\{\rho_{e},\rho_{m}\},$ if $V$ is trivial, $\rho=\rho_{m}$
and if $B$ is trivial, $\rho=\rho_{e}.$
\end{proposition}

\begin{proof}
Note that Lemma \ref{representation234} shows that the total scattering
cross-section is finite if $\rho>2.$ Let the magnetic potential $A$ be defined
by the equalities (\ref{basicnotions19})-(\ref{basicnotions21}). Since $B$ is
homogeneous of order $-\rho_{m}-1$, $A$ is homogeneous of order $-\rho_{m}$.
Thus we get $\mathcal{V}_{V,A;\omega}^{\left(  E\right)  }\left(  x\right)
=\left\vert x\right\vert ^{-\rho}\left(  V_{\operatorname*{ang}}\left(
\hat{x}\right)  +\left(  \operatorname*{sgn}E\right)  \left\langle
\omega,A_{\operatorname*{ang}}\left(  \hat{x}\right)  \right\rangle \right)
+W\left(  x\right)  $ for $\left\vert x\right\vert \geq R$, for some
$V_{\operatorname*{ang}}\in C^{\infty}\left(  \mathbb{S}^{2}\right)  $,
$A_{\operatorname*{ang}}\left(  \hat{x}\right)  \in C^{\infty}\left(
\mathbb{S}^{2};\mathbb{R}^{3}\right)  $ and some $W\left(  x\right)  $
homogeneous of order $\rho_{1}=\max\{\rho_{e},\rho_{m}\}.$ Note that if
$\rho_{e}=\rho_{m},$ $W\left(  x\right)  \equiv0.$ Then we have
\begin{equation}
\left.
\begin{array}
[c]{c}%
\left\vert \left(  \mathcal{FV}_{V,A;\omega}^{\left(  E\right)  }\right)
\left(  -\nu\left(  E\right)  \tilde{\theta}\right)  P_{\omega}\left(
E\right)  u\right\vert =\frac{1}{\left(  2\pi\right)  ^{3/2}}\left\vert
{\displaystyle\int}
\left(  e^{i\nu\left(  E\right)  \left\langle \tilde{\theta},x\right\rangle
}\mathbf{V}_{h}^{\omega}\left(  x\right)  u\right)  dx-%
{\displaystyle\int\limits_{\left\vert x\right\vert \leq R}}
\left(  e^{i\nu\left(  E\right)  \left\langle \tilde{\theta},x\right\rangle
}\mathbf{V}_{h}^{\omega}\left(  x\right)  u\right)  dx\right. \\
\left.  +%
{\displaystyle\int\limits_{\left\vert x\right\vert \geq R}}
\left(  e^{i\nu\left(  E\right)  \left\langle \tilde{\theta},x\right\rangle
}W\left(  x\right)  P_{\omega}\left(  E\right)  u\right)  dx+%
{\displaystyle\int\limits_{\left\vert x\right\vert \leq R}}
\left(  e^{i\nu\left(  E\right)  \left\langle \tilde{\theta},x\right\rangle
}\mathcal{V}_{V,A;\omega}^{\left(  E\right)  }\left(  x\right)  P_{\omega
}\left(  E\right)  u\right)  dx\right\vert ,
\end{array}
\right.  \label{representation232}%
\end{equation}
where $\mathbf{V}_{h}^{\omega}=\left(  V_{h}^{\left(  1\right)  }+\left(
\operatorname*{sgn}E\right)  \left\langle \omega,V_{h}^{\left(  2\right)
}\right\rangle \right)  P_{\omega}\left(  E\right)  $ and $V_{h}^{\left(
1\right)  }\left(  x\right)  =\left\vert x\right\vert ^{-\rho}%
V_{\operatorname*{ang}}\left(  \hat{x}\right)  ,$ $V_{h}^{\left(  2\right)
}\left(  x\right)  =\left\vert x\right\vert ^{-\rho}A_{\operatorname*{ang}%
}\left(  \hat{x}\right)  $ for $x\in\mathbb{R}^{3}.$

Note that if a function $f\left(  x\right)  :=\left\vert x\right\vert ^{-\rho
}f_{\operatorname*{ang}}\left(  \hat{x}\right)  ,$ where
$f_{\operatorname*{ang}}\left(  \hat{x}\right)  $ belongs to $C^{\infty
}\left(  \mathbb{S}^{2}\right)  $ and it is non-trivial, then its Fourier
transform is given by $\hat{f}\left(  \xi\right)  =\left\vert \xi\right\vert
^{-3+\rho}\hat{f}\left(  \hat{\xi}\right)  $ and $\hat{f}\left(  \hat{\xi
}\right)  $ is also a non-trivial, $C^{\infty}\left(  \mathbb{S}^{2}\right)  $
function. This means that $\hat{V}_{h}^{\left(  j\right)  }\left(  \xi\right)
\in C^{\infty}\left(  \mathbb{R}^{3}\backslash\{0\}\right)  $ and $\hat{V}%
_{h}^{\left(  j\right)  }\left(  \xi\right)  =\left\vert \xi\right\vert
^{-3+\rho}\hat{V}_{h}^{\left(  j\right)  }\left(  \hat{\xi}\right)  $ for
$\xi\neq0,$ $j=1,2$, and hence, $\mathbf{\hat{V}}_{h}^{\omega}\left(
\xi\right)  =\left\vert \xi\right\vert ^{-3+\rho}\mathbf{\hat{V}}_{h}^{\omega
}\left(  \hat{\xi}\right)  $.

Suppose that $\rho=\rho_{e}$, what implies that $\hat{V}_{h}^{\left(
1\right)  }\left(  \hat{\xi}\right)  $ is non-trivial. We take $\hat{\xi}$
such that $\hat{V}_{h}^{\left(  1\right)  }\left(  \hat{\xi}\right)  \neq0.$
Let $\omega_{1}$ be orthogonal to $\hat{\xi}$ and suppose that $\mathbf{\hat
{V}}_{h}^{\omega_{1}}\left(  \hat{\xi}\right)  \neq0.$ By continuity we get
$\left\vert \mathbf{\hat{V}}_{h}^{\omega}\left(  \hat{\zeta}\right)
\right\vert >c,$ for some constant $c>0,$ $\left\vert \omega-\omega
_{1}\right\vert <\delta_{1},$ for some $\delta_{1}>0,$ as in Proposition
\ref{representation231}, and $\hat{\zeta}$ such that $\left\langle \hat{\zeta
},\hat{\xi}\right\rangle >1-\varepsilon$ for some $\varepsilon>0.$ Hence,
$\left\vert \mathbf{\hat{V}}_{h}^{\omega}\left(  \zeta\right)  \right\vert
>c\left\vert \zeta\right\vert ^{-3+\rho},$ for $\left\vert \omega-\omega
_{1}\right\vert <\delta_{1},$ $\delta_{1}>0,$ and $\zeta$ such that
$\left\langle \zeta,\xi\right\rangle >\left(  1-\varepsilon\right)  \left\vert
\zeta\right\vert \left\vert \xi\right\vert \mathbf{.}$ Note that there are
$W_{1}$ and $W_{2}$ such that $\mathcal{V}_{V,A;\omega}^{\left(  E\right)
}\left(  x\right)  F\left(  \left\vert x\right\vert \leq R\right)  +W\left(
x\right)  F\left(  \left\vert x\right\vert \geq R\right)  =W_{1}\left(
x\right)  +W_{2}\left(  x\right)  ,$ where $W_{1}$ satisfies (\ref{eig32})
with $\rho_{1},$ and $W_{2}$ is a $C^{\infty}$ function for $\left\vert
x\right\vert \leq R,$ and $W_{2}\equiv0,$ for $\left\vert x\right\vert >R$
($F\left(  \cdot\right)  $ is the characteristic function of the correspondent
set). Observe that, if $\rho_{e}=\rho_{m}$, then we have $W_{1}\left(
x\right)  \equiv0$ and $W_{2}\left(  x\right)  =\mathcal{V}_{V,A;\omega
}^{\left(  E\right)  }\left(  x\right)  $ for $\left\vert x\right\vert \leq
R.$ It follows from Lemma \ref{representation201} that $\left\vert
{\displaystyle\int}
\left(  e^{i\nu\left(  E\right)  \left\langle \tilde{\theta},x\right\rangle
}W_{1}\left(  x\right)  u\right)  dx\right\vert \leq C\left(  \rho_{1}%
-\rho\right)  \left\vert \omega-\theta\right\vert ^{-3+\rho_{2}},$ where
$\rho_{2}=\rho_{1},$ for $\ \rho_{1}<3$, and $\rho_{2}=3-\varepsilon,$
$\varepsilon>0$ for $\rho_{1}\geq3.$ We note that $\left\vert -%
{\displaystyle\int\limits_{\left\vert x\right\vert \leq R}}
\left(  e^{-i\left\langle y,x\right\rangle }\mathbf{V}_{h}^{\omega}\left(
x\right)  u\right)  dx+%
{\displaystyle\int\limits_{\left\vert x\right\vert \leq R}}
\left(  e^{-i\left\langle y,x\right\rangle }W_{2}\left(  x\right)  P_{\omega
}\left(  E\right)  u\right)  dx\right\vert \leq C,$ uniformly in $y$ and
$\omega.$ Moreover, if $\left\vert \omega-\omega_{1}\right\vert <\delta_{1}$
and $-\left\langle \tilde{\theta}/\left\vert \tilde{\theta}\right\vert
,\hat{\xi}\right\rangle >1-\varepsilon$ we have $\frac{1}{\left(  \nu\left(
E\right)  \left\vert \tilde{\theta}\right\vert \right)  ^{3-\rho}}\left\vert
\left\vert \mathbf{\hat{V}}_{h}^{\omega_{1}}\left(  -\tilde{\theta}/\left\vert
\tilde{\theta}\right\vert \right)  \right\vert -C(\left\vert \rho_{e}-\rho
_{m}\right\vert \right.  $ $\left.  \times\left(  \nu\left(  E\right)
\left\vert \tilde{\theta}\right\vert \right)  ^{-3+\rho_{2}}+1)\left(
\nu\left(  E\right)  \left\vert \tilde{\theta}\right\vert \right)  ^{3-\rho
}\right\vert $ $\geq\frac{c}{2\left(  \nu\left(  E\right)  \left\vert
\tilde{\theta}\right\vert \right)  ^{3-\rho}},$ for $\left\vert \tilde{\theta
}\right\vert <\varepsilon_{1}$ and some $\varepsilon_{1}>0$. Then, from
relations (\ref{representation232}) and (\ref{representation102}) we get%
\begin{equation}
\left.
\begin{array}
[c]{c}%
\left\vert s^{\operatorname{int}}\left(  \omega,\theta;E\right)  \right\vert
\geq\left\vert \left(  2\pi\right)  ^{-1/2}\upsilon\left(  E\right)  ^{2}%
\frac{\nu\left(  E\right)  }{\left\vert E\right\vert }\left(  \mathcal{FV}%
_{V,A;\omega}^{\left(  E\right)  }\right)  \left(  -\nu\left(  E\right)
\tilde{\theta}\right)  P_{\omega}\left(  E\right)  \right\vert \\
-\left\vert \left(  s^{\operatorname{int}}\left(  \omega,\theta;E\right)
-\frac{1}{i}\left(  2\pi\right)  ^{-1/2}\upsilon\left(  E\right)  ^{2}%
\frac{\nu\left(  E\right)  }{\left\vert E\right\vert }\left(  \mathcal{FV}%
_{V,A;\omega}^{\left(  E\right)  }\right)  \left(  -\nu\left(  E\right)
\tilde{\theta}\right)  P_{\omega}\left(  E\right)  \right)  \right\vert
\geq\frac{c}{\left(  \nu\left(  E\right)  \left\vert \tilde{\theta}\right\vert
\right)  ^{3-\rho}},
\end{array}
\right.  \label{representation235}%
\end{equation}
for $\left\vert \tilde{\theta}\right\vert <\varepsilon_{2}\leq\varepsilon_{1}$
and some $\varepsilon_{2}>0$. Let $\delta_{1}$ be such that the set
$\Theta_{\omega}:=\{\theta\in\Omega_{+}\left(  \omega,\delta\right)
|-\left\langle \tilde{\theta}/\left\vert \tilde{\theta}\right\vert ,\hat{\xi
}\right\rangle >1-\varepsilon$ and $\left\vert \tilde{\theta}\right\vert
<\varepsilon_{2}\}$ is of positive measure. Then, using
(\ref{representation235}) we obtain $\int\limits_{\mathbb{S}^{2}}%
\int\limits_{\mathbb{S}^{2}}\left\vert s^{\operatorname{int}}\left(
\omega,\theta;E\right)  \right\vert ^{2}d\theta d\omega\geq\frac{c}{\nu\left(
E\right)  ^{6-2\rho}}\int\limits_{\left\vert \omega-\omega_{1}\right\vert
\leq\delta_{1}}\int\limits_{\Theta_{\omega}}\frac{1}{\left\vert \tilde{\theta
}\right\vert ^{6-2\rho}}d\theta d\omega.$ As $\int_{\Theta_{\omega}}\frac
{1}{\left\vert \omega-\theta\right\vert ^{6-2\rho}}d\theta$ is infinite, for
$\rho\leq2,$ then using relation (\ref{representation99}) we conclude that
$\int_{\mathbb{S}^{2}}\int_{\mathbb{S}^{2}}\left\vert s^{\operatorname{int}%
}\left(  \omega,\theta;E\right)  u\right\vert ^{2}d\theta d\omega=\infty$.

Suppose again that $\rho=\rho_{e}$ and $\hat{\xi}$ is such that $\hat{V}%
_{h}^{\left(  1\right)  }\left(  \hat{\xi}\right)  \neq0,$ but $\mathbf{\hat
{V}}_{h}^{\omega_{1}}\left(  \hat{\xi}\right)  =0.$ Noting that $\mathbf{\hat
{V}}_{h}^{-\omega_{1}}\left(  \hat{\xi}\right)  \neq0$ and proceeding
similarly as above we get $\int_{\mathbb{S}^{2}}\int_{\mathbb{S}^{2}%
}\left\vert s^{\operatorname{int}}\left(  \omega,\theta;E\right)  u\right\vert
^{2}d\theta d\omega=\infty$ if $\rho=\rho_{e}$.

Now suppose that $\rho=\rho_{m}$. Let us take $\hat{\xi}_{1}$ such that
$\hat{V}_{h}^{\left(  2\right)  }\left(  \hat{\xi}_{1}\right)  \neq0.$ By
continuity $\left\vert \hat{V}_{h}^{\left(  2\right)  }\left(  \hat{\xi
}\right)  \right\vert >c>0$ for all $\hat{\xi}$ close enough to $\hat{\xi}%
_{1}$. Consider the set $\Psi\subset\mathbb{S}^{2}$ of all $\hat{\xi}$ such
that $\hat{V}_{h}^{\left(  2\right)  }\left(  \hat{\xi}\right)  \neq0.$ We
claim that there is $\omega_{1}$ orthogonal to some $\hat{\xi}\in\Psi$ such
that $\left\langle \omega_{1},\hat{V}_{h}^{\left(  2\right)  }\left(  \hat
{\xi}\right)  \right\rangle \neq0.$ Suppose that this is not true. That is,
for every $\hat{\xi}\in\Psi$, $\left\langle \omega,\hat{V}_{h}^{\left(
2\right)  }\left(  \hat{\xi}\right)  \right\rangle =0,$ for all $\omega$
orthogonal to $\hat{\xi}.$ This implies that $\left\langle \hat{\xi},\hat
{V}_{h}^{\left(  2\right)  }\left(  \hat{\xi}\right)  \right\rangle
=\pm\left\vert \hat{V}_{h}^{\left(  2\right)  }\left(  \hat{\xi}\right)
\right\vert .$ Hence, taking in account that $\hat{V}_{h}^{\left(  2\right)
}\left(  \hat{\xi}\right)  =0,$ for $\hat{\xi}\in\mathbb{S}^{2}\diagdown\Psi,$
we have $\xi\times\hat{V}_{h}^{\left(  2\right)  }\left(  \xi\right)  =0,$ for
all $\xi,$ what implies that $\operatorname{curl}V_{h}^{\left(  2\right)
}=B=0,$ for $\left\vert x\right\vert \geq R.$ This is a contradiction. Then
there is $\omega_{1}$ orthogonal to some $\hat{\xi}$ such that $\left\langle
\omega_{1},\hat{V}_{h}^{\left(  2\right)  }\left(  \hat{\xi}\right)
\right\rangle \neq0.$ Similarly to the case when $\rho=\rho_{e}$ we obtain
$\int_{\mathbb{S}^{2}}\int_{\mathbb{S}^{2}}\left\vert s^{\operatorname{int}%
}\left(  \omega,\theta;E\right)  u\right\vert ^{2}d\theta d\omega=\infty$.
\end{proof}

\subsection{Reconstruction of the electric potential and the magnetic field
from the high energy limit.}

Now we consider a special limit when $\left\vert E\right\vert \rightarrow
\infty$ and $\omega\left(  E\right)  ,\theta\left(  E\right)  \rightarrow
\omega,$ for an arbitrary $\omega\in\mathbb{S}^{2},$ in such way that
$\eta:=\nu\left(  E\right)  \left(  \omega\left(  E\right)  -\theta\left(
E\right)  \right)  \neq0$ is fixed (see \cite{75}). Let us take two families
of vectors $\omega\left(  E\right)  ,\theta\left(  E\right)  \in\mathbb{S}%
^{2}$ with these properties. We obtain the following result

\begin{proposition}
Let the magnetic potential $A\left(  x\right)  $ and the electric potential
$V\left(  x\right)  $ satisfy the estimates (\ref{eig31}) and (\ref{eig32})
respectively. For $\eta\in\mathbb{R}^{3}\backslash\{0\},$ let $\omega\left(
E\right)  ,\theta\left(  E\right)  \in\mathbb{S}^{2}$ be as above. Then, we
have%
\begin{equation}
\left.  \lim_{\substack{\pm E\rightarrow\infty}}\upsilon\left(  E\right)
^{-2}s\left(  \omega\left(  E\right)  ,\theta\left(  E\right)  ;E\right)
=\left(  2\pi\right)  ^{-1}\mathcal{F}\left(  e^{-iR\left(  y,\omega;\pm
\infty\right)  }P_{\omega}^{\pm}\left(  \infty\right)  \right)  \left(
\eta\right)  \right.  , \label{representation114}%
\end{equation}
where $R\left(  y,\omega;\pm\infty\right)  :=\int\limits_{-\infty}^{\infty
}\left(  V\left(  y+t\omega\right)  \pm\left\langle \omega,A\left(
y+t\omega\right)  \right\rangle \right)  dt$ and $P_{\omega}\left(  \pm
\infty\right)  =\frac{1}{2}\left(  I\pm\left(  a\cdot\omega\right)  \right)
.$
\end{proposition}

\begin{proof}
We follow the proof of Proposition 6.7 of \cite{30} for the Schr\"{o}dinger
equation. For a fixed $\omega\in\mathbb{S}^{2},$ we take a cut-off function
$\Psi_{+}\left(  \omega,\theta;\omega\right)  ,$ supported, as function of
$\theta,$ on $\Omega_{+}\left(  \omega,\delta\right)  ,$ such that it is equal
to $1$ in $\Omega_{+}\left(  \omega,\delta^{\prime}\right)  $, for some
$\delta^{\prime}>\delta.$ Let the first coordinate axis in $\Pi_{\omega}$ be
directed along $\eta.$ Then, integrating by parts in the R.H.S of relation
(\ref{representation248}) (understood as an oscillatory integral), with
respect to $y_{1},$ we get%
\begin{equation}
\left.  s_{\operatorname{sing}}^{(N)}\left(  \omega\left(  E\right)
,\theta\left(  E\right)  ;E;\omega\right)  =\left(  2\pi\right)  ^{-2}%
\upsilon\left(  E\right)  ^{2}\left(  i\left\vert \eta\right\vert \right)
^{-n}\Psi_{+}\left(  \omega\left(  E\right)  ,\theta\left(  E\right)
;\omega\right)
{\displaystyle\int\limits_{\Pi_{\omega}}}
e^{-i\left\langle y,\eta\right\rangle }\partial_{y_{1}}^{n}\mathbf{h}%
_{N}\left(  y,\omega\left(  E\right)  ,\theta\left(  E\right)  ;E;\omega
\right)  dy.\right.  \label{representation194}%
\end{equation}
For $n\geq2$ the integral in the last relation is absolutely convergent, as
\begin{equation}
\left\vert \partial_{y_{1}}^{n}\mathbf{h}_{N}\left(  y,\omega\left(  E\right)
,\theta\left(  E\right)  ;E;\omega\right)  \right\vert \leq C_{n}\left\langle
y\right\rangle ^{-\left(  \rho-1\right)  -n}. \label{representation214}%
\end{equation}
We have $\left\vert \partial_{y}^{\alpha}\partial_{\zeta}^{\beta}%
\partial_{\zeta^{\prime}}^{\gamma}\left(  \tilde{\Psi}_{+}\left(  \zeta\left(
E\right)  ,\zeta^{\prime}\left(  E\right)  ;\omega\right)  \left(
\mathbf{\tilde{h}}_{N}-\mathbf{\tilde{h}}_{0}\right)  \right)  \left(
y,\zeta,\zeta^{\prime};E\right)  \right\vert \leq$ $C_{\alpha,\beta,\,\gamma
}\nu\left(  E\right)  ^{-1}\left\langle y\right\rangle ^{-\rho-\left\vert
\alpha\right\vert },$ for some constants $C_{\alpha,\beta,\,\gamma},$
independent of $\zeta$ and $\zeta^{\prime}.$ Then, Theorem
\ref{representation179} implies that
\begin{equation}
\lim_{\substack{\left\vert E\right\vert \rightarrow\infty}}\upsilon\left(
E\right)  ^{-2}s\left(  \omega\left(  E\right)  ,\theta\left(  E\right)
;E\right)  =\lim\limits_{_{\substack{\left\vert E\right\vert \rightarrow
\infty}}}\upsilon\left(  E\right)  ^{-2}s_{\operatorname{sing}}^{(0)}\left(
\omega\left(  E\right)  ,\theta\left(  E\right)  ;E;\omega\right)  .
\label{representation115}%
\end{equation}
Using equality (\ref{representation251}) we see that%
\begin{equation}
\lim_{\pm E\rightarrow\infty}\left(  \Psi_{+}\left(  \omega\left(  E\right)
,\theta\left(  E\right)  ;\omega\right)  \partial_{y_{1}}^{n}\mathbf{h}%
_{0}\left(  y,\omega\left(  E\right)  ,\theta\left(  E\right)  ;E;\omega
\right)  \right)  =\left(  \partial_{y_{1}}^{n}e^{-iR\left(  y,\omega
;\pm\infty\right)  }\right)  P_{\omega}\left(  \pm\infty\right)  .
\label{representation116}%
\end{equation}
Estimate (\ref{representation214}) allows us to take the limit in
(\ref{representation194}), as $\pm E\rightarrow\infty$. Therefore, equalities
(\ref{representation115}) and (\ref{representation116}) imply%
\begin{equation}
\left.  \lim\limits_{\pm E\rightarrow\infty}\upsilon\left(  E\right)
^{-2}s^{\operatorname{int}}\left(  \omega\left(  E\right)  ,\theta\left(
E\right)  ;E;\omega\right)  =\left(  2\pi\right)  ^{-2}\left(  i\left\vert
\eta\right\vert \right)  ^{-n}%
{\displaystyle\int_{\Pi_{\omega}}}
e^{-i\left\langle y,\eta\right\rangle }\left(  \partial_{y_{1}}^{n}%
e^{-iR\left(  y,\omega;\pm\infty\right)  }\right)  P_{\omega}\left(  \pm
\infty\right)  dy,\right.  \label{representation195}%
\end{equation}
Integrating back by parts in the R.H.S. of \ref{representation195}, we obtain
(\ref{representation114}).
\end{proof}

Let us prove that we can uniquely reconstruct the electric potential $V$ and
the magnetic field $B$ from the limit (\ref{representation195}). The integral
in the R.H.S. of (\ref{representation195}) is, up to a coefficient, the
two-dimensional Fourier transform of $\left(  \partial_{y_{1}}^{n}%
e^{-iR\left(  y,\omega;\pm\infty\right)  }\right)  P_{\omega}^{\pm}\left(
\infty\right)  .$ The matrices $P_{\omega}^{\pm}\left(  \infty\right)  $ are
of the form $\frac{1}{2}%
\begin{pmatrix}
I & \pm\sum_{j=1}^{3}\sigma_{j}\omega_{j}\\
\pm\sum_{j=1}^{3}\sigma_{j}\omega_{j} & I
\end{pmatrix}
.$ Taking, for example, the first component of the matrix $\left(
\partial_{y_{1}}^{n}e^{-iR\left(  y,\omega;\pm\infty\right)  }\right)
P_{\omega}^{\pm}\left(  \infty\right)  $ we recover the function $\left(
\partial_{y_{1}}^{n}e^{-iR\left(  y,\omega;\pm\infty\right)  }\right)  $. Let
us take $y:=\left(  y_{1},y_{2},y_{3}\right)  \in\Pi_{\omega}.$ Since
$\partial_{y_{1}}^{n-1}e^{-iR\left(  y,\omega;\pm\infty\right)  }$ tends to
$0,$ as $\left\vert y\right\vert \rightarrow\infty,$ then $\left(
\partial_{y_{1}}^{n-1}e^{-iR\left(  y,\omega;\pm\infty\right)  }\right)
=-\left(  \int_{0}^{\infty}\partial_{y_{1}}^{n}e^{-iR\left(  y\left(
t\right)  ,\omega;\pm\infty\right)  }dt\right)  ,$ where $y\left(  t\right)
=\left(  y_{1}+t,y_{2},y_{3}\right)  \in\Pi_{\omega}.$ Applying this argument
$n-1$ times we get the function $\partial_{y_{1}}e^{-iR\left(  y,\omega
;\pm\infty\right)  }.$ Since $R\left(  y,\omega;\pm\infty\right)  $ tends to
$0,$ as $\left\vert y\right\vert \rightarrow\infty,$ we have $e^{-iR\left(
y,\omega;\pm\infty\right)  }-1=-\left(  \int_{0}^{\infty}\partial_{y_{1}%
}e^{-iR\left(  y\left(  t\right)  ,\omega;\pm\infty\right)  }dt\right)  .$
Thus, we recover the function $e^{-iR\left(  y,\omega;\pm\infty\right)  }%
\ $from the limit (\ref{representation195}). Since $R\left(  y,\omega
;\pm\infty\right)  $ is a continuous function of $y\in\Pi_{\omega}$ and tends
to $0,$ as $\left\vert y\right\vert \rightarrow\infty,$ we can determine the
function $R\left(  y,\omega;\pm\infty\right)  $ from the function
$e^{-iR\left(  y,\omega;\pm\infty\right)  }.$ Similarly to \cite{79} or
\cite{25}, noting that $R_{e}\left(  \omega,y;V\right)  :=\int\limits_{-\infty
}^{\infty}V\left(  y+t\omega\right)  dt$ is even in $\omega$ and $R_{m}\left(
\omega,y;A\right)  :=\int\limits_{-\infty}^{\infty}\left\langle \omega
,A\left(  y+t\omega\right)  \right\rangle dt$ is odd, we get $R_{e}$ and
$R_{m}$ from the limit (\ref{representation195}), and moreover, inverting the
Radon transform we recover the electric potential $V$ and the magnetic field
$B,$ from the scattering amplitude. We can formulate the obtained results as
the following

\begin{theorem}
\label{representation218}Suppose that the electric potential $V\left(
x\right)  $ and the magnetic field $B\left(  x\right)  $ satisfy the estimates
(\ref{eig32}) and (\ref{basicnotions18}), for all $\alpha$ and $d,$
respectively. Then the scattering amplitude $s\left(  \omega,\theta;E\right)
,$ known in some neighborhood of the diagonal $\omega=\theta,$ for every
$E\geq E_{0}$ or $-E\geq E_{0},$ for some $E_{0}>m,$ uniquely determines the
electric potential $V\left(  x\right)  $ and the magnetic field $B\left(
x\right)  .$ Moreover, one can reconstruct $V\left(  x\right)  $ and $B\left(
x\right)  $ from the high-energy limit (\ref{representation195}).
\end{theorem}

Using the high-energy asymptotics of the resolvent, Ito \cite{15} gave the
relation (\ref{representation114}) and he proved Theorem
\ref{representation218} for smooth electromagnetic potentials with $\rho
>3.$\ Jung \cite{52}, calculating the high-velocity limit for the scattering
operator, by the time-dependent method of \cite{62}, for continuous, Hermitian
matrix valued potentials $\mathbf{V}\left(  x\right)  ,$ satisfying the
conditions (\ref{intro7}) and $\left\Vert VF\left(  \left\vert x\right\vert
\geq R\right)  \right\Vert \in L^{1}\left(  \left[  0;\infty\right)
;dR\right)  ,$ where $F\left(  \left\vert x\right\vert \geq R\right)  $ is the
characteristic function of the set $\left\vert x\right\vert \geq R,$ presents
a reconstruction formula, that allows to uniquely recover the electric
potential and magnetic field from the scattering operator. Also using the
time-dependent method, Ito \cite{57} showed that for potentials $\mathbf{V}$
of the form (\ref{intro4}), satisfying $\left\vert \partial_{x}^{\alpha
}V\left(  x\right)  \right\vert +\left\vert \partial_{x}^{\alpha}A\left(
x\right)  \right\vert \leq C_{\alpha}\left\langle x\right\rangle ^{-\rho},$
$\rho>1,$ $\left\vert \alpha\right\vert \leq1,$ for all $x\in\mathbb{R}^{3},$
the electric potential $V\left(  x\right)  $ and the magnetic field $B\left(
x\right)  $ can be completely reconstructed from the scattering operator. Ito also considered time-dependent potentials in \cite{57}.

\subsection{Inverse problem at fixed energy for homogeneous potentials.}

We follow the approach of \cite{25}. For a fixed $\omega\in\mathbb{S}^{2},$ we
take a cut-off function $\Psi_{+}\left(  \omega,\theta;\omega\right)  ,$
supported, as function of $\theta,$ on $\Omega_{+}\left(  \omega
,\delta\right)  ,$ such that it is equal to $1$ in $\Omega_{+}\left(
\omega,\delta^{\prime}\right)  $, for some $\delta^{\prime}>\delta.$ It is
convenient for us to reformulate Theorem \ref{representation179} in terms of
asymptotic series. Let us rewrite formula (\ref{representation252}) in terms
of powers of the potential $W\left(  x\right)  :=$ $\left(  V\left(  x\right)
,A\left(  x\right)  \right)  .$ Note first that for $\left\vert E\right\vert
>m,$ $e^{i\Phi^{\pm}\left(  x,\xi;E\right)  }=\sum_{j=0}^{\infty}\frac{1}%
{j!}\left(  \pm i%
{\displaystyle\int\limits_{0}^{\infty}}
\left(  \frac{\left\vert E\right\vert }{\left\vert \xi\right\vert }V\left(
x\pm\left(  \operatorname*{sgn}E\right)  t\omega\right)  +\left(
\operatorname*{sgn}E\right)  \left\langle \omega,A\left(  x\pm\left(
\operatorname*{sgn}E\right)  t\omega\right)  \right\rangle \right)  dt\right)
^{j}.$ Introducing this expression in (\ref{eig33}) we can write $a_{N}^{\pm
}\left(  x,\xi;E\right)  $ as an asymptotic expansion $a_{N}^{\pm}\left(
x,\xi;E\right)  \simeq\sum_{m=0}^{\infty}a_{N,m}^{\pm}\left(  x,\xi\right)  ,$
where $a_{N,m}^{\pm}\left(  x,\xi\right)  $ is of order $m$ with respect to
$W\left(  x\right)  $.\ Plugging this expansion in the representation
(\ref{representation252}) for $\mathbf{h}_{N}\left(  y,\theta,\omega
;E\,;\omega\right)  $ and collecting together the terms of the same power with
respect to $W\left(  x\right)  $ we see that $\mathbf{h}_{N}\left(
y,\theta,\omega;E\,;\omega\right)  $ admits the expansion into the asymptotic
series
\begin{equation}
\mathbf{h}_{N}\left(  y,\theta,\omega;E\,;\omega\right)  \simeq\sum
\limits_{n=0}^{\infty}\mathbf{a}_{n}\left(  y,\omega,\theta;E\right)  ,
\label{direct17}%
\end{equation}
where $\mathbf{a}_{n}\left(  y,\omega,\theta;E\right)  \ $is of order $n$ with
respect to $W\left(  x\right)  .$ Note that $\mathbf{a}_{0}\left(
y,\omega,\theta;E\right)  =\left(  \operatorname*{sgn}E\right)  P_{\omega
}\left(  E\right)  \left(  \alpha\cdot\omega\right)  P_{\theta}\left(
E\right)  $ and that the lineal term with respect to $W\left(  x\right)  $ is
given by the relation%

\begin{equation}
\left.
\begin{array}
[c]{c}%
\mathbf{a}_{1}\left(  y,\omega,\theta;E\right)  =\\
-i\left(  \operatorname*{sgn}E\right)  \left(
{\displaystyle\int\limits_{0}^{\infty}}
\left(  \frac{\left\vert E\right\vert }{\left\vert \xi\right\vert }V\left(
y\pm t\omega\right)  \pm\left\langle \omega,A\left(  y\pm t\omega\right)
\right\rangle +\frac{\left\vert E\right\vert }{\left\vert \xi\right\vert
}V\left(  y\mp t\theta\right)  \pm\left\langle \theta,A\left(  y\mp
t\theta\right)  \right\rangle \right)  dt\right)  P_{\omega}\left(  E\right)
\left(  \alpha\cdot\omega\right)  P_{\theta}\left(  E\right)  ,
\end{array}
\right.  \label{direct1}%
\end{equation}
for $\pm E>m,$ up to a term from the class $\mathit{S}^{-\rho}$ (on variable
$y$). Moreover, it follows that $\mathbf{a}_{n}\in\mathit{S}^{-\left(
\rho-1\right)  n}.$

We will recover the asymptotics of the potential $W\left(  x\right)  $ from
the linear part, with respect to $W\left(  x\right)  $, of the symbol of the
operator $S_{\operatorname{sing}}\left(  E\right)  ,$ with kernel
$s_{\operatorname{sing}}^{(N)}.$ We compute the symbol $a\left(
y,\omega;E\right)  $ of $S_{\operatorname{sing}}\left(  E\right)  $ from its
amplitude $\upsilon\left(  E\right)  ^{2}\Psi_{+}\left(  \omega,\theta
;\omega\right)  $ $\times\mathbf{h}_{N}\left(  y,\omega,\theta;E;\omega
\right)  ,$ by making the change of variables $z=-\nu\left(  E\right)  y$ in
the definition (\ref{representation248}) of $\mathbf{h}_{N}$ and applying
(\ref{basicnotions45}). Recall the notation of Remark \ref{representation75}.
We get,
\begin{equation}
\tilde{a}\left(  y,\zeta;E\right)  \simeq\sum_{\beta}\frac{\upsilon\left(
E\right)  ^{2}}{\beta!}\left(  -i\nu\left(  E\right)  \right)  ^{-\left\vert
\beta\right\vert }\left.  \partial_{y}^{\beta}\partial_{\zeta^{\prime}}%
^{\beta}\mathbf{\tilde{h}}_{N}\left(  y,\zeta,\zeta^{\prime};E;\omega\right)
\right\vert _{\zeta^{\prime}=\zeta}, \label{direct16}%
\end{equation}
for $y\in\Pi_{\omega}$ (here we used that $\Psi_{+}\left(  \omega
,\theta;\omega\right)  =1,$ for $\theta\in\Omega_{+}\left(  \omega
,\delta^{\prime}\right)  $)$.$ The explicit formula for the symbol $a\left(
y,\omega;E\right)  $ can be obtained by plugging the expansion (\ref{direct17}%
) in the relation (\ref{direct16}). Note that $s_{\operatorname{sing}}^{(N)}$
is related with $a\left(  y,\omega;E\right)  $ by the expression
$s_{\operatorname{sing}}^{(N)}\left(  \omega,\theta;E\,;\omega\right)
=\left(  2\pi\right)  ^{-2}\int\limits_{\Pi_{\omega}}e^{i\nu\left(  E\right)
\left\langle y,\theta\right\rangle }a\left(  y,\omega;E\right)  dy.$ This
relation together with Theorem \ref{representation179} show that we can
recover the function $a\left(  y,\omega;E\right)  $ from the scattering
amplitude $s\left(  \omega,\theta;E\,\right)  $, up to a function from the
class $\mathit{S}^{-p\left(  N\right)  },$ that is, $a\left(  y,\omega
;E\right)  =\int\limits_{\Pi_{\omega}}e^{i\nu\left(  E\right)  \left\langle
y,\theta\right\rangle }s\left(  \omega,\theta;E\,\right)  \Psi_{+}\left(
\omega,\theta;\omega\right)  d\theta+a_{\operatorname*{reg}}^{\left(
N\right)  }\left(  y,\omega;E\right)  ,$ where $a_{\operatorname*{reg}%
}^{\left(  N\right)  }\in\mathit{S}^{-p\left(  N\right)  }\ $and $p\left(
N\right)  \rightarrow\infty,$ as $N\rightarrow\infty.$

Using the relations (\ref{representation113}) and (\ref{direct1}), we get the
following result for the function $a\left(  y,\omega;E\right)  $:

\begin{proposition}
\label{direct19}Let the magnetic potential $A\left(  x\right)  $ and the
electric potential $V\left(  x\right)  $ satisfy the estimates (\ref{eig31})
and (\ref{eig32}) respectively. Then, the function $a\left(  y,\omega
;E\right)  $ admits the expansion
\begin{equation}
a\left(  y,\omega;E\right)  \simeq\frac{\nu\left(  E\right)  }{\left\vert
E\right\vert }P_{\omega}\left(  E\right)  +\sum_{n=1}^{\infty}h_{n}\left(
y,\omega;E\right)  , \label{direct9}%
\end{equation}
where $h_{n}\left(  y,\omega;E\right)  $ is of order $n$ with respect to
$W\left(  x\right)  $ and $h_{n}\in\mathit{S}^{-\left(  \rho-1\right)  n}.$
Moreover, $(h_{1}+i\frac{\nu\left(  E\right)  }{\left\vert E\right\vert
}R\left(  y,\omega;E\right)  P_{\omega}\left(  E\right)  )\in\mathit{S}%
^{-\rho},$ where $R$ is defined by (\ref{representation219}), and
$(a-\frac{\nu\left(  E\right)  }{\left\vert E\right\vert }P_{\omega}\left(
E\right)  +i\frac{\nu\left(  E\right)  }{\left\vert E\right\vert }R\left(
y,\omega;E\right)  P_{\omega}\left(  E\right)  )\in\mathit{S}^{-\rho
+1-\varepsilon},$ $\varepsilon=\min\{\rho-1,1\}.$
\end{proposition}

Let the electric potential $V\left(  x\right)  \in C^{\infty}\left(
\mathbb{R}^{3}\right)  $ and the magnetic field $B\left(  x\right)  \in
C^{\infty}\left(  \mathbb{R}^{3}\right)  ,$ with $\operatorname{div}B=0,$
admit the asymptotic expansions%
\begin{equation}
V\left(  x\right)  \simeq\sum_{j=1}^{\infty}V_{j}\left(  x\right)  ,
\label{direct2}%
\end{equation}
and
\begin{equation}
B\left(  x\right)  \simeq\sum_{j=1}^{\infty}B_{j}\left(  x\right)  ,
\label{direct3}%
\end{equation}
respectively, where the functions $V_{j}\left(  x\right)  $ are homogeneous of
order $-\rho_{j}^{\left(  e\right)  },$ with $1<\rho_{j}^{\left(  e\right)
}<\rho_{k}^{\left(  e\right)  },$ and the functions $B_{j}\left(  x\right)  $
are homogeneous of order $-r_{j}^{\left(  m\right)  },$ with $2<r_{j}^{\left(
m\right)  }<r_{k}^{\left(  m\right)  }$ for $k>j.$

It follow from relations (\ref{basicnotions19})-(\ref{basicnotions21}) that
the magnetic field $B(x)$ is homogeneous of order $-r^{\left(  m\right)  }$
$<-2$ if and only if the magnetic potential $A(x)$ is homogeneous of order
$-\rho^{\left(  m\right)  }=-r^{\left(  m\right)  }+1<-1$. Therefore, if $B$
satisfies relation (\ref{direct3}), $A(x),$ defined by (\ref{basicnotions19}%
)-(\ref{basicnotions21}), is an asymptotic sum
\begin{equation}
A\left(  x\right)  \simeq\sum_{j=1}^{\infty}A_{j}\left(  x\right)  ,
\label{direct20}%
\end{equation}
where the functions $A_{j}\left(  x\right)  $ are homogeneous of order
$-\rho_{j}^{\left(  m\right)  }$ with $1<\rho_{j}^{\left(  m\right)  }%
<\rho_{k}^{\left(  m\right)  }$ for $k>j.$

Let $V\left(  x\right)  $ and $B\left(  x\right)  $ be as above. We define the
magnetic potential $A\left(  x\right)  $ by (\ref{basicnotions19}%
)-(\ref{basicnotions21}). Thus, we obtain a potential $\mathbf{V}\left(
x\right)  $ of the form (\ref{intro4}), where $A$ and $V$ satisfy the
estimates (\ref{eig31}) and (\ref{eig32}), respectively. Moreover, $V$ admits
the expansion (\ref{direct2}) and $A$ satisfies the relation (\ref{direct20}).

Actually, adding terms which are equal to zero in (\ref{direct2}) and
(\ref{direct3}) we can suppose that $r_{j}^{\left(  m\right)  }=\rho
_{j}^{\left(  e\right)  }+1.$ Then, expansions (\ref{direct2}) and
(\ref{direct20}) are equivalent to the expansion
\begin{equation}
W\left(  x\right)  \simeq\sum_{j=1}^{\infty}W_{j}\left(  x\right)
\label{direct4}%
\end{equation}
where $W_{j}\left(  x\right)  =\left(  V_{j}\left(  x\right)  ,A_{j}\left(
x\right)  \right)  $ is homogeneous of order $-\rho_{j}=-\rho_{j}^{\left(
m\right)  }=-\rho_{j}^{\left(  e\right)  }.$

Plugging (\ref{direct4}) in the expansion (\ref{direct9}) we get the following
result for the function $a\left(  y,\omega;E\right)  $, analogous to Theorem
3.4 of \cite{25} for the Schr\"{o}dinger equation:

\begin{theorem}
\label{direct7}Suppose that an electric potential $V\left(  x\right)  $ and a
magnetic field $B\left(  x\right)  ,$ with $\operatorname{div}B=0,$ are
$C^{\infty}\left(  \mathbb{R}^{3}\right)  $-functions and that they admit the
asymptotic expansions (\ref{direct2}) and (\ref{direct3}), where $V_{j}$ and
$B_{j}$ are homogeneous functions of orders $-\rho_{j}$ and $-r_{j}=-\rho
_{j}-1,$ respectively, where $1<\rho_{1}<\rho_{2}<\cdots.$ Let the magnetic
potential $A(x)$ be defined by the equalities (\ref{basicnotions19}%
)-(\ref{basicnotions21}). Then, the function $a\left(  y,\omega;E\right)  $
admits the asymptotic expansion
\begin{equation}
a\left(  y,\omega;E\right)  \simeq\frac{\nu\left(  E\right)  }{\left\vert
E\right\vert }P_{\omega}\left(  E\right)  +\sum_{n=1}^{\infty}\sum
_{m=0}^{\infty}\sum_{j_{1},...,j_{n}}h_{n,m;j_{1},...,j_{n}}\left(
y,\omega;E\right)  , \label{direct11}%
\end{equation}
where, for each $k=1,2,...,n,$ $j_{k}\ $takes values from $1$ to $\infty.$ The
functions $h_{n,m;j_{1},...,j_{n}}\left(  y,\omega;E\right)  $ are of order
$n$ with respect to the potential $W\left(  x\right)  ,$ they only depend on
$W_{j_{1}},W_{j_{2}}...,W_{j_{k}},$ and they are homogeneous functions of
order $n-m-\sum_{k=1}^{n}\rho_{j_{k}}$ with respect to the variable $y.$ We
note that $h_{1,0;j}\left(  y,\omega;E\right)  =-i\frac{\nu\left(  E\right)
}{\left\vert E\right\vert }R\left(  y,\omega;E;W_{j}\right)  P_{\omega}\left(
E\right)  $ (here we emphasize the dependence of the function $R\left(
y,\omega;E\right)  $ on $W_{j}$) is homogeneous function of order $1-\rho_{j}$
with respect to $y$.
\end{theorem}

Suppose that we know the matrix $-i\frac{\nu\left(  E\right)  }{\left\vert
E\right\vert }R\left(  y,\omega;E;W\right)  P_{\omega}\left(  E\right)  $.
Note that the first diagonal component of the matrix $P_{\omega}\left(
E\right)  $ is equal to $\left(  \frac{1}{2}+\frac{m}{2E}\right)  .$ As
$\left(  \frac{1}{2}+\frac{m}{2E}\right)  \neq0,$ for $\left\vert E\right\vert
>m,$ we recover the function $\int\limits_{-\infty}^{\infty}\mathcal{V}%
_{V,A,\omega}^{\left(  E\right)  }\left(  t\omega+y\right)  dt$ from the first
diagonal component of the matrix $-i\frac{\nu\left(  E\right)  }{\left\vert
E\right\vert }R\left(  y,\omega;E;W\right)  P_{\omega}\left(  E\right)  $.
Therefore, similarly to \cite{79} or \cite{25}, we determine the electric
potential $V$ and the magnetic field $B$ from $R_{e}$ and $R_{m},$
respectively, by using the Radon transform.

Let us define a mapping that sends a potential $W\left(  x\right)  =\left(
V\left(  x\right)  ,A\left(  x\right)  \right)  $ into $a^{\left(  N\right)
}\left(  y,\omega;E\right)  +a_{\operatorname*{reg}}^{\left(  N\right)
}\left(  y,\omega;E\right)  -\frac{\nu\left(  E\right)  }{\left\vert
E\right\vert }P_{\omega}\left(  E\right)  \ $by $T\left(  y,\omega;E;W\right)
=a^{(N)}\left(  y,\omega;E\right)  +a_{\operatorname*{reg}}^{\left(  N\right)
}\left(  y,\omega;E\right)  -\frac{\nu\left(  E\right)  }{\left\vert
E\right\vert }P_{\omega}\left(  E\right)  $ (here we emphasize the dependence
of $a$ on $N$). As $a^{(N)}\left(  y,\omega;E\right)  $ is defined in terms of
the asymptotic expansion (\ref{direct16}), then $T$ is defined only up to a
function from the class $\mathit{S}^{-\infty}.$ We separate the linear part
$-i\frac{\nu\left(  E\right)  }{\left\vert E\right\vert }R\left(
y,\omega;E;W\right)  P_{\omega}\left(  E\right)  $ of $T$, with respect to
$W\left(  x\right)  ,$ and define $Q\left(  y,\omega;E;W\right)  =T\left(
y,\omega;E;W\right)  +i\frac{\nu\left(  E\right)  }{\left\vert E\right\vert
}R\left(  y,\omega;E;W\right)  P_{\omega}\left(  E\right)  .$ We now can
formulate the reconstruction result.

\begin{theorem}
\label{18}Suppose that an electric potential $V\left(  x\right)  $ and a
magnetic field $B\left(  x\right)  ,$ with $\operatorname{div}B=0,$ are
$C^{\infty}\left(  \mathbb{R}^{3}\right)  $-functions and that they admit the
asymptotic expansions (\ref{direct2}) and (\ref{direct3}), where $V_{j}$ and
$B_{j}$ are homogeneous functions of orders $-\rho_{j}$ and $-r_{j}=-\rho
_{j}-1,$ respectively, where $1<\rho_{1}<\rho_{2}<\cdots.$ Let the magnetic
potential $A(x)$ be defined by the equalities (\ref{basicnotions19}%
)-(\ref{basicnotions21}). Then, the kernel $s\left(  \omega,\theta;E\right)  $
of the scattering matrix $S\left(  E\right)  $ for fixed $E\in(-\infty
,m)\cup(m,+\infty)$ in a neighborhood of the diagonal $\omega=\theta$ uniquely
determines each one of $V_{j}\left(  x\right)  $ and $B_{j}\left(  x\right)
$. Moreover, $V_{1}\left(  x\right)  $ and $B_{1}\left(  x\right)  $ can be
reconstructed from the formula $-i\frac{\nu\left(  E\right)  }{\left\vert
E\right\vert }R\left(  y,\omega;E;W_{1}\right)  P_{\omega}\left(  E\right)
=h_{1}\left(  y,\omega;E\right)  ,$ where $h_{1}$ is the term of highest
homogeneous order with respect to $y$ in the expansion (\ref{direct11}). The
functions $V_{j}\left(  x\right)  $ and $B_{j}\left(  x\right)  ,$ for
$j\geq2,$ can be recursively reconstructed from the formula $-i\frac
{\nu\left(  E\right)  }{\left\vert E\right\vert }R\left(  y,\omega
;E;W_{j}\right)  P_{\omega}\left(  E\right)  =\left(  a\left(  y,\omega
;E\right)  -\frac{\nu\left(  E\right)  }{\left\vert E\right\vert }P_{\omega
}\left(  E\right)  -T\left(  y,\omega;E;\sum_{i=1}^{j-1}W_{i}\right)  \right)
^{%
{{}^\circ}%
}.$ Here we denote by $f^{%
{{}^\circ}%
}$, the highest order homogeneous term $f_{k}$ in (\ref{basicnotions16}) that
is not identically zero.
\end{theorem}

\begin{proof}
The proof is analogous to Theorem 4.2 of \cite{25}, for the case of the
Schr\"{o}dinger equation.
\end{proof}

\begin{corollary}
Let $W^{\left(  j\right)  }$ satisfy the assumptions of Theorem \ref{18} for
$j=1,2.$ If $s_{1}\left(  \omega,\theta;E\right)  -s_{2}\left(  \omega
,\theta;E\right)  \in C^{\infty}\left(  \mathbb{S}^{2}\times\mathbb{S}%
^{2}\right)  $ for some $\left\vert E\right\vert >m,$ then $V_{1}-V_{2}$ and
$B_{1}-B_{2}$ belong to the Schwartz class $\mathit{S.}$
\end{corollary}

\section{Completeness of averaged scattering solutions.}

If the potential $\mathbf{V}$ satisfies Condition \ref{basicnotions26} and
decreases as $\left\vert x\right\vert ^{-\rho},$ when $\left\vert x\right\vert
\rightarrow\infty,$ with $\rho>2,$ then, for all $E\in\{(-\infty
,-m)\cup(m,\infty)\}\backslash\sigma_{p}\left(  H\right)  ,$ the scattering
solution $u_{\pm}\left(  x,\omega;E\right)  $ are defined by
(\ref{representation29}) and they satisfy the asymptotic expansion (\ref{re16}).

\bigskip If $\rho\leq2,$ then relation (\ref{representation29}) makes no
sense. However, similarly to the Schr\"{o}dinger operator case (\cite{37}%
,\cite{38}), we can generalize definition (\ref{representation29}).

We define the \textquotedblleft unperturbed averaged scattering
solutions\textquotedblright\ by $\psi_{0,f}\left(  x;E\right)  :=\int%
_{\mathbb{S}^{2}}e^{i\nu\left(  E\right)  \left\langle \omega,x\right\rangle
}P_{\omega}\left(  E\right)  f\left(  \omega\right)  d\omega,$ for any $f\in
L^{2}\left(  \mathbb{S}^{2};\mathbb{C}^{4}\right)  .$ Note that up to a
coefficient $\psi_{0,f}$ is equal to $\Gamma_{0}^{\ast}\left(  E\right)  f,$
defined by (\ref{basicnotions39}). Then, it follows that $\psi_{0,f}%
\in\mathcal{H}^{1,-s}\left(  \mathbb{R}^{3};\mathbb{C}^{4}\right)  ,$ $s>1/2,$
and $H_{0}\psi_{0,f}=E\psi_{0,f}.$

Let the potential $\mathbf{V}$ satisfy Condition \ref{basicnotions26}. The
\textquotedblleft perturbed averaged\ scattering solutions\textquotedblright%
\ are defined by
\begin{equation}
\psi_{+,f}\left(  x;E\right)  :=[I-R_{+}\left(  E\right)  \mathbf{V}%
]\psi_{0,f},\text{ }E\in\{(-\infty,-m)\cup(m,\infty)\}\backslash\sigma
_{p}\left(  H\right)  ,\text{ }f\in L^{2}\left(  \mathbb{S}^{2};\mathbb{C}%
^{4}\right)  . \label{completeness6}%
\end{equation}
Note that $\psi_{+,f}\in\mathcal{H}^{1,-s}\left(  \mathbb{R}^{3}%
;\mathbb{C}^{4}\right)  ,$ $1/2<s\leq s_{0},$ and $H\psi_{+,f}=E\psi_{+,f}.$
If $\rho>2,$ the formula (\ref{representation29}) holds and we have
$\psi_{+,f}\left(  x;E\right)  =\int_{\mathbb{S}^{2}}\psi^{+}\left(
x,\omega;E\right)  f\left(  \omega\right)  d\omega.$ The last equality
justifies the name averaged scattering solutions.

Using the stationary representation (\ref{basicnotions14}) we can write the
scattering matrix $S\left(  E\right)  $ in terms of the averaged solutions.
For $f,g\in L^{2}\left(  \mathbb{S}^{2};\mathbb{C}^{4}\right)  $ we have
\begin{equation}
\left.  \left(  S\left(  E\right)  f,g\right)  _{\mathcal{H}\left(  E\right)
}=\left(  f,g\right)  _{\mathcal{H}\left(  E\right)  }-i\left(  2\pi\right)
^{-2}\upsilon\left(  E\right)  ^{2}\left(  \mathbf{V}\psi_{+,f},\psi
_{0,g}\right)  _{L^{2}}.\right.  \label{completeness5}%
\end{equation}

For us, the important property of the averaged scattering solutions is that
the set (\ref{completeness6}) is dense on the set of all solutions to the
Dirac equation (\ref{eig1}) in $L^{2}\left(  \Omega;\mathbb{C}^{4}\right)  ,$
where $\Omega$ is a connected open bounded set with smooth boundary
$\partial\Omega$. We present this assertion as:

\begin{theorem}
\label{completeness12}Let $\mathbf{V}\ $satisfies Condition
\ref{basicnotions26} and the following estimate
\begin{equation}
\left\vert \partial_{x}^{\alpha}\mathbf{V}\left(  x\right)  \right\vert \leq
C_{\alpha}\left(  1+\left\vert x\right\vert \right)  ^{-\rho},\text{ }%
\rho>1,\text{ }\left\vert \alpha\right\vert \leq1,\text{ for }x\in
\mathbb{R}^{3}\backslash\Omega, \label{completeness2}%
\end{equation}
where $\Omega$ is a connected open bounded set with smooth boundary
$\partial\Omega$. Then, the set of averaged scattering solutions $\left\{
\psi_{+,f},\right.  $ $\left.  f\in\mathcal{H}\left(  E\right)  \right\}  $ is
strongly dense on the set of all solutions to (\ref{eig1}) in $L^{2}\left(
\Omega;\mathbb{C}^{4}\right)  $ for all fixed $E\in\{(-\infty,-m)\cup
(m,\infty)\}\backslash\sigma_{p}\left(  H\right)  .$
\end{theorem}

Let us note that the result of Theorem \ref{completeness12} holds for all
$\left\vert E\right\vert >m,$ if some result on absence of eigenvalues on
$(-\infty,-m)\cup(m,\infty)$ is applied. For example, if $\mathbf{V}\in
L_{\operatorname*{loc}}^{5}\left(  \mathbb{R}^{3}\right)  $ satisfies relation
(\ref{completeness2}), then the result of Theorem \ref{completeness12} remains
true for all $\left\vert E\right\vert >m$ (see Section 2)$.$

\begin{proof}
We proceed as in the proof of Theorem 3.1 of \cite{26} (see also \cite{58})
for the Schr\"{o}dinger case. Let us take a solution $\chi\in L^{2}\left(
\Omega;\mathbb{C}^{4}\right)  $ that is orthogonal to $\psi_{+,f}$ for all
$f\in\mathcal{H}\left(  E\right)  $. Then, as $\psi_{+,f}=\psi_{+,P_{\omega
}\left(  E\right)  f},$ for all $f\in L^{2}\left(  \mathbb{S}^{2}%
;\mathbb{C}^{4}\right)  ,$ it follows that
\begin{equation}
\left(  \chi,\psi_{+,f}\right)  _{L^{2}\left(  \Omega;\mathbb{C}^{4}\right)
}=0,\text{ for }f\in L^{2}\left(  \mathbb{S}^{2};\mathbb{C}^{4}\right)  .
\label{completeness1}%
\end{equation}
We extend $\chi$ by zero to $\mathbb{R}^{3}\backslash\Omega$ and then,
$\chi\in L_{s}^{2}$ for all $s.$ Let us define $\psi:=R_{+}\left(  E\right)
\chi.$ Note that $\psi$ satisfy the equation
\begin{equation}
\left(  H-E\right)  \psi=\chi. \label{completeness10}%
\end{equation}
Suppose that $\psi\in L_{-\sigma}^{2}$ for some $\sigma<1/2.$ Then, as
$\chi=0$ on $\mathbb{R}^{3}\backslash\Omega,$ multiplying
(\ref{completeness10}) from the left side by $H_{0}+E$ we get that $\psi$
satisfies the following equation%
\begin{equation}
-\Delta\psi-\left(  E^{2}-m^{2}\right)  \psi+\left(  H_{0}+E\right)  \left(
\mathbf{V}\psi\right)  =0, \label{completeness11}%
\end{equation}
for $x\in$ $\mathbb{R}^{3}\backslash\overline{\Omega}$. As the principal part
$\Delta\psi$ of (\ref{completeness11}) is diagonal, then, under assumption
(\ref{completeness2}), the proofs of \cite{45} in the case of a scalar
Schr\"{o}dinger equation apply for a system of Schr\"{o}dinger equations
(\ref{completeness11}) and thus, we get that $\psi$ vanishes identically on
$\mathbb{R}^{3}\backslash\overline{\Omega}.$ In particular, $\psi=0$ on
$\partial\Omega$, in the trace sense. Then, as the boundary $\partial\Omega$
is smooth, we can approximate $\psi$ in the norm of $\mathcal{H}^{1}$ by
functions $\psi_{n}\in C_{0}^{\infty}\left(  \Omega\right)  ,$ $n\in
\mathbb{N}.$ Noting that $\left(  \left(  H-E\right)  \psi_{n},\chi\right)
_{L^{2}\left(  \Omega;\mathbb{C}^{4}\right)  }=\left(  \psi_{n},\left(
H-E\right)  \chi\right)  _{L^{2}\left(  \Omega;\mathbb{C}^{4}\right)  }$ for
all $n,$ we prove that $\left\Vert \chi\right\Vert _{L^{2}\left(
\Omega;\mathbb{C}^{4}\right)  }^{2}=\left(  \left(  H-E\right)  \psi
,\chi\right)  _{L^{2}\left(  \Omega;\mathbb{C}^{4}\right)  }=\left(
\psi,\left(  H-E\right)  \chi\right)  _{L^{2}\left(  \Omega;\mathbb{C}%
^{4}\right)  }=0,$ and hence, $\chi=0.$

Thus, to complete the proof we need to show that $\psi\in L_{-\sigma}^{2},$
for some $\sigma<1/2.$ Note that relation (\ref{completeness1}) implies that
$\Gamma_{-}\left(  E\right)  \chi=0$. Moreover, as the operator $\mathcal{F}%
_{-},$ defined by relation (\ref{basicnotions40}), gives a spectral
representation of $H$ and $\Gamma_{-}\left(  \lambda\right)  $ is locally
H\"{o}lder continuous, it follows from the Privalov's theorem that $\psi
=\psi_{1}+\psi_{2},$ where
\begin{equation}
\psi_{1}=\int_{I}\Gamma_{-}^{\ast}\left(  \lambda\right)  \left(  \frac
{1}{\lambda-E}\left(  \Gamma_{-}\left(  \lambda\right)  -\Gamma_{-}\left(
E\right)  \right)  \chi\right)  d\lambda, \label{completeness14}%
\end{equation}
and $\psi_{2}=R\left(  E\right)  E_{H}\left(  \mathbb{R}\backslash I\right)
\chi,$ for some neighborhood $I$ of the point $E.$ Here $E_{H}$ is the
resolution of the identity for $H.$ Note that $\psi_{2}$ is already from
$L^{2}.$

Let us define the operator $J$ by
\begin{equation}
Jg:=\int_{I}\Gamma_{-}^{\ast}\left(  \lambda\right)  g\left(  \lambda\right)
d\lambda. \label{completeness15}%
\end{equation}
Since $\mathcal{F}_{-}$ is unitary from $\mathcal{H}_{ac}$ onto $\mathcal{\hat
{H}}$ and $Jg=E_{H}\left(  I\right)  \mathcal{F}_{-}^{\ast}g,$ the operator
$J$ is bounded from $L^{2}\left(  I;L^{2}\left(  \mathbb{S}^{2};\mathbb{C}%
^{4}\right)  \right)  $ into $L^{2}.$ Moreover, as the operator $\Gamma
_{-}^{\ast}\left(  \lambda\right)  $ is bounded from $L^{2}\left(
\mathbb{S}^{2};\mathbb{C}^{4}\right)  $ into $L_{-s}^{2},$ for $1/2<s\leq
s_{0}$, then $J$ is bounded from $L^{1}\left(  I;L^{2}\left(  \mathbb{S}%
^{2};\mathbb{C}^{4}\right)  \right)  $ into $L_{-s}^{2}.$ Thus, by
interpolation (see, for example, \cite{2}), $J$ is bounded from $L^{p}\left(
I;L^{2}\left(  \mathbb{S}^{2};\mathbb{C}^{4}\right)  \right)  $ into
$L_{-\sigma}^{2},$ with $\sigma=\left(  2/p-1\right)  s$ and $1\leq p\leq2.$

Let us take $s_{1}=\frac{1}{2}+\vartheta,$ $\vartheta<\min\{s_{0}-1/2,1/2\}.$
Note that $\Gamma_{-}\left(  \lambda\right)  $ is locally H\"{o}lder
continuous from $L_{s_{1}}^{2}$ to $L^{2}\left(  \mathbb{S}^{2};\mathbb{C}%
^{4}\right)  $ with exponent $\vartheta.$ Then, as $\chi\in L_{s_{1}}^{2}\ $we
get $\frac{1}{\lambda-E}\left(  \Gamma_{-}\left(  \lambda\right)  -\Gamma
_{-}\left(  E\right)  \right)  \chi\in L^{p}\left(  I;L^{2}\left(
\mathbb{S}^{2};\mathbb{C}^{4}\right)  \right)  ,$ where $p<\frac
{1}{1-\vartheta}.$ Taking $p=\frac{1}{1-\vartheta/2}$ and $s=1/2+\vartheta/2$
we get that $\sigma\ <\frac{1}{2}.$ Using relations (\ref{completeness14}) and
(\ref{completeness15}) we get $\psi_{1}=J\left(  \frac{1}{\lambda-E}\left(
\Gamma_{-}\left(  \lambda\right)  -\Gamma_{-}\left(  E\right)  \right)
\chi\right)  .$ Therefore, we conclude that $\psi\in L_{-\sigma}^{2}$. Theorem
is proved.
\end{proof}

\section{Uniqueness of the electric potential and the magnetic field at fixed
energy.}

In this Section we aim to show that the scattering matrix $S\left(  E\right)
,$ given for some energy $E,$ determines uniquely the electric potential $%
\begin{pmatrix}
V_{+} & 0\\
0 & V_{-}%
\end{pmatrix}
$ and the magnetic field $B\left(  x\right)  =\operatorname*{rot}A\left(
x\right)  .$ The averaged scattering solutions (\ref{completeness6}) result
useful here. With the help of these solutions we can prove that $S\left(
E\right)  $ determines uniquely the Dirichlet to Dirichlet map (Definition
(\ref{completeness9})). Then, supposing that the electric potential $V_{\pm
}\left(  x\right)  $ and the magnetic field $B\left(  x\right)  $ are known
outside some connected open bounded set $\Omega_{E}$ with smooth boundary
$\partial\Omega_{E},$ we show that $V_{\pm}\left(  x\right)  $ and $B\left(
x\right)  $ are uniquely determined everywhere in $\mathbb{R}^{3}.$

Let us consider the free Dirac operator $L_{0,\Omega_{E}}=\left(
\begin{array}
[c]{cc}%
0 & -i\sigma\cdot\nabla\\
-i\sigma\cdot\nabla & 0
\end{array}
\right)  $ on $L^{2}\left(  \Omega_{E};\mathbb{C}^{2}\right)  \times
L^{2}\left(  \Omega_{E};\mathbb{C}^{2}\right)  ,$ where $\Omega_{E}$ is a
connected open bounded set with smooth boundary $\partial\Omega_{E}$ and
$\sigma=\left(  \sigma_{1},\sigma_{2},\sigma_{3}\right)  $ are the Pauli
matrices (\ref{basicnotions22}). $L_{0}$ is a self-adjoint operator on (see
\cite{19}) $D_{L_{0,\Omega_{E}}}:=\{u=\left(  u_{+},u_{-}\right)  \mid
u_{+}\in\mathcal{H}_{0}^{1}\left(  \Omega_{E};\mathbb{C}^{2}\right)  ,u_{-}%
\in\mathcal{H}\left(  \Omega_{E};\mathbb{C}^{2}\right)  \},$ where
$\mathcal{H}_{0}^{1}\left(  \Omega_{E};\mathbb{C}^{2}\right)  $ is the closure
of $C_{0}^{\infty}\left(  \Omega_{E}\right)  $ in the space $\mathcal{H}%
^{1}\left(  \Omega_{E};\mathbb{C}^{2}\right)  $ and $\mathcal{H}\left(
\Omega_{E};\mathbb{C}^{2}\right)  $ is the closure of $\mathcal{H}^{1}\left(
\Omega_{E};\mathbb{C}^{2}\right)  $ in the norm $\left\Vert \cdot\right\Vert
_{\mathcal{H}\left(  \Omega_{E};\mathbb{C}^{2}\right)  }:=\left\Vert \left(
\sigma\cdot\nabla\right)  \cdot\right\Vert _{L^{2}\left(  \Omega
_{E};\mathbb{C}^{2}\right)  }+\left\Vert \cdot\right\Vert _{L^{2}\left(
\Omega_{E};\mathbb{C}^{2}\right)  }.$

Let $\mathbf{V}$ be an Hermitian $4\times4$-matrix-valued function whose
entries belong to $L^{\infty}\left(  \Omega_{E}\right)  .$ Then,
$L_{\mathbf{V,}\Omega_{E}}:=L_{0,\Omega_{E}}+\mathbf{V}$ is self-adjoint on
$D_{L_{0,\Omega_{E}}}$. Consider the following Dirichlet problem%
\begin{equation}
\left\{
\begin{array}
[c]{c}%
\left(  L_{\mathbf{V,}\Omega_{E}}-E\right)  \left(  u_{+},u_{-}\right)
=0,\text{ in }\Omega_{E},\\
\left.  u_{+}\right\vert _{\partial\Omega_{E}}=g\in h\left(  \partial
\Omega_{E}\right)  ,\text{ on }\partial\Omega_{E}.
\end{array}
\right.  \label{completeness4}%
\end{equation}
Here $h\left(  \partial\Omega_{E}\right)  $ is defined as the trace on
$\partial\Omega_{E}$ of $\mathcal{H}\left(  \Omega_{E};\mathbb{C}^{2}\right)
.$ Suppose that $E$ belongs to the resolvent set of $L_{\mathbf{V,\Omega_{E}}%
}$. Then, from Proposition 4.11 of \cite{19} we get that for every $f\in
h\left(  \partial\Omega_{E}\right)  $ there exist an unique solution $\left(
u_{+},u_{-}\right)  \in\mathcal{H}\left(  \Omega_{E};\mathbb{C}^{2}\right)
\times\mathcal{H}\left(  \Omega_{E};\mathbb{C}^{2}\right)  $ of the equation
(\ref{completeness4}).

For any $g\in h\left(  \partial\Omega_{E}\right)  ,$ we define the Dirichlet
to Dirichlet (up-spinor to down-spinor) map by%
\begin{equation}
\Lambda_{\mathbf{V}}g=\left.  u_{-}\right\vert _{\partial\Omega_{E}}\in
h\left(  \partial\Omega_{E}\right)  , \label{completeness9}%
\end{equation}
where $\left(  u_{+},u_{-}\right)  $ is the unique solution of
(\ref{completeness4}).

The uniqueness result of the potential at fixed energy is the following

\begin{theorem}
\label{completeness8}Let the potentials $\mathbf{V}_{j}\left(  x\right)  ,$
$j=1,2,$ be given by%
\begin{equation}
\mathbf{V}_{j}\left(  x\right)  =%
\begin{pmatrix}
V_{+}^{\left(  j\right)  } & \sigma\cdot A^{\left(  j\right)  }\\
\sigma\cdot A^{\left(  j\right)  } & V_{-}^{\left(  j\right)  }%
\end{pmatrix}
\label{completeness17}%
\end{equation}
with real functions $V_{\pm}^{\left(  j\right)  },A_{k}^{\left(  j\right)
}\in C^{\infty}\left(  \mathbb{R}^{3}\right)  ,$ $k=1,2,3,$ such that $V_{\pm
}^{\left(  j\right)  }$ and $A_{k}^{\left(  j\right)  },$ $j=1,2,$ satisfy
(\ref{completeness2}). Let $S_{j}\left(  E\right)  $ be the scattering
matrices corresponding to $\mathbf{V}_{j},$ $j=1,2.$ Suppose that for some
$E\in\left(  -\infty,-m\right)  \cup\left(  m,+\infty\right)  ,$ $S_{1}\left(
E\right)  =S_{2}\left(  E\right)  ,$ and there is a connected open bounded set
$\Omega_{E}$ with smooth boundary $\partial\Omega_{E}$ $,$ such that $E$
belongs to the resolvent set of $L_{\mathbf{V}_{j},\Omega_{E}}$ for both
$j=1,2.$ Let $\mathbf{V}_{1}\left(  x\right)  $ be equal to $\mathbf{V}%
_{2}\left(  x\right)  $ for $x\in\mathbb{R}^{3}\diagdown\Omega_{E}.$ Then we
have that $V_{\pm}^{\left(  1\right)  }\left(  x\right)  =V_{\pm}^{\left(
2\right)  }\left(  x\right)  $ and $\operatorname*{rot}A^{\left(  1\right)
}\left(  x\right)  =\operatorname*{rot}A^{\left(  2\right)  }\left(  x\right)
$ for all $x\in\mathbb{R}^{3}.$
\end{theorem}

\begin{proof}
We follow the proof of \cite{26} for the Schr\"{o}dinger case. Let us first
show that $\psi_{+,f}^{\left(  1\right)  }\left(  x;E\right)  =\psi
_{+,f}^{\left(  2\right)  }\left(  x;E\right)  $\ for $x\in\mathbb{R}%
^{3}\diagdown\Omega_{E},$ $f\in\mathcal{H}\left(  E\right)  ,$ where
$\psi_{+,f}^{\left(  j\right)  }$ are the averaged scattering solutions
corresponding to $\mathbf{V}_{j}$ for $j=1,2.$ Denote$\ \psi:=\psi
_{+,f}^{\left(  2\right)  }-\psi_{+,f}^{\left(  1\right)  }$ and
$\eta:=\mathbf{V}_{1}\psi_{+,f}^{\left(  1\right)  }-\mathbf{V}_{2}\psi
_{+,f}^{\left(  2\right)  }.$ Note that $\psi\in\mathcal{H}^{1,-s},$ $\eta\in
L_{s}^{2},$ $\frac{1}{2}<s\leq s_{0},$ and moreover,
\begin{equation}
\left(  H_{0}-E\right)  \psi=\eta. \label{completeness3}%
\end{equation}
Since $S_{1}\left(  E\right)  =S_{2}\left(  E\right)  ,$ it follows from
(\ref{completeness5}) that $\Gamma_{0}\left(  E\right)  \eta=0.$ This implies
that $\Gamma_{0}\left(  E\right)  \mathcal{F}^{\ast}\hat{\eta}=0.$ Then, as
$\hat{\eta}\in\mathcal{H}^{s}$ and $\Gamma_{0}\left(  E\right)  \mathcal{F}%
^{\ast}$ is bounded from $\mathcal{H}^{s}$ into $L^{2}\left(  \mathbb{S}%
^{2};\mathbb{C}^{4}\right)  ,$ we conclude that $\hat{\eta}\left(  \xi\right)
=0$ in the trace sense on the sphere of radius $\left\vert \xi\right\vert
=\nu\left(  E\right)  .$

Note that $\hat{\psi}=P^{+}\left(  \xi\right)  \hat{\psi}+P^{-}\left(
\xi\right)  \hat{\psi}.$ From (\ref{completeness3}) we get
\begin{equation}
\left(  \pm\sqrt{\xi^{2}+m^{2}}-E\right)  P^{\pm}\left(  \xi\right)  \hat
{\psi}=P^{\pm}\left(  \xi\right)  \hat{\eta}. \label{representation220}%
\end{equation}
For $\pm E>m,$ the function $\frac{1}{\mp\sqrt{\xi^{2}+m^{2}}-E}$ is bounded.
Then, as $P^{\pm}\left(  \xi\right)  \hat{\eta}\in\mathcal{H}^{s},$ we get,
for $\pm E>m,$ $P^{\mp}\left(  \xi\right)  \hat{\psi}\in\mathcal{H}^{s}.$
Moreover, it follows from (\ref{representation220}) that, $P^{\pm}\left(
\xi\right)  \hat{\psi}=\frac{\pm\sqrt{\xi^{2}+m^{2}}+E}{\left\vert
\xi\right\vert ^{2}-\nu\left(  E\right)  ^{2}}P^{\pm}\left(  \xi\right)
\hat{\eta},$ for $\pm E>m.$ As $P^{\pm}\left(  \xi\right)  \hat{\eta}%
\in\mathcal{H}^{s}$ and $P^{\pm}\left(  \xi\right)  \hat{\eta}=0$ in the trace
sense on the sphere of radius $\left\vert \xi\right\vert =\nu\left(  E\right)
,$ then Theorem 3.2 of \cite{4} implies that for $\pm E>m,$ $P^{\pm}\left(
\xi\right)  \hat{\psi}\in\mathcal{H}^{s-1},$ $s>\frac{1}{2}.$ Therefore, we
obtain that $\hat{\psi}\in\mathcal{H}^{s-1}$ and $\psi\in L_{s-1}^{2}.$ Note
that $s-1>-1/2.$ As $\mathbf{V}_{1}\left(  x\right)  =\mathbf{V}_{2}\left(
x\right)  ,$ for $x\in\mathbb{R}^{3}\diagdown\Omega_{E}$ we see that $\left(
H_{0}+\mathbf{V}_{1}\right)  \psi=E\psi,$ for $x\in\mathbb{R}^{3}%
\diagdown\Omega_{E}.$ Thus, similarly to the proof of Theorem
\ref{completeness12}, as $\psi$ satisfies (\ref{completeness11}), for
$x\in\mathbb{R}^{3}\diagdown\Omega_{E},$\ we conclude that $\psi$ is
identically zero for $x\in\mathbb{R}^{3}\diagdown\overline{\Omega_{E}}.$ In
particular we obtain $\psi_{+,f}^{\left(  1\right)  }\left(  x;E\right)
=\psi_{+,f}^{\left(  2\right)  }\left(  x;E\right)  $ in the trace sense on
$\partial\Omega_{E}$

Let $\tau$ be the trace map $\tau:\mathcal{H}\left(  \Omega_{E};\mathbb{C}%
^{2}\right)  \rightarrow h\left(  \partial\Omega_{E}\right)  .$ Note that the
scattering solution $\psi_{+,f}^{\left(  j\right)  }\left(  x;E\right)
\in\mathcal{H}^{1}\left(  \Omega_{E};\mathbb{C}^{4}\right)  ,$ $j=1,2,$ solves
(\ref{completeness4}) with $g=\tau\left(  \left(  \psi_{+,f}^{\left(
j\right)  }\right)  _{+}\right)  $ (here $(u)_{\pm}$ denotes the first or the
last two components of a $\mathbb{C}^{4}$ vector respectively). As $\psi
_{+,f}^{\left(  1\right)  }\left(  x;E\right)  =\psi_{+,f}^{\left(  2\right)
}\left(  x;E\right)  $ in the trace sense on $\partial\Omega_{E},$ we get
\begin{equation}
\left.  \Lambda_{\mathbf{V}_{2}}\left(  \tau\left(  \left(  \psi
_{+,f}^{\left(  1\right)  }\right)  _{+}\right)  \right)  =\Lambda
_{\mathbf{V}_{2}}\left(  \tau\left(  \left(  \psi_{+,f}^{\left(  2\right)
}\right)  _{+}\right)  \right)  =\tau\left(  \left(  \psi_{+,f}^{\left(
2\right)  }\right)  _{-}\right)  =\tau\left(  \left(  \psi_{+,f}^{\left(
1\right)  }\right)  _{-}\right)  =\Lambda_{\mathbf{V}_{1}}\left(  \tau\left(
\left(  \psi_{+,f}^{\left(  1\right)  }\right)  _{+}\right)  \right)
.\right.  \label{completeness13}%
\end{equation}

For any solution $u^{\left(  j\right)  }=\left(  u_{+}^{\left(  j\right)
},u_{-}^{\left(  j\right)  }\right)  \in\mathcal{H}\left(  \Omega
_{E};\mathbb{C}^{2}\right)  \times\mathcal{H}\left(  \Omega_{E};\mathbb{C}%
^{2}\right)  $ of $\left(  L_{\mathbf{V}_{j}}-E\right)  u^{\left(  j\right)
}=0,$ $j=1,2,$ we have (see Lemma 2.1 of \cite{19})%
\begin{equation}
_{h\left(  \partial\Omega_{E}\right)  }\left\langle \overline{u_{+}^{\left(
2\right)  }},\left(  i\sigma\cdot N\right)  \left(  \Lambda_{\mathbf{V}_{1}%
}-\Lambda_{\mathbf{V}_{2}}\right)  u_{+}^{\left(  1\right)  }\right\rangle
\text{ }_{h\left(  \partial\Omega_{E}\right)  ^{\ast}}=\int_{\Omega_{E}%
}\left(  u^{\left(  2\right)  },\left(  \mathbf{V}_{1}-\mathbf{V}_{2}\right)
u^{\left(  1\right)  }\right)  dx, \label{completeness16}%
\end{equation}
where $h\left(  \partial\Omega_{E}\right)  ^{\ast}$ is the dual space to
$h\left(  \partial\Omega_{E}\right)  $ with respect to the duality $_{h\left(
\partial\Omega_{E}\right)  }\left\langle u,\overline{v}\right\rangle
_{h\left(  \partial\Omega_{E}\right)  ^{\ast}}=\int_{\partial\Omega_{E}}%
u\cdot\overline{v}dS$ and $N$ is the unit outer normal vector to
$\partial\Omega_{E}.$ Taking $u^{\left(  1\right)  }=\psi_{+,f}^{\left(
1\right)  }$ in (\ref{completeness16}) and using (\ref{completeness13}) we get
$\int_{\Omega_{E}}\left(  u^{\left(  2\right)  },\left(  \mathbf{V}%
_{1}-\mathbf{V}_{2}\right)  \psi_{+,f}^{\left(  1\right)  }\right)  dx=0$ for
all averaged scattering solutions $\psi_{+,f}^{\left(  1\right)  }.$ Since
these solutions are dense on the set of all solutions to (\ref{completeness4})
(here we used Theorem \ref{completeness12}, observing that $(-\infty
,-m)\cup(m,\infty)\}\cap\sigma_{p}\left(  H\right)  =\varnothing$ if
$\mathbf{V}$ satisfies relation (\ref{basicnotions46})), $\int_{\Omega_{E}%
}(u^{\left(  2\right)  },\left(  \mathbf{V}_{1}-\mathbf{V}_{2}\right)
u^{\left(  1\right)  })dx=0,$ for any solution $u^{\left(  j\right)  }=\left(
u_{+}^{\left(  j\right)  },u_{-}^{\left(  j\right)  }\right)  \in
\mathcal{H}\left(  \Omega_{E};\mathbb{C}^{2}\right)  \times\mathcal{H}\left(
\Omega_{E};\mathbb{C}^{2}\right)  $ of $\left(  L_{\mathbf{V}_{j}}-E\right)
u^{\left(  j\right)  }=0,$ $j=1,2.$ Thus, it follows from relation
(\ref{completeness16}) that $_{h\left(  \Gamma\right)  }\left\langle
\overline{u_{+}^{\left(  2\right)  }},\left(  i\sigma\cdot N\right)  \left(
\Lambda_{\mathbf{V}_{1}}-\Lambda_{\mathbf{V}_{2}}\right)  u_{+}^{\left(
1\right)  }\right\rangle _{h\left(  \Gamma\right)  ^{\ast}}=0$ for any
solution $u^{\left(  j\right)  }=\left(  u_{+}^{\left(  j\right)  }%
,u_{-}^{\left(  j\right)  }\right)  $ of $\left(  L_{\mathbf{V}_{j}}-E\right)
u^{\left(  j\right)  }=0,$ $j=1,2.$ As for any $f_{j}\in h\left(
\Gamma\right)  ,$ $j=1,2,$ there exist an unique solution $u^{\left(
j\right)  }$ of $\left(  L_{\mathbf{V}_{j}}-E\right)  u^{\left(  j\right)
}=0,$ such that $\left.  u_{+}^{\left(  j\right)  }\right\vert _{h\left(
\Gamma\right)  }=f_{j},$ we conclude that $\left(  i\sigma\cdot N\right)
\left(  \Lambda_{\mathbf{V}_{1}}-\Lambda_{\mathbf{V}_{2}}\right)  f_{1}\in
h\left(  \Gamma\right)  ^{\ast}$ is the functional $0,\ $for all $f_{1}\in
h\left(  \Gamma\right)  ,$ and hence, $\Lambda_{\mathbf{V}_{1}}=\Lambda
_{\mathbf{V}_{2}}.$ Since $\mathbf{V}_{1}\left(  x\right)  =\mathbf{V}%
_{2}\left(  x\right)  $ for $x\in\mathbb{R}^{3}\diagdown\Omega_{E}$ and
$\mathbf{V}_{j}\in C^{\infty}\left(  \mathbb{R}^{3}\right)  ,$ $j=1,2,$ we get
that $A_{k}^{\left(  1\right)  }=A_{k}^{\left(  2\right)  },$ for $k=1,2,3,$
to infinite order on $\partial\Omega_{E}.$ Thus, it follows from Theorem 1 of
\cite{19} that $V_{1}\left(  x\right)  =V_{2}\left(  x\right)  ,$
$\operatorname*{rot}A_{1}\left(  x\right)  =\operatorname*{rot}A_{2}\left(
x\right)  $ on $\Omega_{E}.$ Theorem is proved.
\end{proof}

\begin{corollary}
\label{22}Let the potentials $\mathbf{V}_{j}\left(  x\right)  ,$ $j=1,2,$ be
given by (\ref{completeness17}), with real functions $V_{\pm}^{\left(
j\right)  },A_{k}^{\left(  j\right)  }\in C^{\infty}\left(  \mathbb{R}%
^{3}\right)  ,$ $k=1,2,3,$ such that $V_{\pm}^{\left(  j\right)  }$ satisfy
(\ref{completeness2}) and the magnetic fields $B_{j},$ with
$\operatorname{div}B_{j}=0,$ satisfy the estimate (\ref{basicnotions18}) with
$d=1$, $j=1,2$. Suppose that for some $E\in\left(  -\infty,-m\right)
\cup\left(  m,+\infty\right)  ,$ $S_{1}\left(  E\right)  =S_{2}\left(
E\right)  ,$ and there is a connected open bounded set $\Omega_{E}$ with
smooth boundary $\partial\Omega_{E}$ $,$ such that $E$ belongs to the
resolvent set of $L_{\mathbf{V}_{j},\Omega_{E}}$ for both $j=1,2.$ Let
$V_{\pm}^{\left(  1\right)  }=V_{\pm}^{\left(  2\right)  }$ and $B_{1}=B_{2}$
for $x\in\mathbb{R}^{3}\diagdown\Omega_{E}.$ Then we have that $V_{\pm
}^{\left(  1\right)  }=V_{\pm}^{\left(  2\right)  }$ and $B^{\left(  1\right)
}=B^{\left(  2\right)  }$ for all $x\in\mathbb{R}^{3}.$
\end{corollary}

\begin{proof}
Let us define the magnetic potentials $\tilde{A}_{j},$ $j=1,2,$ from the
magnetic fields $B_{j},$ $j=1,2,$ by the equalities (\ref{basicnotions19}%
)-(\ref{basicnotions21}). Let $\mathbf{\tilde{V}}_{j}\left(  x\right)  ,$
$j=1,2,$ be the correspondent potentials. Observe that $\tilde{S}_{j}=S_{j}$,
where $\tilde{S}_{j}$ is associated to potential $\mathbf{\tilde{V}}_{j},$
$j=1,2.$ Moreover, as $L_{\mathbf{V}_{j},\Omega_{E}}$ and $L_{\mathbf{\tilde
{V}}_{j},\Omega_{E}}$ are unitary equivalent, $E$ belongs to the resolvent set
of $L_{\mathbf{\tilde{V}}_{j},\Omega_{E}}.$ Since $B_{1}=B_{2}$ for
$x\in\mathbb{R}^{3}\diagdown\Omega_{E},$ then by construction $\tilde{A}%
_{1}=\tilde{A}_{2}$ for $x\in\mathbb{R}^{3}\diagdown\Omega_{E},$ and hence
$\mathbf{\tilde{V}}_{1}\left(  x\right)  =\mathbf{\tilde{V}}_{2}\left(
x\right)  ,$ for $x\in\mathbb{R}^{3}\diagdown\Omega_{E}.$ Applying Theorem
\ref{completeness8} we conclude that $V_{\pm}^{\left(  1\right)  }=V_{\pm
}^{\left(  2\right)  }$ and $B^{\left(  1\right)  }=B^{\left(  2\right)  }$
for all $x\in\mathbb{R}^{3}.$
\end{proof}

If the asymptotic expansions (\ref{direct2}) and (\ref{direct3}) actually
converge, in pointwise sense, for $\left\vert x\right\vert $ large enough,
respectively to $V\left(  x\right)  $ and $B\left(  x\right)  ,$ then
collecting the result of Corollary \ref{22} and the reconstruction result of
Theorem \ref{18} we are able to formulate the following uniqueness result for
the inverse scattering problem:

\begin{theorem}
Let the expansion (\ref{direct2}) for the electric potentials $V_{j}\in
C^{\infty}\left(  \mathbb{R}^{3}\right)  $ and the expansion (\ref{direct3})
for the magnetic fields $B_{j}\in C^{\infty}\left(  \mathbb{R}^{3}\right)  ,$
with $\operatorname{div}B_{j}=0,$ $j=1,2,$ hold. Let the magnetic potentials
$A_{j},$ $j=1,2,$ be defined by the equalities (\ref{basicnotions19}%
)-(\ref{basicnotions21}). Suppose that for some $E\in\left(  -\infty
,-m\right)  \cup\left(  m,+\infty\right)  ,$ $S_{1}\left(  E\right)
=S_{2}\left(  E\right)  ,$ and there is a connected open bounded set
$\Omega_{E}$ with smooth boundary $\partial\Omega_{E}$ $,$ such that $E$
belongs to the resolvent set of $L_{\mathbf{V}_{j},\Omega_{E}}$ for both
$j=1,2.$ Moreover, suppose that the asymptotic expansions (\ref{direct2}) and
(\ref{direct3}), for $V$ and $B$, respectively, actually converge in pointwise
sense for $x\in\mathbb{R}^{3}\diagdown\Omega_{E}.$ Then, we have that
$V_{1}\left(  x\right)  =$ $V_{2}\left(  x\right)  $ and $B_{1}\left(
x\right)  =B_{2}\left(  x\right)  $ for all $x\in\mathbb{R}^{3}.$
\end{theorem}

\begin{proof}
As $S_{1}\left(  E\right)  =S_{2}\left(  E\right)  ,$ it follows from the
reconstruction result of Theorem \ref{18} that the asymptotic terms
$V_{k}^{\left(  1\right)  }=V_{k}^{\left(  2\right)  }$ and $B_{k}^{\left(
1\right)  }=B_{k}^{\left(  2\right)  }$ coincide for all $k.$ Moreover, as the
asymptotic expansions (\ref{direct2}) and (\ref{direct3}) converge,
$V_{1}\left(  x\right)  =$ $V_{2}\left(  x\right)  $ and $B_{1}\left(
x\right)  =B_{2}\left(  x\right)  $ for $x\in\mathbb{R}^{3}\diagdown\Omega
_{E}$. Using Corollary \ref{22} we conclude that $V_{1}\left(  x\right)  =$
$V_{2}\left(  x\right)  $ and $B_{1}\left(  x\right)  =B_{2}\left(  x\right)
$ for all $x\in\mathbb{R}^{3}.$
\end{proof}

\end{document}